\documentclass[12pt]{article}
\pdfoutput=1

\newlength{\abstractwidth}
\usepackage{epsf}
\usepackage{color}
\usepackage{graphicx}
\usepackage{hyperref}
\usepackage{dsfont}

\usepackage{amsmath, nccmath}
\usepackage{amsmath, amsthm}
\usepackage{amssymb}
\usepackage{amsfonts}
\usepackage{latexsym}
\usepackage{mathtools}
\usepackage{shuffle}
\usepackage[all,cmtip]{xy}

\usepackage{tikz, pgf}
\usepackage{tkz-fct}
\usepackage{pgfplots}
\usetikzlibrary{shapes.misc}
\usetikzlibrary{shapes,snakes}
\usetikzlibrary{decorations.pathmorphing}	
\usetikzlibrary{decorations.markings}

\usetikzlibrary{arrows,shapes,positioning}
\usetikzlibrary{decorations.markings}
\usepackage[rightcaption]{sidecap}
\tikzstyle arrowstyle=[scale=1]
\tikzstyle directed=[postaction={decorate,decoration={markings,
    mark=at position .65 with {\arrow[arrowstyle]{stealth}}}}]
\tikzstyle reverse directed=[postaction={decorate,decoration={markings,
    mark=at position .65 with {\arrowreversed[arrowstyle]{stealth};}}}]
\usetikzlibrary{positioning}

\renewcommand{\thefootnote}{\fnsymbol{footnote}}
\renewcommand{\thanks}[1]{\footnote{#1}}
\newcommand{\starttext}{
\setcounter{footnote}{0}
\renewcommand{\thefootnote}{\arabic{footnote}}}

\numberwithin{equation}{section}
\newcommand{\bea}{\begin{eqnarray}}
\newcommand{\eea}{\end{eqnarray}}
\newcommand{\be}{\begin{eqnarray}}
\newcommand{\ee}{\end{eqnarray}}
\newcommand{\<}{\langle}
\renewcommand{\>}{\rangle}
\newcommand{\bma}{\begin{matrix}}
\newcommand{\ema}{\end{matrix}}

\def\cA{{\cal A}}

\def\cC{{\cal C}}
\def\cD{{\cal D}}
\def\cE{{\cal E}}

\def\cG{{\cal G}}

\def\cI{{\cal I}}
\def\cJ{{\cal J}}

\def\cL{{\cal L}}
\def\cM{{\cal M}}
\def\cN{{\cal N}}
\def\cO{{\cal O}}
\def\cP{{\cal P}}
\def\cQ{{\cal Q}}
\def\cR{{\cal R}}

\def\cT{{\cal T}}
\def\cU{{\cal U}}
\def\cV{{\cal V}}
\def\cW{{\cal W}}

\def\bG{{\bf G}}

\def\bL{{\bf L}}

\def\mA{\mathfrak{A}}
\def\mB{\mathfrak{B}}

\def\mD{{\cal D}}

\def\mJ{\mathfrak{J}}

\def\mN{\cN}

\def\mT{\mathfrak{T}}

\def\mt{{\mathfrak{t}}}
\def\muu{\mathfrak{u}}

\def\CC{{\mathbb C}}

\def\RR{{\mathbb R}}
\def\ZZ{{\mathbb Z}}

\def\vA{{\vec{A}}}
\def\vB{{\vec{B}}}
\def\vC{{\vec{C}}}

\def\vI{{\vec{I}}}
\def\vJ{{\vec{J}}}
\def\vK{{\vec{K}}}

\def\vP{{\vec{P}}}
\def\vQ{{\vec{Q}}}
\def\vR{{\vec{R}}}
\def\vS{{\vec{S}}}

\def\vU{{\vec{U}}}
\def\vV{{\vec{V}}}

\def\vX{{\vec{X}}}
\def\vY{{\vec{Y}}}
\def\vZ{{\vec{Z}}}

\def\Im{{\rm Im \,}}

\def\det{{\rm det \,}}

\def\half{{1\over 2}}
\def\thalf{\tfrac{1}{2}}
\def\p{\partial}

\def\a{\alpha}
\def\b{\beta}

\def\tet{\vartheta}
\def\ep{\varepsilon}
\def\om{\omega}

\def\pbx{\p _{\bar x}}
\def\pby{\p _{\bar y}}

\def\pbw{\p _{\bar w}}
\def\pbz{\p _{\bar z}}

\def\xx{{\boldsymbol x}}

\def\zz{{\boldsymbol z}}

\def\bPhi{\boldsymbol{\Phi}}

\def\btau{{\boldsymbol{\tau}}}

\def\LL{{L}}

\def\CS{{\rm CS}}
\def\Sp{{\rm Sp}}
\def\SUB{{\boldsymbol{\Theta}}}

\def\der{{\mathfrak{der}}}
\def\otherep{\epsilon}

\def\no{\nonumber}
\def\sm{\smallskip}

\definecolor{Cyan}{cmyk}{1.,0,0,0}
\definecolor{Magenta}{cmyk}{0,1.,0,0}
\definecolor{Yellow}{cmyk}{0,0,1.,0}
\definecolor{White}{cmyk}{0,0,0,0}
\definecolor{Orange}{cmyk}{0,0.61,0.87,0}
\definecolor{RedOrange}{cmyk}{0,0.77,0.87,0}
\definecolor{Red}{cmyk}{0,1.,1.,0}
\definecolor{Purple}{cmyk}{0.45,0.86,0,0}
\definecolor{Violet}{cmyk}{0.79,0.88,0,0}
\definecolor{Blue}{cmyk}{1,0.5,0,0}
\definecolor{ProcessBlue}{cmyk}{0.96,0,0,0}
\definecolor{GreenYellow}{cmyk}{0.6,0,1.,0}
\definecolor{Black}{cmyk}{0,0,0,1}
\definecolor{dgreen}{rgb}{0,0.70,0.30}
\newtheorem{theorem}{Theorem}

\newtheorem{thm}{Theorem}[section]
\newtheorem{lem}[thm]{Lemma}
\newtheorem{prop}[thm]{Proposition}
\newtheorem{cor}[thm]{Corollary}

\newtheorem{deff}[thm]{Definition}
\newtheorem{rmk}[thm]{Remark}
\newtheorem{setup}[thm]{Schemata}

\begin{document}
\starttext
\setcounter{footnote}{0}

\begin{flushright}
2026 August 29  \\
UUITP-19/26
\end{flushright}

\vskip 0.30in

\begin{center}

{\LARGE \bf Flat connections on moduli spaces I}

\vskip 0.1in

{\large \bf  Local (1,0)-extension of the DHS connection}

\vskip 0.35in

{\large Eric D'Hoker${}^{a}$, Benjamin Enriquez${}^{b}$, Oliver Schlotterer${}^{c}$, Federico Zerbini${}^{d}$} 

\vskip 0.15in

{ \sl ${}^{a}$Mani L. Bhaumik Institute for Theoretical Physics}\\
{\sl  Department of Physics and Astronomy}\\
{\sl University of California, Los Angeles, CA 90095, USA}

\vskip 0.1in 

{\sl ${}^b$IRMA and Universit\'e de Strasbourg}\\
{\sl  7, Rue Descartes, 67084 Strasbourg, France}

\vskip 0.1in

{\sl ${}^c$Department of Physics and Astronomy,} \\
  { \sl Department of Mathematics,} \\
  { \sl Centre for Geometry and Physics,} \\ 
  {\sl Uppsala University, 75120 Uppsala, Sweden}
 
 \vskip 0.1in

{\sl ${}^{d}$Departamento de Matem\'aticas Fundamentales, UNED}\\
{\sl Calle de Juan del Rosal, 28040 Madrid, Spain}

\vskip 0.15in 

{\tt \small dhoker@physics.ucla.edu, enriquez@math.unistra.fr, oliver.schlotterer@physics.uu.se, f.zerbini@mat.uned.es}

\vskip 0.2in

\begin{abstract}
\vskip 0.1in
The flat DHS connection $\mathcal J_{\mathrm{DHS}}$ constructed in arXiv:2602.01461 is smooth on the configuration space of $n$ points on a fixed compact Riemann surface $\Sigma$ of arbitrary genus~$h$, takes values in an infinite-dimensional Lie algebra $\hat \mt_{h,n}$ and is invariant under the modular group $\mathrm{Sp}(2h,\ZZ)$. This paper is the first in a series for a program whose goal is to extend the connection $\mathcal J_{\mathrm{DHS}}$ to a global flat connection on the Teichm\"uller space $\mathcal T_{h,n}$ valued in the Lie algebra of derivations of $\hat{\mathfrak t}_{h,n}$. 
Upon the choice of local coordinates adapted to the map $\mathcal T_{h,n}\to \mathcal T_h$, such a connection splits into three pieces: $\mathcal J_{\mathrm{DHS}}$, a piece $\cL$ corresponding to holomorphic directions of $\cT_h$  and a third piece corresponding to anti-holomorphic directions in $\mathcal T_h$. In this paper, we isolate the system of equations satisfied by  $\cL$ and obtain its solution locally and explicitly. The construction of a global extension of $\mathcal J_{\mathrm{DHS}}$  to  $\mathcal T_{h,n}$ and the extension  of the meromorphic connection of arXiv:1112.0864 to  $\mathcal T_{h,n}$, are relegated to future publications in this series.
\end{abstract}

\end{center}

\newpage

\setlength{\textheight}{8in}

\setcounter{tocdepth}{2} 
\tableofcontents

\baselineskip=15pt
\setcounter{equation}{0}
\setcounter{footnote}{0}

\newpage

\section{Introduction}
\setcounter{equation}{0}
\label{sec:1}

The construction of flat connections on various manifolds, with values in suitable infinite-dimensional Lie algebras, is of considerable interest both from the viewpoints of physics and mathematics. In physics, state-of-the-art computations in perturbative quantum field theory and string theory benefit from systematic integration methods in terms of polylogarithms, namely iterated integrals arising from path-ordered exponentials of such flat connections. In mathematics, they encode information on the mixed motives which originate from fundamental groups, and give rise to proofs of 1-formality of the latter. 

\sm

Recently, flat connections on the configuration spaces $\mathrm{Cf}_n(\Sigma)$ of $n$ points on a
Riemann surface $\Sigma$ of genus $h\geq1$, valued in the Lie algebra $\hat \mt_{h,n}$, were studied from various viewpoints.\footnote{The Lie algebra $\hat \mt_{h,n}$ is isomorphic to the Lie algebra of the pro-unipotent completion of the braid group $\pi_1({\rm Cf}_n(\Sigma))$ \cite{Bezruk}. A definition in terms of generators and relations will be given in section \ref{sec:2.1}.} Modular invariant connections that are real analytic and single-valued on $\mathrm{Cf}_n(\Sigma)$ were produced for $n=1$ in \cite{DHoker:2023vax} and for general~$n$ in~\cite{DHoker:2026lgg}.  Multivalued meromorphic connections with non-trivial automorphy conditions were introduced in \cite{Enriquez:2011}.\footnote{See \cite{Baune:2024biq, DHoker:2025dhv, Enriquez:2021} for explicit representations of the differential forms in the meromorphic connection of~\cite{Enriquez:2011}.} Single-valued meromorphic variants were introduced in \cite{Enriquez:2021, Enriquez:2022}.
 Although these flat connections differ in their defining properties,  they are related to one another by composing a gauge transformation and a Lie algebra automorphism \cite{Enriquez:2021, DHoker:2025szl}, and this can be exploited to relate the corresponding associated spaces of polylogarithms \cite{DHoker:2025szl}.

\sm

The above-mentioned connections were all studied for a fixed complex structure on~$\Sigma$. The goal of this work and its sequels is to extend these connections by components that account  for variations in complex-structure moduli while preserving flatness. This task decomposes as follows: to first assemble the individual flat connections into a vertical connection relative to the fibration, 
\bea
\varpi:\mathcal T_{h,n}\to\mathcal T_h
\label{varpifib}
\eea
Here, $\cT_{h,n}$ is the Teichm\"uller space of compact Riemann surfaces of genus $h$ with $n$ marked points, $\mathcal T_h$ denotes the space $\mathcal T_{h,0}$, and the fibers of~$\varpi$ are isomorphic to ${\rm Cf}_n(\Sigma)$, as represented schematically in the lower box of figure~\ref{fig:1a} (see for example \cite{Hubbard}).  The next task is that of 
producing a \textit{global extension to Teichm\"uller space} of one of the above connections, namely,  the search for a flat connection on $\mathcal T_{h,n}$ with values in the Lie 
algebra\footnote{In this paper, we shall denote by $\mathfrak{der}(\mathfrak g)$  the Lie algebra of {\it continuous} derivations of a graded complete Lie algebra $\mathfrak g$; 
it is therefore the graded completion of the Lie algebra of derivations of the direct sum of its graded components. The latter algebra is supported in degree $\geq 0$ if $\mathfrak g$ is generated in degree 1. }   
of derivations $\mathfrak{der}(\hat{\mathfrak t}_{h,n})$, whose restriction to the fibers of $\varpi$ coincides with the image of the initial connection under the adjoint Lie algebra morphism $\hat{\mathfrak{t}}_{h,n}\to\mathfrak{der}(\hat{\mathfrak{t}}_{h,n})$.

\begin{figure}[htb]
\begin{center}
\tikzpicture[scale=0.98]
\scope[xshift=-4.3cm,yshift=6cm]
\filldraw [fill=blue!6, draw=black] (4.3,3) rectangle (0,0);
\draw [ultra thick] (0,0) -- (4.3,0);
\draw [draw=black] (4.3,3) rectangle (0,0);
\draw (1,-0.4) node{$\cT_h$};
\node [rotate=90] at (-0.4,2.3) {${\rm Cf}_n (\Sigma)$};
\draw [ultra thick, red, dashed] (0,1) .. controls (1,2) ..  (4.3,1);
\draw (6.3,1.8) node{\large $ \phi_n$};
\draw (6.3,1) node{\large $s \cdot \phi_n = s'$};
\draw [very thick, <->] (4.8, 1.4) -- (7.6, 1.4);
\draw (2.2,2.5) node{$d_{{\rm Cf}_n (\Sigma)} + \partial _{\cT_h} - \cJ^s$};
\endscope
\scope[xshift=4.3cm,yshift=6cm]
\filldraw [fill=blue!6, draw=black] (4.3,3) rectangle (0,0);
\draw [ultra thick] (0,0) -- (4.3,0);
\draw [draw=black] (4.3,3) rectangle (0,0);
\draw (3.5,-0.4) node{$\cT_h$};
\node [rotate=90] at (-0.4,2.3) {${\rm Cf}_n (\Sigma)$};
\draw [ultra thick, red, dashed] (0,1) .. controls (1,0.2) and (2,2.5)  ..  (4.3,1);
\draw (2.2,2.5) node{$d_{{\rm Cf}_n (\Sigma)} + \partial _{\cT_h} - \cJ^{s'}$};
\endscope
\scope[xshift=0cm,yshift=0cm]
\filldraw [fill=blue!6, draw=black] (4.5,4) rectangle (0,0);
\draw [ultra thick] (0,0) -- (4.5,0);
\draw [ultra thick] (1,0.7) -- (1,3.2);
\draw [ultra thick] (2.4,0.7) -- (2.4,3.2);
\draw [ultra thick] (1.8,2) node{\bf $\cdots$};
\node [rotate=90] at (-1.6,1.4) {projection $\varpi$};
\draw [ultra thick, ->] (-1,2) -- (-1,0.5);
\draw [draw=black] (4.5,4) rectangle (0,0);
\draw (4,3.5) node{$\cT_{h,n}$};
\draw (4,-0.4) node{$\cT_h$};
\draw (0.9,3.5) node{\small $\mathrm{Cf}_n(\Sigma)$};
\draw (2.4,3.5) node{\small $\mathrm{Cf}_n(\Sigma)$};
\node [rotate=45] at (4.9,5.2) {\large ${\rm triv}_{s'}$};
\draw [very thick, <-] (4.6,4.2) -- (5.9,5.5);
\node [rotate=-45] at (-0.5,5.1) {\large ${\rm triv}_s$};
\draw [very thick, <-] (-0.2,4.2) -- (-1.5,5.5);
\draw (3.4,0.5) node{$d_\SUB - \cJ_\SUB$};
\endscope
\endtikzpicture
\caption{The lower box gives a schematic representation of the Teichm\"uller spaces $\cT_{h,n}$ and $\cT_h$ and the projection $\varpi : \cT_{h,n} \to \cT_h$ whose fibers are isomorphic to ${\rm Cf}_n(\Sigma)$. The upper boxes give a schematic representation of the trivialization maps ${\rm triv}_s$ and ${\rm triv}_{s'}$ induced by two different sections $s$ and $s'$ of $\cT_h \to \CS(\Sigma)$, respectively. The construction of the connections $d_{{\rm Cf}_n (\Sigma)} + \partial _{\cT_h} - \cJ^s$ and  $d_{{\rm Cf}_n (\Sigma)} + \partial _{\cT_h} - \cJ^{s'}$ as well as their interrelation by the fibered diffeomorphism $\phi_n$ are the subjects of the present paper, while the construction of the connection $d_\SUB - \cJ_\SUB$ for the lower box is relegated to subsequent work.
\label{fig:1a}}
\end{center}
\end{figure}


\subsection{Motivation and context}

In genus $h=1$, the \textit{global extension problem} of the connections has already been studied from different perspectives. In the multivalued meromorphic context, it was solved for $n=1$ in \cite{LR} and for arbitrary $n$  in \cite{CEE}. In the smooth single-valued context,  it was solved  for $n=1$ in \cite{Schlotterer:2025qjv}. 
 
 \sm
 
In genus $h\geq 2$, the motivation for solving the extension problem is threefold:

\sm

(a) Its solution should lead to the introduction of new modular functions and modular tensors on $\mathcal T_h$, beyond the ones identified in \cite{Kawazumi:lecture, Kawazumi} at low rank and more generally in  \cite{DHoker:2015wxz, DHoker:2016mwo, DHoker:2017pvk, DHoker:2020uid}. The solution is also expected to explain further interrelations between modular tensors  and the moduli variation of polylogarithms.

\sm

(b) In topology, it is expected to give rise to an explicit flat connection on the  moduli space $\cM _{h,n}$ such as the one constructed in section 15 of \cite{Hain:1997} for $h\geq2$,  possibly giving an alternative proof of the results of \cite{Hain:1997}  on relative completions of mapping class~groups.

\sm

(c) In physics, the genus-one global connection of \cite{CEE}, and its component along the direction of variations of the modulus of $\Sigma$ in particular, has become a central tool used to organize the spaces of elliptic Feynman integrals \cite{Weinzierl:2022eaz, Bourjaily:2022bwx} and genus-one string amplitudes \cite{Berkovits:2022ivl, DHokerkaidi} in terms of iterated integrals over modular forms. A corresponding impact is expected to result from the solution of the higher-genus globalization problem.

\sm

It is worthwhile to also highlight links of the global extension problem to the earlier conformal field theory literature: in
\cite{Hitchinref, Bernref1, Bernref2}, a bundle
of conformal blocks on $\mathcal T_{h,n}$, depending on Lie algebraic data,\footnote{More precisely, an affine Kac-Moody Lie algebra and a collection of $n$ integrable 
representations at the same level, all this being related with the WZW models.} 
was constructed, and
equipped with a projectively flat connection (sometimes called the 
Knizhnik-Zamolodchikov-Bernard (KZB)/Hitchin connection), thus generalizing the Knizhnik-Zamolodchikov 
connection obtained at $h=0$ in \cite{kzref}. In \cite{drinref}, a universal (with respect to the 
Lie algebraic data) flat connection was constructed for $h=0$. At $h=1$, the analogous step was performed in \cite{CEE}. At general $h$, one could expect that solutions of the global extension problems for the connections in \cite{Enriquez:2011, DHoker:2026lgg} should be likewise related with the KZB/Hitchin connection.

\subsection{The global extension problem for the DHS connection}
\label{sec:1.1}

In this paper, we shall address the extension problem in the smooth single-valued and modular invariant context provided by the multivariable DHS connection  which was introduced in \cite{DHoker:2026lgg}. 

\sm

To set the stage, denote by $\Sigma$ a fixed smooth, compact, orientable  surface without boundary. Recall that the DHS construction   assigns to a complex structure $J$ on $\Sigma$ a smooth connection on $\mathrm{Cf}_n(\Sigma)$,  valued in the Lie algebra $\hat{\mathfrak t}_{h,n}$. Within this context, the list of tasks outlined above decomposes as:

\sm 

(a) the construction 
of a $\hat{\mathfrak t}_{h,n}$-valued 1-form $\cJ_\text{DHS}$ on $\mathcal T_{h,n}$ as above;

\sm 

(b) the construction of a $\mathfrak{der}(\hat{\mathfrak t}_{h,n})$-valued 
1-form $\mathcal J_{\rm global}$  over $\mathcal T_{h,n}$, such that

\sm 

(c) $d-\mathcal J_{\rm global}$ is a flat connection on $\mathcal T_{h,n}$, and

\sm 

(d) the restriction of $\mathcal J_{\rm global}$ to the vertical subbundle 
$\mathrm{ker}(T\varpi)$ of the fibration \eqref{varpifib} coincides with the image of 
$\cJ_\text{DHS}$ by the adjoint Lie algebra morphism 
$\hat{\mathfrak t}_{h,n}\to\mathfrak{der}(\hat{\mathfrak t}_{h,n})$.

\sm 

We relegate the solution of these tasks to subsequent papers in the series of articles for this program mentioned in the abstract.

\subsection{The local $(1,0)$ extension problem for the DHS connection}
\label{sec:1.1a}

As a first step towards a solution of (b) above, we will construct in a subsequent paper in this series a bundle $\Theta^*$ over $\mathcal T_{h,n}$ of  ``1-forms modulo the $(0,1)$-direction in the base  of \eqref{varpifib}", which will be a quotient of the bundle of 1-forms and therefore  equipped with a differential $d_\Theta$, 
and show that any solution $\mathcal J_{\rm global}$ of (b) leads to:

\sm 

(b$'$) a $\mathfrak{der}(\hat{\mathfrak t}_{h,n})$-valued section 
$\mathcal J_\Theta$ of $\Theta^*$ over $\mathcal T_{h,n}$ such that

\sm 

(c$'$) $d_\Theta-\mathcal J_{\Theta}$ is flat in a sense to be explained  therein, 
and

\sm 

(d$'$) the restriction of $\mathcal J_{\Theta}$  to $\mathrm{ker}(T\varpi)$ 
coincides with the image of $\mathcal J_{\mathrm{DHS}}$ outlined in (d).  

\sm 

We will construct the connection $d _\SUB - \mathcal J_\Theta$ with the properties (b$'$), (c$'$), (d$'$) in subsequent papers in the series and refer to it as the solution to the {\it global $(1,0)$ extension problem} as the extension is only in the $(1,0)$ direction on $\cT_{h}$. 

\sm

In this paper, we shall choose to replace the setting of the fibration  $\mathcal T_{h,n}\to \mathcal T_h$ by a trivialized setting, working locally on the product $\mathrm{Cf}_n(\Sigma)\times\mathcal T_h$ instead of 
$\mathcal T_{h,n}$. In this setting, we may assemble the DHS connections at  given complex structures into a family of connections on this product structure as follows. We shall denote by $\CS(\Sigma)$ the space of complex structures on $\Sigma$ and by ${\rm Diff}_0(\Sigma)$  the group of diffeomorphisms connected to the identity. An arbitrary local holomorphic section,
\bea
\label{local:sec}
 s: \cT_h \to \CS(\Sigma)
 \eea
of the projection $\mathrm{CS}(\Sigma) \to \mathrm{CS}(\Sigma)/\mathrm{Diff}_0(\Sigma)=\mathcal T_h$ gives rise to a vertical 1-form $\cJ^s_\text{DHS}$ on $\mathrm{Cf}_n(\Sigma)\times\mathcal T_h$; namely, the restriction of 
$\cJ^s_\text{DHS}$ to $\mathrm{Cf}_n( \Sigma) \times\{\btau\}$ is the DHS one-form attached to the complex structure  $s(\btau)$.
Another such section $s'$ is related to  $s$ by a $\cT_h$-dependent diffeomorphism, namely, by a  local map 
$\phi : \mathcal T_h\to\mathrm{Diff}_0(\Sigma)$ such that  $s'(\btau)=s(\btau)\phi(\btau)$ for $\btau \in \cT_h$  (henceforth simply denoted $s'=s\cdot \phi$). The domain $U_{ss'}$ of  the map $\phi$ is the intersection of the domains of $s$ and $s'$. The vertical 1-forms $\cJ^s_\text{DHS}$ for different sections $s$ and $s'$ are then related by, 
\bea
\mathcal J_{\mathrm{DHS}}^{s'}=\phi_{n}^* \, \cJ_{\mathrm{DHS}}^s
\eea
where $\phi_n$ is the fibered diffeomorphism of ${\rm Cf}_n(\Sigma)\times U_{ss'}$ induced by $\phi$, which takes
$(x_1, \cdots, x_n ,\btau) \in \mathrm{Cf}_n(\Sigma)\times U_{ss'}$ to 
\bea
\label{def:phin}
\phi_n: (x_1,\cdots,x_n,\btau)\mapsto (\phi x_1,\cdots,\phi x_n,\btau)\in \mathrm{Cf}_n(\Sigma)\times U_{ss'}
\eea
A schematic representation of the above set-up is shown in figure \ref{fig:1a}.

\sm

The cotangent bundle of the product space ${\rm Cf}_n (\Sigma) \times \cT_h$ is given by, 
\bea
T^*\big ( {\rm Cf}_n (\Sigma) \times \cT_h \big ) = T^* {\rm Cf}_n (\Sigma)  \, \oplus \, T^* \cT_h
\label{tstareq}
\eea
Since $\cT_h$ is a complex (K\"ahler) manifold, its complexified cotangent bundle splits into holomorphic and anti-holomorphic components, 
\bea
T^* _\CC \cT_h = T^{* (1,0)} \cT_h \, \oplus \, T^{*(0,1)} \cT_h
\eea
giving rise to the bundle,
\bea
T^*_\CC {\rm Cf}_n (\Sigma)  \oplus T^{* (1,0)} \cT_h
\eea
over $\mathrm{Cf}_n(\Sigma)\times\mathcal T_h$. The total differential on $\cT_h$ splits into $(1,0)$ and $(0,1)$ components $ \p _{\cT_h}$ and $ \bar{\p} _{\cT_h}$ according to the decomposition (\ref{tstareq}). It turns out that for any  $\phi : \mathcal T_h\to \mathrm{Diff}_0(\Sigma)$, the diffeomorphism  $\phi_n$ of $\mathrm{Cf}_n(\Sigma)\times\mathcal T_h$ lifts to this bundle.\footnote{This situation will be explained as follows in a subsequent paper in this series. Each section  $s$ gives rise to a local isomorphism  $\mathrm{triv}^s:\mathrm{Cf}_n(\Sigma)\times\mathcal T_h\to \mathcal T_{h,n}$, one has  $\mathrm{triv}^{s'}= \mathrm{triv}^{s} \phi_n$ if $s'=s\cdot \phi$, and 
$\mathcal J^s_{\mathrm{DHS}}=(\mathrm{triv}^s)^*\mathcal J_{\mathrm{DHS}}$; finally,
$T^* {\rm Cf}_n (\Sigma) \oplus T^{* (1,0)} \cT_h=\mathrm{triv}_s^*\Theta^*$.  }

\subsubsection{Formulation of the local $(1,0)$ extension problem}
\label{sec:1.111}

Taking advantage of the product structure $\mathrm{Cf}_n(\Sigma)\times\mathcal T_h$, the problem then becomes
the {\it local $(1,0)$ extension problem}, which may be formulated as follows;
\begin{description}
\itemsep=0in
\item (b$''$) for any section $s$, to construct\footnote{In the abstract, $\mathcal L^s $ is simply  denoted by 
$\mathcal L$ for the sake of simplicity.} a  $\mathfrak{der}(\hat{\mathfrak t}_{h,n})$-valued local section  
$\mathcal L^s $ of $T^{* (1,0)} \cT_h$  over $\mathrm{Cf}_n(\Sigma)\times \cT_h$, in such a way that 
for any sections $s$ and $s'$ related by a diffeomorphism $s' = s \cdot \phi$, one has 
$\mathcal L^{s'}=\phi_n^* \, \cL^{s}$; 
\item (c$''$) the connection component $\cJ^s$ defined by,
\bea
\cJ^s = \cJ^s_\text{DHS} + \cL^{s}
\eea
corresponds to a flat connection $d_{{\rm Cf}_n(\Sigma)} + \p _{\cT_h} - \cJ^s$. 
\end{description}

\sm

We shall pick local coordinates $(\xx |\btau)$ on $\mathrm{Cf}_n(\Sigma)\times\mathcal T_h$, where $\xx= (x_1, \cdots, x_n)$ are local coordinates on ${\rm Cf}_n(\Sigma)$ and $\btau=(\tau_1, \cdots , \tau_\sigma)$ with $\sigma = \dim \cT_h$ are local coordinates on $\cT_h$, with $\dim \cT_h=3h-3$ for $h\geq 2$ and $\dim \cT_1 = 1$. The flatness condition (c$''$) may then be expressed more explicitly on $\cJ^s = \cJ^s(\xx|\btau)$ by,\footnote{Here and throughout, $\cJ^s \wedge \cJ^s$ will be  understood as the wedge product with coefficients in the universal enveloping algebra of $\mathfrak{der}(\hat{\mathfrak t}_{h,n})$.}
\bea
\label{intr.02}
(d_\xx + \partial_\btau)  \, \cJ^s - \cJ^s \wedge \cJ^s =0
\eea
where $d_\xx = \sum_{i=1}^n (dx^i \, \p_{x_i} + d\bar x^i \, \p_{\bar x_i})$ and $\p_\btau = \sum _{\a=1}^\sigma d\tau_\a \p_{\tau_\a}$ represent the total differential $d_{{\rm Cf}_n(\Sigma)}$ and the Dolbeault differential $\p _{\cT_h}$, respectively. Taking into account the flatness of the DHS connection itself, the flatness condition (c$''$) becomes equivalent to a set of partial differential equations for $\cL^{s}$, given by, 
\bea
\label{intro.02ab}
(d_\xx + \p_\btau) \cL^{s} + \p_\btau \cJ^s_\text{DHS} - \cJ^s_\text{DHS} \wedge \cL^{s} - \cL^{s} \wedge \cJ^s_\text{DHS} - \cL^{s} \wedge \cL^{s}=0
\eea
Therefore, the first task will be to solve the equation of \eqref{intro.02ab}, which will be done by solving for the components $L_\a$ in the decomposition of $\cL^{s}$, 
\bea
\label{intr.01}
\cL^{s} (\xx|\btau) = \sum_{\alpha=1}^\sigma d \tau_\alpha \, L^s_{\alpha}(\xx | \btau) 
\eea
The final  tasks will be to prove that the solution is indeed valued in $\der(\hat \mt_{h,n})$ and transforms under a change of slice $s$ as specified by (b$''$).

\subsection{The main result}

The main result of this paper is stated in the following theorem.\footnote{One can show that a solution $\mathcal J_\Theta$ of (b$'$) gives rise to a solution $s\mapsto \mathcal J^s$  by $\mathcal J^s =\mathrm{triv}_s^*\mathcal J_\Theta$. In subsequent work, we intend to construct such a solution of (b$'$) and to relate it in this way with the solution presently found for (b$''$).}
\begin{theorem}
\label{1.thm:main}
For any $\btau_0 \in\mathcal T_h$, there exists a neighborhood $U_{\btau_0}$ of $\btau_0$ in which the
\textit{local $(1,0)$ extension problem}, formulated in items (b$''$) and (c$''$) of subsection \ref{sec:1.111}, admits a solution, namely, for an arbitrary section  $s : U_{\btau_0} \to \mathrm{CS}(\Sigma)$, there exists $\mathcal L^s$ as in 
(b$''$) and satisfying (c$''$). In local coordinates, this means that the differential equation (\ref{intro.02ab}) has a solution taking values in $\der(\hat \mt_{h,n})$ that satisfies condition (b$''$) upon changing sections. 
\end{theorem}

The proof of Theorem \ref{1.thm:main} will be presented in section \ref{sec:7} by combining the results of Theorems \ref{7.thm:20} and \ref{7.thm:30}. In fact, the components $L_\alpha = L^s_\alpha$ of the solution (see (\ref{intr.01})) are obtained in explicit form given in Corollary \ref{7.cor:30}, given there on the generators of $\hat{\mathfrak{t}}_{h,n}$ and extended to the entire algebra by the results of section \ref{sec:66}.

\subsection{Further results, methods and organization}
\label{summsec}

Henceforth, we shall drop the superscript $^s$ referring to a choice of section unless indicated otherwise.
To solve (\ref{intro.02ab}) for $\cL$ requires evaluating the moduli variation $\p_\btau \cJ_\text{DHS}$. In section~\ref{sec:3a}, we review the theory of complex structure variations in terms of Beltrami and quadratic differentials. Section \ref{sec:3a} then provides explicit formulas for the variations of Abelian differentials $\om_I$, the Arakelov Green function $\cG(x,y)$, the DHS kernels\footnote{The DHS kernels are the coefficients of the expansion of $\cJ_\text{DHS}$ as a Lie series in the Lie algebra~$\hat \mt_{h,n}$. They are single-valued real analytic functions and differential 1-forms out of which the polylogarithms associated with the connection $\cJ_\text{DHS}$ can be constructed with the help of iterated integrals \cite{DHoker:2023vax}, whence the name \textit{DHS integration kernels}, or simply \emph{DHS kernels}. They will be defined in section \ref{sec:2.2}, and their moduli variations will be given in section \ref{sec:3.6}.} and the connection $\cJ_\text{DHS}$. The results may all be expressed in terms of DHS kernels. 

\sm

In section \ref{sec:4}, we obtain necessary conditions to be satisfied by a potential solution to (\ref{intr.02}). These necessary conditions may be used to produce explicit expressions for the action of $\cL$ on the generators of $\hat{\mt}_{h,n}$, assuming that $\cL$ exists. The explicit form of $\cL$ is given in (\ref{7.0c}), the constituents $L_\a$ of $\mathcal L$ decompose into $L_\a = L_\a ^0 + L_\a ^1$ with $L_\a^1$ given in (\ref{7.cor.33}) and the action of $L_\a ^0$ on the generators of $\mt_{h,n}$ is given by (\ref{7.cor.30}).
This reveals the functional dependence of $\cL$ in terms of DHS kernels as well as differentials $\partial_\btau \cA$ and $\partial_\btau \mD$ of new \emph{modular tensors} $\cA$ and $\mD$. 
This functional dependence of $\cL$ in terms of $\partial_\btau \cA$ and $\partial_\btau \mD$ is the counterpart to the functional dependence of $\cJ_\text{DHS}$ in terms of the DHS kernels. The modular tensors $\cA$ and $\mD$ are built from the convolution integrals of DHS kernels, and their components, which are single-valued real-analytic functions on ${\cal T}_h$, are given explicitly as follows,\footnote{The rank 4 modular tensor $\cA^M{}_{KL}{}^N$ was introduced by Kawazumi in \cite{Kawazumi:lecture, Kawazumi}. Its double trace $\cA^L{}_{KL}{}^K$ is proportional to the Kawazumi-Zhang invariant \cite{Kawazumi:2008, Zhang} which, in the special case of genus two, plays an important role in the low energy expansion of Type IIB string theory \cite{D'Hoker:2013eea,DHoker:2014oxd,Pioline:2015qha}.}
\bea
\label{intr.ad}
\cA^M {}_{K I_1\cdots I_r L}{\,}^N & = & 
\int _\Sigma d^2 x \, \int _\Sigma d^2 y \, \om^M(x) \bar \om_K(x)  \, \cG_{I_1\cdots I_r}(x,y) \, \bar \om_L(y) \om^N(y) 
\no \\
\mD_{I_1\cdots I_r} & = & \int _\Sigma d^2 x \, \kappa (x) \, \cG_{I_1\cdots I_r} (x,x)
\eea
where $r\geq 0$, all indices belong to $\{1,\cdots ,h\}$, $\omega^M$ and $ \bar \omega_K$ are Abelian differentials, $\kappa$ is the K\"ahler form on $\Sigma$, and $\cG_{I_1 \cdots I_r}(x,y)$  may be defined recursively starting from the Arakelov Green function $\cG(x,y)$ by setting $\cG_\emptyset (x,y) = \cG(x,y)$ for $r=0$ and for $r \geq 1$ by,
\bea
\cG_{I_1 \cdots I_r} (x,y) = \int _\Sigma d^2 z \, \cG(x,z) \bar \om_{I_1} (z) \p_z \cG_{I_2 \cdots I_r}(z,y)
\label{defcgte}
\eea
while separate definitions for $\cA^M{}_\emptyset {}^N$,  $\cA^M {}_{P} {}^N$ (the latter transforming non-tensorially under ${\rm Sp}(2h,\mathbb Z)$) are given in Theorem \ref{7.thm:1}. The modular tensors $\cA$ and $\mD$  obey a wealth of linear relations, some of which are given in Proposition \ref{7.prop:1re} of section \ref{sec:2}. 

\sm

In sections \ref{sec:6},  \ref{sec:66} and \ref{sec:7}, we shall show that the necessary local conditions on $\cL$ obtained in section~\ref{sec:4} are actually sufficient for the existence and flatness of a local connection~$\cJ$. A major task will be to demonstrate in section \ref{sec:66} that  $\cL$  takes values in $\mathfrak{der}(\hat \mt_{h,n})$. For this purpose, $\mt_{h,n}$ will be uplifted to a larger algebra $\muu_{h,n}$ in intermediate steps. The corresponding extended components of the local connection $\cJ$ are shown in section \ref{sec:6} to take values in the Lie algebra $\mathfrak{der}(\hat {\muu}_{h,n})$ before obtaining the desired $\cJ \in \mathfrak{der}(\hat \mt_{h,n})$,  through the Lie algebra projection $\muu_{h,n} \to \mt_{h,n}$, and proving in section \ref{sec:7} that $\cJ$, or more precisely the assignment $s\mapsto \cJ^s$ where $s$ is a local section as in section \ref{local:sec}, satisfies the flatness conditions (\ref{intro.02ab}), as summarized in Theorem \ref{7.thm:20}. Together with the extendability established in Theorem \ref{7.thm:30}, we arrive at the main result of this work previewed in Theorem~\ref{1.thm:main}.  

\sm

In Proposition \ref{5.prop.lww}, we identify  modular properties of $\cL^s$. This 
partially\footnote{\label{modularfoot}The other ingredients of this discussion, namely the construction 
of group-valued cocycles on the space $\mathrm{Cf}_n(\Sigma)\times \mathrm{CS}(\Sigma)$ and discussion 
of their behavior under the diffeomorphisms $(\phi_{ss'})_n$, will be treated in forthcoming work.} 
sets the stage for discussing modularity properties of the assignment $s\mapsto \cJ^s$.

\sm

It was established in \cite{DHoker:2026ggx} that the non-trivial components of the flatness conditions on the multi-variable DHS connection $\cJ_\text{DHS}$ are equivalent to the union of all interchange and Fay identities, namely non-trivial bilinear relations between the DHS kernels out of which $\cJ_\text{DHS}$ is constructed, that were derived via other methods in  \cite{DHoker:2024ozn}. The flat connection $\cJ$ extends $\cJ_\text{DHS}$ by including the directions in which  moduli are varied, with corresponding \textit{integration kernels} $\partial_\btau \cA$ and $\partial_\btau \mD$. The flatness conditions on $\cJ$ impose linear and bilinear relations between $\partial_\btau \cA$,  $\partial_\btau \mD$ and the familiar DHS kernels. A preliminary account of these relations is given in section \ref{sec:77}, but a systematic analysis is relegated to future work.   

\sm

In section \ref{sec:8}, we give a detailed account of the global flat connection $\cJ$ in the special case of genus one. We discuss how the Lie algebra of the Calaque-Enriquez-Etingof connection \cite{CEE} emerges and show its relation to Tsunogai's derivation algebra \cite{Tsuongai:1995}, explaining why the contributions to $\cJ$ in Remark \ref{innerrmk} are viewed as generalizing Tsunogai's derivations to arbitrary~genus. 

\sm

Ten appendices elaborate on technical aspects of this work, with combinatorial prerequisites in appendix \ref{sec:AA}, coincident limits of Fay identities in appendix \ref{sec:A} and several proofs in the remaining appendices.

\subsection{Ongoing work}

In this series of papers dedicated to the topic of \textit{Flat connections on moduli spaces}, of which this paper is the first, we plan to extend the scope of the present work by the investigations of the following problems,
\begin{itemize}
\itemsep= -0.02in
\item the global $(1,0)$ extension problem for the DHS connection;
\item the full global extension problem, including the $(0,1)$, namely, the $d \bar \tau_\alpha$ 
analogue of~$\mathcal L$ arising from the trivialization of $\mathcal J_{\mathrm{global}}$; 
\item the extension problem in the  meromorphic context of the\footnote{See also \cite{BT}, p.\ 669, footnote 1.} Enriquez connection~\cite{Enriquez:2011}.
\end{itemize}

\subsection{Future directions}

An important long-term goal is the construction of homotopy-invariant iterated integrals based on the flat connections $\cJ$ and $\cJ_\text{global}$. The flatness of $\cJ_\text{global}$ on the space $\cT_{h,n}$ guarantees that the path-ordered exponential of $\cJ_\text{global}$ along a path $\gamma \subset \cT_{h,n}$, 
\bea
\Gamma [\gamma] = {\rm Pexp} \int _\gamma \cJ_\text{global}
\label{intro:gam}
\eea
is homotopy invariant and may be utilized to investigate the moduli dependence of polylogarithms on ${\rm Cf}_n(\Sigma)$. 

\sm

More specifically, by equating the path-ordered exponential of (\ref{intro:gam}) associated with different paths $\gamma_1$ and $\gamma_2$ with the same endpoints (see figure \ref{fig:2}), polylogarithms may be expressed via iterated integrals over modular tensors in complex structure moduli. These representations provide powerful tools to unravel the relations among special values of polylogarithms at arbitrary genus in the same way as the systematics of elliptic polylogarithms at rational points on a torus is governed by iterated integrals over modular forms \cite{Adams:2017ejb, Broedel:2018iwv}. 
In the special case of paths $\gamma_1$ or $\gamma_2$ in the fundamental group of $\mathrm{Cf}_n(\Sigma)$, 
flatness of $\cJ_{\rm global}$ implies iterated-integral representations for higher-genus versions of
multiple zeta values \cite{Baune:2025sfy} and associators \cite{Gonzalez:2020, DHoker:2023vax, Tani:2025, Tani:2026},
generalizing the genus one results on elliptic multiple zeta values \cite{Enriquez:Emzv, Broedel:2015hia, Lochak:2017} and elliptic associators \cite{CEE, Enriquez:2014}.

\begin{figure}[htb]
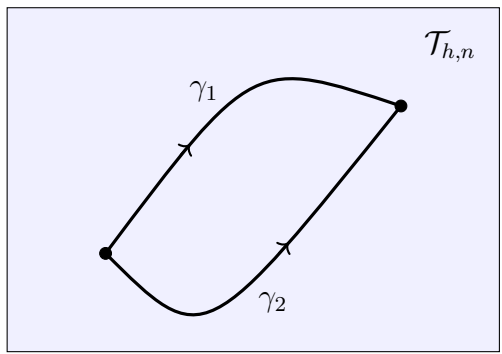

\begin{center}
\tikzpicture[scale=1.3]
\scope[xshift=0cm,yshift=0cm]
\filldraw [fill=blue!6, draw=black] (5.5,3.5) rectangle (0.5,0);
\draw[ black, fill] (1.5,1) circle (0.06);
\draw[ black, fill] (4.5,2.5) circle (0.06);
\draw[very thick] (1.5,1) .. controls (2.5, 0) .. (4.5,2.5);
\draw[very thick] (1.5,1) .. controls (3, 3) .. (4.5,2.5);
\draw [ thick,->] (2.26,2) -- (2.36,2.11);
\draw [ thick,->] (3.26,1) -- (3.36,1.11);
\draw (5,3.1) node{$\cT_{h,n}$};
\draw (2.5,2.65) node{$\gamma_1$};
\draw (3.2,0.5) node{$\gamma_2$};
\endscope
\endtikzpicture
\caption{Flatness of $\cJ_{\rm global}$ guarantees equality of the path-ordered exponentials $\Gamma [\gamma_1]=\Gamma[\gamma_2]$ for any pair of homotopic curves $\gamma _1$ and $\gamma _2$ in $\cT_{h,n}$. \label{fig:2}}
\end{center}
\end{figure}

Moreover, combining the path-ordered exponentials of ${\cal J}_{\rm global}$ in (\ref{intro:gam}) with their complex conjugates will offer a starting point to investigate iterated-integral representations of the non-holomorphic modular tensors in \cite{Kawazumi:lecture, Kawazumi, DHoker:2017pvk, DHoker:2020uid} and the systematics of their relations.
In this way, ${\cal J}_{\rm global}$ is expected to be a stepping stone to explore echoes beyond genus one of the
link \cite{Gerken:2020yii, Dorigoni:2022npe, Dorigoni:2024oft} between modular graph forms \cite{DHoker:2015wxz, DHoker:2016mwo} and equivariant iterated Eisenstein integrals \cite{Brown:2017qwo, Brown:2017mgf}.

\subsection*{Acknowledgements}

We are grateful to Martin Raum and Yoann Sohnle for valuable discussions. We would like to thank the Erwin Schr\"odinger International Institute for Mathematics and Physics (ESI), University of Vienna (Austria), for the opportunity to participate in the Thematic Programme ``Amplitudes and Algebraic Geometry'' in 2026 where a significant part of this work has been accomplished and for the support given. The research of ED was supported in part by NSF grant PHY-26-09924. The research of OS is funded by the European Union under ERC Synergy Grant MaScAmp 101167287. Views and opinions expressed are however those of the author(s) only and do not necessarily reflect those of the European Union or the European Research Council. Neither the European Union nor the granting authority can be held responsible for them. The research of FZ was partly supported by the Spanish Ministry of Science, Innovation and Universities under the 2023 grant ``Proyecto de generaci\'on de conocimiento'' PID2023-152822NB-I00.

\subsection*{Tool and computational resource disclosure}
Suggestions by AI tools such as ChatGPT and Claude were used in the research process. 
However, all the statements of this work, as well as their proofs, are due to the authors.

\newpage

\section{The DHS connection on a fixed Riemann surface}
\setcounter{equation}{0}
\label{sec:2}

In this section, we review the DHS connection in $n \geq 2$ variables on a fixed compact Riemann surface $\Sigma$ of genus $h \geq 1$ without punctures, including the DHS integration kernels and the structure relations of the infinite-dimensional Lie algebra $\hat \mt_{h,n}$ from which the connection is constructed. Since the DHS  connection in $n$ variables on a surface with $p$ punctures may be obtained from the connection in $n+p$ variables on a surface without punctures by freezing $p$ of the variables, there is no need to consider  the case with punctures separately \cite{DHoker:2026lgg}.

\subsection{The DHS kernels on a Riemann surface $\Sigma$}
\label{sec:2.2}

Throughout this section, we shall consider a fixed compact Riemann surface $\Sigma$ of arbitrary genus $h \geq 1$ without punctures.  The homology group $H_1(\Sigma, \ZZ)$ is equipped with a non-degenerate  intersection pairing $\mJ$ with respect to which a canonical basis of cycles $\mA^I$ and $\mB_I$ has intersection pairings  $\mJ(\mA^I, \mA^J) = \mJ(\mB_I, \mB_J) =0$ and $\mJ(\mA^I, \mB_J) = \delta ^I_J$ for $I,J \in \{ 1,\cdots, h\}$. Choosing loops $\mA^I$ and $\mB_I$ on $\Sigma$ with a common base point $p \in \Sigma$ to represent these cycles may be used to promote them to generators of the fundamental group $\pi_1 (\Sigma, p)$.  A basis of the Dolbeault cohomology group $H^{(1,0)} (\Sigma, \CC)$ is given by $h$ holomorphic $(1,0)$-forms  $\om_I$ normalized by their $\mA$-periods and whose $\mB$-periods give the components $\Omega_{IJ}$ of the period matrix  $\Omega$,
\bea
\label{2.a.1}
\oint _{\mA^I} \om_J = \delta ^I_J \hskip 1in \oint _{\mB_I} \om_J = \Omega _{IJ}
\eea 
Similarly, a basis of $H^{(0,1)}(\Sigma, \CC)$ is given by the anti-holomorphic $(0,1)$ forms $\bar \om_I$.  The Riemann relations ensure symmetry of $ \Omega$ and positivity of $ \Im(\Omega)$. The pull-back of the translation invariant K\"ahler form on the Jacobian $\CC^h/(\ZZ^h + \Omega \ZZ^h)$ to $\Sigma$ provides a canonical volume form $\kappa$ on $\Sigma$, written in local complex coordinates\footnote{\label{ftglobal} Throughout, we shall equivalently represent a globally defined differential form in terms of local coordinates with the help of local coefficient functions  that have suitable transformation properties under a change of coordinates. For example, Abelian differentials $\om_I$ are expressed locally as $\omega_I = \omega_I(x) dx$ in terms of coefficient function $\om_I(x)$. The total differential $d$ on $\Sigma$ decomposes into $d = \p + \bar \p$ where the Dolbeault differentials are expressed as $\p = dx  \, \p_x $ and $\bar \p = d \bar x \, \p_{\bar x}$ in local complex coordinates $x, \bar x$ where $\p_x, \pbx$ stand for partial derivatives (see section \ref{sec:3.aa} below for more details). The coordinate volume form will be denoted $d^2x = { i \over 2} dx \wedge d \bar x$ and the Dirac delta function is normalized by $\int _\Sigma d^2 x \, \delta(x,y) = 1$.}
as $\kappa (x) \, d^2x$, which may be normalized as follows,\footnote{Throughout, a pair of a repeated upper and lower index is to be summed over by the Einstein convention without writing the sum symbol. Indices $I,J \in \{ 1,\cdots, h \}$ may be lowered and raised using the metric $Y = \Im (\Omega)$ and its inverse $Y^{-1}$ with components $Y_{IJ}$ and $Y^{IJ}$, respectively. For example, using these notations we have $Y_{IJ} Y^{JK} = \delta ^K_I$ and $\bar \om^I =  Y^{IJ} \, \bar \om_J$. }
\bea
\label{2.a.2}
\kappa(x) = { 1 \over h} \om_I(x) \bar \om^I(x) \hskip 1in
\int _\Sigma d^2 x \, \kappa(x) =1
\eea
The Arakelov Green function $\cG(x,y)$ is a symmetric single-valued  function of $x,y \in \Sigma$, which is smooth outside the diagonal, has a logarithmic singularity at $x=y$, and is uniquely determined by the following equations  (for an explicit construction see \cite{DHoker:2017pvk,DHoker:2023vax}), 
\bea
\label{2.a.3}
\pbx \p_x \cG(x,y) = - \pi \delta (x,y) + \pi \kappa(x) 
\hskip 1in 
\int _\Sigma d^2 x \, \kappa (x) \cG(x,y)=0
\eea
In terms of the Arakelov Green function $\cG_\emptyset (x,y) = \cG(x,y)$ we define the following smooth functions for $r \geq 1$ and $I_1,\cdots ,I_r,J \in\{1,\cdots ,h\}$ recursively by convolutions over~$\Sigma$, 
\bea
\label{2.a.4}
\cG_{I_1 \cdots I_r} (x,y) & = & \int _\Sigma d^2 z \, \cG(x,z) \, \bar \om_{I_1} (z) \, \p_z \cG_{I_2 \cdots I_r} (z,y)
\no \\
\Phi _{I_1 \cdots I_r}{}^J(x) & = & \int _\Sigma d^2 z \, \cG_{I_1 \cdots I_{r-1}}(x,z) \, \bar \om_{I_r} (z) \, \om^J(z)
\eea
The function $\Phi$ is traceless in its last two indices, while $\cG$ satisfies a reflection symmetry, 
\bea
\label{2.a.5}
\Phi _{I_1\cdots I_{r-1} J}{}^J(x)=0
\hskip 1in
\cG_{I_1 \cdots I_r} (x,y) = (-)^r \cG_{I_r \cdots I_1}(y,x)
\eea 
The DHS kernels for $r \geq 1$ are given in terms of $\cG$ and $\Phi$ by,\footnote{The DHS kernels in this paper are related to the ones introduced in \cite{DHoker:2023vax} and further developed in \cite{DHoker:2026lgg,DHoker:2026ggx}  by $f_{I_1 \cdots I_r} {}^J(x,y) = Y_{I_1 K_1} \cdots Y_{I_r K_r} \, f^{K_1 \cdots K_r}{}_L \, Y^{JL}$ (see Remark \ref{remark:2.1}). \label{foot:7}}
\bea
\label{2.a.6}
f_{I_1 \cdots I_r} {}^J(x,y)  =  \p_x \Phi _{I_1 \cdots I_r} {}^J (x) - \p_x \cG_{I_1 \cdots I_{r-1}} (x,y) \, \delta _{I_r}^J
\eea
and inherit their single-valuedness in $x,y \in \Sigma$ from the Arakelov Green functions in the convolutions (\ref{2.a.4}).
Conversely, the functions $\p_x \Phi$ and $\p_x \cG$ may be obtained from the traceless and trace parts of~$f$, respectively.   It will often be convenient to amend these relations by setting $f_\emptyset {}^J(x,y) = \p_x \Phi _\emptyset {}^J (x) = \om^J(x)$ for the case corresponding to $r=0$.

\subsubsection{Differential relations amongst DHS kernels}

The integral relations of (\ref{2.a.4}), along with the definition of $\cG(x,y)$ in (\ref{2.a.3}), readily imply the following Massey system of differential equations for $r \geq 1$, 
\bea
\label{2.a.7}
\pbx \, \p_x \, \Phi _{I_1 \cdots I_r} {}^J(x) & = & - \pi \, \bar \om_{I_1} (x) \, \p_x \Phi _{I_2 \cdots I_r} {}^J(x) 
+ \pi \, \delta_{r,1} \, \delta _{I_1} ^J \, \kappa(x)
\no \\
\pbx \, \p_x \, \cG_{I_1 \cdots I_r} (x,y) & = & - \pi \, \bar \om_{I_1} (x) \, \p_x \cG_{I_2 \cdots I_r} (x,y) 
\no \\
\pby \, \p_x \, \cG_{I_1 \cdots I_r} (x,y) & = & - \pi \, \bar \om_K(y) \, f_{I_1 \cdots I_r} {}^K(x,y)
\eea
The last two equations are extended to $r=0$ via (\ref{2.a.3}) and,
\bea
\label{2.a.7a}
\partial_x \partial_{\bar y} \cG(x,y) = \pi \delta(x,y) - \pi \bar \omega_K(y) \omega^K(x)
\eea
For later use, it will be convenient to have the differential equations directly in terms of the DHS kernels $f$, which are again given by a Massey type system, 
\bea
\pbx \, f_{I_1 \cdots I_r} {}^J(x,y) & = & - \pi \, \bar \om_{I_1} (x) \, f_{I_2 \cdots I_r} {}^J(x,y) + \pi \, \delta_{r,1} \, \delta _{I_1}^J \, \delta (x,y)
\no \\
\pby \, f_{I_1 \cdots I_r} {}^J(x,y) & = & \pi \, f_{I_1 \cdots I_{r-1}} {}^ K (x,y) \, \bar \om_K(y) \delta ^J_{I_r} 
- \pi \, \delta _{r,1} \, \delta ^J_{I_1} \, \delta(x,y)
\label{fmass}
\eea
Further properties of the DHS kernels may be found in \cite{DHoker:2023vax,DHoker:2024ozn}.

{\rmk 
\label{remark:2.1}
The multiplets $\Phi$, $\cG$ and $f$ defined above all transform under non-linear tensor representations of the modular group $\Sp(2h,\ZZ)$ \cite{DHoker:2020uid} to be specified in section \ref{sec:mod} below.
The choice to invert the positions of the indices on these tensors, indicated in footnote \ref{foot:7},  is made so that each tensor transforms under $\Sp(2h,\ZZ)$ with automorphy factors (\ref{modsec.03}) (which are themselves matrices) that are purely anti-holomorphic in the moduli of $\Sigma$. The advantage of this choice is that holomorphic derivatives in the moduli are covariant without the need for a connection, as will be detailed in section \ref{sec:3a}.  
}

\subsection{The modular tensors $\cA$, $\mD$ and $\mN$}
\label{sec:adn}

With the help of further convolution integrals,  the DHS kernels $\cG_\vI \, (x,y)$ produce  infinite families of modular graph functions \cite{DHoker:2017pvk} and modular graph tensors \cite{Kawazumi:lecture, Kawazumi, DHoker:2020uid} that  generalize to higher genus the genus-one modular graph functions \cite{DHoker:2015wxz} and forms \cite{DHoker:2016mwo}, respectively, and similarly depend only on the moduli of $\Sigma$.  Of immediate interest in the present work will be the two families defined for arbitrary $\vI= I_1 \cdots I_r$ with $r \geq 0$ and $\vJ = J_1\cdots J_s$ with $s\geq 2$ by,
\bea
\label{7.f.1}
\cA^M {}_{K \vI L}{\,}^N & = & 
\int _\Sigma d^2 x \, \int _\Sigma d^2 y \, \om^M(x) \, \bar \om_K(x)  \, \cG_\vI (x,y) \, \bar \om_L(y) \, \om^N(y) 
\no \\
\mD_\vJ & = & \int _\Sigma d^2 x \, \kappa (x) \, \cG_\vJ (x,x)\, , \ \ \ \ \ \  \mD_K = 0
\eea
with $K,L,M,N \in \{ 1, \cdots, h\} $, as well as the following family defined for $r \geq 3$ by,\footnote{The notation $\mN$ used here corresponds to the notation $\hat {\mathfrak{N}}$ used in \cite{DHoker:2024ozn} for the same object.}
\bea
\label{defcn}
 \mN_{I_1 I_2} & = & \int _\Sigma d^2 z_1 \, \bar \om_{I_1} (z_1) \int _\Sigma d^2 z_2 \, \bar \om_{I_2}(z_2) 
\Big ( \p_{z_1} \cG(z_1, z_2) \p_{z_2} \cG(z_2,z_1) - \p_{z_1} \p_{z_2} \cG(z_1, z_2) \Big )
\no \\
\mN_{I_1 \cdots I_r} & = & \int _\Sigma d^2 z_1 \, \bar \om_{I_1} (z_1) \int _\Sigma d^2 z_2 \, \bar \om_{I_2}(z_2) \,
\p_{z_1} \cG(z_1, z_2) \, \p_{z_2} \cG_{I_3 \cdots I_r}(z_2,z_1)
\eea
In terms of two-dimensional conformal field theory, the modular tensors $\cA$ correspond to tree-level Feynman graphs while $\mD$ and $\mN$ correspond to one-loop graphs.  The combinations $\mN$  occur naturally in the coincident limits $y \rightarrow x$ of $\p_x \cG_{I_1 \cdots I_s} (x,y)$  with $s\geq 1$   \cite{DHoker:2024ozn}, as will be reviewed in appendix \ref{app:B.cin} below. The simplest instance $\cA^M {}_{K L}{\,}^N $ of (\ref{7.f.1}) produces the Kawazumi-Zhang invariant \cite{Kawazumi:2008, Zhang, D'Hoker:2013eea, DHoker:2017pvk} upon contraction with~$\delta^K_N \delta^L_M$.

\sm

{\prop
\label{7.prop:1re}
The integral representations for $\cA, \mD$ and $\mN$ given in (\ref{7.f.1}) and (\ref{defcn}) are absolutely convergent. While $\mN$ enjoys cyclic symmetry, 
\bea
\mN_{\vI J} =  \mN_{J \vI}
\eea
$\cA, \mD$ and $ \mN$ have the following antipode symmetries,\footnote{The antipode $\theta$ is an automorphism of the shuffle algebra of words and may be defined on a word of length $r$ by $\theta(I_1 \cdots I_r) = (-)^r I_r \cdots I_1$. The summation over $K \vI L = \vX \vY \vZ$ in (\ref{7.f.3}) instructs to sum over all combinations of words $\vX$, $\vY$, $\vZ$ whose concatenation $\vX \vY \vZ$ equals the word $K \vI L$ for given letters $K,L$ and word $\vI$, including the cases where some of $\vX$, $\vY$, $\vZ$ are empty. A selection of combinatorial definitions and identities is reviewed in appendix \ref{sec:AA}, and a more comprehensive collection can be found in appendix A of \cite{DHoker:2026ggx}. \label{other.4}} 
\bea
\label{7.f.2}
\cA^M {}_{K \vI L}{\,}^N = \cA^N{}_{L \theta(\vI) K}{}^M
\hskip 0.8in
\mD_\vI = \mD_{\theta(\vI)}
\hskip 0.8in
\mN_\vI =  \mN_{\theta(\vI)} 
\eea
and obey the following linear interrelation,
\bea
\label{7.f.3}
 \mD_{K \vI} \, \delta ^J_L - \delta _K^J \, \mD_{\vI L} 
& = & 
 \sum_{K \vI L = \vX \vY \vZ} \cA^J {}_{ \big ( \vX \shuffle \theta(\vZ) \big ) M \big ( \vY + \theta (\vY) \big )} {}^M
\eea 
}

\begin{proof}
\vskip -0.1in
Given that the only singularity of $\p_x \cG(x,y)$ is a simple pole at $x=y$, absolute convergence is manifest. 
The reflection symmetry properties follow from the property $\cG_{\theta(\vI)} (y,x) = \cG_\vI(x,y)$ given in (\ref{2.a.5}). The proof of (\ref{7.f.3}) is more involved and is relegated to appendix \ref{sec:B}.
\end{proof}

\sm

In section \ref{sec:77}  we will discuss both corollaries of the identities (\ref{7.f.3}) and evidence 
for the existence of additional as-of-yet unknown relations among $\cA$ and $\mD$.

\subsection{The Lie algebra $\mt_{h,n}$}
\label{sec:2.1}

The DHS connection in $n$ variables on a compact Riemann surface $\Sigma$ of genus $h \geq 1$ without punctures  takes values in the infinite-dimensional Lie algebra $\hat \mt_{h,n}$.

{\deff
\label{2.def:1}
The generators of the Lie algebra $\mt_{h,n}$ are $a_{iI}, b_i^I$ and $t_{ij}= t_{ji}$, for  $I\in \{ 1,\cdots, h\}$ and $i,j \in  \{1, \cdots, n\}$ with $j \not= i$ and obey the defining structure relations \cite{Enriquez:2011},
\begin{align}
 \big [ a_{iI} , a_{jJ}  \big ]  & = 0 &   \big [ b_i^I , b_j^J \big ] & =  0 
&  
\big [b_i^I , a_{jJ} \big ]  & = \delta ^I_J \, t_{ij} 
\no \\
 \big [ a_{iI}, t_{jk} \big ] & = 0 &  \big [b_i^I, t_{jk} \big] & = 0 
&
\big [ b_{iI} , a^I _i \big ] & = - \sum_{j \not = i} t_{ij}   
\label{11.1}
\end{align}
where $i,j,k$ are mutually distinct.  The Lie algebra $\mt_{h,n}$ admits a positive bi-grading, denoted $| \cdot |$,  that assigns the bi-degrees $|a_{iI}|=(1,0)$, $|b_i^I|=(0,1)$ and $| t_{ij}|=(1,1)$ to the generators of $\mt_{h,n}$. Elements $X\in \mt_{h,n}$ of bi-degree $|X|=(k,\ell)$ are said to have $a$-degree $k$ and $b$-degree $\ell$.
The Lie algebra $\hat \mt_{h,n}$ is the degree completion of $\mt_{h,n}$.\footnote{The DHS connection takes values in the degree completion $\hat \mt_{h,n}$ rather than in $\mt_{h,n}$ since it is expressed as an infinite Lie series whose convergence is subject to the completion condition. }}

\sm

Further structure relations are obtained from (\ref{11.1}) using the Jacobi identity \cite{Enriquez:2011}, 
\begin{align}
\label{2.prop.1}
 \big [ a_{iI} + a_{jI}, t_{ij} \big ] & = 0 &  \big [ t_{ij} + t_{ik} , t_{jk} \big ] & = 0
\no \\
 \big [ b_{i}^I + b_{j}^I , t_{ij} \big ] & = 0  &  \big [ t_{ij}, t_{k\ell} \big ] & =0  
\end{align}
for $i,j,k,\ell$ mutually distinct. We note that,  for genus $h\geq 2$, the relations $[ a_{iI}, t_{jk}  ]  =  [b_i^I, t_{jk} ] = 0 $ of (\ref{11.1}) may be deduced from the other defining relations.

\subsection{The multivariable DHS connection}
\label{sc:2.3}

The multivariable DHS connection $\cJ_\text{DHS}$ defined in \cite{DHoker:2026lgg} takes values in the Lie algebra $\hat \mt_{h,n}$ given in Definition \ref{2.def:1} and is a single-valued one-form of $n$ variables  $\xx= (x_1, \cdots, x_n)  \in \Sigma^n$ which is smooth on the configuration space, 
\bea
\label{2.b.0}
\text{Cf}_n(\Sigma) =  \Sigma ^n \setminus \{ \hbox{diagonals} \} 
\eea 
with singularities along the diagonals $x_i=x_j$ given by holomorphic simple poles with residue $- t_{ij}$.
It will be convenient to locally decompose $\cJ_\text{DHS}$ into a sum of $(1,0)$-forms and $(0,1)$-forms,  
\bea
\label{2.b.1}
\cJ_\text{DHS} &= \sum_{i=1}^n \Big ( J_i^{(1,0)} (\xx) \, dx_i  + J_i^{(0,1)} (\xx) \, d \bar x_i  \Big )
 \eea
 The $(0,1)$ components are defined to be,
 \bea
J_i^{(0,1)} (\xx) = - \pi \, \bar \om_I(x_i) \, b_i^I
\label{defj01}
\eea
whereas the $(1,0)$ components are defined to be, 
\bea
\label{2.b.4}
J_i^{(1,0)} (\xx) = \sum _{r=0} ^\infty \p_i \Phi _{I_1 \cdots I_r} {}^J (x_i)  B_i^{I_1}  \cdots B_i ^ {I_r} a_{iJ} 
+ \sum _{j \not= i} \sum _{r=0}^\infty \p_i \cG_{I_1 \cdots I_r} (x_i, x_j)  B_i^{I_1}  \cdots B_i ^ {I_r}  t_{ij}
\qquad
\eea
where $\p_x \Phi _\emptyset {}^J(x) = \om^J(x)$ and $B_i^I X = [b_i^I,X]$ for arbitrary $X \in \hat \mt_{h,n}$.

\sm

The main result \cite{DHoker:2026lgg} is that $\cJ_\text{DHS}$ satisfies the Maurer-Cartan equation,
\bea
\label{2.a.1b}
d_\xx \, \cJ_\text{DHS} - \cJ_\text{DHS} \wedge \cJ_\text{DHS} =0
\eea 
where $d_\xx =\sum _i ( dx_i \, \p_{x_i} + d \bar x_i \, \p_{\bar x_i} ) $ is the total differential on $\text{Cf}_n(\Sigma)$ (at fixed moduli), so that $d_\xx- \cJ_\text{DHS}$ defines a flat connection over $\text{Cf}_n(\Sigma)$. In view of the proof of the main result of this article, it may be instructive for the reader to see how the construction of $\cJ_\text{DHS}$ may be regarded as a solution of the following problem.
 
{\prop 
 \label{2.thm:40}
There exists a single-valued smooth one-form $\cJ_{\rm DHS}$ on ${\rm Cf}_n(\Sigma)$ which takes values in $\hat \mt_{h,n}$, such that $d_\xx- \cJ_{\rm DHS}$ defines a flat connection over ${\rm Cf}_n(\Sigma)$, and whose components with respect to the decomposition (\ref{2.b.1}) into $(1,0)$ and $(0,1)$ forms satisfy the following conditions,

(a)  the components $J_i^{(0,1)} (\xx)$ are given by (\ref{defj01});

(b) the components $J_i^{(1,0)} (\xx)$ have $a$-degree 1;

(c) the components $J_i^{(1,0)} (\xx)$ have singularities along the diagonals $x_i=x_j$, for $j \not= i$, given by holomorphic simple poles with residue $- t_{ij}$.}

\begin{proof}
\vskip 0in
As a result of assumption $(b)$, the flatness condition (\ref{2.a.1b}) on ${\rm Cf}_n (\Sigma)$ decomposes into four sets of decoupled equations, each set of uniform $a$-degree. The first is $d_\xx\, J^{(0,1)}_j =~0 $ which is solved trivially by the expression (\ref{defj01}) prescribed in (a). The remaining components of the flatness conditions on ${\rm Cf}_n(\Sigma)$ are given as follows on $\Sigma ^n$,
\begin{subequations}
\label{2.b.3}
\begin{align}
{}  [ J^{(1,0)} _i , J^{(1,0)} _j  ] & =  0
\label{2.b.3aa} \\
\p_i J^{(1,0)} _j - \p_j J^{(1,0)} _i & =  0
\label{2.b.3bb} \\
\bar \p_j J^{(1,0)} _i - [J^{(0,1)}_j , J^{(1,0)} _i] & =  \pi \delta(x_i, x_j) \, t_{ij}  &  j & \not = i
\label{2.b.3cc} \\
\bar \p_i J^{(1,0)} _i - [J^{(0,1)}_i , J^{(1,0)} _i] & =  - \sum_{k \not= i} \delta(x_i, x_k) t_{ik}
\label{2.b.3dd}
\end{align}
\end{subequations}
Note that the $\delta$-functions on the right side explicitly account for the poles at coincident points as prescribed by the residue condition $(c)$ and vanish on ${\rm Cf}_n(\Sigma)$.  The $a$-degrees are two in the first line and one in the last three lines. A particular solution to the conditions of the theorem is given by $J^{(1,0)}_i$ in \eqref{2.b.4}, as proven in Theorem 4.1 of \cite{DHoker:2026lgg}. Indeed, the system  of equations (\ref{2.b.3cc}) and (\ref{2.b.3dd})  is solved by the Massey structure of the differential equations satisfied by the DHS kernels in (\ref{2.a.7}) and the residue condition of~$(c)$;  the system (\ref{2.b.3bb}) is solved by (\ref{2.b.4}) by the consequence
$\p_i\p_j \cG_{I_1 \cdots I_r}(x_i,x_j) =(-1)^r \p_j \p_i \cG_{I_r \cdots I_1}(x_j,x_i)$ of (\ref{2.a.5}) and, 
\bea
\label{2.b.7}
B^{I_1}_i \cdots B^{I_r}_i t_{ij} & = & (-)^r B^{I_r}_j \cdots B^{I_1}_j t_{ij}
\eea
which in turn follows from (\ref{2.prop.1}). Finally, the system of equations (\ref{2.b.3aa}) was proven in Theorem 4.1 of \cite{DHoker:2026lgg} and was shown in \cite{DHoker:2026ggx} to be equivalent to the set of all interchange and Fay identities on DHS kernels. \end{proof}

\sm

To express the components $J_i^{(1,0)} (\xx)$ in terms of DHS kernels succinctly, we use the following generating function whose lowest order is the Arakelov Green function $\cG(x,y)$,
\bea
\label{2.b.3a}
\bG (x,y;B) & =  & \sum _{r=0}^\infty \p_x \cG _{I_1 \cdots I_r}  (x,y) B^{I_1} \cdots B^{I_r}
\, = \sum _{\vI \in \cW_h} \p_x \cG _\vI (x,y) B^{\vI}
\eea
and the generating function whose lowest order is the holomorphic $(1,0)$ form $\om^J(x)$,
\bea
\label{2.b.3b}
\bPhi ^J (x;B) & = & \sum _{r=0}^\infty \p_x \Phi _{I_1 \cdots I_r} {}^J  (x) B^{I_1} \cdots B^{I_r}
\, = \sum _{\vI \in \cW_h} \p_x \Phi _\vI {}^J (x) B^{\vI} 
\eea
where $\cW_h$ denotes the set of all words in the alphabet of $h$ letters $I \in \{ 1,\cdots,h \} $, including the empty word $\emptyset$.\footnote{The summation over words $\vI \in \cW_h$ represents the double sum over the length $r \geq 0$ of the words $\vI= I_1 \cdots I_r$ and, for given $r$, the sum over all its individual letters $I_1, \cdots ,I_r  \in \{1,\cdots,h \} $, as illustrated in the right equalities of (\ref{2.b.3a}) and (\ref{2.b.3b}). The symbols $\Phi_\vI{}^J(x)$, $\cG_\vI(x,y)$ and $B^\vI$ stand for  $\Phi_{I_1 \cdots I_r}{}^J(x)$, $\cG_{I_1 \cdots I_r}(x,y)$ and for the concatenation product $B^\vI = B^{I_1} \cdots B^{I_r}$, respectively. \label{foot.4}} The differential equations for the DHS kernels of (\ref{2.a.7}) translate into the following differential equations for the generating functions, 
\bea
\label{2.gen.3}
\pbx  \bG (x,y;B) & = & \pi \kappa (x) - \pi \delta(x,y) - \pi \bar \om _I(x) B^I \, \bG(x,y;B)
\no \\
\pbx \bPhi ^J(x;B) & = & \pi \kappa(x) B^J - \pi \bar \om_I(x) B^I \, \bPhi^J (x;B)
\no \\
\pby  \bG(x,y;B) & = & \pi \delta (x,y) - \pi \bar \om_I(y) \, \big ( \bPhi ^I (x;B) - \bG(x,y;B) B^I \big ) 
\eea
In terms of these generating functions, the components (\ref{2.b.4}) of 
the $(1,0)$ part of the multivariable DHS connection take the form,   
\bea
\label{2.b.5}
J_i^{(1,0)} (\xx)  &= 
\bPhi^J(x_i; B_i) \, a_{iJ} + \sum_{j \neq i} \bG (x_i,x_j;B_i) \, t_{ij}  
\eea

\subsection{Modular invariance of $\cJ_\text{DHS}$}
\label{sec:mod}

Modular transformations leave the intersection pairing $\mJ$ invariant and map canonical bases of homology cycles to canonical homology bases. Assembling the cycles $\mA^I$ and $\mB_I$ into column matrices $\mA, \mB$ a modular transformation acts by a matrix $M \in \Sp (2h,\ZZ)$, 
\bea
\label{modsec.01}
\left ( \bma \mB  \cr \mA \ema \right ) 
\, \rightarrow \,
M \left ( \bma \mB  \cr \mA \ema \right )  
\hskip 0.7in 
M = \left ( \bma A & B \cr C & D \ema \right ) 
\hskip 0.7in 
M^t \mJ M = \mJ
\eea 
Under $\Sp(2h,\ZZ)$, the row matrix  of holomorphic Abelian differentials $\om$, the period matrix $\Omega$ of (\ref{2.a.1}), and its imaginary part $Y= \Im (\Omega)$ transform with the following $\Omega$-dependent automorphy factors that take values in $\mathrm{GL}(h,\CC)$, 
\bea
\label{modsec.05}
\om & \to &  \om (C \Omega {+}D)^{-1}
\no \\
\Omega & \to &   (A\Omega{+}B)(C\Omega{+}D)^{-1}
\no \\
Y & \to &  (\bar \Omega C^t {+} D^t)^{-1} Y (C \Omega {+}D)^{-1} 
\eea
Modular tensors transform under more general non-linear representations of $\Sp(2h,\ZZ)$. Using contractions with  the metric $Y$ or with its inverse, we may convert an arbitrary modular tensor to one whose modular transformation law involves only holomorphic automorphy factors (which is the choice made in \cite{DHoker:2026lgg}), or to a tensor whose transformation law involves only anti-holomorphic automorphy factors. Here, we shall make the latter choice (see Remark \ref{remark:2.1}) so that  the derivations with respect to holomorphic moduli that we shall use in the sequel are then covariant without the need for a connection.  

\sm

Therefore, more generally, we define a modular tensor $\mT$ of rank $(r,s)$ with components $\mT_{I_1 \cdots I_r} {} ^{J_1 \cdots J_s}$ to transform as follows under $M \in \Sp(2h,\ZZ)$,
\bea
\label{modsec.02}
\mT_{I_1 \cdots I_r}{}^{J_1\cdots J_s}(\Omega) & \to & 
\tilde Q_{I_1}{}^{\! K_1}  \cdots \tilde Q_{I_r}{}^{\! K_r}  \, 
\mT_{K_1 \cdots K_r}{}^{L_1\cdots L_s}(\Omega) \, 
\tilde R_{L_1}{}^{\! J_1}  \cdots \tilde R_{L_s}{}^{\! J_s} 
\qquad
\eea
where we use shorthands for the automorphy factor $\tilde Q$ and its inverse $\tilde R$,
\bea
\tilde Q & = & \tilde Q(M, \Omega) =  \big ( \bar  \Omega C^t {+} D^t \big )^{-1}  
\no \\
\tilde R & = & \tilde R(M, \Omega)  = \bar  \Omega C^t {+} D^t = \tilde Q^{-1}
\label{modsec.03}
\eea
As defined in (\ref{modsec.02}), modular tensors are sections of an anti-holomorphic vector bundle over Torelli space  (which is the moduli space of compact Riemann surfaces equipped with a canonical basis of $\mA$ and $\mB$ cycles). These conventions are related to the definition of modular tensors in \cite{DHoker:2020uid} by complex conjugation and ensure that the holomorphic moduli derivatives in  section \ref{sec:3a} preserve tensorial properties.

{\lem
\label{lem:MT}
The DHS kernels $\Phi$ and $f$ defined in (\ref{2.a.4}) and (\ref{2.a.6}) transform as modular tensors of rank $(r,1)$, while $\cG$ transforms as a tensor of rank $(r,0)$,
\bea
f_{I_1 \cdots I_r}{}^{J}(x,y) & \rightarrow & 
 \tilde Q_{I_1}{}^{K_1} \cdots \tilde Q_{I_r}{}^{K_r} f_{K_1 \cdots K_r}{}^{L}(x,y) \tilde R_{L}{}^{J} 
\no \\
\cG_{I_1 \cdots I_r} (x,y)  & \rightarrow & \tilde Q_{I_1}{}^{K_1} \cdots \tilde Q_{I_r}{}^{K_r} \cG_{K_1 \cdots K_r}(x,y)
 \label{modsec.04}
\eea
Similarly, the quantities $\cA^M {}_{P_1\cdots P_r}{\,}^N$, $\mD_{I_1\cdots I_s}$ and $\mN_{I_1\cdots I_t}$ defined by (\ref{7.f.1}) and (\ref{defcn}) for $r,s,t\geq 2$ transform as modular tensors of rank $(r,2)$, $(s,0)$ and $(t,0)$, respectively.}

\sm

\begin{proof}
\vskip 0in
The modular tensor properties (\ref{modsec.04}) follow from those of $\om^M, \, \bar \om_K$ together with modular
invariance of $\cG(x,y)$ and readily imply the tensorial $\Sp(2h,\ZZ)$ transformation of 
$\cA^M {}_{P_1\cdots P_r}{\,}^N$, $\mD_{I_1\cdots I_s}$ and $\mN_{I_1\cdots I_t}$
through their integral representations (\ref{7.f.1}) and~(\ref{defcn}).
\end{proof}

The tensorial transformation (\ref{modsec.04}) of DHS kernels leads to the following proposition proven in section 4.5 of \cite{DHoker:2026lgg}. 

{\prop
\label{2.prop:1}
The  modular transformation properties of the DHS kernels, given in (\ref{modsec.04}), lead to a modular invariant multi-variable DHS connection $\cJ_{\rm DHS}$, provided the generators $a,b,t$ of the Lie algebra $\mt_{h,n}$ transform as follows under $M \in \Sp(2h,\ZZ)$,
\bea
\label{,odsec.0.6}
a_{iI} \to \tilde Q(M, \Omega) _I{}^ K a_{iK} 
\hskip 0.8in 
b_i^J \to b_i^L \tilde R(M, \Omega) _L {}^J 
\hskip 0.8in 
t_{ij} \to t_{ij}
\eea
These transformations leave the structure relations of $\mt_{h,n}$ given in Definition \ref{2.def:1} invariant, so that $\Sp(2h,\ZZ)$ is an automorphism group of $\mt_{h,n}$.  }

\newpage

\section{Moduli variations of the DHS connection}
\setcounter{equation}{0}
\label{sec:3a}

In this section, we shall begin by briefly reviewing the theory of complex structure deformations on Riemann surfaces in sections \ref{sec:3.aa} to \ref{sec:difdtau}. We then proceed by providing explicit formulas for the variational derivatives of Abelian differentials,  the Arakelov Green function, and the DHS kernels in section \ref{sec:3.4}, and for the modular tensors $\cA$ and $\cD$ in section \ref{sec:3.3}. Combining these results with the vanishing of the variational derivatives of the generators of the Lie algebra $\mt_{h,n}$ discussed in section \ref{sec:3.8}, we obtain the variational derivative of the multivariable DHS connection in \ref{sec:3.9}.

\subsection{Complex structures and Teichm\"uller spaces}
\label{sec:3.aa}

Throughout, we shall consider a smooth orientable compact surface $\Sigma$ of genus $h \geq 1$. An almost complex structure $J$ on $\Sigma$ maps the tangent space $T_p (\Sigma)$ to $\Sigma$ at each point $p \in \Sigma$ to itself $J : T_p(\Sigma) \to T_p(\Sigma)$,  correspondingly maps the tangent bundle $T \Sigma$  to itself $J : T \Sigma \to T \Sigma$, and squares to minus the identity map so that $J$ is an automorphism of both $T_p (\Sigma)$ and  $T \Sigma$.  Its integrability being automatic, $J$ defines a complex structure on $\Sigma$, so that $\Sigma$ equipped with $J$ is a Riemann surface.  A complex structure $J$ splits the (complexified) tangent  bundle into holomorphic and anti-holomorphic components denoted $T^{(1,0)}_J \Sigma$ and $T^{(0,1)}_J \Sigma$ on which the eigenvalue of $J$ is $+i$ and $-i$, respectively, 
\bea
\label{3.b.1}
T_\CC \Sigma = T^{(1,0)}_J \Sigma \oplus T^{(0,1)}_J \Sigma
\eea
and similarly for the cotangent bundle; we shall often  use the notation,
\bea
\label{not:KJ}
K_J =T^{*(1,0)}_J \Sigma \hskip 1in  \overline K_J =T^{*(0,1)}_J \Sigma
\eea
for their dual bundles. This splitting depends upon the complex structure, as indicated by the subscript $J$.  The total differential $d$ splits accordingly, 
\bea
\label{3.b.1a}
d = \p _J + \bar \p_J
\eea
and this splitting again depends on $J$, as indicated by the subscript $J$.

\sm

A complex structure on $\Sigma$ is equivalent to a conformal structure on $\Sigma$ which, in turn, is specified by a Riemannian metric $g$ modulo Weyl transformations $ g \to e^\phi \, g$. In real local coordinates $\xi^m$ with $m \in \{1,2\}$, the metric may be expressed as $g = g_{mn} d\xi^m d \xi^n$ while the complex structure $J$ is given by \cite{Belavin:1986cy, DHoker:1988pdl},
\bea
\label{3.b.2}
J_m{}^n = \ep_{mp} \sqrt{\det g} \, g^{pn} 
\eea
where $g^{pn}$ are the components of the inverse of $g$ so that $g_{mn} g^{np} =\delta ^p_m$ while $\ep_{mn}=-\ep _{nm}$ with $ \ep_{12}=1$ specifies the orientation of $\Sigma$. In these coordinates, the components $\p_J$ and $\overline \p_J$ of the splitting of the total differential $d = \p_J + \overline \p_J$ introduced in (\ref{3.b.1a})  are given by,
\bea
\label{3.b.4old}
\p_J = \half d \xi^m \Big (  \p_m - i J_m{}^n \p_n \Big )
\hskip 1in 
\overline \p_J = \half d \xi^m \Big (  \p_m + i J_m{}^n \p_n \Big )
\eea
with $\p_m = \p / \p \xi^m$. The local coordinates $\xi^m$ are arbitrary and may be chosen to be independent of the  complex structure. Instead, however, we may choose \textit{local coordinates adapted to the complex structure $J$} in which both $J$ and the metric $g$ take simple forms. Such a choice of coordinates will now depend on $J$. We choose real adapted local coordinates $\xi^m$ in which the components of the complex structure $J$ are given by,
\bea
\label{3.b.3a}
J_1{}^1=J_2{}^2=0 \hskip 1in  J_1{}^2=-J_2{}^1=1
\eea 
and the metric takes the form $g = e^\phi |d\xi^1 + i d \xi^2 |^2$.  Correspondingly, in terms of the complex adapted local coordinates $z = \xi^1 + i \xi^2$ and $\bar z = \xi^1 - i \xi^2$   the components of the complex structure are given by,\footnote{The relation between $J$ of  (\ref{3.b.3a}) and (\ref{3.b.3b}) is $J_{\xi^1 + \eta i \xi^2}{}^{\xi^1 + \eta' i \xi^2} = \half ( J_1{}^1 + i \eta J_1{}^2 - i \eta' J_2{}^1 + \eta \eta' J_2{}^2 ) $.}

\bea
\label{3.b.3b}
J_z{}^z= - J_{\bar z} {}^{\bar z} =i 
\hskip 1in
J_z{}^{\bar z} = J_{\bar z} {}^z=0
\eea  
and the metric takes the form $g = e^\phi |dz|^2$. Equation (\ref{3.b.3b}) reflects the fact that diagonalizing  $J$ of (\ref{3.b.3a}) acting on $T\Sigma$ requires complexification to $T_\CC \Sigma$, thereby extending the validity of (\ref{3.b.4old}) to arbitrary real or complex local coordinates. The differential operators $\p_J$ and $\overline \p_J$ reduce to the familiar expressions $\p_J = dz \, \p_z $ and $\bar  \p_J = d \bar z \, \pbz$ so that locally holomorphic functions $f$ satisfy the familiar Cauchy-Riemann equation $\bar \p_J f=0$. 
 
\sm

We denote by $\CS(\Sigma)$ the infinite-dimensional complex manifold of complex structures on $\Sigma$. The infinite-dimensional Lie group $\mathrm{Diff}_0(\Sigma)$ of diffeomorphisms of $\Sigma$ that are continuously connected to the identity naturally acts on $\CS( \Sigma)$. The Teichm\"uller space  of genus $h$ 
is obtained from  $\CS(\Sigma)$ by the quotient,
\bea
\label{3.b.4}
\mathcal T_h=\CS(\Sigma)/\mathrm{Diff}_0(\Sigma)
\eea
The  Teichm\"uller space $\cT_h$ is a contractible finite-dimensional K\"ahler manifold, whose dimensions are $\dim_\CC \cT_1 =1$ and $\dim_\CC \cT_{h}=3h-3$ for $h\geq 2$ (see, for example,  \cite{Hubbard} and references therein).  The manifold structure of $\mathcal T_h$ may be obtained out of \eqref{3.b.4} since the action of $\mathrm{Diff}_0(\Sigma)$ is free, continuous, proper, and with local sections \cite{earle:eells}.

\subsection{Variations of complex structure}
\label{sec:3.a}

An infinitesimal variation of the complex structure $J$ may be parametrized by,
\bea
 J_m {}^n  \to J_m {}^n  + \otherep \, \delta J_m {}^n 
 \eea 
subject to the condition that the deformed complex structure squares to minus the identity to leading order in the infinitesimal parameter $\otherep$, 
\bea
\label{3.c.1}
J_m{}^p \, \delta J_p {}^n + \delta J_m {}^p \, J_p {}^n=0
\eea
In terms of the complex adapted local coordinates in which $J$ takes the form of (\ref{3.b.3b}), this condition requires $\delta J_z{}^z=\delta J _{\bar z} {}^{\bar z} =0$. Recall from (\ref{3.b.1}) that $J$ splits the complexified tangent bundle $T_\CC \Sigma$  into two sub-bundles $T^{(1,0)}_J \Sigma $ and $ T^{(0,1)}_J \Sigma$  on which $J$ takes the eigenvalues $+ i$ and $-i$, respectively.
Complex conjugation  exchanges the sub-bundles $T^{(1,0)}_J\Sigma$ and $T^{(0,1)}_J \Sigma$ of 
$T_{\mathbb C}\Sigma$, as well as their duals.

\sm

The infinitesimal variation $\delta J$ is a linear map from the complexified tangent bundle $T_\CC \Sigma$ to itself which can be equivalently described as the pair $(\delta J_{\bar z} {}^z, \delta J_z {}^{\bar z})$ of mutually complex conjugate maps between the two sub-bundles, 
\bea
\delta J : T_\CC \Sigma \to T_\CC \Sigma 
\hskip 1in 
\begin{cases} \delta J_{\bar z} {}^z : T^{(1,0)}_J \Sigma  \to  T^{(0,1)}_J \Sigma \cr
\delta J_z {}^{\bar z} : T^{(0,1)}_J \Sigma  \to  T^{(1,0)}_J \Sigma \end{cases}
\eea
The maps $\delta J_{\bar z} {}^z$ and $\delta J_z {}^{\bar z}$ may also be viewed as sections of the bundles $T^{(1,0)}_J \Sigma \otimes T_J^{* (0,1)}\Sigma$ and $T^{(0,1)}_J \Sigma \otimes T_J^{* (1,0)}\Sigma$, respectively, and are referred to as \textit{Beltrami differentials} $\mu$ and $\bar \mu$. In adapted local complex coordinates they are defined by,
\begin{align}
\label{3.c.1a}
\mu & = \mu(z) \, d\bar z/ dz & \delta J_{\bar z}{}^z & = - 2 i \, \mu (z) 
\no \\
\bar \mu & = \overline{\mu(z)} \, dz/ d \bar z & \delta J_{z}{}^{\bar z} & = + 2 i \, \overline{\mu (z)} 
\end{align} 
The normalization factors are conventional and have been chosen so that, under a variation $\delta J$ of the complex structure $J$, the metric $g= e^\phi |dz|^2$ associated with $J$ maps to the metric $g' = e^{\phi'} |dz - \otherep \mu(z) d\bar z|^2$ associated with $J +\epsilon \delta J$. Thus, the complex structure variation $\delta J$ is equivalent to a pair $(\mu, \bar \mu)$ of a Beltrami differential and its complex conjugate and the (real) tangent space $T_J ( \CS(\Sigma)  )$ to $\CS(\Sigma)$ is generated by all such pairs $(\mu, \bar \mu)$.

\sm

The (complexified) tangent space to  $\CS(\Sigma)$ at $J$ splits into its holomorphic and anti-holomorphic components,
\bea
\label{3.c.1b}
T_J \big (  \CS(\Sigma) \big ) \otimes \CC = T^{(1,0)}_J \big ( \CS(\Sigma) \big ) \oplus 
T^{(0,1)}_J \big ( \CS(\Sigma) \big )
\eea
Its component $T^{(1,0)}_J  ( \CS(\Sigma)  )$ and $T^{(0,1)}_J  ( \CS(\Sigma)  )$ are generated by the pairs of Beltrami differentials $(\mu, 0)$ and $(0,\bar \mu)$, respectively, namely,\footnote{The notation $\Gamma(M,V)$ stands for the space of smooth global sections of a bundle $V$ over a manifold~$M$.}
\bea
\label{3.c.2}
T^{(1,0)}_J \big ( \CS(\Sigma) \big ) & = & \Gamma \big ( \Sigma,T^{(1,0)}_J  \Sigma \otimes T^{*(0,1)}_J \Sigma \big )
\no \\
T^{(0,1)}_J \big ( \CS(\Sigma) \big ) & = & \Gamma \big ( \Sigma,T^{(0,1)}_J  \Sigma \otimes T^{*(1,0)}_J \Sigma \big )
\eea
relative to the $(1,0)\oplus(0,1)$ decomposition of  (\ref{3.b.1}). We also use the notation
\bea
\label{not:BD}
\mathrm{BD}_J:=\Gamma \big ( \Sigma,T^{(1,0)}_J  \Sigma \otimes T^{*(0,1)}_J \Sigma \big )
\eea

\subsection{The variational derivative}
\label{sec:intdww}

In the sequel, the only variations in the complex structure that will be needed to solve the partial globalization problem  are along the component $T^{(1,0)}_J \big ( \CS(\Sigma) \big ) $, and we shall henceforth restrict our discussion to this case. 

\sm

The corresponding variation of an arbitrary function $\phi$ of $\CS(\Sigma)$ under a variation of $J$ given by a Beltrami differential $\mu$ is then given by pairing $\mu$ against a $(2,0)$-form $q$ which is a section of the bundle $T^{* (1,0)}_J  \Sigma \, \otimes  T^{* (1,0)}_J \Sigma$. In local complex coordinates the quadratic differential is given by $q= q_{ww} (w) dw^2$ and the pairing takes the form,\footnote{This pairing is non-degenerate. The normalization by $\pi$  will conveniently simplify many formulas in the sequel. Henceforth, we shall use the local coordinates $w,\bar w$ to express the dependence of the Beltrami differentials on $\Sigma$ for added clarity, unless otherwise indicated.}
\bea
\label{3.c.30}
\< \mu | q \> = { 1 \over \pi} \int _\Sigma d^2 w \, \mu(w) \, q_{ww}(w)
\eea
The variation $\delta _\mu \phi$ of $\phi$ under a variation $\delta J _{\bar z} {}^z$ of the complex structure given by the first line in (\ref{3.c.1a}),  is defined by, 
\bea
\label{3.c.3}
\delta_\mu \, \phi  = 
{ \p \over \p \otherep}  \phi(J + \epsilon \, \delta J) \Big | _{\otherep = 0} 
\eea
The linear pairing defines a $(2,0)$ form $\delta_{ww} \phi$ by the relation, 
\bea
\label{3.c.40}
\delta_\mu \phi = \< \mu | \delta _{ww} \phi \> = { 1 \over \pi} \int _\Sigma d^2 w \, \mu(w) \, \delta_{ww} \phi 
\eea 
The $(2,0)$ form $\delta_{ww} \phi$ is referred to as the \textit{variational derivative of $\phi$} at the point $w \in \Sigma$. Reformulating (\ref{3.c.40}) in terms of variations of the metric, and taking the normalizations of (\ref{3.b.1}), (\ref{3.c.1a}) and $\ep_{z \bar z} = { i \over 2}$ into account gives,
\bea
\label{3.c.40a}
\delta_\mu \phi =  { 1 \over 4 \pi} \int _\Sigma d^2 w \, \sqrt{\det g} \, \delta g^{ww}  \, \delta_{ww} \phi 
\eea 
from which we recover the familiar expression \cite{Friedan:1982is,Verlinde:1986kw} (see also \cite{Bernref2}), 
\bea
\label{3.c.40b}
\delta _{ww} = { 4 \pi \over \sqrt{\det g} } { \delta \over \delta g^{ww}}
\eea
in terms of the functional derivative with respect to the metric on $\Sigma$. Here, the sub/superscript specify both the point $w$ at which the functional derivative is evaluated, and the fact that the result of the functional derivative is a $(2,0)$ form. The variational derivative satisfies the functional version of Schwarz's theorem, \footnote{The identity may be proven by observing that the relation is ultra-local, namely proportional to $\delta(v,w)$, upon which the remaining  components $g^{ww}$ and $g^{vv}$ at $w=v$ are equal to one another so that their derivatives commute with one another. A more detailed and more rigorous proof will be given in a subsequent paper in this series. } 
\bea
\delta _{vv} \delta _{ww} - \delta _{ww} \delta _{vv}=0
\label{comdww}
\eea

\sm

The variations of the differentials $\p_J$ and $\bar \p_J$ of (\ref{3.b.4old}) are similarly evaluated using a system of real coordinates $\xi^m$ that are independent of the complex structure to obtain, 
\bea
\label{3.c.5}
\delta_\mu \, \p_J & = & - \mu \, \p_J 
\no \\
 \delta_\mu \, \bar \p_J & = &  \mu \,  \p_J 
\eea
The corresponding variation of the local coordinate derivatives on scalars is readily obtained from these relations using the relation $\delta_\mu dz = - \mu (z) d\bar z$, and is given by,
\bea
\label{muder1a}
\delta_\mu \, \p_z & = &0 
\no \\
\delta_\mu \, \pbz & = & \mu(z) \p_z
\eea 
The extension of these relations to the case of coordinate derivatives acting on the coefficient functions of forms of more general type is well-known and given by \cite{DHoker:1988pdl, Friedan:1982is,Verlinde:1986kw}, 
\begin{align}
\label{muder1}
\delta_\mu \p_z & =  0 & \hbox{for } (0,n) \hbox{ forms}
\no \\
\delta _\mu \pbz  & =   \mu (z) \p_z + n \big (\p_z \mu (z) \big ) & \hbox{for } (n,0) \hbox{ forms}
\end{align}
A proof of the second line of (\ref{muder1}) can be found in appendix \ref{app:new}.

\subsection{Diffeomorphisms and derivatives with respect to moduli}
\label{sec:difdtau}

We shall mostly be interested in functions on  the space of complex structures $\CS(\Sigma)$ that are invariant under the group of diffeomorphisms ${\rm Diff}_0(\Sigma)$. This invariance reduces their dependence on the complex structure  to their dependence on Teichm\"uller space, as discussed in section~\ref{sec:3.aa} above (see also  section I.E in \cite{DHoker:1988pdl}).

\sm

The Lie algebra of $\mathrm{Diff}_0(\Sigma)$ is the space $\mathfrak{vec}(\Sigma)$ of all smooth vector fields on $\Sigma$. For each $J\in\ \CS (\Sigma)$, the action of $\mathrm{Diff}_0(\Sigma)$ on $\CS(\Sigma)$ induces a linear map $\mathfrak{vec}(\Sigma)\to T_J^{(1,0)}\CS(\Sigma)$ into the space of Beltrami differentials, given by the following $J$-dependent composition, 
\bea
\label{3.f.1}
\mathfrak{vec}(\Sigma) 
\to 
\Gamma \big ( \Sigma,T^{(1,0)}_J \Sigma \big ) 
\xrightarrow{~ \overline \partial_J ~} 
\Gamma \big ( \Sigma,T^{(1,0)}_J \Sigma \otimes T^{*(0,1)}_J \Sigma \big )
\eea
which we will also denote $\bar \partial_J$. The first map in \eqref{3.f.1} is the inclusion of $\mathfrak{vec}(\Sigma)=\Gamma ( \Sigma, T \Sigma)$ into the space of sections $\Gamma \big ( \Sigma,T_\CC \Sigma \big )$ of the complexified tangent bundle $T_\CC \Sigma$, which is then decomposed according to (\ref{3.b.1}) and projected onto its holomorphic component,  
\bea
\label{3.f.2}
\Gamma \big ( \Sigma,T \Sigma  \big ) 
\subset 
\Gamma \big ( \Sigma,T_\CC \Sigma \big )
=\Gamma \big (\Sigma,T^{(1,0)}_J \Sigma \oplus T^{(0,1)}_J \Sigma \big )
\to \Gamma \big ( \Sigma,T^{(1,0)}_J \Sigma)
\eea
The second map in (\ref{3.f.1}) is given by the $\bar \partial_J$ operator associated with the holomorphic bundle structure of $T^{(1,0)}_J \Sigma $. 

\sm

Translated into concrete formulas, we represent an element of $\mathfrak{vec}(\Sigma)$ by a vector field $v= v^m \p_m$ in real coordinates $\xi^m$, complexify $v$ to $ v(w) \p_w + \bar v(w) \pbw$, project to the holomorphic component $v(w) \p_w$, and map this into the space of Beltrami differentials. The action of $\mathfrak{vec}(\Sigma)$ on the complex structure $J$ is given by the following Beltrami differential,
\bea
\label{3.f.3}
\mu(w) = \pbw v(w)
\eea  
The formula may also be derived from the well-known action of a vector field $v^m$ on the metric, modulo Weyl transformations. Note that, for genus $h \geq 2$, there are no solutions to $\pbw v(w)=0$, namely, $\Sigma$ has no conformal Killing vectors and no continuous symmetries. 

\sm

Viewing a function $\psi$ on Teichm\"uller space $\cT_h$ as a ${\rm Diff}_0(\Sigma)$ invariant function on the space of complex structures $\CS(\Sigma)$, its invariance under $\mathfrak{vec}(\Sigma)$ combined with  the pairing relation (\ref{3.c.40}) implies that, 
\bea
\< \bar \p  v | \delta _{ww} \psi \>=0
\eea
for all $v \in \mathfrak{vec}(\Sigma)$. Since the pairing is non-degenerate, this implies that $\delta_{ww} \psi$ is a \textit{holomorphic quadratic differential}. The (1,0) differential of~$\psi$ on $\cT_h$  corresponds to a $\mathrm{Diff}_0(\Sigma)$-equivariant assignment that maps $J \in \CS(\Sigma)$ to a holomorphic quadratic differential $\delta _{ww} \psi $ given by the relation \eqref{3.c.40} for an arbitrary Beltrami differential $\mu$.  In particular, if $(\tau_1,\cdots ,\tau_\sigma)$ with $\sigma=\dim_\CC (\cT_h)$ is a local system of complex coordinates on $\mathcal T_h$  in a  neighborhood of (the equivalence class $[J]$) of $J$, then the derivative of $\psi$ with respect to a holomorphic modulus $\tau_\a$, for $\a  \in \{ 1,\cdots, \sigma \}$,  is given by, 
\bea
\label{3.f.9}
 \frac{\p }{ \p \tau _ \alpha }  \, \psi
=\< \mu _\a | \delta _{ww} \psi \>
= { 1 \over \pi} \int _\Sigma d^2 w \, \mu_\a (w) \, \delta _{ww} \psi
\eea
where $\mu_\alpha\in T^{(1,0)}_J(\CS(\Sigma))$ is an arbitrary  Beltrami differential representing the vector field $ \partial/\partial\tau_\alpha$ on  the Teichm\"uller space $\cT_h$.

\sm

{\rmk \label{rem:tww}
The problem of computing  moduli variations of an arbitrary function $\phi$ has been reduced to evaluating the  \text{variational derivative} $\delta _{ww} \phi $ which is a $(2,0)$ form in a point $w$ on $\Sigma$. We note that, in two-dimensional conformal field theory, the variational derivative  is obtained by inserting the traceless stress tensor $T_{ww}$ into the expectation value of an arbitrary operator $\cO$ with the following normalization $\delta_{ww} \< \cO \> = \< T_{ww} \cO \>$. This can also be seen from the representation (\ref{3.c.40b}) of $\delta_{ww}$ as a functional derivative with respect to the metric. The splitting of the tangent and cotangent bundles of (\ref{3.b.1}) was articulated in \cite{Friedan:1982is} and used extensively in \cite{Alvarez:1982zi, DHoker:1986eaw}. The use of the traceless stress tensor to vary moduli also dates back to \cite{Friedan:1982is} and was elaborated in \cite{Belavin:1986cy,Verlinde:1986kw} and reviewed in \cite{DHoker:1988pdl}. }

\subsection{Evaluation of auxiliary variational derivatives}
\label{sec:3.4}

In this section, we shall evaluate the variational derivatives of various key ingredients in the DHS connection, including local coordinate derivatives, Abelian differentials, the canonical volume form, the Arakelov Green function, and the DHS integration kernels. 

\sm

Converting the moduli variations of Cauchy-Riemann operators of (\ref{muder1})  into their variational derivatives gives the following relations,
\bea
\label{muder.1}
\delta _{ww} \p_z =0 
& \hskip 1in & 
\delta _{ww} \p_{\bar z} =  \pi \delta (w,z) \p_z +  \pi n \big ( \p_z \delta (w,z) \big )
\eea
We recall that the left formula applies to the derivative $\p_z$ acting on scalars, while the right formula applies to the $\pbz$ derivative acting on the coefficient function of a $(n,0)$ form.

\subsubsection{Abelian differentials}

The formulas of (\ref{muder1}) allow us to compute the complex structure and moduli 
variations of a holomorphic $(1,0)$-form $\om_I= \om_I(z)  \, dz$ as  follows \cite{Verlinde:1986kw} (see also Theorem 3 in \cite{Rauch}),
\bea
\label{3.a.3}
0 = \delta _\mu \,  \big(\pbz \, \om _I (z)  \big) & = & (\delta _\mu \,  \pbz) \, \om _I (z)  + \pbz \big(\delta _\mu \, \om _I (z)  \big) 
\no \\ & = & \p_z \big( \mu  \, \om_I (z) \big) +  \pbz \big(\delta _ \mu \, \om _I (z) \big) 
\eea
 which, using the normalization $\oint _{\mA^J} \delta _\mu \, \om_I=0$,  is solved uniquely as follows,
\bea
\label{3.a.4}
\delta _\mu \, \om_I (z) = { 1 \over \pi} \int _\Sigma d^2 w \, \mu (w)  \,  \om _I(w)  \, \p_z \p_w \ln E(z,w) 
\eea
where $E(z,w) $ is the prime form, see for instance \cite{Faysbook}.
From (\ref{3.a.4}) and (\ref{3.c.40}) we read off the variational derivative $\delta_{ww} \om_I(z)$, and we find,
\bea
\label{3.a.5}
\delta_{ww} \, \om_I (z) =  \om _I(w) \, \p_z \p_w \ln E(z,w) 
\eea
Integrating (\ref{3.a.5}) over a $\mB_J$ cycle gives the variational derivative of the periods,
\bea
\delta_{ww} \Omega _{IJ} = 2 \pi i \, \om_I(w) \om_J(w)
\label{domij}
\eea
while $\delta_{ww} \bar \Omega_{IJ}=0$. Further useful formulas are as follows, 
\begin{align}
\label{muder.2}
\delta _{ww} \, Y_{IJ} & =  \pi  \om_I(w) \om_J(w) & \delta _{ww} \, \om^I(z) & =  - \om^I(w) \p_z \p_w \cG(z,w) 
\no \\
\delta _{ww} \, Y^{IJ} & =  - \pi \om^I(w) \om^J(w)  & \delta _{ww} \, \bar \om_I(z) & =  0
\end{align} 
where $\cG$ is the Arakelov Green function defined by (\ref{2.a.3}). The variational derivative of the volume form $\kappa$ in (\ref{2.a.2}) may be obtained from these formulas and is given~by,
\bea
\label{muder.3}
\delta _{ww} \, \kappa (z) & = & - { 1 \over  h} \om^I(w) \, \bar \om _I (z) \, \p_z\p_w \cG(z,w)
\eea

\subsubsection{The Arakelov Green function}
\label{sec:3.5}

The variational derivative of the prime form was obtained in \cite{Verlinde:1986kw}. The variational derivative of the Arakelov Green function is given by the following proposition. 

{\prop
\label{3.prop:40}
The variational derivative of the Arakelov Green function is given by, 
\bea
\label{muder.4}
\delta _{ww} \, \cG (x,y) & = & - \p_w \cG (w,x) \, \p_w \cG (w,y) + \gamma _{ww} (x) + \gamma _{ww} (y)
\eea
where $\gamma _{ww} (x) $ is given in terms of the DHS kernel $\cG_I(x,y)$ by,
\bea
\label{muder.5}
\gamma _{ww} (x) = - \frac{1}{h} \, \omega^I(w) \, \partial_w \cG_I (w,x)
\eea}

This proposition will be proven in appendix \ref{sec:Arak} using the explicit construction of the Arakelov Green function in terms of Abelian differentials and the prime form in  \cite{DHoker:2017pvk}.

\subsubsection{The DHS kernels}
\label{sec:3.6}

The variational derivatives of the DHS kernels defined by (\ref{2.a.4}) and (\ref{2.a.6}) are given by the relation (29) of \cite{DHoker:2025dhv}, 
\begin{align}
\label{Var.3}
\delta_{ww} \, f_\vI{\, }^J(x,y)  = - \sum_{\vI = \vP \vQ} \p_w \p_x \cG_\vP(x,w)  f_\vQ{}^J(w,y)
\end{align}
where the subscript $\vI = \vP \vQ$ stands for the summation over all possible deconcatenations of the given word $\vI$ into words $\vP$ and $\vQ$ where $\vP$ or $\vQ$ are allowed to be the empty word~$\emptyset$. 
Decomposing this result into trace and traceless parts with respect to the rightmost indices, we obtain, 
\bea
\label{Var.20}
\delta_{ww} \, \p_x \cG_\vI (x,y) & = & 
- \sum_{\vI = \vP \vQ} \p_w \p_x \cG_\vP (x,w) \p_w \cG_\vQ(w,y)
+ { 1 \over h} \p_w \p_x \cG_{\vI J} (x,w) \om^J(w)
 \\
\delta_{ww} \, \p_x \Phi _{\vI K}{}^J(x) & = & 
- \sum_{\vI K = \vP \vQ} \p_w \p_x \cG_\vP (x,w) \p_w \Phi _\vQ{}^J(w)
+ { 1 \over h }  \p_w \p_x \cG_{\vI L}(x,w) \om^L(w) \, \delta ^J _K
\qquad
\no
\eea
where we recall the relations $f_\emptyset {}^J(x,y) = \p_x \Phi _\emptyset {}^J(x) = \om^J(x)$ and $\cG_\emptyset(x,y) = \cG(x,y)$. We shall also make use of the integrated version of these relations, which are given as follows,
\bea
\label{7.c.8}
\delta _{ww} \cG_\vI (x,y) & = &
- \sum_{\vI = \vP \vQ} \p_w \cG_\vP (x,w) \p_w \cG_\vQ(w,y) - { 1 \over h} \om^L(w) \Big ( \p_w \cG_{L \vI} (w,y) - \p_w \cG_{\vI L } (x,w) \Big ) 
\no \\
\delta_{ww} \Phi _{\vI K} {}^J(x) & = & - \sum_{\vI = \vP \vQ} \p_w \cG_\vP(x,w) \p_w \Phi _{\vQ K} {}^J(w)
- \p_w \cG_{\vI K}(x,w) \om^J(w)
\no \\ && \qquad
+{ 1 \over h} \p_w \cG_{\vI L} (x,w) \om^L(w) \delta^J_K
-{ 1 \over h} \om^L(w) \p_w \Phi _{L\vI K}{}^J (w)
\eea
Since $[\delta_{ww}, \p_x]=0$ by (\ref{muder1}), we recover (\ref{Var.20}) from (\ref{7.c.8}) by differentiating term by term, and the $x$-independent terms on the right side are readily validated by integrating against $\kappa(x)$ and
using the variational derivative (\ref{muder.3}).

\subsection{Variational derivatives of $\cA$ and $\mD$}
\label{sec:3.3}

The modular tensors $\cA$ and $\mD$ were introduced in section \ref{sec:adn}. Their variational derivatives will play a key role in the extension of the DHS connection to holomorphic moduli variations, and are given by the proposition below. 
{\prop
\label{7.prop:2}
The variational derivatives of the modular tensors $\cA$  and $\mD$ for an arbitrary word $\vI = I_1 \cdots I_r$ with $r \geq 0$  are given by,
\bea
\label{7.prop.2a}
\delta_{ww}  \, \cA^M {} _{K \vI L} {\, }^N & = & 
\sum_{K \vI L = \vR \vS} \p_w \Phi _{\theta(\vR)} {}^M (w) \, \p_w \Phi _\vS{}^N(w)
- { 1 \over h} \delta ^M _K \, \om^J(w) \, \p_w \Phi _{J \vI L}{}^N(w)
\no \\ && 
- { 1 \over h} \om^J(w) \, \p_w \Phi _{J \theta(\vI) K }{}^M (w) \, \delta ^N_{L} 
\eea
and 
\bea
\label{7.prop.2b}
\delta_{ww} \mD_\vI  &= &
- \frac{1}{h} \sum_{\vI = \vX \vY \vZ} \Big( \p_w \Phi_{ \vZ M \vX J }{}^M(w) \p_w \Phi_{\theta(\vY)}{}^J(w) 
+ \p_w \Phi_{ \vZ M \vX }{}^J(w) \p_w \Phi_{\theta(\vY) J}{}^M(w) \Big) 
\no \\ && 
+ \frac{1}{h} \om^J(w) \sum_{\vI = \vX \vY} \Big( \om^M(w) \mN_{M \vY J \vX} + \p_w \Phi_{M \vY J \vX}{}^M(w) \Big)  
\no \\ &&
 + \frac{1}{h} \,  \om^J(w) \sum_{\vI = \vX \vY \vZ} \p_w \Phi_{J (\vX \shuffle \theta(\vZ)) M (\vY + \theta(\vY)) }{}^M(w) 
\eea
where the modular tensor $\mN$ is defined in (\ref{defcn}), also see (\ref{A.4}) of appendix \ref{sec:A}.}

\sm

The proof of Proposition \ref{7.prop:2}  is relegated to appendix \ref{sec:BB}. 

\sm

{\rmk Note that the expression for $\delta_{ww}  \, \cA^M {} _{K \vI L} {\, }^N$ on the right side of  (\ref{7.prop.2a}) manifestly has the antipode symmetry of (\ref{7.f.2}). The expression for $\delta_{ww} \mD_\vI$ in (\ref{7.prop.2b}), however,  is not manifestly invariant under $ \vI \to \theta(\vI)$ since early steps of the proof in appendix \ref{sec:BB} involve coincident limits of the Fay identities in (\ref{A.10}) where the symmetry of the left side under $\vI \leftrightarrow \vJ$ is not manifest on the right side. Manifest symmetry of (\ref{7.prop.2b}) may be achieved simply by symmetrizing the right side, though this clearly lengthens  the formula.}

\subsection{Variational derivatives of Lie algebra generators of $\mt_{h,n}$}
\label{sec:3.8}

Modular invariance of the DHS connection $\cJ_\text{DHS}$, which was defined in  Theorem \ref{2.thm:40},  requires the automorphy factors $\tilde Q$, $\tilde R$ of Proposition \ref{2.prop:1} in the modular transformation law of the generators $a_{iI}, b_i^I$ and $t_{ij}$ of $\mt_{h,n}$. As a result, the generators $a_{iI}$ and $b_i^I$ may be viewed as sections of a non-trivial anti-holomorphic vector bundle over Teichm\"uller  space on which the variational derivative $\delta_{ww}$ acts covariantly without the need for a connection, 
\bea
\delta_{ww} \, a_{iI} \to \tilde Q_I {}^K \, (\delta _{ww} \, a_{iK}) 
\hskip 0.8in 
\delta _{ww} \, b_i^I \to (\delta_{ww} \, b_i^K) \tilde R_K{}^I
\eea  
where $\tilde R = \bar \Omega C^t + D^t$ and $\tilde Q = \tilde R^{-1}$ were given in (\ref{modsec.03}). Covariance of the variational derivative allows us to set these derivatives to zero in a manner that is consistent with modular invariance of $\cJ_\text{DHS}$, as summarized by the definition below.

{\deff
\label{3.def:3}
We define the variational derivative $\delta_{ww}$ of the generators of the Lie algebra $\mt_{h,n}$ to vanish,
\bea
\delta_{ww} \, a_{iI} =0 
\hskip 1in 
\delta _{ww} \, b_i^I =0 
\hskip 1in 
\delta_{ww} \, t_{ij}=0
\eea
This assignment is consistent with the structure relations of $\mt_{h,n}$.}

\subsection{Variational derivative of the DHS connection}
\label{sec:3.9}

To evaluate the variational derivative of $\cJ_\text{DHS}$ given in (\ref{2.b.5}) in terms of generating functions, we begin by evaluating the variational derivative of the generating functions $\bPhi$ and $\bG$. To do so we exploit the vanishing variational derivatives of the generators of $\mt_{h,n}$ by Definition \ref{3.def:3}. The result may be readily obtained by recasting the relations of (\ref{Var.20}) in terms of generating functions $\bPhi$, $\bG$ in (\ref{2.b.3a}), (\ref{2.b.3b}), and we obtain, 
\bea
\label{3.d.1}
\delta_{ww}  \bG(x,y;B)
& = & - \p_x \theta \bG(w,x;B) \bG(w,y;B)  + { 1 \over h} \om^K(w) \sum_{\vI} \p_w \p_x \cG_{\vI K} (x,w)  B^\vI
\no \\
\delta _{ww} \bPhi ^J(x;B) 
& = & 
- \p_x \theta \bG(w,x;B) \bPhi ^J(w;B) + { 1 \over h} \om^K(w) \sum_{\vI} \p_w \p_x \cG_{\vI K} (x,w)  B^{\vI J} 
\qquad
\eea
where $\theta$ is the antipode, whose definition was given in footnote \ref{other.4}, and whose action is on the function that it immediately precedes.  Combining these variational derivatives to obtain the one for $J_i^{(1,0)}$, we use the fact that the  second term of both lines in (\ref{3.d.1}) cancel one another, so that we get, 
\bea
\label{3.d.2}
\delta_{ww} \, J_i^{(1,0)} (\xx) = - \p_i \theta \bG(w,x_i;B_i) \, J_i^{(1,0)} \big(\xx_i(w) \big) 
\eea
where $\xx_i(w)$ is obtained from $\xx$ by substituting $w$ for $x_i$, leaving all other entries unaltered,
\bea
\label{3.d.3}
\xx_i(w) = (x_1, \cdots, x_{i-1} , w, x_{i+1}, \cdots, x_n)
\eea
Furthermore, in view of the vanishing variational derivative $\delta_{ww}$ of both $\bar \om_I(x_i)$ and $b^I_i$, we obtain the vanishing of the $(0,1)$ component of $\cJ_\text{DHS}$ in (\ref{defj01}), namely,
\bea
\label{3.d.4}
\delta _{ww} \, J_i^{(0,1)} (\xx) =0
\eea
Putting it all together, we obtain, 
\bea
\label{3.d.5}
\delta_{ww} \, \cJ_\text{DHS} (\xx) = - \sum _i dx_i \, \p_i \theta \bG(w,x_i;B_i) \, J_i^{(1,0)} \big(\xx_i(w) \big) 
\eea
In the next sections, we shall use these results to extend $\cJ_\text{DHS}$ to a flat connection with components in the holomorphic moduli directions.

\newpage

\section{Extending the DHS connection by varying moduli}
\label{sec:4}
\setcounter{equation}{0}

The main goal of this paper, and the topic of the present and subsequent sections,  is to provide an explicit construction of a local flat connection that extends the multivariable DHS connection $\cJ_\text{DHS}$ on the configuration space ${\rm Cf}_n (\Sigma)$ of $n$ variables on a fixed Riemann surface $\Sigma$ by including holomorphic variations in the complex structure of $\Sigma$. 

\subsection{Choice of local sections and strategy}
\label{sec:46.1}

Fix a point $\btau_0 \in \cT_h$ in Teichm\"uller space and a complex structure $J_0\in\mathrm{CS}(\Sigma)$ above $\btau_0$. 
As explained in section \ref{sec:1.1a}, a central role is played in Theorem \ref{1.thm:main} by local sections 
\bea
\label{46.a.1}
s : \mathcal T_h\to\mathrm{CS}(\Sigma)
\eea
such that $s(\btau_0)=J_0$.

\sm

To a local section $s$, given in (\ref{46.a.1}), and a point $\btau \in \cT_h$, is attached the space of  {\it Beltrami differentials} ${\rm BD}_{s,\btau}$ defined as the image of  the differential $d_\btau s : T_\btau\mathcal T_h\to \mathrm{BD}_{s(\btau)}$ using the notation \eqref{not:BD}.  A basis $(e_\alpha)_\alpha$ of $T_\btau\mathcal T_h$ being given, the space ${\rm BD}_{s,\btau}$ associated with  $(s,\btau)$ is spanned by the Beltrami differentials $\mu^s _\a$  defined~by, 
\begin{equation}
\label{BDs:section}
\mu_\alpha^{s,\btau} =(d_\btau s)(e_\alpha)\in \mathrm{BD}_{s(\btau)}
\end{equation}
where $\a \in \{ 1, \cdots , \dim \cT_h\}$.

\sm

In order to be in a position where we can easily translate the problem in terms of coordinates,  we will work not with just one but with several sections. For each tuple 
of open subsets $\underline U=(U_1,\cdots,U_n)$ of $\Sigma$ such that 
the union of their closures does not cover $\Sigma$, namely $\cup_i\overline U_i\neq \Sigma$,\footnote{The relevance of this condition will be made clear in section \ref{sect:bers}.}  we will construct a section $s_{\underline U}$, 
with the property that the corresponding Beltrami differentials satisfy, 
\bea
\label{vanishing:cond}
\mu_\a ^{s_{\underline U},\btau} (x) =0 \quad \hbox{ for all }\btau\hbox{ and } x \in U_i \hbox{ with } i \in \{1, \cdots, n\}
\eea
and all $\a\in \{ 1, \cdots , \dim \cT_h\}$. The vanishing condition (\ref{vanishing:cond}) for the Beltrami differentials is advantageous to prevent a mixing of $\partial_{x_i}$ and $\partial_{\bar x_i}$ in the differential equations of DHS kernels and related quantities under complex structure variations. The existence of sections that enable (\ref{vanishing:cond}) as well as their construction are discussed in section~\ref{sect:bers} below.

\sm

In sections \ref{sec:gfc} and \ref{sec:locf} through Lemma \ref{7.lem:5a},  we shall work locally in the following sense. We fix a collection of tuples of open sets 
$\underline U=(U_1,\cdots,U_n)$, work with $\mathcal J_{\mathrm{DHS}}^{s_{\underline U}}$, 
and study the 
system of partial differential equations satisfied by the extension of $\mathcal J_{\mathrm{DHS}}^{s_{\underline U}}$
by choosing  coordinates $(\xx |\btau)$ on $U_1\times \cdots \times U_n\times \mathcal T_h$,  where $\xx= (x_1, \cdots, x_n)$ is a set of complex coordinates  for~$U_1\times...\times U_n$  and $\btau = (\tau_1, \cdots,  \tau_\sigma)$ is a set of complex coordinates on $\cT_h$. Suppressing the superscript  $s_{\underline U}$ for clarity, the problem to be solved in those sections will therefore be  the construction of a $(1,0)$-form component $\cL(\xx|\btau)$ in $d\tau_\alpha$ directions 
 such that the connection, 
\bea
\label{7.0.x}
\cJ (\xx| \btau) = \cJ_\text{DHS} (\xx| \btau) + \cL(\xx|\btau)
\eea
satisfies the following flatness condition, 
\bea
\label{7.0.a}
(d_\xx + \p _\btau ) \cJ- \cJ \wedge \cJ=0
\eea
expressed in terms of  the total differential $d_\xx$ on ${\rm Cf} _n(\Sigma)$ defined after (\ref{2.a.1b}) and the differential in holomorphic moduli $\p _\btau = \sum _\a d\tau _\a \, \p_{\tau_\a}$ (excluding the anti-holomorphic variations $ \sum _\a d \bar \tau _\a \, \p_{\bar \tau_\a}$). 
The flatness condition (\ref{2.a.1b}) satisfied by  $\cJ_\text{DHS}$ together with the flatness condition (\ref{7.0.a}) satisfied by $\cJ$ imply the differential equation for the component $\cL$,
\bea
\label{7.0.b}
(d_\xx + \p _\btau )  \cL + d _\btau \, \cJ_\text{DHS} - \cJ _\text{DHS} \wedge \cL 
- \cL  \wedge  \cJ _\text{DHS} - \cL \wedge \cL =0
\eea 
This differential equation (in the variables $\xx, \bar \xx$, and $\btau$) satisfies the Frobenius integrability conditions, as  may be shown by applying the  differential $d_\xx + \p _\btau $ to the left side of (\ref{7.0.b}), eliminating $(d_\xx + \p _\btau )  \cL$ using (\ref{7.0.b}), and using the flatness condition (\ref{2.a.1b}).

\sm

While integrability of (\ref{7.0.b}) guarantees the existence of local solutions for $\cL$, it does not guarantee that these local solutions extend to regular global single-valued solutions on the Teichm\"uller space $\cT_{h,n}$ of compact Riemann surfaces of genus $h$ with $n$ marked points.  In the remainder of this section, we shall pinpoint necessary conditions for the existence of a single-valued and modular invariant solution for $\cL$ and we shall exhibit the solution to these necessary conditions explicitly in terms of DHS kernels.

\subsection{Construction of local sections}
\label{sect:bers}

In this subsection, we shall provide an explicit construction of the local sections $s: \cT_h \to \CS(\Sigma)$ in (\ref{46.a.1}), subject to the conditions of (\ref{vanishing:cond}) on their Beltrami differentials,  for a given tuple 
$\underline U=(U_1,\cdots,U_n)$ of open subsets of $\Sigma$ with $\cup_i\overline U_i\neq\Sigma$. This will justify the possibility to impose the vanishing conditions (\ref{vanishing:cond}) on Beltrami differentials that facilitate the computations in the remainder of this work.

\sm

The space $\mathrm{CS}(\Sigma)$ of complex structures may be identified with the space $\mathcal P (\Sigma)$ of endomorphisms $p$ of  the complexified tangent bundle $T_\CC \Sigma$, such that $p^2=p=\overline p$  taking $J$ to $p=(J+i)/(2i)$.   The corresponding bijection between their tangent spaces is induced as follows. 
Recall that ${\rm BD}_J$ is the space of Beltrami differentials $\mu$ at complex structure $J$ 
(see \ref{not:BD}) and denote by $X_p$ the following space,
\bea
X_p = \{ \delta q : p \delta q = \delta q (1-p) =0 \}
\eea
at endomorphism $p \in \cP(\Sigma)$. The tangent spaces to ${\rm CS}(\Sigma)$ and $\cP(\Sigma)$ are  given by,\footnote{In the second line, the direct sum notation is justified by the relation $X_p\cap \overline X_p=0$, 
which follows from  $\overline X_p=X_{1-p}$, and from the fact that $p\delta x=(1-p)\delta x=0$ implies $\delta x=0$.  The equality  follows from $\delta p+\overline{\delta p}=0=(p-1)\delta p+\delta p\cdot p =p\delta p+\delta p(p-1)$, which  implies, with $\delta q=\big ( (1-p)\delta p\big )  p$, both $\delta q\in X_p$ and, since  
$\overline{\delta q}=(1-\overline p)\overline{\delta p}\cdot \overline p=-p\cdot \delta p(1-p)=-p\delta p=\delta p p-\delta p
=(1-p)\delta p\cdot p-\delta p=\delta q-\delta p$, the relation  $\delta p=\delta q-\overline{\delta q}$.}
\bea
T_J \big ( \CS (\Sigma) \big ) & = & {\rm BD}^\Delta _J 
= \big \{ (\mu, - \bar \mu) \in {\rm BD} _J \oplus \overline{{\rm BD}}_J \big  \}
\no \\
T_p \big ( \cP(\Sigma) \big ) & = &  X^\Delta _p ~ \, = \big \{ (\delta q, - \overline{\delta q} ) \in X_p \oplus \bar X_p \big \}
\eea
The bijection between tangent spaces results from the above expressions and the bijection $\mathrm{BD}_{J}=\mathrm{Hom}(K_J,\overline K_J)\to X_p$ (homomorphisms of smooth bundles over $\Sigma$) induced by  $\mu\mapsto i_{\overline K_J}\mu p_{K_J}$. Here, $K_J=T^{\ast(1,0)}\Sigma$  is the canonical bundle of $\Sigma$ 
(see \ref{not:KJ}) and $p_{K_J}$ and $i_{\overline K_J}$ are the projection and injection of  
$T_\CC^*\Sigma$ to and from the relevant summands of its decomposition $K_J\oplus \overline K_J$.

\sm

For an arbitrary  $J_0\in \mathrm{CS}(\Sigma)$, we set $K_0=K_{J_0}$ and define a map, 
\bea
e:U_0\to  \mathrm{CS}(\Sigma) \hskip 1in U_0 \subset \mathrm{BD}_{J_0}^\Delta
\eea
on the subset $U_0$ of Beltrami pairs $(\mu,-\overline\mu)\in \mathrm{BD}_{J_0}^\Delta$ such that $|\mu\overline\mu|<1$ (recall that $\mu$ is a smooth bundle morphism $K_0\to\overline K_0$, therefore  
$\mu\overline\mu\in C^\infty(\Sigma,\mathbb R_+)$) by the condition that 
$e(\mu,-\overline\mu)\in \mathrm{CS}(\Sigma)$ is the complex structure defined by,  
\bea
p(\mu,-\overline\mu)=
(i_{K_0}+i_{\overline K_0}\mu)(id_{K_0}-\overline \mu\mu)^{-1}(p_{K_0}-\overline\mu p_{\overline K_0})\in\mathcal P(\Sigma)
\eea
where $i_{K_0}$ (resp.\ $p_{K_0}$) is the injection (resp.\ projection) relating $K_0$ and $K_0\oplus\overline K_0$, and similarly for $i_{\overline K_0}$ (resp. $p_{\overline K_0}$). 
 The differential, 
\bea
d_{(\mu,-\overline\mu)} e : \mathrm{BD}_{J_0}^\Delta\to \mathrm{BD}_{e(\mu,\overline\mu)}^\Delta
\eea
is such that 
\begin{align}
 \mathrm{BD}_{J_0}^\Delta\ni (\delta\mu,-\overline{\delta\mu})&\mapsto 
\big((i_{\overline K_0}+i_{K_0}\overline\mu)\delta\mu(id_{K_0}-\overline\mu\mu)^{-2}(p_{K_0}-\overline\mu p_{\overline K_0}),\\ & \quad \quad
(i_{K_0}+i_{\overline K_0}\mu)\overline{\delta\mu}(id_{\overline K_0}-\overline\mu\mu)^{-2}(\mu p_{K_0}-p_{\overline K_0}) \big)
\in X_{p(\mu,-\overline\mu)}^\Delta\simeq \mathrm{BD}_{e(\mu,\overline\mu)}^\Delta\notag 
\end{align}
It is therefore a $\mathbb C$-linear map. Its composition with a $\mathbb C$-linear section of the projection 
$\mathrm{BD}_{J_0}^\Delta\to T_{[J_0]}\mathcal T_h$ is then a locally defined holomorphic map 
$T_{[J_0]}\mathcal T_h\to \mathrm{CS}(\Sigma)$.\footnote{The classical literature (see for example \cite{Bers:bullLMS, EE:DG}) contains other examples of local sections of the  projection $\mathrm{CS}(\Sigma)\to\mathcal T_h$, 
however, global sections of this projection are known not to exist in general \cite{Earle:cross}.}  Let 
$\mathrm{BD}_{J_0}^{\underline U}\subset \mathrm{BD}_{J_0}$ be the subspace of Beltrami differentials which vanish on $\cup_i U_i$. By the assumption on $\underline U$, the restriction  
$\mathrm{BD}_{J_0}^{\underline U}\to T_{[J_0]}\mathcal T_h$ is surjective. Fix a section $\sigma_{\underline U}$ 
of it. Then $s_{\underline U}=e\circ \sigma_{\underline U} : T_{[J_0]}\mathcal T_h\to\mathrm{CS}(\Sigma)$ is such that the Beltrami differentials at $(\mu,-\overline\mu)$ are $\mu_\alpha$ corresponding~to, 
\bea
(i_{\overline K_0}+i_{K_0}\overline\mu)\sigma_{\underline U}(e_\alpha)(id_{K_0}-\overline\mu\mu)^{-2}(p_{K_0}-\overline\mu p_{\overline K_0})
\label{ik0eq}
\eea
which vanishes on $\cup_i U_i$, as announced. 

\sm

For future use, let us note here that local sections $s_{\underline U},s_{\underline U'}$ 
relative to different tuples of open subsets are related by, 
\bea
s_{\underline U'}(\btau)=s_{\underline U}(\btau)
\phi_{\underline U\underline U'}(\btau)
\eea
where $\phi_{\underline U\underline U'} : \mathcal T_h\to\mathrm{Diff}_0(\Sigma)$ is a smooth map, with $\phi_{\underline U\underline U'}(\btau_0)=1$.  For the remainder of this section, until before Lemma \ref{7.e.2A}, 
we will choose and install on $\mathrm{Cf}_n(\Sigma)  \times \mathcal T_h$ the vertical connection 
$\cJ_\text{DHS}^{\underline U}$, whose restriction to $\mathrm{Cf}_n(\Sigma) \times  \{\btau\}$ corresponds  to the complex structure $s_{\underline U}(\btau)$ over $\Sigma$; 
and furthermore restrict our study to the product of $U_1\times\cdots\times U_n$
with the defining set of $s_{\underline U}$, parametrized by $(\xx|\btau)$.

\subsection{The generalized flatness conditions}
\label{sec:gfc}

The flatness conditions of $\cJ$ in (\ref{7.0.x}) consist of the flatness conditions of $\cJ_\text{DHS}$ at fixed moduli, given in (\ref{2.b.3}), supplemented by the conditions in (\ref{7.0.b}) that involve the components $\LL_\a$ in the directions of holomorphic moduli variations,
\bea
\label{7.0c}
\cL (\xx|\btau) =   \sum_{\a =1} ^{\dim \cT_h}  \LL_\a (\xx| \btau) d \tau_\a
\eea
or moduli derivatives $\p_{\a} = \p_{\tau_\a}$ of $\cJ_\text{DHS}$,\footnote{As an illustration of the convenience of the vanishing condition  (\ref{vanishing:cond}) on the Beltrami differentials, we note that the middle line of (\ref{7.a.0}) would take the following more involved form $ \bar \p_i L_\a - \p_\alpha J_i^{(0,1)} - [J_i^{(0,1)}, L_\a ] + \mu_\a (x_i) \, J_i ^{(1,0)} =0$ for generic choices of the section $s: \mathcal T_h\to\mathrm{CS}(\Sigma)$.}
\bea
\label{7.a.0}
\p_i \LL_\a - \p_\a J_i^{(1,0)} - [ J_i ^{(1,0)} , \LL_\a] & = & 0
\no \\
\bar \p_i \LL_\a - \p_\a J_i^{(0,1)} - [ J_i ^{(0,1)} , \LL_\a] & = & 0
\no \\
\p_\a \LL_\b - \p_\b \LL_\a - [\LL_\a, \LL_\b] & = & 0
\eea
where the second equation follows from \eqref{vanishing:cond}, and the arguments $(\xx| \btau)$ are suppressed for brevity here and below unless indicated otherwise. Integrability of the system of differential equations for $\LL_\a$ in (\ref{7.a.0}) is guaranteed by the integrability of (\ref{7.0.b}) established in the preceding subsection.

\sm

It remains to specify the Lie algebra  in which the connection $\cJ$ takes values. Clearly this algebra must include $\hat \mt_{h,n}$ in which $\cJ_\text{DHS} $ is valued. Two options present themselves. 

\sm

In a first option, the connection $\cJ$ takes values in $\hat \mt_{h,n}$. The components $\LL_\a$ may be obtained by integrating the system of equations (\ref{7.a.0}) with the help of the path-ordered exponential of the connection $\cJ_\text{DHS}$ and its associated polylogarithms \cite{DHoker:2023vax}. This construction produces $\LL_\alpha$ in terms of non-local expressions in DHS kernels which preserve neither the $a$-grading of $\hat \mt_{h,n}$ nor the single-valuedness of $\cJ_\text{DHS}$ on ${\rm Cf}_n(\Sigma)$.

{\setup 
\label{setup:1}
The second option is the one we shall choose and is formulated as follows. The connection takes values in the  Lie algebra $\der(\hat \mt_{h,n})$ of derivations that contains $\hat \mt_{h,n}$ as a  non-trivial subalgebra and inherits the bi-grading of $\hat \mt_{h,n}$ by $a$- and $b$-degrees in Definition~\ref{2.def:1}.
The connection is assumed to preserve  the $a$-grading  by assigning $a$-degree one to $\LL_\a$ and splitting the flatness conditions of (\ref{7.a.0}) into their parts of fixed $a$-degree,
\begin{subequations}
\label{7.a.1}
\begin{align}
\p_i \LL_\a - \p_\a J_i^{(1,0)} & = 0 
\label{7.a.1a} \\ 
\p_\alpha \LL_\b - \p_\b \LL_\a & =  0
\label{7.a.1b} \\
 [ J_i ^{(1,0)} , \LL_\a] & = 0
\label{7.a.1c} \\
[\LL_\a, \LL_\b] & =0
\label{7.a.1d} \\
 \p_\a J^{(0,1)}_i & = 0 
\label{7.a.1e} \\
 \bar \p_i \LL_\a  - [ J_i ^{(0,1)} , \LL_\a]  & =  0
\label{7.a.1f} 
\end{align}
\end{subequations}
In the sequel, we shall show that the system of equations (\ref{7.a.1}), in which $\der(\hat \mt_{h,n})$, $J^{(1,0)}_i$ and $J^{(0,1)}_i$ are considered as given and $\LL_\a$ as the unknowns,  may be solved for functions $\LL_\a$ that are single-valued on $\cT_{h,n}$ and may be expressed  in terms of DHS kernels. In particular, assuming $L_\alpha$ to have $a$-degree one will lead to a unique solution of (\ref{7.a.1}).}

{\rmk  The second option generalizes the construction for genus one in terms of meromorphic and multiple-valued Kronecker-Eisenstein integration kernels  \cite{CEE} and Tsunogai's derivation algebra \cite{Tsuongai:1995}. Here, by contrast, the connection will be single-valued and modular invariant but non-meromorphic. Its restriction to genus one will be presented in explicit form in section~\ref{sec:8}.}

\subsection{Local formulation of the flatness conditions}
\label{sec:locf}

We shall represent the derivative of a function~$f$ with respect to a modulus $\tau_\a$ in terms of a pairing of the variational derivative $\delta_{ww} f$ of section \ref{sec:intdww} against the corresponding Beltrami differential $\mu_\a(w)$ (see section II.E in \cite{DHoker:1988pdl}). We shall express the components $\LL_\a$ of $\cJ$ with the help of these Beltrami differentials as well,
\bea
\label{7.a.2}
\p_\alpha f = { 1 \over \pi} \int _\Sigma d^2 w \, \mu_\a(w) \delta_{ww} f
\hskip 1in 
\LL_\a = { 1 \over \pi} \int _\Sigma d^2 w \, \mu_\a(w) \LL_{ww} 
\eea 
where $\mu_\a(w)$'s are given by \eqref{BDs:section}.
The flatness conditions of (\ref{7.a.1a}) and (\ref{7.a.1b}) may be expressed in terms of the variational derivative $\delta_{ww}$ and the $(2,0)$-form components $\LL_{ww}$, 
\begin{subequations}
\label{7.a.3}
\begin{align}
\p_i \, \LL_{ww} - \delta_{ww} \, J_i^{(1,0)}  &= 0 
\label{7.a.3aa} \\
\delta_{vv} \, \LL_{ww} - \delta_{ww} \, \LL_{vv}  & =  0
\label{7.a.3bb}
\end{align}
\end{subequations}
while the flatness conditions of (\ref{7.a.1c}) and (\ref{7.a.1d}) become two algebraic relations, 
\begin{subequations}
\label{7.a.4}
\begin{align}
{} \big [ J_i ^{(1,0)} , \LL_{ww} \big ]  & = 0
\label{7.a.4aa} \\
{} \big [\LL_{vv}, \, \LL_{ww} \big ]  & =0
\label{7.a.4bb}
\end{align}
\end{subequations}
Equation (\ref{7.a.1e}) is equivalent to,
\bea
\label{7.a.4a}
\delta_{ww} \, J_i^{(0,1)}=0
\eea
which is the relation (\ref{3.d.4}) derived in the previous section. The local form of (\ref{7.a.1f}) is more delicate, as we have to allow for contact terms when $w$ coincides with one of the variables $x_i$. This may be seen by taking the variational derivatives of the last two flatness conditions in (\ref{2.b.3}) which, upon using  (\ref{7.a.3aa}) and (\ref{7.a.4a}), become, 
\bea
\label{7.a.4b}
\p_j \Big ( \bar \p_i L_{ww} - [  J_i^{(0,1)  } , L_{ww} ] + \pi \delta (w,x_i) J_i^{(1,0)} \Big ) =0
\eea
for all $j  \in \{ 1,\cdots, n \}$. The combination inside the parentheses is independent of $x_j$ for $j \not= i$ since it is a scalar in those $x_j$ and is an anti-holomorphic $(0,1)$ form in $x_i$ which will be shown to vanish in the next subsection. As a result, the local version of (\ref{7.a.1f}) is,
\bea
\label{7.a.5}
\bar \p_i \LL_{ww}  - \big [ J_i ^{(0,1)} , \LL_{ww} \big ] + \pi \delta (w,x_i) J_i^{(1,0)} =0
\eea
Equation (\ref{7.a.1f}) is recovered by integrating (\ref{7.a.5}) against Beltrami differentials $\mu_\a(w)$ which are known to  vanish in the open subsets $U_i$'s where the $x_i$'s belong, given the choices made.

\sm

{\setup
\label{setup:2}
The flatness conditions  (\ref{7.a.3}), (\ref{7.a.4}), (\ref{7.a.4a}) and (\ref{7.a.5}) are equivalent to the flatness conditions of (\ref{7.a.1}) stated in Schemata \ref{setup:1}. They will constitute a convenient starting point for the construction of the connection $\cJ$ which we shall undertake in the subsections below and in subsequent sections.}

\subsection{Integrating the differential equations for $\LL_{ww}$}
\label{sec:4.3}

We begin by solving the system of differential equations (\ref{7.a.3aa}) which determines the $x_i$-dependence of $L_{ww}(\xx| \btau)$ from the variational derivative (\ref{3.d.2}) of the components $J_i^{(1,0)}$ in~(\ref{2.b.5}). The general solutions may be decomposed as follows, 
\bea
\label{7.b.1}
L_{ww}(\xx| \btau) = L^0_{ww}( \btau) + L^1_{ww}(\xx| \btau) 
\eea 
where $L_{ww}^0(\btau)$ is independent of $\xx$ and takes values in $\der(\hat \mt_{h,n})$,  while $L^1_{ww}(\xx| \btau)$ is given by the proposition below. 

{\prop
\label{7.prop:1}
A particular  solution for  $L_{ww}^1 (\xx| \btau)$ to the inhomogeneous linear differential equation (\ref{7.a.3aa}) that takes values in $\hat \mt_{h,n}$ is given by, 
\bea
\label{7.b.2}
L_{ww} ^1 (\xx| \btau) = \cV_{ww} (\xx| \btau) + \cV^1_{ww}(\btau)
\eea 
where each term may be expressed in terms of DHS kernels,
\begin{subequations}
\label{7.b.9}
\begin{align}
\cV_{ww} (\xx| \btau) & =  
- \sum_k \theta \bG(w,x_k;B_k) \Big (  \bPhi^J(w;B_k) a_{kJ} 
+ \half \sum_{\ell  \not= k}  \theta \bG(w,x_\ell; B_\ell) t_{k\ell} \Big )
\label{7.b.9aa} \\
\cV^1 _{ww} (\btau)  & =
- \om^J(w)  \sum_k \bigg \{ { 1 \over h} \sum _{\vI \not= \emptyset}  \p_w \Phi _{J \vI}{}^K(w) B_k^ \vI \, a_{kK}  
+ {1\over h-1} \p_w \Phi _J{}^K(w) a_{kK} \bigg \}
\label{7.b.9bb}
\end{align}
\end{subequations}
The expression (\ref{7.b.9bb}) is valid for $h \geq 2$ while we define $\cV_{ww}^1=0$ for $h=1$.}

\sm 

Note that the term $\cV^1_{ww}(\btau)$, which is independent of $\xx$, was included into (\ref{7.b.2}) to render
the particular solution $L_{ww} ^1 (\xx| \btau) $ meromorphic in $w$ as will be proven in Corollary \ref{7.cor:1} below, and curl-free as will be proven in Corollary \ref{7.cor:2}.

\begin{proof}
To prove the proposition,  we use the expression for the variational derivative of $J_i^{(1,0)}$ computed in (\ref{3.d.2}) and the expression  for $J_i^{(1,0)}$ in terms of generating functions of (\ref{2.b.5})  to recast (\ref{7.a.3aa})  in the following form,
\bea
\label{7.b.8}
\p_i \, \LL_{ww} (\xx| \btau) = - \p_i \, \theta \bG(w,x_i;B_i) \, \Big ( \bPhi^J(w; B_i) \, a_{iJ} + \sum_{j \neq i} \theta \bG (w,x_j;B_j) \, t_{ij}  \Big ) 
\eea
Using the decomposition (\ref{2.b.3b}) and (\ref{2.b.3a}) of the generating functions $\bPhi$ and $\bG$ in terms of DHS kernels, one  verifies that $\p_i \cV_{ww}(\xx|\btau)$ obtained from (\ref{7.b.9aa}) equals the right side of  (\ref{7.b.8}) so that $\p_i \big ( L_{ww} - \cV_{ww} \big )=0$. Since $\cV_{ww}$ is manifestly single-valued in $\xx$ and $L_{ww}$ is assumed to be single-valued, the differential equation $\p_i \big ( L_{ww} - \cV_{ww} \big )=0$ on the compact Riemann surface $\Sigma$ is solved by   $L_{ww} - \cV_{ww} $ being independent of $\xx$.  Since $\cV^1_{ww}$ is independent of $\xx$, there exists an $\xx$-independent $(2,0)_w$ form $L_{ww}^0$ such that $L_{ww}$ given by (\ref{7.b.1}) solves~(\ref{7.a.3aa}). 
\end{proof}

{\cor
\label{777.cor:1}
The decomposition of $L_{ww}$ into $L_{ww}^0$ and a $\hat \mt_{h,n}$-valued solution $ L_{ww}^1$ to (\ref{7.a.3aa}) is unique modulo shifts $(L_{ww}^0, L_{ww}^1) \to (L_{ww}^0 + \ell_{ww}, L_{ww}^1 -\ell_{ww})$ for an arbitrary $\xx$-independent $(2,0)_w$-form $\ell_{ww}$ that takes values in $\hat \mt_{h,n}$.}

\begin{proof}
The shift by an arbitrary $\xx$-independent $(2,0)_w$ form $\ell_{ww}$ with values in $\hat \mt_{h,n}$ leaves the decomposition (\ref{7.b.1}) invariant. Conversely, the $\xx$-dependence of $L_{ww}$ is entirely contained in $L_{ww}^1$ so that any shift  by $\ell_{ww}$ must be $\xx$-independent. Since $L_{ww}^1$ takes values in $\hat \mt_{h,n}$ so must $\ell_{ww}$, which establishes uniqueness modulo shifts as specified above. 
\end{proof}

The result of the following corollary will play a key role in solving the second differential flatness condition (\ref{7.a.3bb}) below. 

{\cor
\label{7.cor:1}
The component $L^1_{ww}(\xx| \btau)$ in (\ref{7.b.2}) and (\ref{7.b.9}) is a meromorphic $(2,0)$ form in $w$ whose only singularities in $w$ are simple poles at the variables $x_k$ with residues $J_k^{(1,0)}$, which is regular as $x_j \rightarrow x_i$ and which satisfies,
\bea
\label{7.c.4}
\pbw \LL_{ww}^1 (\xx| \btau) = \pi \sum_k \delta (w,x_k) \, J_k^{(1,0)} (\xx| \btau)
- \pi \, \delta_{h,1} \, \kappa (w) \sum _k \om^K(w)  a_{kK} 
\eea \sm}

\begin{proof}
\vskip -0.3in
To prove (\ref{7.c.4}), we evaluate $\pbw \cV_{ww}$ and $\pbw \cV^1 _{ww} $ with the help of (\ref{2.gen.3}),
 \bea
 \label{7.c.5a}
\pbw \cV_{ww} (\xx| \btau) & = & \pi \sum_k \delta (w,x_k) J^{(1,0)}_k (\xx| \btau) 
- \pi \kappa (w) \sum _k \bPhi^K(w;B_k) a_{kK} 
\no \\
\pbw \cV^1 _{ww} (\btau) & = &  \pi \kappa (w) \sum _k  \bPhi^K(w; B_k) \, a_{kK} 
\eea
where the second relation holds for $h \geq 2$ while $\cV^1 _{ww}$  vanishes for $h=1$. 
The sum of these contributions gives (\ref{7.c.4}) for all $h\geq 1$. The simple poles in
$w$ at $x_k$ with residues $J_k^{(1,0)}$ are exhibited by the delta distributions
in (\ref{7.c.4}), and regularity as $x_j \rightarrow x_i$ is manifest from (\ref{7.b.9}) since none of
the integration kernels depend simultaneously on $x_i$ and $x_j$.
\end{proof}

{\cor
\label{7.cor:2}
The component $L^1_{ww}(\xx| \btau)$ in (\ref{7.b.2}) and (\ref{7.b.9}) satisfies,
\bea
\label{7.c.5}
\delta_{vv} \LL_{ww}^1 - \delta_{ww} \LL^1 _{vv} =0
\eea}

\begin{proof}
To prove (\ref{7.c.5}) we decompose $J_i^{(1,0)}(\xx, \btau)$ in (\ref{2.b.4}) as follows,
\bea
\label{7.c.6}
J_i^{(1,0)} (\xx| \btau) = \om^K(x_i) \, a_{iK} + \p_i \cW(\xx| \btau)
\eea
where $\cW(\xx|\btau)$ is a single-valued function of $\xx \in {\rm Cf}_n(\Sigma)$ given by,
\bea
\label{7.c.7}
\cW(\xx|\btau) = \sum_k \sum _{\vI \not = \emptyset} \Phi _\vI {}^K(x_k) B_k^\vI \, a_{kK}
+ \half \sum _{k \not= \ell} \sum_\vI \cG_\vI (x_k, x_\ell) B_k^\vI \, t_{k \ell}
\eea 
With the help of the variational derivatives given in (\ref{7.c.8}), and the expressions for $\cV_{ww}$ and $\cV_{ww}^1$ in (\ref{7.b.9}), we obtain the following expression for $\LL^1_{ww}$, 
\bea
\label{7.c.9}
\LL_{ww}^1(\xx| \btau) = \delta_{ww} \cW(\xx| \btau) - \cU_{ww}(\xx|  \btau) 
\eea
where $\cU_{ww}$ is given by,
\bea
\label{7.c.10}
\cU_{ww}(\xx| \btau) & = & 
\sum_k \Big \{ \p_w \cG(w,x_k) \om^K(w) +c_h \,  \om^J(w) \p_w \Phi _J{}^K(w) \Big \} a_{kK}
\eea
with $c_1=0$ and $c_h=1/(h-1)$ for $h \geq 2$. 
One readily verifies that the  function $\cU_{ww}$ satisfies the variational curl condition, 
\bea
\label{7.c.11}
\delta _{vv} \, \cU_{ww} - \delta _{ww} \, \cU_{vv}=0
\eea
Using the variational derivative identity $\delta_{vv} \delta_{ww} \cW = \delta_{ww} \delta_{vv} \cW$ 
of (\ref{comdww}), together with (\ref{7.c.11}),  establishes (\ref{7.c.5}) and completes the proof of the corollary.
\end{proof}

{\rmk
The $(2,0)_w$ form $\cU_{ww}$ in (\ref{7.c.10}) may be expressed as the $\delta_{ww}$ variational derivative of a function  that  is neither single-valued on ${\rm Cf}_n(\Sigma)$ nor well-defined on $\cM_{h}$.  Thus, we expect $\cU_{ww}$ to correspond to a non-trivial cohomology class in $H^1(\cM_h, \CC)$ with values in $\hat \mt_{h,n}$. }

{\setup
\label{setup:3}
At this point in the development of the paper, we have obtained the full $\xx$-dependence of the component $L_{ww}$ of the connection, if it exists, in the form of $L^1_{ww}$ given in (\ref{7.b.2}) and (\ref{7.b.9}) and thereby integrated the differential equations (\ref{7.a.3aa}), (\ref{7.a.4a}) and (\ref{7.a.5}) for $L_{ww}$, as well as (\ref{7.c.5}) for $L^1_{ww}$.  It now remains to determine the $\xx$-independent part $L^0_{ww}$ of (\ref{7.b.1}).}

\subsection{Action of $\LL_{ww}^0$ on the generators of $\mt_{h,n}$}
\label{sec:4.4}

We shall now initiate the determination of the $\xx$-independent component $L^0_{ww}$ in the decomposition (\ref{7.b.1}) of $L_{ww}$. Our procedure is as follows.
\begin{enumerate}
\itemsep -0.03in
\item In this subsection, we derive a number of \textit{necessary conditions} on $L^0_{ww}$ required by the system of equations (\ref{7.a.4aa}) and (\ref{7.a.5});
\item In the next sections \ref{sec:6} to \ref{sec:7}, we shall prove that these necessary conditions are \textit{sufficient conditions} to solve (\ref{7.a.3bb}) and (\ref{7.a.4}), to prove the existence of $L_{ww}^0$ and to show that it takes values in $\der(\hat \mt_{h,n})$.
\end{enumerate}

\subsubsection{Action of $L^0_{ww}$ on $b_i^I$}

In this subsection, we establish the necessary conditions imposed on $L^0_{ww}$ to guarantee that equation (\ref{7.a.5}) is solved by the function $L_{ww}$ whose $\xx$-dependent constituent $L^1_{ww} = \cV_{ww}+\cV^1_{ww}$ is determined by Proposition \ref{7.prop:1}.  By applying $\bar \p_i$ to (\ref{7.b.9aa}) and using the formulas of (\ref{2.gen.3}), one verifies that the function $\cV_{ww}$ satisfies the following equation (suppressing the common argument $(\xx| \btau)$  for brevity),  
\bea
\label{7.b.11}
\bar \p_i \cV_{ww} - \big [ J_i ^{(0,1)} , \cV_{ww} \big ] +  \pi \, \delta (x_i, w) \, J_i ^{(1,0)} 
& = &  \pi \, \bar \om_I(x_i) \, \theta \bPhi ^I(w;B_i) \bPhi ^J (w;B_i) a_{iJ}
\eea
In view of the simple form (\ref{defj01}) of $J_i ^{(0,1)}$, this can be conveniently solved for 
 $[ \LL^0_{ww} (\btau) , b_i^I]$. Eliminating $\cV_{ww}$ in favor of $\LL_{ww}$, $\LL_{ww}^0$  and $\cV^1_{ww}$ using (\ref{7.b.1}) and (\ref{7.b.2}) we recover (\ref{7.a.5}) provided we impose the following action of $\LL^0_{ww}$ on the generators $b_i^I$ of $\hat \mt_{h,n}$,
\bea
\label{7.b.12}
{} \big [ \LL^0_{ww} (\btau) , b_i^I \big ]  = \theta \bPhi ^I(w;B_i) \, \bPhi ^J (w;B_i) \, a_{iJ} - \big [ \cV^1_{ww} (\btau) , b_i^I \big ] 
\eea
with $\cV^1_{ww}$ given by (\ref{7.b.9bb}).
Expanding the generating functions on the right of (\ref{7.b.12}) in powers of $b_i$ shows  the lowest $b$-degree contribution to be $ \om^I(w) \om^J(w)  a_{iJ}$. 

{\rmk 
\label{remark:4.20}
The relation (\ref{7.b.12}) reveals  that $\LL^0_{ww}$ cannot be valued in $\hat \mt_{h,n}$ as this algebra contains no element with bi-grading $(1,-1)$ that would be required to satisfy the relation (\ref{7.b.12}) to lowest $b$-degree, namely $  [ \LL^0_{ww} (\btau) , b_i^I  ] = \om^I(w) \om^J(w) a_{iJ}+\cdots$ with $b$-degrees $\geq 1$ in the ellipsis. This is the simplest instance in which $L^0_{ww}$ is found to take values in a Lie algebra that must be strictly larger than $\hat \mt_{h,n}$, and that will be identified with $\der (\hat \mt_{h,n})$ in subsequent sections.}

\subsubsection{Action of $L^0_{ww}$ on $a_{iI} $ and $t_{ij}$}

As a necessary condition for $\LL_{ww}$ in (\ref{7.b.1}) to solve the flatness condition (\ref{7.a.4aa}),
the action of $\LL^0_{ww}$ on the generators $a_{iI} $ and $t_{ij}$ of $\mt_{h,n}$ must be given by the following lemmas.
{\lem
\label{7.lem:5a}
The\footnote{The argument of the integral of the right side should not be mistaken for a local object; rather, a collection $\underline U_0=(U^0_1,\cdots,U^0_n)$ of open subsets is chosen; 
$[ \LL_{ww}^1(\xx_0| \btau) , J_i^{(1,0)} (\xx_0 |\btau) \big ]$ is relative to $\underline U_0$ and to 
$\mathcal J_{\mathrm{DHS}}^{s_{\underline{U}_0}}$; to make this into a global object, one transports to various $[ \LL_{ww}^1(\xx_0| \btau) , J_i^{(1,0)} (\xx_0 |\btau) \big ]$ attached to 
$\underline U=(U_1,\cdots,U_n)$ by applying the vertical diffeomorphism taking $\mathcal J_{\mathrm{DHS}}^{s_{\underline U}}$ to $\mathcal J_{\mathrm{DHS}}^{s_{\underline U_0}}$; the resulting partial functions are checked to glue into a global object, denoted $[ \LL_{ww}^1(\xx| \btau) , J_i^{(1,0)} (\xx |\btau) \big ]$; a similar viewpoint is understood in the 
next equation.}  following combinations are independent of $\xx$,
\begin{subequations}
\label{7.e.2A}
\begin{align}
Z^a_{iI} & =  \int _\Sigma d^2 x_i \, \bar \om_I(x_i) \, \big [ \LL_{ww}^1(\xx| \btau) , J_i^{(1,0)} (\xx |\btau) \big ]
\label{7.e.2Aaa} \\
Z^t_{ij} & = \lim_{x_j \to x_i} \big [ \LL_{ww} ^1(\xx| \btau), t_{ij} \big ]
\label{7.e.2Abb}
\end{align}
\end{subequations}}
\begin{proof}
\vskip -0.1in
To  prove $\xx$-independence of  $Z^a_{iI}$, we first note that  the variable $x_i$ has been integrated over so that $\xx$-independence may be investigated by  applying $\p_j$ for $j \not= i$ (suppressing the  arguments $\xx| \btau$),
\bea
\p_j Z^a_{iI} & = & 
\int _\Sigma d^2 x_i \, \bar \om_I(x_i) \, \Big ( \big [ \p_j \LL_{ww}^1 , J_i^{(1,0)}  \big ]
+ \big [ \LL_{ww}^1 , \p_j J_i^{(1,0)}  \big ] \Big )
\no \\ & = & 
\int _\Sigma d^2 x_i \, \bar \om_I(x_i) \, \delta_{ww} \big [ J_j^{(1,0)} , J_i^{(1,0)}  \big ] =0
\eea
In going from the first line to the second,  we have used (\ref{2.b.3bb}) to convert $\p_j J_i^{(1,0)}$ into $\p_i J_j^{(1,0)}$, integrated by parts in $x_i$, and used (\ref{7.a.3aa}) to convert $\p_i L^1_{ww} = \p_i L_{ww}$ into $ \delta _{ww} J^{(1,0)}_i$ in both terms. The vanishing of the second line follows from the flatness condition (\ref{2.b.3aa}), thereby proving $\xx$-independence of the right side of (\ref{7.e.2Aaa}). To prove $\xx$-independence of $Z^t_{ij}$ in (\ref{7.e.2Abb}), we evaluate its   derivative with respect to $x_k$ for $k \not= i,j$  using  (\ref{7.a.3aa}) to obtain,
\bea
\p_k Z^t_{ij} = \delta_{ww} \lim_{x_j \to x_i} \big [ J_k^{(1,0)} (\xx| \btau), t_{ij} \big ]
\eea
which vanishes by using the relations of (\ref{11.1}). The derivative with respect to $x_i$ gives,
\bea
\p_i Z^t_{ij}
= \delta_{ww} \lim_{x_j \to x_i} \big [ J_i^{(1,0)} (\xx| \btau) + J_j^{(1,0)} (\xx| \btau), t_{ij} \big ]
\eea
The vanishing of the right side was proven in appendix A of \cite{DHoker:2026lgg}. 
\end{proof}

{\lem
\label{7.lem:5}
Necessary\footnote{Caveats similar to those of Lemma \ref{7.lem:5a} are in order.} 
conditions on $\LL_{ww}^0$  to solve the flatness condition $[J_i^{(1,0)}, \LL_{ww} ]=0$ of (\ref{7.a.4aa}) are given by the following actions of $\LL^0_{ww}$ on $a_{iI}$ and $t_{ij}$,
\begin{subequations}
\label{7.e.2}
\begin{align}
{} \big [ \LL_{ww}^0 (\btau) , a_{iI} \big ] & =  - \int _\Sigma d^2 x_i \, \bar \om_I(x_i) \, \big [ \LL_{ww}^1(\xx| \btau) , J_i^{(1,0)} (\xx |\btau) \big ]
\label{7.e.2aa} \\
{} \big [ \LL_{ww}^0 (\btau) , t_{ij} \big ] & =  - \lim_{x_j \to x_i} \big [ \LL_{ww} ^1(\xx| \btau), t_{ij} \big ]
\label{7.e.2bb}
\end{align}
\end{subequations}}
 That these conditions, along with the action (\ref{7.b.12}) of $L^0_{ww}$ on $b_i^I$, are  
sufficient will be proven in section \ref{sec:vanLJ}.

\begin{proof}
\vskip 0in
To prove the necessity of the relations (\ref{7.e.2}) we show that they are implied by the flatness condition $[\LL_{ww}, J_i^{(1,0)} ]=0$ of (\ref{7.a.4aa}),  which may equivalently be expressed as,
\bea
\label{7.e.3}
{}  \big [ \LL_{ww}^0 (\btau) , J_i^{(1,0)} (\xx| \btau) \big ] = - \big [ \LL_{ww}^1 (\xx|  \btau) , J_i^{(1,0)} (\xx | \btau) \big ] 
\eea
The first relation (\ref{7.e.2aa}) is obtained by integrating both sides of (\ref{7.e.3}) over $x_i$ against $\bar \om_I(x_i)$, using the fact that $\LL_{ww}^0 $ is independent of $\xx$, and that the remaining integral of $J_i^{(1,0)}$ on the left gives $a_{iI}$. The second relation (\ref{7.e.2bb}) is obtained by evaluating the residue in $x_i$ at $x_j$ with $j \not= i$ on both sides, using the fact that the residue of $J_i^{(1,0)}(\xx| \btau)$ is $- t_{ij}$, that $\LL^0_{ww}(\btau)$ is independent of $\xx$ and that $\LL^1_{ww}(\xx|\btau)$ is regular as $x_j \to x_i$ by Corollary \ref{7.cor:1}. 
\end{proof}

Note that the expressions (\ref{7.b.12}) and (\ref{7.e.2}) for the action of $\LL_{ww}^0 $ on the generators of $\mt_{h,n}$ obtained as necessary conditions expose that the solution of the flatness conditions (\ref{7.a.3aa}) and (\ref{7.a.4aa}) is unique if it exists. The flatness conditions (\ref{7.a.3aa}) and (\ref{7.a.4aa}) in turn follow from our requirement on $L_\alpha$ and therefore $L_{ww}$ to have $a$-degree one, see Schemata \ref{setup:1}. On these grounds, the uniqueness of $L_{ww}$ and thus ${\cal L}$ is contingent on our assumption that its $a$-degree is~one.

\subsubsection{Holomorphicity and variational curl of $L_{ww}^0$}

Next, we shall prove two lemmas that will play a key role in the sequel.

{\lem 
\label{7.lem:6}
If the necessary conditions of (\ref{7.b.12}) and (\ref{7.e.2}) have a solution for $L^0_{ww}$, then the action of $\LL_{ww}^0$ on the generators of $\mt_{h,n}$, namely on all $X \in \{ a_{iI}, b_i^I, t_{ij}   \}$ with $I \in \{ 1,\cdots, h\}$ and $ i,j \in \{ 1,\cdots, n\}$,  is given by a holomorphic $(2,0)$ form in~$w$,  
\bea
\label{7.e.7}
\pbw \big [ L_{ww}^0 (\btau) , X \big ]=0 
\eea \sm}

\begin{proof}
\vskip -0.1in
For $X= a_{iI}$, the relation (\ref{7.e.7}) follows by substituting (\ref{7.c.4}) into the $\pbw$ derivative of (\ref{7.e.2aa}) and using the flatness conditions $[ J_k^{(1,0)}, J_i^{(1,0)} ]=0$. For $X= b_i^I$, we evaluate the left side of (\ref{7.e.7}) by taking the $\pbw$ derivative of (\ref{7.b.12}), using the second relation in (\ref{2.gen.3}) to evaluate the derivatives of the first term in (\ref{7.b.12}) and the second equation in (\ref{7.c.5a}) to evaluate the $\pbw$ derivative of $\cV^1_{ww}$, to obtain,
\bea
\pbw \big [ L^0_{ww}, b_i^I \big ] = \pi \kappa (w)  \, \bigg ( \sum _{k \not= i} B_i^I \bPhi ^K(w;B_k) a_{kK} + 
\theta \bPhi ^I(w;B_i) B_i^J a_{iJ} \bigg ) 
\eea
Commuting $B_i^I$ through $\bPhi$ in the first term, using $B_i^I a_{kK} = \delta ^I_K t_{ik}$ and then using $B_i^Ja_{iJ} = - \sum _{k \not= i} t_{ik}$, we confirm that the right side vanishes.  Finally, for $X= t_{ij}$ in (\ref{7.e.7}), we use (\ref{7.c.4}) to eliminate $\pbw L_{ww}^1$ from the $\pbw$ derivative of (\ref{7.e.2bb}) to obtain, 
\bea
\label{7.e.8}
\pbw \big [ \LL_{ww}^0 (\btau) , t_{ij} \big ] & = & - \pi \lim_{x_j=x_i} \sum_k \delta (w,x_k) \,  \big [ J_k^{(1,0)} (\xx| \btau), t_{ij} \big ]
\no \\ & = & 
- \pi  \delta (w,x_i) \,  \lim _{x_j \to x_i} \big [ J_i^{(1,0)} (\xx| \btau) + J_j^{(1,0)} (\xx| \btau), t_{ij} \big ]
\eea
The limit was shown to vanish in appendix A of \cite{DHoker:2026lgg}.
\end{proof}

{\lem 
\label{7.lem:7}
If the necessary conditions of (\ref{7.b.12}) and (\ref{7.e.2}) have a solution for $L^0_{ww}$, then the action of $\LL_{ww}^0$ on the generators of $\mt_{h,n}$, namely on all $X \in \{ a_{iI}, b_i^I, t_{ij} \}$ with $I \in \{ 1,\cdots, h \}$ and $ i,j  \in \{ 1,\cdots, n \}$,  has vanishing variational curl, 
\bea
\label{7.e.9}
\delta_{vv}  \big [ L_{ww}^0 (\btau) , X \big ] - \delta_{ww}  \big [ L_{vv}^0 (\btau) , X \big ]  =0 
\eea}

\begin{proof}
For $X= a_{iI}$, the left side of (\ref{7.e.9}) is computed  by applying the variational derivative to  (\ref{7.e.2aa}), using the relation $\delta _{ww} \bar \om _I(x_i)=0$ on the third line of (\ref{muder.2}) as well as the relations (\ref{7.c.5}) and (\ref{7.a.3aa}), to obtain\footnote{A caveat similar to that of Lemma \ref{7.lem:5a} is in order.} , 
\bea
\delta_{vv}  \big [ L_{ww}^0(\btau) ,  a_{iI}  \big ] - \delta_{ww}  \big [ L_{vv}^0 (\btau) ,  a_{iI}  \big ] 
= 
\int _\Sigma d^2 x_i \, \bar \om_I(x_i) \, \p_i \big [ \LL^1_{vv} (\xx| \btau), \LL^1 _{ww} (\xx| \btau) \big ] 
\eea
which vanishes upon integration by parts. For $X= b_i^I$, equation (\ref{7.e.9}) follows from applying the variational derivative to $\bPhi$ and $\cV^1_{ww}$ in (\ref{7.b.12}) and using the second line of (\ref{3.d.1}). Finally, for $X = t_{ij}$, equation (\ref{7.e.9}) is an immediate consequence of applying the variational derivative to  (\ref{7.e.2bb}) and using the relation (\ref{7.c.5}). 
\end{proof}

{\setup \label{schem.15}
Combining the result of Lemma \ref{7.lem:7} with that of Corollary \ref{7.cor:2} implies $\delta_{vv}  \big [ L_{ww}  , X \big ] - \delta_{ww}  \big [  L_{vv}   , X \big ]  =0 $ for $X \in \{ a_{iI}, b_i^I, t_{ij} \}$ with $I\in \{ 1,\cdots, h\} $ and $ i,j \in \{ 1,\cdots, n \} $. Since the center of $\mt_{h,n}$ vanishes, this solves (\ref{7.a.3bb}) when applied to $X \in \{ a_{iI}, b_i^I, t_{ij} \}$. 
However, it will require the proof of section \ref{sec:six3} that the action (\ref{7.b.12}), (\ref{7.e.2}) 
of $L_{ww}^0$ defines a derivation of $\hat{\mt}_{h,n}$ to deduce that 
the component $L_{ww}$ of this section solves (\ref{7.a.3bb}) 
when applied to arbitrary $X \in \hat{\mt}_{h,n}$.
}

\subsection{Action of $\LL_{ww}^0$ via modular tensors $\cA$ and $\mD$}
\label{sec:4.5}

We now formulate the fundamental theorem below which gives the necessary conditions of (\ref{7.b.12}) and (\ref{7.e.2}) for the action of $\LL_{ww}^0$ on the generators $a, b$ and $t$ in terms of the modular tensors $\cA$ and $\mD$ introduced in section \ref{sec:adn}, and their variational derivatives obtained in section \ref{sec:3.3}.

{\thm
\label{7.thm:1}
If a solution to the flatness conditions (\ref{7.a.4aa}) and (\ref{7.a.5}) exists, then the action of $\LL_{ww}^0$ on the generators of $\mt_{h,n}$ may be expressed in terms of the  variational derivatives of the modular tensors $\cA$ and $\mD$ defined in (\ref{7.f.1})  as follows, 
\bea
\label{9.q.20}
{} [\LL^0_{ww} , a_{iI} ] & = &
 \half \sum_{\vP, \vQ } \bigg \{  \delta_{ww} \, \cA^M{}_{\theta(\vP) I \vQ} {}^N \, 
\big [ B_i^\vP a_{iM}, B_i^\vQ a_{iN} \big ]
- \sum_{k\not= i} 
 \delta_{ww} \mD_{ \theta(\vP) I \vQ} \,   \big [ B_i^\vP t_{ik} , B_i^\vQ t_{ik}  \big ] \bigg \}
\no \\
{} [\LL^0_{ww} , b_i^I ] & = & \sum_\vP  \delta_{ww} \, \cA^I {}_\vP {}^J \, B_i^\vP a_{iJ}
\no \\
{} [\LL^0_{ww}, t_{ij} ] & = & -  \sum_{\vP \not= \emptyset}   \delta _{ww} \, \mD_\vP \, \big [  B_i^\vP t_{ij}, t_{ij} \big ]
\eea
We have extended the definition of $\cA$ in (\ref{7.f.1}) to include the rank 2 case for all $h \geq 1$,
\bea
\label{7.f.5}
\delta _{ww} \cA ^M {} _\emptyset {}^N & = & \om^M(w) \om^N(w) 
\hskip 1in 
\cA^M{} _ \emptyset {}^N= - { 1 \over \pi} \, Y^{MN}
\eea
and the rank 3 case for $h \geq 2$,
\bea
\label{7.f.6}
\delta _{ww} \cA ^M {} _I{}^N & = & \om^M(w) \, \p_w \Phi _I{}^N(w) 
+{ 1 \over h-1} \, \om^L(w) \p_w \Phi _L{}^M (w) \, \delta ^N_I - (M \leftrightarrow N) 
\eea
while $\cA^M{}_I {}^N=0$ for genus $h=1$.  
}

\sm

The proof of the theorem will be given in appendix \ref{sec:C}. 

\sm

{\rmk 
While the expressions for $\delta _{ww} \cA$ and $\delta _{ww} \mD$ in Proposition \ref{7.prop:2} give the action of $\LL^0_{ww}$ as a holomorphic  linear combination  of bilinears in the DHS kernels $\p_w \Phi_{\vP}{}^Q(w)$ and  $\om^I(w)\om^J(w)$ with the modular tensors $\mN_{\vR}$ of (\ref{defcn}) as coefficients, it would be desirable to obtain expressions whose $w$-dependence is a linear combination of  $\om^I(w)\om^J(w)$ alone, so that, by (\ref{domij}), the moduli variations of these functions can be readily read off in terms of the corresponding variations in the period matrix.}

{\rmk \label{innerrmk}
The action of $L^0_{ww}$ on $t_{ij}$ in the third line of (\ref{9.q.20}) is by inner derivation. As a result, we may use the freedom exposed in
Corollary \ref{777.cor:1} to shift $L^0_{ww}$ by an $\xx$-independent $\ell_{ww}$ that takes values in $\hat \mt_{h,n}$. Denoting the shifted generator by $\bL^0_{ww}$, we~have,
\bea
\bL^0_{ww} = L^0_{ww} + \ell_{ww} 
\hskip 1in 
\ell_{ww} =  \half\sum_{\vI \neq \emptyset} \delta _{ww} \mD_\vI  \sum_{k \not= \ell}  B_k^\vI  \, t_{k \ell} 
\label{defbL0}
\eea
In addition to the relation $[ \bL^0_{ww}, t_{ij}]=0$, the action of $\bL^0_{ww}$ on $a_{iI}$ and $b_i^I$ is given by,
\bea
\label{7.thm.2}
{} \big [ \bL_{ww}^0 , a_{iI} \big ] & = & 
 \half \sum_{\vP, \vQ }  \delta_{ww} \, \cE^M{}_{\theta(\vP) I \vQ} {}^N \, 
\big [ B_i^\vP a_{iM}, B_i^\vQ a_{iN} \big ]
\no \\
{} \big [ \bL_{ww}^0 , b_i^I \big ] & = & 
\sum_\vP   \delta _{ww} \, \cE^I{}_\vP {}^J \,  B_i^\vP a_{iJ}
\eea
where the modular tensors $\cE$ are given by (with $ \mD_{\emptyset}=0= \mD_P$),
\begin{align}
\cE^I {} _\emptyset {}^J &= \cA^I {}_\emptyset {}^J
\hskip 0.6in
\cE^I {} _P {}^J = \cA^I {}_P {}^J
\hskip 0.6in
\cE^I {}_{K \vP L} {}^J = \cA^I {}_{K \vP L} {}^J + \delta ^I_K \, \mD_\vP \, \delta ^J_L 
 \label{22.a} 
\end{align} 
The proof is straightforward and left to the reader. }

{\rmk
Note that by the right sides of (\ref{7.thm.2}), the action of the shifted generator $\bL_{ww}^0$ does not mix the $b_i^I$, $a_{iJ}$ with different values of  $i \in \{ 1,\cdots,n \} $. This decoupling is preserved by the vanishing $\bL_{ww}^0$
action on arbitrary brackets $[b_i^I , a_{jJ}]$ with $i \neq j$ and $[b_i^I , a_{iI}]$ implied by $[ \bL_{ww}^0 , t_{ij} ] =0 $.
As will be detailed in section \ref{sec:idts}, these properties of the bracket relations (\ref{7.thm.2}) lead us to view the expansion  coefficients of $\bL_{ww}^0$ as generalizations of Tsunogai's derivations on $\mt_{1,n}$ from genus 
one to arbitrary genus.}

\subsection{Formulation in terms of moduli variations}
\label{sec:modvar}

Having obtained explicit solutions for $\LL_{ww}^1 $ in (\ref{7.b.2}) and (\ref{7.b.9}) or (\ref{7.c.9}) in terms of the generators of $\mt_{h,n}$ and for the action of  $\LL^0_{ww}$ on the generators of $\mt_{h,n}$ in (\ref{9.q.20}),  we may now pair these expressions against Beltrami differentials $\mu_\a(w)$ to express our results in terms of the components $ \LL_\a = \LL_\a ^0 + \LL^1_\a$ defined by the second equation in (\ref{7.a.2}) for each function and the moduli derivatives of the modular tensors $\cA$ and $\mD$. We note that the modular tensors $\cA$ and $\mD$ have been defined in (\ref{7.f.1}) as sections of an anti-holomorphic vector bundle (see section \ref{sec:mod}) so that the derivatives with respect to holomorphic moduli $\tau_\a$ act covariantly without the need for a connection. The results are as follows.

{\cor
\label{7.cor:30}
The action of $\LL_\a^0$ on the generators of $\mt_{h,n}$ may be expressed in terms of moduli variations of the modular tensors $\cA$ and $\mD$ defined in (\ref{7.f.1})  as follows, 
\bea
\label{7.cor.30}
{} [\LL^0_\a , a_{iI} ] & = &
{ 1 \over 2}  \sum_{\vP, \vQ } \bigg \{  \p_\a  \cA^M{}_{\theta(\vP) I \vQ} {}^N 
\big [ B_i^\vP a_{iM}, B_i^\vQ a_{iN} \big ]
-  \p_\a \mD_{ \theta(\vP) I \vQ} \sum_{k\not= i}  \big [ B_i^\vP t_{ik} , B_i^\vQ t_{ik}  \big ] \bigg \}
\no \\
{} [\LL^0_\a , b_i^I ] & = &  \sum_\vP  \p_\a \, \cA^I {}_\vP {}^J \,   B_i^\vP a_{iJ}
\no \\
{} [\LL^0_\a, t_{ij} ] & = & -    \sum_{\vP \not= \emptyset} \p_\a \mD_\vP  \, \big [  B_i^\vP \, t_{ij}, t_{ij} \big ]
\eea
The expression for $\LL^1_\a$ is similarly read off from (\ref{7.c.9}),
\bea
\label{7.cor.33}
\LL^1_\a(\xx| \btau) = { 1 \over \pi} \p_\a \cW(\xx| \btau) - { 1 \over \pi } \int _\Sigma d^2 w \, \mu_\a (w) \, \cU_{ww} (\xx | \btau)
\eea
with $\cW(\xx| \btau)$ and $\cU_{ww} (\xx | \btau)$ given by (\ref{7.c.7}) and (\ref{7.c.10}), respectively,
where it is understood that the variables $\xx$ are left unchanged under infinitesimal moduli variations, since we have consistently chosen Beltrami differentials $\mu_\a (w)$ that vanish in the neighborhood of each variable $x_i$, as was discussed after equation (\ref{7.a.5}). }

\sm 

Note that the one-form $\cL  = \sum _\a \LL_\a d \tau_\a$ that enters the connection $\cJ$ in (\ref{7.0.x}) 
is  independent of the choice of the local complex coordinates $\tau_\a$ on the moduli space $\cM_h$ since it is a linear combination of the forms $\sum_\a d \tau_\a \, \p_\a \cA $ and $\sum_\a d \tau _\a \, \p_\a \mD$, as may be seen from (\ref{7.cor.30}), which are both manifestly invariant under locally holomorphic reparametrizations of $\tau_\a$.

\newpage

\section{Uplifting  to a larger algebra}
\setcounter{equation}{0}
\label{sec:6}

In this section and the next ones, we shall prove that the $\xx$-independent part $L^0_{ww}$
of the extended flat connection (\ref{7.0.x}) acts as a derivation of $\hat \mt_{h,n}$ and that the remaining flatness conditions $[L_{ww}, J^{(1,0)}_i]=0$, $[L_{vv}, L_{ww}]=0$ and $\delta_{vv} L_{ww} - \delta_{ww} L_{vv}=0$ are satisfied. The starting point will be the action of $L^0_{ww}$ on the generators $a,b,t$, derived in Theorem \ref{7.thm:1} from necessary conditions for the flatness condition $[L_{ww}, J^{(1,0)}_i]=0$ to hold, and the relations between the modular tensors $\cA$ and $\mD$ proven in Proposition \ref{7.prop:1re}. It turns out, however, that the flatness conditions require additional relations between $\cA$ and $\mD$ which are akin to the Fay identities between DHS kernels \cite{DHoker:2024ozn}, and whose explicit form is as yet not known explicitly. To circumvent these issues, we shall proceed as follows.

\subsection{The Lie algebra $\muu_{h,n}$}

In this subsection, we  introduce an auxiliary Lie algebra $\muu_{h,n}$ which will provide an essential stepping stone
towards proving, in section \ref{sec:six3} below, that the action of $L^0_{ww}$ in Theorem \ref{7.thm:1} defines a
derivation on $\hat \mt_{h,n}$.

{\deff
\label{6.def:1}
The Lie algebra $\mathfrak{lie}(\tilde a, \tilde b, \tilde t)$ is defined to be freely generated by the elements $\tilde a_{iI}, \tilde b_i ^I$ and $\tilde t_{ij}= \tilde t_{ji}$ for $I \in \{ 1, \cdots, h \} $ and $i,j \in \{ 1, \cdots, n\} $ with $j \not= i$. The Lie algebra $\muu_{h,n}$ is the quotient of $\mathfrak{lie}(\tilde a, \tilde b, \tilde t)$  by the following set of relations for mutually distinct $i,j,k$,\footnote{The relations of (\ref{6.z.1}) are the ones used to define $\mt_{h,n}$ in Definition \ref{2.def:1}, except for the relation $ [ a_{iI},  a_{jJ}]=0$ of $\mt_{h,n}$ for $i \not= j$ whose counterpart is being omitted in the definition of  $\muu_{h,n}$.}
\begin{subequations}
\label{6.z.1} 
\begin{align}
\big [ \tilde b_i^I , \, \, \tilde b_j^J \big ] & =  0 
\label{6.z.1aa} \\
\big [\tilde b_i^I, \, \tilde t_{jk} \big] & = 0 
\label{6.z.1bb} \\
\big [\tilde b_i^I , \tilde a_{jJ} \big ]  & = \delta ^I_J \, \tilde t_{ij} 
\label{6.z.1cc} \\
\big [ \tilde b_{iI} , \tilde a^I _i \big ] & = - \sum_{j \not = i} \tilde t_{ij}   
\label{6.z.1dd} \\
\big [ \tilde a_{iI}, \tilde t_{jk} \big ] & = 0 
\label{6.z.1ee}
\end{align}
\end{subequations}
The Lie algebra $\muu_{h,n}$ is equipped with a positive bi-grading $| \tilde a_{iI} |=(1,0)$, $| \tilde b_i^I|=(0,1)$ and $ |\tilde t_{ij} | = (1,1)$ whose first and second entries are referred to as $\tilde a$-degree and $\tilde b$-degree, respectively. The degree completion of $\muu_{h,n}$ is denoted $\hat {\muu}_{h,n}$. We  define the ideal $\cI \subset \muu_{h,n}$~by,
\bea
\label{6.z.3}
\cI = \hbox{ideal generated by } [\tilde a_{iI}, \tilde a_{jJ}]  \hbox{ with $i\neq j$}
\eea
The Lie algebra  $\mt_{h,n}$ may be obtained as the quotient $\mt_{h,n} = \muu_{h,n}/\cI$ via the exact sequence
\bea
\label{6.z.3a}
0 \to \cI \to \muu_{h,n} \xrightarrow[\text{}]{\text{\, $\pi$ \, }}  \mt_{h,n} \to 0
\eea 
The projection $\pi:  \muu_{h,n} \to \mt_{h,n}$ maps the generators as follows,
\bea
\pi ( \tilde a_{iI}) =  a_{iI}
\hskip 1in 
\pi(\tilde b_i^I)=  b_i^I
\hskip 1in 
\pi(\tilde t_{ij}) =  t_{ij}
\eea
 and is an isomorphism in $a$-degrees zero and one}.

{\prop
\label{prop:6.z1}
 The generators $\tilde a$, $\tilde b$ and $\tilde t$ of $\muu_{h,n}$ satisfy the following relations for arbitrary $i,j,k,\ell$ mutually distinct, 
\begin{subequations}
\label{6.z.2}
\begin{align}
{} \big [ \tilde b_i^I+ \tilde b_j^I, \tilde t_{ij} \big ] & =0 
\label{6.z.2aa} \\
{} \big [ \tilde t_{ik}+ \tilde t_{jk}, \tilde t_{ij} \big ] & =0
\label{6.z.2bb} \\
{} \big [ \tilde t_{i\ell}, \tilde t_{jk} \big ] & =0
\label{6.z.2cc} \\
{} \big [ \tilde a_{iI}+ \tilde a_{jI}, \tilde t_{ij} \big ] &  \in \cI
\label{6.z.2dd}
\end{align}
\end{subequations}
\sm }

\begin{proof}
\vskip -0.2in
 To prove (\ref{6.z.2aa}), we use the structure relations of (\ref{6.z.1}),
\bea
{} \big [ \tilde b_i^I+ \tilde b_j^I, \tilde t_{ij} \big ]  =  
\big [ \tilde b_i^J, [ \tilde b_j^I, \tilde a_{iJ} ] \big ] + \big [ \tilde b_j^J , [\tilde b_i^I , \tilde a_{jJ} ]  \big ] 
= 
- \sum_{ k \not= i } \big [ \tilde b_j^I , \tilde t_{ik} \big ] - \sum_{ \ell \not= j } \big [ \tilde b_i^I , \tilde t_{j\ell} \big ] 
\eea
where we have used (\ref{6.z.1cc}) and $\tilde t_{ij}= \tilde t_{ji}$ to obtain the first equality and $[\tilde b_i^I, \tilde b_j^J]=0$ and (\ref{6.z.1dd}) to obtain the second equality. Only $k=j$ and $\ell=i$ contribute non-trivially to the sums in view of (\ref{6.z.1bb}) which gives (\ref{6.z.2aa}). The relation (\ref{6.z.2bb}) is obtained by taking the commutator of the first relation with $\tilde a_{kK}$ for $k \not= i,j$, while (\ref{6.z.2cc}) is obtained by taking the commutator of  (\ref{6.z.1bb})  with $\tilde a_{\ell L}$ for $\ell \not= i,j,k$. Finally, to prove (\ref{6.z.2dd}), we use (\ref{6.z.1cc}) to eliminate $\tilde t_{ij}$ in favor of $[\tilde b, \tilde a]$, and the Jacobi identity to obtain,
\bea
{} \big [\tilde a _{iI} + \tilde a_{jI}, \tilde t_{ij} \big ] & = &
\big [ \tilde b_i^J, [ \tilde a_{iJ}, \tilde a_{jI} ] \big ] + \sum _{ k \not = i} \big [ \tilde t_{ik} , \tilde a_{jI} \big ]
+ \big [ \tilde b_j^J, [ \tilde a_{jJ}, \tilde a_{iI} ] \big ] + \sum _{ \ell  \not = j} \big [ \tilde t_{j \ell } , \tilde a_{iI} \big ]
\eea
Only the terms with $k=j$ and $\ell=i$ contribute and the resulting terms are the opposite of the left side so that we obtain 
$\big [\tilde a _{iI} + \tilde a_{jI}, \tilde t_{ij} \big ] =
\thalf \big [ \tilde b_i^J, [ \tilde a_{iJ}, \tilde a_{jI} ] \big ] 
+ \thalf \big [ \tilde b_j^J, [ \tilde a_{jJ}, \tilde a_{iI} ] \big ]   \in \cI$.
\end{proof}

\subsection{$\tilde L^0_{ww}$ is a derivation of $\muu_{h,n}$}

We define a $(2,0)_w$ form $\tilde L_{ww}^0$ in terms of its action on the generators $\tilde a, \tilde b, \tilde t$ of $\muu_{h,n}$ by the relations mirroring those on $a,b,t$ in (\ref{9.q.20}), 
\bea
\label{6.z.4}
\tilde \LL^0_{ww} ( \tilde a_{iI} ) & = &
 \half \sum_{\vP, \vQ } \bigg \{  \delta_{ww} \, \cA^M{}_{\theta(\vP) I \vQ} {}^N \, 
\big [ \tilde B_i^\vP \tilde a_{iM}, \tilde B_i^\vQ \tilde a_{iN} \big ]
- \sum_{k\not= i} 
 \delta_{ww} \, \mD_{ \theta(\vP) I \vQ} \,   \big [ \tilde B_i^\vP \tilde t_{ik} , \tilde B_i^\vQ \tilde t_{ik}  \big ] \bigg \}
\no \\
\tilde \LL^0_{ww} ( \tilde b_i^I ) & = & \sum_\vP  \delta_{ww} \, \cA^I {}_\vP {}^J \,   \tilde B_i^\vP \tilde a_{iJ}
\no \\
\tilde \LL^0_{ww}( \tilde t_{ij} ) & = & -  \sum_{\vP \not= \emptyset}   \delta _{ww} \mD_\vP  \, \big [ \tilde B_i^\vP \tilde t_{ij}, \tilde t_{ij} \big ]
\eea
where $\tilde B_i^I X = [\tilde b_i^I , X]$ for arbitrary $X \in \hat{\muu}_{h,n}$, and we write $\tilde B^\vI_i = \tilde B_i^{I_1}\cdots \tilde B_i^{I_r}$ for $\vI = I_1\cdots I_r$.
The action of $\tilde L^0_{ww}$ extends to a derivation of the freely generated Lie algebra $\mathfrak{lie}(\tilde a, \tilde b, \tilde t)$. It remains to  establish that $\tilde L^0_{ww}$ acts as a derivation on $\hat{\muu}_{h,n}$.\footnote{For this reason, the action of $\tilde L^0_{ww}$ on $\hat{\muu}_{h,n}$ is being denoted by $\tilde L^0_{ww} (X)$ for $X \in  \hat{\muu}_{h,n}$ rather than $[\tilde L^0_{ww}, X]$ which would presume that $\tilde L^0_{ww}$ acts as a derivation of $\hat{\muu}_{h,n}$. } Doing so is the purpose of the remainder of the  present section and the following theorem.

{\thm
\label{6.thm:1}
$\tilde \LL_{ww}^0$ acts on $\hat{\muu}_{h,n}$ by derivation, namely, $\tilde L^0_{ww}$ takes values in $\mathfrak{der}(\hat{\muu}_{h,n})$ and preserves every structure relation of $\muu_{h,n}$ given in Definition \ref{6.def:1} for mutually distinct $i,j,k$, 
\begin{subequations}
\label{6.z.5}
\begin{align}
\Big [ \tilde \LL_{ww}^0 (\tilde b_i^I) , \tilde b_j^J \Big ] + \Big [ \tilde b_i^I, \tilde \LL_{ww}^0( \tilde b_j^J) \Big ] & =0
\label{6.z.5aa} \\
\Big [ \tilde \LL_{ww}^0 (\tilde b_i^I) , \tilde t_{jk} \Big ]  + \Big [ \tilde b_i^I,  \tilde  \LL_{ww}^0( \tilde t_{jk}) \Big ]  & = 0  
\label{6.z.5bb} \\
\Big [ \tilde \LL_{ww}^0 (\tilde b_i^I) , \tilde a_{jJ} \Big ]  + \Big [ \tilde b_i^I,  \tilde  \LL_{ww}^0( \tilde a_{jJ} )\Big ]  - \delta ^I_J \tilde L^0_{ww} (\tilde t_{ij})  & =0 
\label{6.z.5cc} \\
\Big [ \tilde \LL_{ww}^0 (\tilde b_i^I) , \tilde a_{iI} \Big ]  + \Big [ \tilde b_i^I,  \tilde \LL_{ww}^0( \tilde a_{iI} )\Big ]  + \sum_{j \not= i} \tilde L^0_{ww} (\tilde t_{ij})  & = 0
\label{6.z.5dd} \\
\Big [ \tilde \LL_{ww}^0( \tilde a_{iI}) , \tilde t_{jk}  \Big ] + \Big [ \tilde a_{iI} ,  \tilde \LL_{ww}^0( \tilde   t_{jk} ) \Big ] & =0
\label{6.z.5ee} 
\end{align}
\end{subequations}}

\begin{proof}
To prove the theorem, we proceed to proving each one of the relations of (\ref{6.z.5}) in turn, by substituting the corresponding expressions for the actions of $\tilde L^0_{ww}$ given by (\ref{6.z.4}). 

\sm

$\bullet$ \textit{(\ref{6.z.5aa})} The left side  becomes,
\bea
\label{6.b.2}
 \sum_\vP  \delta _{ww} \cA^J {}_\vP {}^K \,  \tilde B_i^I \tilde B_j^\vP \tilde a_{jK} 
- \sum_\vP  \delta _{ww} \cA^I {}_\vP {}^K \,  \tilde B_j^J \tilde B_i^\vP \tilde a_{iK} 
\eea
Using the structure relation of (\ref{6.z.1cc}) on both terms, we obtain, 
\bea
\label{6.b.3}
\sum_\vP \delta _{ww} \cA^J {}_\vP {}^I \,  \tilde B_j^\vP \tilde t_{ij} 
- \sum_\vP \delta _{ww} \cA^I {}_\vP {}^J \,  \tilde B_i^\vP \tilde t_{ij} 
\eea
For $|\vP |\geq 2$, both contributions cancel one another using the corollary $\tilde B_j ^\vP \tilde t_{ij} = \tilde B_i ^{\theta(\vP)} \tilde t_{ij}$ of (\ref{6.z.2aa}) and the antipodal property of $\cA$ in (\ref{7.f.2}). For $|\vP|=0$, the cancellation is trivial using (\ref{7.f.5}) while for $|\vP|=1$ it holds using (\ref{7.f.6}) along with (\ref{6.z.2aa}).

\sm

$\bullet$ \textit{(\ref{6.z.5bb})} The first commutator in (\ref{6.z.5bb}) vanishes using the second equality in (\ref{6.z.4}) and the relation $[\tilde B_i^\vP \tilde a_{iJ}, \tilde t_{jk}]=0$ which follows from (\ref{6.z.1bb}) and (\ref{6.z.1ee}). The second commutator vanishes using the last equality in (\ref{6.z.4}) and the identity 
$\big [ \tilde b_i^I,  [ \tilde B_j ^\vP  \tilde t_{jk}, \tilde t_{jk}  ] \big ]  =0$  for $i \not= j,k$.

\sm

$\bullet$ \textit{(\ref{6.z.5cc})} This relation requires a lengthy proof, which relies on the relations
between the modular tensors $\cA$ and $\mD$ of Proposition \ref{7.prop:1re} and is given in appendix  \ref{sec:E}.

\sm

$\bullet$ \textit{(\ref{6.z.5dd})}  The part of the left side  that involves $\cA$ is obtained by substituting the first and second lines of (\ref{6.z.4}) into the second and first commutators, respectively, and  produces two terms that are equal to one another and may be regrouped as follows, 
\bea
\label{6.d.3}
 \sum_\vP  \delta_{ww} \cA^M{}_\vP{\, }^N \,  \big [ \tilde a_{iM}, \tilde B_i^\vP \tilde a_{iN} \big ]
- \sum_{\vP, \vQ}  \delta_{ww} \, \cA^M{}_{\theta(\vP) I \vQ} {}^N \, 
\big [ \tilde B_i^\vP \tilde a_{iM}, \tilde B_i^{I\vQ} \tilde  a_{iN} \big ]
\eea
Replacing the sum over $I\vQ$ by a sum over $\vQ'$ minus the correction term corresponding to $\vQ'= \emptyset$ which is not part of the sum over $I \vQ$,  the second term of (\ref{6.d.3}) becomes, 
\bea
\label{6.d.4}
- \sum_{\vP, \vQ'}  \delta_{ww} \, \cA^M{}_{\theta(\vP)  \vQ'} {}^N \, 
 \big [ \tilde  B_i^\vP \tilde  a_{iM}, \tilde  B_i^{\vQ'} \tilde  a_{iN} \big ]
+ \sum_\vP \delta_{ww} \, \cA^M{}_{\theta(\vP)} {}^N \, 
\big [ \tilde  B_i^\vP \tilde a_{iM}, \tilde a_{iN} \big ]
\eea
The first sum vanishes as the first factor is invariant under swapping $(M, \vP) \leftrightarrow (N, \vQ')$ but the commutator changes sign. The second term in (\ref{6.d.4}) cancels the first term in (\ref{6.d.3}) so that all dependence on $\cA$ of the left side of (\ref{6.z.5dd}) cancels. It remains to collect the contributions from terms involving the modular tensor $\mD$, 
\bea
\label{6.d.5old}
- \sum_{k\not= i} \sum_{\vP , \vQ } 
 \delta_{ww} \mD_{ \theta(\vP) I \vQ} \,  \big [ \tilde B_i^\vP \tilde t_{ik} , \tilde B_i^{I \vQ} \tilde t_{ik}  \big ] 
- \sum_{k \not= i} \sum_{\vP \not= \emptyset}  \delta _{ww}  \mD_\vP  \, \Big [  \tilde B_i^\vP \tilde t_{ik}, \tilde t_{ik} \Big ]
\eea
Replacing the sum over $I \vQ$ by a sum over $\vQ'$ we recognize the second term of (\ref{6.d.5old}) as the $\vQ'= \emptyset$ contribution, so that the sum of both contributions may be regrouped as follows,
\bea
\label{6.d.5}
- \sum_{k\not= i} \sum_{\vP , \vQ' } 
 \delta_{ww} \mD_{ \theta(\vP)  \vQ' } \,  \big [ \tilde B_i^\vP \tilde t_{ik} , \tilde B_i^{\vQ'} \tilde t_{ik}  \big ] 
\eea
Since the first factor under the double sum is invariant under $\vP \leftrightarrow \vQ'$ but the commutator changes sign, the sum of (\ref{6.d.5}) vanishes.

\sm

$\bullet$ \textit{(\ref{6.z.5ee})} Substituting the expression on the first line of (\ref{6.z.4}) into the first commutator, we see that the contributions from the modular tensor $\cA$ cancel right away. The remaining contributions from  the modular tensor $\mD$ give, 
\bea
\label{6.c.2}
\half \sum_{\ell\not= i} \sum_{\vP , \vQ } 
 \delta_{ww} \mD_{ \theta(\vP) I \vQ}  \Big [ \tilde t_{jk}, \big [ \tilde B_i^\vP \tilde t_{i\ell} , \tilde B_i^\vQ \tilde t_{i\ell}  \big ]  \Big ]
\eea
Moving $\tilde t_{jk}$ through the $\tilde B_i$ factors and using $[\tilde t_{jk}, \tilde t_{i\ell}]= 0$ when $\ell \not = j,k$ shows that only the contributions from $\ell =j,k$ are non-vanishing and equal to one another. Furthermore, using
$[\tilde t_{jk}, \tilde t_{ij}]= - [\tilde t_{jk}, \tilde t_{ik} ]$, the first commutator of  (\ref{6.z.5ee}) becomes,
\bea
\label{6.c.3}
- \sum_{\vP , \vQ } 
 \delta_{ww} \mD_{ \theta(\vP) I \vQ}   \Big [ \tilde t_{jk}, \big [ \tilde B_i^\vP \tilde t_{ij} , \tilde B_i^\vQ \tilde t_{ik}  \big ]  \Big ]
\eea
Substituting the last line of (\ref{6.z.4}) into the second commutator of  (\ref{6.z.5ee}), we obtain, 
\bea
\label{6.c.4}
- \sum_\vP  \delta _{ww} \mD_\vP  \, \Big [ \big [ \tilde a_{iI} , \tilde B_j^\vP \tilde t_{jk} \big ] , \tilde t_{jk} \Big ]
= - \sum_{\vR, \vQ}  \delta _{ww}\mD_{\vR I \vQ}  \,
\Big [ \tilde B_j^\vR \tilde B_k^{\theta(\vQ)} [  \tilde t_{jk}, \tilde t_{ij}] , \tilde t_{jk} \Big ]
\eea
where we have used the extension of the $\mathfrak{t}_{h,n}$ identity in (\ref{A.Lie.2}) to an identity in $\muu_{h,n}$  to evaluate the inner commutator.\footnote{The rewriting in (\ref{6.c.4}) relies on the extension to $\muu_{h,n}$ of the $\mathfrak{t}_{h,n}$ identity in (\ref{A.Lie.2}) with $X= t_{jk}$ and only requires the structure relations of Definition \ref{6.def:1}, without generating any terms in the ideal (\ref{6.z.3}).} Using (\ref{6.z.2bb}) in the inner commutator,  moving the factors of $\tilde B_j$ and $\tilde B_k$ back into the inner commutator, converting them to $\tilde B_i$, and relabeling $\vR = \theta(\vP)$, the second commutator of (\ref{6.z.5ee}) becomes, 
\bea
\label{6.c.6}
 \sum_{\vP, \vQ}  \delta _{ww}  \mD_{\theta(\vP) I \vQ} \, 
\Big [  \tilde t_{jk},  [ \tilde B_i^\vP \tilde t_{ij}, \tilde B_i^\vQ \tilde t_{ik} ]  \Big ]
\eea
The simplified expressions (\ref{6.c.3}) and (\ref{6.c.6}) for the two commutators in (\ref{6.z.5ee}) add up to zero
which completes our proof of the final relation in (\ref{6.z.5}). 
\end{proof}

\newpage

\section{Establishing $\LL_{ww}^0 \in \mathfrak{der}(\mt_{h,n})$}
\setcounter{equation}{0}
\label{sec:66}

In this section, we shall establish the following results, in the order of the subsections.
\begin{enumerate}
\itemsep =-0.03in
\item Prove the existence and uniqueness of lifting the $\hat \mt_{h,n}$-valued components $J^{(1,0)}_i, J_i^{(0,1)}$ and $L^1_{ww}$ of the connection $\cJ$ to $\hat {\muu}_{h,n}$-valued components $\tilde J^{(1,0)}_i, \tilde J_i^{(0,1)}$ and $\tilde L^1_{ww}$;
\item Prove that $\tilde L_{ww} (\tilde J^{(1,0)}_i)$ is contained in the ideal $\cI$ defined by (\ref{6.z.3});
\item Prove $\tilde L^0_{ww} ([\tilde a_{iI}, \tilde a_{jJ} ]) \in \cI$, so that $\tilde L^0_{ww}$ and $\tilde L_{ww}$  project to derivations of $\hat \mt_{h,n}$;
\item Deduce $[L_{ww}  , J^{(1,0)}_i] = 0$ from the first three results.
\end{enumerate}
These results will furthermore
set the stage for proving the remaining flatness conditions  $\delta_{vv} L_{ww} - \delta_{ww} L_{vv}=0$ and $[L_{vv}, L_{ww}]=0$ as well as modularity of $L_{ww}$ in section \ref{sec:7}.

\subsection{Lifting $J^{(1,0)}_i, J_i^{(0,1)}$ and $L^1_{ww}$ from $\hat \mt_{h,n}$ to $\hat{\muu}_{h,n}$}

We begin the proof of the existence and uniqueness of the lift to $\hat{\muu}_{h,n}$ of the component of the 
extended flat connection $\cJ$ (\ref{7.0.x}) with the following results.

\sm

Since the lift $\muu_{h,n} \to \mt_{h,n}$ is an isomorphism in $a$-degrees zero and one, there exists a unique lift 
of the relations of Definition \ref{3.def:3}, which is given by,
\bea
\delta_{ww} \, \tilde a_{iI} =0 
\hskip 1in 
\delta _{ww} \, \tilde b_i^I =0 
\hskip 1in 
\delta_{ww} \, \tilde t_{ij}=0
\eea
The isomorphism also implies that  there exists a unique lift of the $a$-degree zero component $ J_i^{(0,1)} $ to $\tilde J_i^{(0,1)}$ and of the $a$-degree one component $ J^{(1,0)}_i$ to $ \tilde J^{(1,0)}_i$, given as follows,
\begin{subequations}
\label{66.a.2}
\begin{align}
\tilde J^{(0,1)} _i & = 
- \pi \bar \om_I(x_i) \, \tilde b_i^I
\label{66.a.2aa} \\
\tilde J_i^{(1,0)}  &=  
\bPhi^J(x_i; \tilde B_i) \, \tilde a_{iJ} + \sum_{j \neq i} \bG (x_i,x_j; \tilde B_i) \, \tilde t_{ij}  
\label{66.a.2bb} 
\end{align}
\end{subequations} 
The flatness conditions of (\ref{2.b.3})  for the DHS connection uplift  as follows, 
\begin{subequations}
\label{89.b.3}
\begin{align}
{}  [ \tilde J^{(1,0)} _i , \tilde J^{(1,0)} _j  ] & \in \cI
\label{89.b.3aa} \\
\p_i \tilde J^{(1,0)} _j - \p_j \tilde J^{(1,0)} _i & =  0
\label{89.b.3bb} \\
\bar \p_j \tilde J^{(1,0)} _i - [\tilde J^{(0,1)}_j , \tilde J^{(1,0)} _i] & =  \pi \delta(x_i, x_j) \, \tilde t_{ij}  &  j & \not = i
\label{89.b.3cc} \\
\bar \p_i \tilde J^{(1,0)} _i - [\tilde J^{(0,1)}_i , \tilde J^{(1,0)} _i] & =  - \sum_{k \not= i} \delta(x_i, x_k) \tilde t_{ik}
\label{89.b.3dd}
\end{align}
\end{subequations}
In particular, we have,
\bea
\delta_{ww} \tilde J^{(0,1)}_i=0
\eea

The lemma below gives the uplift of the component $L^1_{ww}$ to $\tilde L^1_{ww}$. Combined with the action of $\tilde L^0_{ww}$ on $\muu_{h,n}$ given in (\ref{6.z.4}), this provides $\tilde L_{ww}$, as follows,
\bea
\tilde L_{ww} = \tilde L^0_{ww} + \tilde L^1_{ww}
\eea
{\lem
\label{66.lem:1}
There exists a unique lift of the $\hat \mt_{h,n}$-valued component $L^1_{ww}$ to the $\hat{\muu}_{h,n}$-valued component $\tilde L_{ww}^1$ such that the following relations are satisfied in $\hat{\muu}_{h,n}$,
\begin{subequations}
\label{66.a.1}
\begin{align}
\p_i \tilde L_{ww}^1 - \delta_{ww} \tilde J_i^{(1,0)} & = 0 
\label{66.a.1aa} \\
\delta_{vv} \tilde L_{ww}^1 - \delta_{ww} \tilde L_{vv} ^1 & =0
\label{66.a.1bb} \\
\bar \p_i \tilde L_{ww}^1  - \big [ \tilde J^{(0,1)}_i , \tilde L_{ww}^1 \big ] - \pi \bar \om_I(x_i) \tilde L^0_{ww} (\tilde b_i^I) & = -  \pi \delta(w,x_i) \tilde J_i^{(1,0)} 
\label{66.a.1dd} 
\end{align}
\end{subequations}
where $\tilde L^0_{ww} (\tilde b_i^I) $ was given in (\ref{6.z.4}). 
The explicit expressions for $\tilde L_{ww}^1$ is as follows,
\bea
\label{66.a.2cc}
\tilde L_{ww} ^1 & = & \tilde \cV_{ww} + \tilde \cV_{ww}^1 
\eea
where the constituents $\tilde \cV_{ww}$ and $ \tilde \cV_{ww}^1 $ are given by,
\bea
\label{66.a.8}
\tilde \cV_{ww} & = & - \sum_k \theta \bG(w,x_k; \tilde B_k) \Big (  \bPhi^J(w; \tilde B_k) \tilde a_{kJ} 
+ \half \sum_{\ell  \not= k}  \theta \bG(w,x_\ell; \tilde B_\ell) \tilde t_{k\ell} \Big )
\no \\ 
 \tilde \cV_{ww}^1 & =&  
- \om^J(w)  \sum_k \bigg \{ { 1 \over h} \sum _{\vI \not= \emptyset}  \p_w \Phi _{J \vI}{}^K(w) \tilde B_k^ \vI \, \tilde a_{kK}  
+ c_h\p_w \Phi _J{}^K(w) \tilde a_{kK} \bigg \}
\eea
where $c_1=0$ and $c_h=1/(h-1)$ for $h \geq 2$.}

\sm

Equations (\ref{66.a.1}) and (\ref{66.a.2cc})  are the uplifts of the relations (\ref{7.a.3}),  (\ref{7.a.5}) and the combined equations (\ref{7.b.2}) with (\ref{7.b.9}), respectively. 

\begin{proof}
To prove (\ref{66.a.1aa}), we use the fact that the variational derivative (\ref{3.d.2}), by its $a$-degree
one, straightforwardly lifts to,
\bea
\label{66.b.8}
\delta _{ww} \tilde J_i^{(1,0)}  = - \p_i \, \theta \bG(w,x_i;\tilde B_i) \, \Big ( \bPhi^J(w; \tilde B_i) \, \tilde a_{iJ} + \sum_{j \neq i} \theta \bG (w,x_j; \tilde B_j) \, \tilde t_{ij}  \Big ) 
\eea
and verify that it equals $\p_i \tilde L^1_{ww}$ obtained from differentiating (\ref{66.a.2cc}), as was done in the
proof of Proposition \ref{7.prop:1}.  To prove (\ref{66.a.1bb}), we closely follow the calculations used to prove Corollary \ref{7.cor:2}, which exclusively involves expressions of $a$-degree one. More specifically, we decompose $\tilde J_i^{(1,0)}$ and $\tilde L^1_{ww}$, which are defined in (\ref{66.a.2bb}), and (\ref{66.a.2cc}), respectively, as follows,
\bea
\label{66.c.6}
\tilde J_i^{(1,0)}  & = & \om^K(x_i) \, \tilde a_{iK} + \p_i \tilde \cW
\no \\
\tilde \LL_{ww}^1 & = & \delta_{ww} \tilde \cW - \tilde \cU_{ww} 
\eea
where $\tilde \cW$  and $\tilde \cU_{ww}$ are given by the following uplifts of (\ref{7.c.7}) and (\ref{7.c.10}),
\bea
\label{66.c.7}
\tilde \cW & = & \sum_k \sum _{\vI \not = \emptyset} \Phi _\vI {}^K(x_k) \tilde B_k^\vI \, \tilde a_{kK}
+ \half \sum _{k \not= \ell} \sum_\vI \cG_\vI (x_k, x_\ell) \tilde B_k^\vI \, \tilde t_{k \ell}
\no \\
\tilde \cU_{ww}& = & 
\sum_k \Big \{ \p_w \cG(w,x_k) \om^K(w) +c_h \,  \om^J(w) \p_w \Phi _J{}^K(w) \Big \} \tilde a_{kK}
\eea 
and we reiterate that $c_1=0$ and $c_h=1/(h-1)$ for $h \geq 2$. The last line of (\ref{66.c.6}) is obtained by evaluating $\delta_{ww} \tilde \cW$ with the help of the variational derivatives given in (\ref{7.c.8}) to express $\tilde \LL^1_{ww}$ with the help of  $\tilde \cU_{ww}$. One verifies that the  function $\tilde \cU_{ww}$ satisfies the variational curl condition $\delta _{vv} \, \tilde \cU_{ww} = \delta _{ww} \, \tilde \cU_{vv}$ which, together with the variational derivative identity $\delta_{vv} \delta_{ww} \tilde \cW = \delta_{ww} \delta_{vv} \tilde \cW$ of (\ref{comdww}), establishes  (\ref{66.a.1bb}).
To prove (\ref{66.a.1dd}), one evaluates the combination $\bar \p_i \tilde L_{ww}^1  - \big [ \tilde J^{(0,1)}_i , \tilde L_{ww}^1 \big ] +  \pi \delta(w,x_i) \tilde J_i^{(1,0)} $ and shows that the result is given by $\pi \bar \om_I(x_i) \tilde L^0_{ww} (\tilde b_i^I) $ of (\ref{6.z.4}) using the variational derivative of $\cA$ given in (\ref{7.prop.2a}) of Proposition~\ref{7.prop:2}. 
\end{proof}

\subsection{Proof of $[\tilde L_{ww}, \tilde J_i^{(1,0)}]   \in  \cI$}

The relations established in Lemma \ref{7.lem:5}  may be uplifted to the action of $\tilde \LL^0_{ww}$ on the generators $\tilde a_{iI} $ and $\tilde t_{ij}$ of $\muu_{h,n}$ as given by the lemmas below.

{\lem
\label{66.lem:5a}
The\footnote{A caveat similar to that of Lemma \ref{7.lem:5a} is in order.}  $\xx$-dependence of $\tilde Z^a_{iI}$ and $\tilde Z^t_{ij}$, defined below, belongs to $\cI$,
\begin{subequations}
\label{66.e.2A}
\begin{align}
\tilde Z^a_{iI} & =  \int _\Sigma d^2 x_i \, \bar \om_I(x_i) \, \big [ \tilde L_{ww}^1 (\xx| \btau) , \tilde J_i^{(1,0)} (\xx|\btau) \big ]
\label{7.e.2alt} \\
\tilde Z^t_{ij} & = \lim_{x_j \to x_i} \big [ \tilde L_{ww}^1 (\xx| \btau), \tilde t_{ij} \big ]
\label{66.e.2Abb}
\end{align}
\end{subequations}}

\begin{proof}
\vskip -0.1in
The analytic parts of the proof are identical to those in the proof of Lemma \ref{7.lem:5a}, but the algebraic parts significantly differ.  To investigate the $\xx$-dependence of  $\tilde Z^a_{iI}$, we note that  the variable $x_i$ has been integrated over so that $\tilde Z^a_{iI}$ is independent of $x_i$. The dependence on the remaining $x_j$ may be obtained by  applying $\p_j$ for $j \not= i$,
\bea
\p_j \tilde Z^a_{iI} & = & 
\int _\Sigma d^2 x_i \, \bar \om_I(x_i) \, \Big ( \big [ \p_j \tilde \LL_{ww}^1 , \tilde J_i^{(1,0)}  \big ]
+ \big [ \tilde \LL_{ww}^1 , \p_j \tilde J_i^{(1,0)}  \big ] \Big )
\no \\ & = & 
\int _\Sigma d^2 x_i \, \bar \om_I(x_i) \, \delta_{ww} \big [ \tilde J_j^{(1,0)} , \tilde J_i^{(1,0)}  \big ] 
\eea
In going from the first line to the second,  we have used (\ref{89.b.3bb}) to convert $\p_j \tilde J_i^{(1,0)}$ into $\p_i \tilde J_j^{(1,0)}$, integrated by parts in $x_i$, and used (\ref{66.a.1aa}) to convert $\p_i \tilde L^1_{ww} $ into $ \delta _{ww} \tilde J^{(1,0)}_i$ in both terms. The second line belongs to $\cI$ as a result of the flatness condition (\ref{89.b.3aa}).  
To investigate the $\xx$-dependence of $\tilde Z^t_{ij}$, we evaluate its   derivative with respect to $x_k$ for $k \not= i,j$  using  (\ref{66.a.1aa})  to obtain,
\bea
\p_k \tilde Z^t_{ij} = \delta_{ww} \lim_{x_j \to x_i} \big [ \tilde J_k^{(1,0)} (\xx| \btau), \tilde t_{ij} \big ]
\eea
which vanishes by using (\ref{6.z.1}) and (\ref{6.z.2}). The vanishing of the derivative with respect to $x_i$, modulo terms in $\cI$,  is proven in appendix \ref{sec:F}. 
\end{proof}

{\lem
\label{66.lem:25}
The\footnote{A caveat similar to that of Lemma \ref{7.lem:5a} is in order.}  
actions of $\tilde \LL_{ww}^0$  on $\tilde a_{iI}$ and $\tilde t_{ij}$ satisfy the following relations,
\begin{subequations}
\label{66.b.1}
\begin{align}
{} \big [ \tilde \LL_{ww}^0 (\btau) , \tilde a_{iI} \big ]  + \int _\Sigma d^2 x_i \, \bar \om_I(x_i) \, \big [ \tilde \LL_{ww}^1(\xx| \btau) , \tilde J_i^{(1,0)} (\xx|\btau) \big ]  & \, \in \, \cI
\label{66.b.1aa} \\
{} \big [ \tilde \LL_{ww}^0 (\btau) , \tilde t_{ij} \big ]  + \lim_{x_j=x_i} \big [ \tilde \LL_{ww} ^1(\xx| \btau), \tilde t_{ij} \big ] & \,  \in \, \cI
\label{66.b.1bb}
\end{align}
\end{subequations}}

The proof of this lemma heavily relies on Lemma \ref{66.lem:5a} and is given in appendix \ref{sec:H}.

{\cor
\label{66.cor:1}
The\footnote{A section $s_{\underline{U}_0}$ is understood.} commutator $[\tilde L_{ww}, \tilde J_i^{(1,0)}]$ satisfies the following relations,
\begin{subequations}
\label{66.cor.1}
\begin{align}
\bar \p_i \big [ \tilde L_{ww} , \tilde J_i^{(1,0)} \big ] - \Big [ \tilde J_i^{(0,1)}, \big [ \tilde L_{ww} , \tilde J_i^{(1,0)} \big ]  \Big ] \in \cI
\label{66.cor.1aa} \\
\int _\Sigma d^2 x_i \, \bar \om_I(x_i) \big [ \tilde L_{ww}, \tilde J_i^{(1,0)} \big ] \in \cI
\label{66.cor.1bb}
\end{align}
\end{subequations}
for $\xx \in \Sigma^n$ rather than $\xx \in {\rm Cf}_n(\Sigma)$, namely, all its terms, including those multiplying $\delta$-functions, are in $\cI$.}

\begin{proof}
The notation $[\tilde L_{ww} , \tilde J_i^{(1,0)} ]$ for $[\tilde L_{ww}^1 , \tilde J_i^{(1,0)} ]
+ \tilde L_{ww}^0(\tilde J_i^{(1,0)} )$  is justified in both relations above since  $\tilde L_{ww}^0$ belongs to $\mathfrak{der}(\hat {\muu}_{h,n})$ by Theorem \ref{6.thm:1} and $\tilde J_i^{(1,0)}$ belongs to $\hat {\muu}_{h,n}$. Working out (\ref{66.cor.1aa}) with the help of (\ref{66.a.1dd}), we obtain in parallel to the proof of Lemma \ref{7.lem:5},
\bea
\label{66.a.4}
\bar \p_i \big [ \tilde \LL_{ww}, \tilde J_i^{(1,0)} \big ] 
- \Big [ \tilde J_i^{(0,1)} , \big [ \tilde \LL_{ww}, \tilde J_i^{(1,0)} \big ]  \Big ] 
=
- \pi \sum_{j \not= i} \delta (x_i, x_j) \big [ \tilde \LL_{ww}^0 + \tilde L^1_{ww} , \tilde t_{ij} \big ] 
\eea
The right side belongs to $\cI$ in view of (\ref{66.b.1bb}) of Lemma \ref{66.lem:25}, which proves the corollary.
\end{proof}

We shall now prove the following central theorem.
{\thm 
\label{66.thm:1}
The\footnote{A section $s_{\underline{U}_0}$ is understood.} commutator $[\tilde L_{ww}, \tilde J_i^{(1,0)}]$ satisfies,
\bea
\label{66.thm.1}
[\tilde L_{ww}, \tilde J_i^{(1,0)}]  \, \in \, \cI
\eea}

\begin{proof}
\vskip -0.02in
Combining the relations in (\ref{66.cor.1}) with the fact that the commutator $\big [ \tilde \LL_{ww}, \tilde J_i^{(1,0)} \big ] $ is a Lie series in $\tilde b$ and a single-valued $(1,0)$ form in $x_i$, we proceed to proving the theorem by contradiction. Assuming that (\ref{66.thm.1}) does not hold, so that $\big [ \tilde \LL_{ww}, \tilde J_i^{(1,0)} \big ] \not \in \cI$, the Lie series expansion of the commutator must have a lowest $\tilde b$-degree term which does not belong to $\cI$. We shall refer to this term as $X$ for which we have $X \not \in \cI$, if it exists. Since the 
second term of (\ref{66.cor.1aa}) has $\tilde b$-degree one higher than the left side, $X$ must be a single-valued holomorphic $(1,0)_i$ form  in $x_i$, modulo terms that belong to $\cI$. But in view of (\ref{66.cor.1bb}) the integral of $X$ against every $\bar \om_I(x_i)$ belongs to $\cI$ and therefore $X$ itself must belong to $\cI$,  which contradicts our initial assumption in the proof by contradiction. Therefore, we must have $\big [ \tilde \LL_{ww}, \tilde J_i^{(1,0)} \big ] \in \cI$. 
\end{proof}

\subsection{Proof  of $\tilde L_{ww}^0 \big ( [\tilde a_{iI}, \tilde a_{jJ} ] \big ) \in \cI$}
\label{sec:six3}

We shall now prove the missing link needed to establish that $L_{ww}^0 \in\mathfrak{der} ( \hat \mt_{h,n})$.
We begin with the following lemma.
{\lem
\label{87.lem:1}
The\footnote{A section $s_{\underline{U}_0}$ is understood.}
following relation holds in $ \hat{\muu}_{h,n}$ for arbitrary $i \not= j$,
\bea
\label{87.lem.1}
\int _\Sigma d^2 x_i \,\bar \om_I(x_i) \! \int _\Sigma d^2 x_j \, \bar \om_J(x_j)  
\big [ \tilde J^{(1,0)}_i, \tilde J^{(1,0)} _j \big ]
=
\big [ \tilde a_{iI}, \tilde a_{jJ} \big ]
+ \sum_{\vP, \vQ} \mN_{I \vP J \vQ} \, \big [ \tilde B_i^\vP \tilde t_{ij}, \tilde B_j^\vQ \tilde t_{ij} \big ] 
\qquad
\eea
where the modular tensors $\mN$ were defined in (\ref{defcn}). }

\begin{proof}
The left side of (\ref{87.lem.1}) may be decomposed into its constituents with the help of the expression of (\ref{66.a.2bb}) for $\tilde J_i^{(1,0)}$ and we obtain, 
\bea
\label{6.g.6}
&&
\int_\Sigma d^2 x_i \, \bar \om_I(x_i) \int _\Sigma d^2 x_j \, \bar \om_J(x_j ) \bigg \{ 
\big [ \bPhi^K(x_i; \tilde  B_i) \, \tilde  a_{iK} , \bPhi^L(x_j; \tilde B_j) \, \tilde a_{jL}  \big ]
 \\ && \hskip 2.1in
+ \sum_{\ell \not= j} \big [ \bPhi^K(x_i; \tilde B_i) \, \tilde a_{iK} ,  \bG(x_j,x_\ell; \tilde B_j) \tilde t_{j \ell} \big ] 
\no \\ && \hskip 2.1in
+ \sum_{k \not= i} \big [  \bG(x_i,x_k; \tilde B_i) \tilde t_{ik} , \bPhi^L(x_j; \tilde B_j) \, \tilde a_{jL} \big ] 
\no \\ && \hskip 2.1in
+ \sum_{k \not= i} \sum _{\ell \not= j} \big [ \bG(x_i,x_k; \tilde B_i) \tilde t_{ik} , \bG(x_j,x_\ell; \tilde B_j) \tilde t_{j \ell} \big ] \bigg \}
\no
\eea
The following integrals are obtained by using the fact that $\bPhi ^J(x_i; \tilde B_i)- \om^J(x_i)$ and $\bG(x_i, x_j;\tilde B_i)$ are in the range of $\p_i$, as may be seen by inspection of (\ref{2.b.3a}) and (\ref{2.b.3b}), 
\bea
\label{6.g.7}
\int_\Sigma d^2 x_i \, \bar \om_I(x_i) \, \bPhi^J(x_i; \tilde B_i) & = & \delta ^J_I
\no \\
\int_\Sigma d^2 x_i \, \bar \om_I(x_i) \, \bG(x_i, x_k; \tilde  B_i) & = & 0
\eea
As a result of (\ref{6.g.7}),  the contribution from the first line on the right of (\ref{6.g.6}) gives $\big [ \tilde a_{iI}, \tilde a_{jJ} \big ]$, while the contributions from the second and third lines vanish identically. On the fourth line, the contribution from $\ell \not= i$ vanishes by carrying out the integral over $x_i$ first while the contribution from $k \not= j$ vanishes by carrying out the integral over $x_j$ first. This leaves only the contribution with $k=j$ and $\ell=i$, so that (\ref{6.g.6}) takes the  form,
\bea
\label{6.g.8}
\big [ \tilde a_{iI}, \tilde a_{jJ} \big ]  + 
\int_\Sigma d^2 x_i \, \bar \om_I(x_i) \int _\Sigma d^2 x_j \, \bar \om_J( x_j ) 
 \big [ \bG(x_i,x_j; \tilde B_i) \tilde t_{ij} , \bG(x_j,x_i; \tilde B_j) \tilde t_{ij} \big ] 
\eea
Expanding the generating functions in the second term in (\ref{6.g.8})  in powers of $\tilde B$ the combination in (\ref{6.g.8}) becomes, 
\bea
\label{6.g.9}
\big [ \tilde a_{iI}, \tilde a_{jJ} \big ] \! + \! 
\int_\Sigma d^2 x_i \, \bar \om_I(x_i) \!
 \int _\Sigma d^2 x_j \, \bar \om_J( x_j  ) \sum_{\vP, \vQ} 
\p_i \cG_\vP(x_i, x_j) \p_j \cG_\vQ (x_j, x_i) 
\big [ \tilde B_i^\vP \tilde t_{ij}, \tilde B_j^\vQ \tilde t_{ij} \big ] 
\quad
\eea
The integrals may be carried out in terms of the modular tensors $ \mN$ defined by (\ref{defcn}) to give the desired formula of (\ref{87.lem.1}). 
\end{proof}

Next, we apply Lemma \ref{87.lem:1} to proving the following theorem.

{\thm 
\label{66.thm:2} 
(a) The following identity holds for arbitrary $i\neq j$,
\bea
\label{6.g.0}
{} \big [ \tilde \LL_{ww}^0, [\tilde a_{iI}, \tilde a_{jJ} ] \big ] \in \cI
\eea
(b) $\tilde L_{ww}^0 \in \mathfrak{der}(\hat {\muu}_{h,n})$ induces $L_{ww}^0 \in \mathfrak{der} (\hat \mt_{h,n})$ under the projection $\pi : \muu_{h,n} \to \mt_{h,n}$ in (\ref{6.z.3a}),
\bea
\label{6.g.0a}
\pi \big (\tilde L^0_{ww} ( \tilde X) \big ) = L_{ww}^0 \big(\pi (\tilde X) \big)  
\eea
for arbitrary $\tilde X \in \hat {\muu}_{h,n} $. }

\begin{proof}
To prove \textit{(a)}, we use the relation $[\tilde L_{ww}, \tilde J^{(1,0)} _i] \in \cI$ established in Theorem \ref{66.thm:1}, the fact that $[ \tilde J^{(1,0)}_i, \cI ] \in \cI$ and the statement $\tilde L_{ww} \in \mathfrak{der}(\hat{\muu}_{h,n})$ of Theorem \ref{6.thm:1} to prove the  following relation for the double commutator,
\bea
\label{6.g.1}
\big [ \tilde L_{ww} , [ \tilde J^{(1,0)} _i, \tilde J^{(1,0)} _j ] \big ] =
\big [ [ \tilde L_{ww}, \tilde J^{(1,0)} _i], \tilde J^{(1,0)} _j \big ] 
+ \big [ \tilde  J^{(1,0)} _i, [ \tilde L_{ww}, \tilde J^{(1,0)} _j ] \big ] \in \cI
\eea
We may decompose $\tilde L_{ww}$ into $\tilde L_{ww} = \tilde L_{ww}^0 + \tilde L_{ww}^1$ where $\tilde L_{ww}^0$ is independent of $\xx$, takes values  in $\mathfrak{der} (\hat{\muu}_{h,n})$ and was given on the generators of $\muu_{h,n}$ in (\ref{6.z.4}),  while $\tilde L^1_{ww}$  was defined in (\ref{66.a.2cc}), is $\xx$-dependent and takes values in $\hat{\muu}_{h,n}$. 
The relations $[ \tilde J^{(1,0)} _i, \tilde J^{(1,0)} _j ] \in \cI$ and $\tilde L^1_{ww} \in \hat {\muu}_{h,n}$ imply that 
$\big [ \tilde L_{ww}^1 , [ \tilde J^{(1,0)} _i, \tilde J^{(1,0)} _j ] \big ] \in \cI$. Combining this result with (\ref{6.g.1}) implies in turn, 
\bea
\label{6.g.2}
 \big [ \tilde L_{ww}^0 , [ \tilde J^{(1,0)} _i, \tilde J^{(1,0)} _j ] \big ] \in \cI
\eea
As a result, we also have the following identity,
\bea
\label{6.g.4}
\int_\Sigma d^2 x_i \, \bar \om_I(x_i) \int _\Sigma d^2 x_j \, \bar \om_J( x_j ) 
\big [ \tilde L_{ww}^0 , [ \tilde J^{(1,0)} _i, \tilde J^{(1,0)} _j ] \big ] \in \cI
\eea
Since $\tilde L^0_{ww}$ is independent of $\xx$, we may use (\ref{87.lem.1}) of Lemma \ref{87.lem:1} to evaluate the integrals, 
\bea
\label{6.g.5}
\Big [ \tilde L^0_{ww}, \big [ \tilde a_{iI}, \tilde a_{jJ} \big ] \Big ] 
+ \sum_{\vP, \vQ} \mN_{I \vP J \vQ} \, \Big [ \tilde L^0_{ww}, \big [ \tilde B_i^\vP \tilde t_{ij}, \tilde B_j^\vQ \tilde t_{ij} \big ] \Big ] \in \cI
\eea
Under swapping the summation variables $\vP \leftrightarrow \theta(\vQ)$, the double commutator changes sign in view of the identities $\tilde B_i^{\theta(\vQ)} \tilde t_{ij} = \tilde B_j ^\vQ \tilde t_{ij}$ and $\tilde B_j^{\theta(\vP)} \tilde t_{ij} = \tilde B_i ^\vP \tilde t_{ij}$, while the coefficient $\mN_{I \vP J \vQ}$ is invariant thanks to the antipode and cyclic symmetries $\mN_{I \theta(\vQ) J \theta(\vP)} = \mN_{\vP J \vQ I} = \mN_{I \vP J \vQ } $ of Proposition \ref{7.prop:1re}. As a result, the summand in the second term of (\ref{6.g.5})  is odd under swapping $\vP \leftrightarrow \theta(\vQ)$ so that the sum  vanishes, which proves (\ref{6.g.0}). 

\sm

To prove part \textit{(b)} we combine the results of Theorem \ref{6.thm:1} with the results of part \textit{(a)} which, together, imply that $\tilde L_{ww}^0$ preserves all the structure relations of (\ref{6.z.1}) exactly in $\hat{ \muu}_{h,n}$ as well as that of (\ref{6.g.0}) which is modulo the ideal $\cI$. In the projection $\pi : \muu_{h,n} \to \mt_{h,n}$,  the ideal $\pi(\cI)= 0$ so that $\pi( \tilde J_i^{(1,0)})=  J_i^{(1,0)}$ and $\tilde L_{ww}^0 \to L_{ww}^0$ as given by (\ref{6.g.0a}), which preserves all the structure relations of $\mt_{h,n}$ and therefore acts as a derivation on $\hat \mt_{h,n}$.
\end{proof}

\sm

{\rmk Since it has now been established that $L^0_{ww} \in \mathfrak{der}(\hat \mt_{h,n})$, we shall henceforth represent its action by the customary commutator notation.} 

\subsection{Vanishing of the commutator $[L_{ww},  J_i^{(1,0)}] $}
\label{sec:vanLJ}

With the results $[\tilde L_{ww},  \tilde J_i^{(1,0)}] \in {\cal I}$ and $L_{ww}^0 \in \hat{\mt}_{h,n}$ of Theorems \ref{66.thm:1} and \ref{66.thm:2} in place, we readily deduce the following flatness condition for the extended connection ${\cal J}$ of (\ref{7.0.x}).

{\thm 
\label{12.thm:1}
The  commutator $[L_{ww},  J_i^{(1,0)}] $ vanishes in $\hat \mt_{h,n}$.}

\begin{proof}
The theorem readily follows from $[\tilde L_{ww},  \tilde J_i^{(1,0)}] $ falling into the
ideal $\cI$ by Theorem \ref{66.thm:1} and the fact that $L^0_{ww}$ obtained from
the projection $\pi : \muu_{h,n} \to \mt_{h,n}$ of $\tilde L^0_{ww}$ is a derivation of $\hat{\mt}_{h,n}$ by Theorem \ref{66.thm:2}.
\end{proof}

\newpage

\section{Completing the proof of flatness, of Theorem \ref{1.thm:main}, and stage-setting of modularity}
\setcounter{equation}{0}
\label{sec:7}

To complete the proof that the extended connection $\cJ$ of (\ref{7.0.x})  is flat, it remains to prove the relation  $\delta_{vv} \LL_{ww} - \delta_{ww} \LL_{vv}=0$  of (\ref{7.a.3bb}) and the relation $[\LL_{vv}, \LL_{ww}]=0$ of (\ref{7.a.4bb}). These results will be established in the remainder of this section.

\subsection{Proving the flatness condition $\delta_{vv} \LL_{ww} - \delta_{ww} \LL_{vv}=0$}
\label{sec:7.1}

{\lem
\label{lem:7.22}
We have the flatness condition,
\bea
\label{10.a.1}
\delta_{vv} [\LL_{ww}, X] - \delta_{ww} [\LL_{vv}, X] =0
\eea
}

\begin{proof}
\vskip -0.0in
To prove this flatness condition  means proving that the condition is realized when applied to an arbitrary (moduli-independent) element $X \in \hat{\mt}_{h,n}$. 
Given that $\LL_{ww}$ was established to be a derivation of $\hat{\mt}_{h,n}$ in Theorem \ref{66.thm:2}, 
it suffices to prove this condition on the elementary generators $X \in \{ a_{iI}, b_i^I, t_{ij}\}$ with $I\in \{ 1,\cdots, h\} $ 
and $ i,j \in \{ 1,\cdots, n \} $ which in turn was accomplished in Schemata \ref{schem.15}.
\end{proof}

\subsection{Proving the flatness condition $[\LL_{vv}, \LL_{ww}] =0$}
\label{sec:llcond}

{\lem 
\label{7.lem:33}
We have the flatness condition,
\bea
[\LL_{vv}, \LL_{ww}] =0
\eea}

\begin{proof}
\vskip -0.0in 
We begin by showing that $[\LL_{vv}, \LL_{ww}] $ is independent of $\xx$ by computing its $\p_i$ derivative, and using the first relation in (\ref{7.a.3}) to recast the $\p_i$ derivatives of $\LL_{vv}$ and $\LL_{ww}$ in terms of variational derivatives of $J_i^{(1,0)}$. In this way, we obtain,
\bea
\label{10.c.1}
\p_i \big [ \LL_{vv}, \LL_{ww} \big ] & = & 
\big [ \delta_{vv} J_i^{(1,0)} , \LL_{ww} \big ]  + \big [ \LL_{vv}, \delta_{ww} J_i^{(1,0)} \big ] 
\no \\ & = &
\delta_{vv}  \big [ J_i^{(1,0)} , \LL_{ww} \big ]  + \delta_{ww}  \big [ \LL_{vv}, J_i^{(1,0)} \big ] 
\no \\ && \qquad
+ \big [ J_i^{(1,0)}, \delta_{vv} \LL_{ww} - \delta_{ww} \LL_{vv} \big ] 
\qquad
\eea
The term on the last line vanishes individually in view of the results of the previous subsection, while the terms of the second line vanish in view of Theorem \ref{12.thm:1}. Since $\big [ \LL_{vv}, \LL_{ww} \big ]$ is a scalar function of each $x_i$, which is anti-holomorphic in view of $\p_i \big [ L_{vv}, L_{ww} \big ] =0 $ for all $x_i$ on the compact surface $\Sigma$ with $i \in \{ 1,\cdots,n \} $, we conclude that  $\big [ \LL_{vv}, \LL_{ww} \big ]$ is indeed independent of all~$x_i$. 

\sm

Next, we use the Jacobi identity,
\bea
\label{10.c.2}
\Big [ \big [ \LL_{vv}, \LL_{ww} \big ] , J_i^{(1,0)} \Big ]= 
\Big [ \big [ \LL_{vv}, J_i^{(1,0)} \big ] , \LL_{ww} \Big ] + \Big [ \LL_{vv}, \big [ \LL_{ww} , J_i^{(1,0)} \big ] \Big ] 
\eea
The right side vanishes since the inner commutators in both terms vanish in view of Theorem \ref{12.thm:1}. The vanishing of the left side implies the vanishing of the following combinations,
\bea
\label{10.c.3}
\int _\Sigma d^2 x_i \, \bar \om_I(x_i) \big [  [ \LL_{vv}, \LL_{ww}  ], J_i ^{(1,0)} \big ] & = & 0
\no \\
\mathop{\mathrm{Res}}_{x_j=x_i} \big [  [ \LL_{vv}, \LL_{ww}  ], J_i ^{(1,0)} \big ] & = & 0
\eea
Since $[L_{vv}, L_{ww}]$ is independent of $x_i$, this commutator commutes with the integral in the first line and the residue calculation in the second line, resulting in the relations,
\bea
\label{10.c.4}
\big [  [ \LL_{vv}, \LL_{ww}  ], a_{iI}  \big ] & = & 0
\no \\
\big [  [ \LL_{vv}, \LL_{ww}  ], t_{ij} \big ] & = & 0
\eea
Finally, we shall obtain the commutator of $[\LL_{vv}, \LL_{ww}]$ with $b_i^I$ by considering the complex conjugate derivatives, 
\bea
\label{10.c.5}
\bar \p_i [ \LL_{vv}, \LL_{ww}] =  [ \bar \p_i \LL_{vv}, \LL_{ww}]  + [ \LL_{vv}, \bar \p_i  \LL_{ww}] 
\eea
The left side vanishes since $[ \LL_{vv}, \LL_{ww}]$ is independent of $\xx$, while the right side may be worked out with the help of (\ref{7.a.5}), and we obtain,
\bea
\label{10.c.6}
\big [ [J_i^{(0,1)}, \LL_{vv}] - \pi \delta (v,x_i) J_i^{(1,0)} , \LL_{ww} \big ]  
+ \big [ L_{vv}, [J_i^{(0,1)}, \LL_{ww}] - \pi \delta (w,x_i) J_i^{(1,0)} \big ] =0
\eea
The terms proportional to $\delta$-functions cancel individually since the accompanying commutators $[J_i^{(1,0)} , \LL_{ww} ]$ and $[ L_{vv},  J_i^{(1,0)}]$ vanish. The remaining relation may be recast as follows using the Jacobi identity,
\bea
\label{10.c.7}
\big [ [ \LL_{vv} , \LL_{ww} ] , J_i^{(0,1)} \big ] =0
\eea
and using the explicit form of $J^{(0,1)}_i = - \pi \bar \om_I(x_i) b_i^I$ we obtain,
\bea
\label{10.c.8}
\big [ [ \LL_{vv} , \LL_{ww} ], b_i^I  \big ] =0
\eea
Combined with the results in (\ref{10.c.4}), the above equation shows that we have,
\bea
\label{10.c.9}
\big [ [ \LL_{vv} , \LL_{ww} ], X  \big ] =0
\eea
for all generators $X \in \{ a_{iI}, b_i^I,t_{ij}\}$ of $\mt_{h,n}$. Since the commutator $[ \LL_{vv} , \LL_{ww} ]$ readily inherits the derivation property of $L_{ww}$ in Theorem \ref{66.thm:2}, this implies (\ref{10.c.9}) to hold for arbitrary $X \in  \hat{\mt}_{h,n}$ or, more informally, we obtain the desired flatness condition $[ \LL_{vv}, \LL_{ww}]=0$. 
\end{proof}

\subsection{Summary and completing the proof of Theorem \ref{1.thm:main}}

Combining the results obtained in Proposition \ref{7.prop:1}, Theorem \ref{7.thm:1}, Theorem \ref{12.thm:1} with the results obtained earlier in this section, we obtain the following summary.

{\thm 
\label{7.thm:20}
There exists a unique $\der(\hat \mt_{h,n})$-valued solution $L_{ww}= L^0_{ww} + L^1_{ww}$ to the local flatness conditions stated in (\ref{7.a.3}), (\ref{7.a.4}), (\ref{7.a.4a}) and (\ref{7.a.5}), where $L^1_{ww}$ is given in (\ref{7.b.2}) and (\ref{7.b.9})  while the action of $L^0_{ww}$ on the generators of $\hat \mt_{h,n}$ is given by (\ref{9.q.20}). 
 Moreover, $L_{ww}$ is unique under the assumption that its $a$-degree is one, in the sense detailed in Schemata \ref{setup:1}. }

\sm

The construction of the components $L_\a$ of $\cL$ from $L_{ww}$ via (\ref{7.a.2}) requires choosing a non-canonical section $s: \cT_h \to \CS(\Sigma)$, as already explained in the introduction, and made explicit  in sections  \ref{sec:46.1} and \ref{sect:bers} in terms of Beltrami differentials. The fact that the construction obeys the formulation of the \textit{local $(1,0)$ extension problem} in section \ref{sec:1.111} is the content of the following theorem.

\sm

{\thm
\label{7.thm:30}
For any $\btau_0\in\mathcal T_h$, there exists a neighborhood $U_{\btau_0}$ 
of $\btau_0$ in which problem (b$''$)-(c$''$) in section \ref{sec:1.111} admits a solution; i.e., for any section 
$s : U_{\btau_0}\to \mathrm{CS}(\Sigma)$, there exists $\mathcal L^s$ as in 
(b$''$) and satisfying (c$''$). 
}

\begin{proof}
The starting observation is that it suffices to construct $\mathcal L^{s_0}$
satisfying the conditions of (c$''$) for a {\it particular} section $s_0$,
as $\mathcal L^{s}$ for any other section $s$ is derived from it using the constraints of (b$''$). 
The construction of $\mathcal L^{s_0}$ was accomplished in section \ref{sec:4} to section \ref{sec:llcond} and 
will now be summarized.

\sm

Actually, $\mathcal L^{s_0}$ is not constructed all  at once on the whole of the product of $\mathrm{Cf}_n(\Sigma)$ by a neighborhood of $\btau_0$.  Rather, we first select a complex structure $J_0$ above $\btau_0$, and cover $\Sigma^n$ by  open subsets $U_1\times \cdots \times U_n$ attached to a finite  collection of $n$-tuples 
$\underline U=(U_1,\cdots,U_n)$  of open subsets of $\Sigma$,  each $\underline U$ being such that 
the complement of $\cup_i U_i$ has a nonempty interior in $\Sigma$. 
By restriction to $\mathrm{Cf}_n(\Sigma)$, the finite collection of $n$-tuples $\underline U$
induces a finite cover of $\mathrm{Cf}_n(\Sigma)$ (section \ref{sec:46.1}). 

\sm

We then attach to each $\underline U$ a section $s_{\underline U}$, which is defined on a neighborhood $U_{\btau_0,\underline U}$ of $\btau_0$ such that $\btau_0\mapsto J_0$, and which satisfies the vanishing condition (\ref{vanishing:cond}) on the associated  Beltrami differentials (see section \ref{sect:bers}).  We recall that the spaces $\mathrm{BD}_{s_{\underline U},\btau}$ of these Beltrami differentials, indexed by the elements $\btau$ of 
$U_{\btau_0,\underline U}$, are the images of the differential map of $s_{\underline U}$ at the various points 
 $\btau\in U_{\btau_0,\underline U}$. Specifically, the condition satisfied by $s_{\underline U}$ is that any Beltrami differential in any of these various spaces $\mathrm{BD}_{s_{\underline U},\btau}$ vanishes on $\cup_i U_i$.

\sm

We show that the expression in terms of local coordinates $(\xx|\boldsymbol{\tau})$ on 
$(U_1\times \cdots \times U_n)\times U_{\btau_0,\underline U}$ of the 
system of partial differential equations bearing on $\mathcal L^{s_{\underline U}}$ 
is particularly simple  (sections \ref{sec:gfc} and \ref{sec:locf}).  
From section \ref{sec:4.3} to the end of section \ref{sec:6}, we partially solve the analogue of
this system relative to a Lie algebra $\hat{\mathfrak u}_{h,n}$  equipped 
with a morphism $\hat{\mathfrak u}_{h,n}\to \hat{\mathfrak t}_{h,n}$ obtained by relaxing some relations of the Lie algebra 
$\mathfrak t_{h,n}$ (see Definition \ref{6.def:1}), leading to $\tilde{\mathcal L}^{s_{\underline U}}(\xx|\boldsymbol{\tau})$ such that  the $\Lambda^2(T_{\mathbb C}^*\mathrm{Cf}_n(\Sigma))$ and 
$T_{\mathbb C}^*\mathrm{Cf}_n(\Sigma)\otimes T^{*(1,0)}\mathcal T_h$ parts of the system
\bea
\label{sec7prf.3}
(d_\xx + \p _\btau)  \tilde{\cL}^{s_{\underline U}} + \p_\btau \tilde{\cJ}_\text{DHS}^{s_{\underline U}} -  \tilde{\cJ}_{\rm DHS}^{s_{\underline U}} \wedge \tilde{\cL}^{s_{\underline U}} - \tilde{\cL}^{s_{\underline U}} \wedge  \tilde{\cJ}_{\rm DHS}^{s_{\underline U}} - \tilde{\cL}^{s_{\underline U}} \wedge \tilde{\cL}^{s_{\underline U}}= 0
\eea
which is equivalent to 
\bea 
\label{sec7prf.1}
(d_\xx + \partial_\btau)  \, \tilde{\cJ}^{s_{\underline U}} - \tilde{\cJ}^{s_{\underline U}} \wedge \tilde{\cJ}^{s_{\underline U}} 
= 0
\eea
under 
\bea
\label{sec7prf.2}
\tilde{\mathcal J}^{s_{\underline U}}(\xx|\boldsymbol{\tau})= \tilde{\mathcal J}_{\mathrm{DHS}}^{s_{\underline U}}(\xx|\boldsymbol{\tau}) + \tilde{\mathcal L}^{s_{\underline U}}(\xx|\boldsymbol{\tau})
\eea
are satisfied modulo terms in the ideal $\mathcal I=\mathrm{ker}(\hat{\mathfrak u}_{h,n}\to\hat{\mathfrak t}_{h,n})$ (see \eqref{6.z.3}).
Solutions $\tilde{\mathcal J}^{s_{\underline U}}(\xx|\boldsymbol{\tau})$ 
and $\tilde{\mathcal J}^{s_{\underline U'}}(\xx'|\boldsymbol{\tau})$ corresponding to 
tuples $\underline U,\underline U'$ are shown to 
correspond to one another through the diffeomorphism 
$(\phi_{\underline U,\underline U'})_n$, where 
$\phi_{\underline U,\underline U'}$ is the transition map 
$U_{\btau_0,\underline U}\cap U_{\btau_0,\underline U'}\to\mathrm{Diff}_0(\Sigma)$
from $s_{\underline U}$ to $s_{\underline U'}$, and therefore glue to a 
well-defined $\tilde{\mathcal J}^{s_{{\underline U}_0}}$, ${\underline U}_0$ 
being a particular $n$-tuple.  
However, since the various open sets are slightly displaced by applying the diffeomorphisms, the intersection over all tuples $\underline U$ of open sets 
$U_{\btau_0,\underline U}$ must be further 
restricted to a smaller neighborhood $U_{\btau_0}$ of $\btau_0$ 
in order to ensure that
the vertical trace of the union over all tuples $\underline U$ of the displaced open sets 
$(\phi_{\underline U_0,\underline U})_n(U_1\times \cdots \times U_n\times U_{\btau_0})$  
still covers $\Sigma^n$. Arguments relative to the global behavior of 
$\tilde{\mathcal J}^{s_{\underline U_0}}$ on $\Sigma$ are then 
applied to show that it descends to a derivation of $\hat{\mathfrak t}_{h,n}$ (section 
\ref{sec:66}), therefore producing a 1-form on ${\mathcal J}^{s_{{\underline U_0}}}$ on 
$\mathrm{Cf}_n(\Sigma)\times U_{\btau_0}$ valued in $\mathfrak{der}(\hat{\mathfrak t}_{h,n})$, 
which by the above satisfies the $\Lambda^2(T^*\mathrm{Cf}_n(\Sigma))_{\mathbb C}$ and 
$(T^*\mathrm{Cf}_n(\Sigma))_{\mathbb C}\otimes T^{*(0,1)}\mathcal T_h$ parts of the flatness condition 
\bea
(d_\xx + \partial_\btau)  \, {\mathcal J}^{s_{{\underline U_0}}} - {\mathcal J}^{s_{{\underline U_0}}} \wedge {\mathcal J}^{s_{{\underline U_0}}} 
=0
\eea
In sections \ref{sec:7.1} and \ref{sec:llcond}, we show that ${\mathcal J}^{s_{{\underline U_0}}}$
also satisfies the $\Lambda^2T^{*(0,1)}\mathcal T_h$ part of this condition, and therefore 
is a solution of (c$''$) relative to  $s_{{\underline U_0}}=s_0$.  
This ends the construction of the announced $\mathcal L^{s_0}$ and therefore the proof of the  theorem. 
\end{proof}

\sm

\subsection{Modular properties of $L_{ww}$}
\label{sec:modlww}

This section is dedicated to establishing the transformation properties of $L_{ww}(\xx| \btau)$ under the modular group $ {\rm Sp}(2h,\mathbb Z)$, whose action on the homology was defined in (\ref{modsec.01}). Key ingredients are the modular tensor properties of the DHS kernels, $\cA$ and $\cD$ given by Lemma \ref{lem:MT}, and the modular transformations of the generators of $\mt_{h,n}$ given in Proposition \ref{2.prop:1}. Combined, they lead to the modular invariance of the DHS connection as reviewed in section \ref{sec:mod} and together with Theorem \ref{66.thm:2} imply the following result for $L_{ww}$. 

{\prop
\label{5.prop.lww}
In the decomposition (\ref{7.b.1}) of $L_{ww}(\xx| \btau)$, the $\hat \mt_{h,n}$-valued component $L_{ww}^1(\xx|\btau)$ is modular invariant (as are both of its components $\cV_{ww} (\xx| \btau)$ and $\cV^1 _{ww} (\btau) $ given in (\ref{7.b.9})) while  the action of $L_{ww}^0(\btau) \in \mathfrak{der}(\hat \mt_{h,n})$ determined by Theorem \ref{7.thm:1} is modular covariant, 
provided that $a_{iI}$, $b_i^I$, $t_{ij}$ transform under ${\rm Sp}(2h,\mathbb Z)$ as in~(\ref{,odsec.0.6}).
}

\begin{proof}
Modular invariance of the $\xx$-dependent and $ \hat \mt_{h,n}$-valued parts $\cV_{ww} (\xx| \btau) + \cV^1 _{ww} (\btau)$ of $L_{ww}(\xx| \btau)$ is manifest term-by-term from their expressions in (\ref{7.b.9}): the tensorial ${\rm Sp}(2h,\mathbb Z)$ transformation of DHS kernels in (\ref{modsec.04}) and of $a_{iI}$, $b_i^I$, $t_{ij}$ in (\ref{,odsec.0.6}) imply that
the generating series $\theta \bG(w,x_i;B_i)$ and $\bPhi^J(w;B_k) a_{kJ}$ in (\ref{7.b.9aa}) are individually modular invariant and that the $(2,0)_w$ forms $\om^J(w) \p_w \Phi _{J \vI}{}^K(w)  $ in (\ref{7.b.9bb}) compensate the modular transformation of the accompanying $ \hat \mt_{h,n}$ elements $B_k^ \vI  a_{kK}$.

\sm

It remains to show that, for arbitrary $X \in \hat \mt_{h,n}$, the $\mathfrak{der}(\hat \mt_{h,n})$ action $[ L^0_{ww}(\btau) , X]$ transforms under ${\rm Sp}(2h,\mathbb Z)$ in the same way as $X$ does by (\ref{,odsec.0.6}). At the level of the generators $a_{iI}$, $b_i^I$, $t_{ij}$ of $\hat \mt_{h,n}$, this is the case for each term in the expressions (\ref{9.q.20}) for their brackets with $ L^0_{ww}(\btau) $: by Lemma \ref{lem:MT} together with (\ref{7.f.5}) and (\ref{7.f.6}), the ${\rm Sp}(2h,\mathbb Z)$ transformation of the modular tensors $\delta_{ww} \cA^M{}_{\vP}{}^N$ and $\delta_{ww} \cD_{\vP}$ in (\ref{9.q.20}) is inverse to that of the accompanying $ \hat \mt_{h,n}$ elements, up to the uncontracted lower index $I$ in the first line and the upper index $I$ in the second line. By these uncontracted indices, $\delta_{ww}  \cA^M{}_{\theta(\vP) I \vQ} {}^N [ B_i^\vP a_{iM}, B_i^\vQ a_{iN} ]$ and 
$\delta_{ww} \mD_{ \theta(\vP) I \vQ}  [ B_i^\vP t_{ik} , B_i^\vQ t_{ik} ]$ in the first line of (\ref{9.q.20})  transform like $a_{iI}$ of $ [\LL^0_{ww} , a_{iI} ] $ on the left side
and $\delta_{ww}  \cA^I {}_\vP {}^J B_i^\vP a_{iJ}$ in the second line transforms like
$ b_i^I$ of $[\LL^0_{ww} , b_i^I ]$ on the left side. The contributions $ \delta _{ww}  \mD_\vP  [  B_i^\vP t_{ij}, t_{ij}  ]$ to the right side of the third line are modular invariant just like the $t_{ij}$ of $[\LL^0_{ww} , t_{ij} ]$ on the left side. 

\sm

We have demonstrated that $[\LL^0_{ww}(\btau) , X ]$ and $X$ exhibit the same modular transformation for all generators $X \in \{ a_{iI}, b_i^I,t_{ij}\}$ of $ \hat \mt_{h,n}$ with $I\in \{ 1,\cdots, h\} $ and $ i,j \in \{ 1,\cdots, n \} $.
Given that $L_{ww}(\xx| \btau)$ is a derivation of $\hat \mt_{h,n}$ by Theorem \ref{66.thm:2}, the same is true for arbitrary $X \in \hat \mt_{h,n}$, and we conclude that $L^0_{ww}( \btau)$ enjoys a modular covariant action on $\hat \mt_{h,n}$.
\end{proof}

In order to infer from Proposition \ref{5.prop.lww}, expressing  the modularity of $L_{ww}$, a modularity result 
on the $(1,0)$-form components $L_\alpha$ in holomorphic moduli directions via  (\ref{7.a.2}), it remains to combine this result with the  ${\rm Sp}(2h,\mathbb Z)$ properties of the Beltrami differentials, also see footnote \ref{modularfoot} (section \ref{summsec}).  The leftover work to deduce modularity of ${\cal L} = {\cal L}^s$ is relegated to future work.

\newpage

\section{Relations between $\cA$, $\mD$, $\mN$ and flatness conditions}
\setcounter{equation}{0}
\label{sec:77}

The flatness conditions $[J^{(1,0)}_i, J^{(1,0)}_j]=0$ for the DHS connection $\cJ_\text{DHS}$ given in (\ref{2.b.3aa})   were shown in \cite{DHoker:2026ggx} to be equivalent to the set of all the interchange and Fay identities on DHS kernels that had earlier been derived by independent methods in \cite{DHoker:2024ozn}. In this sense, the flatness relations $[J^{(1,0)}_i, J^{(1,0)}_j]=0$ are generating functions for all interchange and Fay identities. Both the flatness relations and the interchange and Fay identities play a crucial role in making the algebra of DHS kernels together with the associated polylogarithms closed under taking primitives. 

\sm

Having extended the DHS connection $\cJ_\text{DHS}$ on a fixed Riemann surface $\Sigma$ to a flat connection $\cJ$ that includes holomorphic variations in the moduli of $\Sigma$, we dispose of new flatness conditions as well as of the relations ensuring that $L^0_{ww}$ belongs to the derivation algebra $\mathfrak{der}(\hat{\mt}_{h,n})$. The novel protagonist in all these relations is $L^0_{ww}$ determined by its action on the generators $a_{iI}, b_i^I$ and $t_{ij}$ given in Theorem \ref{7.thm:1} in terms of the variational derivatives of the modular tensors $\cA$ and $\mD$ of (\ref{7.f.1}). Therefore, the flatness conditions of the connection $\cJ$ and the conditions for  $L^0_{ww} \in \mathfrak{der}(\hat \mt_{h,n})$ require relations between $\cA$ and $\mD$. The new relations arise in two categories.
\begin{enumerate}
\itemsep=0 in
\item Linear relations between $\cA$ and $\mD$, such as those of Proposition \ref{7.prop:1re} which played, for example, a crucial role in the proof of  Theorem \ref{6.thm:1} given in appendix \ref{sec:E};
\item Bilinear relations between the variational derivatives of  $\cA$ and $\mD$ and DHS kernels that arise from the flatness condition $[L_{vv}, L_{ww}]=0$.
\end{enumerate}
Investigations into the relations of item 2. are relegated to future work. In the remainder of this section, we provide additional discussion of the linear relations of item 1.

\subsection{Linear relations}
\label{sec:77.lin}

While the preservation of the structure relation $\big [ L^0_{ww} , [a_{iI}, a_{jJ}] \big ]=0$  was proven earlier by indirect methods, one may also investigate this relation directly. Being a quantity of $a$-degree 3, it receives contributions of the schematic form $aaa$, $aat$, $att$ and $ttt$, with a Lie series of possible insertions of $B_i$ such as $[ B_{i}^\vP a_{iK} , [ B_{i}^\vQ a_{iL} , B_{i}^\vR t_{ij} ] ]$ in the second category. The results, presented here without proof, are as follows. The contributions of the form $aaa$ and $aat$ cancel from $\big [ L^0_{ww} , [a_{iI}, a_{jJ}] \big ]$ without using any relations between $\cA$ and $\mD$. The contributions of the form $att$ cancel upon a lengthy calculation with the help of identities between $\cA$ and $\mD$ that follow from and generalize those in (\ref{7.f.3}) of Proposition \ref{7.prop:1re}, and are given by, 
\bea
\label{9.lem.44}
\mD_{\vP I \vQ} \, \delta ^L_J - \delta ^L_I \, \mD_{\vQ J \vP}
= \sum _{\vP = \vU \vV} ~ \sum _{\vU I \vQ J \vV = \vX \vY \vZ} 
\cA^L {}_{\big ( \vX \shuffle \theta(\vZ) \big ) M \big ( \vY + \theta(\vY) \big )} {}^M
\eea
The remaining contributions of the form $ttt$ take the following form,
\begin{align}
\big[L_{ww}^0 ,[a_{i I}  ,a_{jJ} ] \big]  &= \sum_{\vP, \vQ,\vR} \Upsilon_{\vP | \vQ | \vR} \, \big[B_{i}^\vP t_{ij} , [ B_{i}^\vQ t_{ij} , B_{i}^\vR t_{ij} ] \big]
\label{tovan} 
\end{align}
The coefficients $\Upsilon_{\vP | \vQ | \vR}$  are complicated but explicitly known in terms of $\delta_{ww} \mD_{\vI}$ and double traces of the variational derivatives of $\cA$, of the form $\delta_{ww}\cA^M{}_{\vI M \vJ N \vK}{}^N$ and $\delta_{ww}\cA^M{}_{\vI N \vJ M \vK}{}^N$. Remarkably, the coefficients of each term in $\Upsilon_{\vP | \vQ | \vR}$ are independent of $h$ for arbitrary $\vP , \vQ , \vR$. The set of identities required for $\Upsilon=0$ remains to be fully understood. 

\sm

Here, we shall present a number of identities that we have so far. They include, 
\begin{itemize}
\item 
Cyclic permutations of (\ref{7.f.3}) in $K \vI L = \vP =P_1 \cdots P_r$ arranged to eliminate the $\mD$,
\bea
\label{aasec.92} 
\sum_{\vP = \vX \vY \vZ}\cA^J {}_{ \big ( \vX \shuffle \theta(\vZ) \big ) N \big ( \vY + \theta (\vY) \big )} {}^N + {\rm cycl}(P_1,\cdots,P_r) =0
\eea
along with further traces over $J$ and $P_i$ that result in doubly-traced $\cA$ terms;
 \item 
Substitution of  $\vI = \vP M \vQ$ into (\ref{7.f.3}) and contraction with $\delta_J^M$,
\bea
\mD_{K \vP L \vQ} - \mD_{ \vP K \vQ L} =  
\sum_{K \vP M \vQ L = \vX \vY \vZ} \cA^M {}_{ \big ( \vX \shuffle \theta(\vZ) \big ) N \big ( \vY + \theta (\vY) \big )} {}^N
 \label{aasec.91}
\eea
and further relations among doubly-traced $\cA$ in which $\mD$ has been eliminated by various traces and cyclic permutations of (\ref{aasec.91});
\item 
The  $\delta^L_J$ trace of (\ref{7.f.3}) gives another relation between $\mD$ and double traces of $\cA$,
\bea
h \mD_{K \vI} - \mD_{\vI K} = \sum_{K \vI M  = \vX \vY \vZ} 
 \cA^M {}_{ \big ( \vX \shuffle \theta(\vZ) \big ) N \big ( \vY + \theta (\vY) \big )} {}^N 
  \label{aasec.55}
\eea
which, for genus $h \geq 2$,  can serve to eliminate any $\mD$ in favor of double traces of $\cA$.
\end{itemize}

\subsection{Bilinear relations}

The flatness condition $[L_{vv}, L_{ww}]=0$ may be expressed in terms of its vanishing action on the generators $a,b$ and $t$ of $\hat \mt_{h,n}$. The resulting relations are all bilinear in the variational derivatives of $\cA$ and $\mD$. These variational derivatives may be expressed in terms of DHS kernels $\p \Phi$ and the modular tensors $\mN$ with the help of (\ref{7.prop.2a}) and (\ref{7.prop.2b}). The nature and structure of the induced relations constitutes an open problem analogous in spirit to the Fay identities between DHS kernels.

\newpage

\section{Restriction to genus one}
\setcounter{equation}{0}
\label{sec:8}

In this section, we shall discuss the special case of genus $h=1$ for which many simplifications take place. In addition, this will give us the opportunity to compare the formal structure and concrete formulas of the single-valued and modular invariant, but non-meromorphic, flat connection of this paper with the meromorphic global connection of \cite{CEE}. In particular, we shall be able to identify Tsunogai's derivations \cite{Tsuongai:1995} as part of the algebra where the connection $\cJ$ at $h=1$ and that of \cite{CEE} takes values.

\sm

We shall represent a genus one Riemann surface $\Sigma$ as the quotient $\Sigma = \CC / (\ZZ + \tau \ZZ)$ where the modulus $\tau$ takes values in the upper half plane, parametrized by $ \tau = \tau _1 + i \tau_2$ with $\tau_1,\tau_2 \in \mathbb R$ and $ \tau_2 >0$.  The customary choice of complex coordinates $z, \bar z$ on the torus implements the quotient representation by identifying $z \cong z+1$ along a homology $\mA^1$ cycle and $z \cong z+\tau$ along a $\mB_1$ cycle so that the normalized holomorphic Abelian differential is $\om_1 = dz$. As a result, the coordinates $z, \bar z$ and the differential $dz$ are $\tau$-dependent.  The choice of complex coordinates $x, \bar x$ adopted in the preceding sections of this paper 
is different from the customary choice in that $x, \bar x$ are independent of the moduli. The change of variables $x \to z(x)$  between the two choices is easily realized in terms of the differential $\om_1= dz = \om_1(x) dx$ with $\om_1(x) = \p z / \p x$.

\subsection{The connection at fixed modulus}

The generators of the genus one Lie algebra $\mt_{1,n}$ are $a_{i1}, b^1_i$ and $t_{ij}$ subject to the structure relations of Definition \ref{2.def:1} and the following echo of translation invariance on the torus,
\bea
\sum_{i=1}^n a_{i1} = \sum_{i=1}^n b^1_i = 0 
\eea 
We shall keep the value of the index $I=1$ explicit, as its position distinguishes between $a_{i1}$ and $a_i^1 = a_{i1} /\tau_2$ and between $b_i^1$ and $b_{i1} = \tau_2 \, b_i^1$.\footnote{We recall that the identities in the third column of  (\ref{11.1}) relate the generators $t_{k\ell}$ of $\mt_{1,n}$ to brackets of $a_i^1,b_{j1}$ or $a_{i1},b_{j}^1$ with opposite positions of their index $I=1$.}

\sm

The restriction of the connection $\cJ_\text{DHS}$ to genus one is still given by (\ref{2.b.1}), but the expression for $J^{(1,0)}_i$ of (\ref{2.b.4}) simplifies in view of the fact that $\Phi _{I_1 \cdots I_r}{}^J=0$ whenever $r \geq 1$. We shall express (\ref{2.b.4}) in terms of customary coordinates $z_i, \bar z_i$ on the torus,
\bea
\label{8.b.1}
\cJ_\text{DHS} (\zz| \tau)&= \sum_{i=1}^n \Big ( J_i^{(1,0)} (\zz | \tau) \, dz_i   - \pi b_i^1 \, d \bar z_i  \Big )
\eea
where we set $\zz= (z_1, \cdots , z_n)$ and $\p_i= \p/ \p z_i$. The components are given as follows, 
\bea
\label{8.b.4}
J_i^{(1,0)} (\zz | \tau) = a_i^1 
+ \sum _{j \not= i} \sum _{r=0}^\infty \p_i \cG_{(r)} (z_i- z_j|\tau )  (B_i^1)^r t_{ij}
\eea
The function $\cG_{(0)} (z|\tau ) = \cG(z|\tau )$ is the genus one Arakelov Green function, given by,
\bea
\label{8.b.5}
\cG(z |\tau ) = - \ln \bigg | {  \tet_1(z | \tau)  \over \eta(\tau) } \bigg |^2 + {2 \pi \over \tau_2} \big ( \Im  z \big )^2 
\eea
where $\eta(\tau)$ is the Dedekind $\eta$-function. 
For $r \geq 1$ the functions $\cG_{(r)}(z|\tau)$ are given recursively by convolution with $\cG$ as in (\ref{2.a.4}), 
\bea
\label{8.b.6}
\cG_{(r)}(z|\tau) = \int _\Sigma d^2 z' \, \cG(z-z' |\tau) \, \p_{z'} \cG_{(r-1)} (z'  | \tau  )
= \int _\Sigma d^2 z' \, \p_z \cG(z-z' |\tau) \, \cG_{(r-1)} (z'  | \tau  ) \ \ 
\eea
where the second equality results from integrating by parts in $z'$ and using translation invariance to convert the derivative $\p_{z'}$ to $ - \p_z$ action on $\cG$. 

\sm

{\rmk 
\label{BLrmk}
The DHS connection (\ref{8.b.1}) and (\ref{8.b.4}) at genus one reduces to the Brown-Levin connection \cite{BrownLevin} once the variables $z_2,\cdots,z_n$ are treated as punctures of the torus (setting all the 
$dz_{i\neq 1}$ to zero) and one of $t_{12},\cdots,t_{1n}$ is eliminated through the $\mt_{1,n}$ relation 
$[b^1_1,a_{11}]=-\sum_{j=2}^n t_{1j}$. This can be seen by identifying the derivatives 
$\p_i \cG_{(r)} (z_i{-} z_j|\tau )$ in (\ref{8.b.4}) as,
\bea
\p_z \cG_{(r)} (z | \tau ) = - \tau_2^r f^{(r+1)}(z|\tau)
\label{dgisf}
\eea
where the single-valued and modular Kronecker-Eisenstein kernels $f^{(s)}(x{-}y|\tau)$ 
are the specializations of the DHS kernel $f^{I_1\cdots I_s}{}_J(x,y)$ to $h=1$ \cite{DHoker:2023vax}.}

\subsection{The flat connection with holomorphic moduli variations}

The flat connection $\cJ$ of (\ref{7.0.x})  for genus one is given in terms of the genus one connection $\cJ_\text{DHS}$ by adding a single component $\cL  \,d \tau$ which we divide by $2 \pi i$ in order to conform to the conventions of \cite{CEE},
\bea
\cJ(\zz| \tau) = \cJ_\text{DHS} (\zz| \tau) + \cL(\zz| \tau) \, { d\tau \over 2 \pi i} 
\label{univh1}
\eea
The generalized flatness conditions are then given by the first two lines of (\ref{7.a.0}), as the conditions on the third line are satisfied trivially. The expressions for the constituents of~$\cL$,
\bea
\cL(\zz| \tau) = \LL^0(\tau) + \LL^1(\zz| \tau)
\label{unig1}
\eea
may be read off directly from the results (\ref{7.cor.30}) of Corollary \ref{7.cor:30}, as these equations were valid for an arbitrary choice of coordinates on moduli space.

\subsubsection{Evaluating the action of $\LL^0$}

To evaluate the action of $\LL^0$ for genus one we avail ourselves of a major simplification of (\ref{7.cor.30}) as a result of the facts that,
\bea
\cA^1 {}_\emptyset {}^1 = - { 1 \over \pi \tau_2} 
\hskip 1in 
\cA^1 {}_{1_r} {}^1 =0 \quad \hbox{ for } r \geq 1
\eea 
where $1_r$ is the word formed out of $r$ letters $1$. Expressing the result in terms of 
\bea
a_i= a_i^1=  \frac{ a_{i1} }{\tau_2} \, , \ \ \ \ \ \
b_i = b_{i1} = \tau_2 b_i^1
\label{redefab}
\eea
we obtain a form that is close to the expression of \cite{CEE}, 
\bea
\label{8.cor.32}
{} [\LL^0 , a_i ] & = &
- \pi i   \sum_{p,q=0} ^\infty
 (-)^p \tau_2 ^{-p-q-1} \p_\tau \mD_{ 1_{p+q+1} } \sum_{k\not= i}  \Big [ B_i ^p  t_{ik} , B_i^q t_{ik}  \Big ] 
\no \\
{} [\LL^0 , b_i ] & = &    a_i
\no \\
{} [\LL^0, t_{ij} ] & = & - 2 \pi i  \sum_{p=1}^\infty  \tau_2^{-p} \, \p_\tau  \mD_{1 _p}    \, \big [  B_i^p \, t_{ij}, t_{ij} \big ]
\eea
The evaluation of $\tau _2^{-r} \p_\tau \mD_{1_r}(\tau)$ for $r \geq 1$ is given by the lemma below in terms of the holomorphic Eisenstein series $G_k(\tau)$ defined for $k \geq 3$ by the absolutely convergent 
sums over $\mathbb Z^2\setminus \{ (0,0) \}$, 
\bea
G_k (\tau) = \sum_{(m,n) \not= (0,0)}  { 1 \over (m + \tau n)^k} 
\label{defgk}
\eea
and related to the normalized Eisenstein series $E_k (\tau)=1+{\cal O}(e^{2\pi i \tau})$ by $G_k(\tau) = 2 \zeta_k E_k(\tau) $. 

{\lem
\label{8.lem:1}
The combination $\tau _2^{-r} \p_\tau \mD_{1_r}$  evaluates to the holomorphic Eisenstein series, 
\bea
\label{8.lem.1}
2 \pi i \,   \tau _2^{-r}  \, \p_\tau \mD_{1_r}(\tau) 
=  (r{+}1)  \, G_{r+2} (\tau) 
\eea
for $r \geq 1$ and therefore vanishes for $r$ odd.} 

\begin{proof}
\vskip 0in
For genus one and $r \geq 1$ the definition of $\mD_{1_r} (\tau) $ given in (\ref{7.f.1}) simplifies with the help of translation invariance on the torus, and we have,
\bea
\label{8.lem.2}
\mD_{1 _r} (\tau) = \cG_{(r)} (0 |\tau)
\eea
To evaluate the right side, we use the following alternative representation of 
the genus one Arakelov Green function that was given in (\ref{8.b.5}),
\bea
\label{8.lem.3}
\cG(z|\tau) = \frac{ \tau_2}{\pi} \sum _{(m,n) \not = (0,0)} { e^{ 2 \pi i (n u - mv)} \over  |m + \tau n|^2} 
\eea
where $z=u+ \tau v$ parametrizes the torus for $u,v \in \RR$ and $0 \leq u,v \leq 1$. The recursion relation for $\cG_{(r)}(z|\tau)$ in terms of the convolutions in (\ref{8.b.6}) is  solved by the following special cases of Zagier's single-valued elliptic polylogarithms~\cite{zagsvempl},
\bea
\label{8.lem.6}
\cG_{(r)} (z |\tau) = \frac{(-)^r \, \tau_2^{r+1}}{ \pi }
\sum _{(m,n) \not = (0,0)} { e^{ 2 \pi i (n u - mv)} \over ( m + \tau n)^{r+1} (m + \bar \tau n) } 
\eea
from which, using (\ref{8.lem.2}),  we readily obtain the one-loop modular graph forms \cite{DHoker:2016mwo},
\bea
\label{8.lem.7}
\mD_{1 _r} (\tau) = \frac{\tau_2^{r+1} }{ \pi }  \sum _{(m,n) \not = (0,0)} {1 \over ( m + \tau n)^{r+1} (m + \bar \tau n) }
\eea
which vanish for odd values of $r$. Evaluating their $\tau$ derivative via,
\bea
\partial_\tau \left ( \frac{\tau_2}{m\tau{+}n} \right ) = \frac{m \bar \tau{+}n}{2i (m\tau{+}n)^2}
\eea 
identifying the Eisenstein series (\ref{defgk}) and multiplying by $\tau_2^{-r}$ implies (\ref{8.lem.1}). 
\end{proof}

As a consequence of Lemma \ref{8.lem:1}, the $\mD$-dependent action (\ref{8.cor.32}) of $L^0$
can be brought into the more explicit form stated in the following corollary.

{\cor \label{nwlzero} The $\zz$-independent part $L^0(\tau)$ of the extended flat connection in (\ref{univh1})
and (\ref{unig1}) acts on the $\mt_{1,n}$ generators via,  
\bea
{} [\LL^0 (\tau) , a_i ] & = & -
\half  \sum_{r=1}^\infty  (2r{+}1) \, G_{2r+2}(\tau) \sum_{{p,q \geq 0 \atop  p+q=2r-1} }
 (-)^p  \sum_{k\not= i}  \big [ B_i ^p  t_{ik} , B_i^q t_{ik}  \big ] 
\no \\
{} [\LL^0 (\tau) , b_i ] & = & a_i
\no \\
{} [\LL^0 (\tau) , t_{ij} ] & = &  \sum_{r=1}^\infty (2r{+}1) \, G_{2r+2}(\tau) \, \big [  t_{ij},  B_i^{2r} \, t_{ij} \big ] 
\label{8.cor.35}
\eea}

\subsubsection{Identifying Tsunogai derivations}
\label{sec:idts}

As a major advantage of the presentation of the $\LL^0$ action in Corollary \ref{nwlzero},
it can be lined up with the following $\zz$-independent part of the $d\tau$ component
of the CEE connection~\cite{CEE},
\bea
\label{CEE.c01}
\LL^0_{\rm CEE}(\tau) = \Delta_0 + \sum_{r=1}^\infty  (2r{+}1) \,  G_{2r+2}(\tau) \, \delta_{2r}
\eea
where the derivations $\Delta_0$ and $\delta_{2r}$ of $\mt_{1,n}$ with $r\geq 1$ are determined by
\begin{align}
\Delta_0(a_i) &= 0 & \delta_{2r}(a_i) &= - \frac{1}{2} \sum_{k\neq i} \sum_{p=0}^{2r-1} (-1)^p \, \big[B_i^{p} t_{ik}  ,  B_{i}^{2r-1-p} t_{ik} \big]
\notag \\
\Delta_0(b_i) &= a_i & \delta_{2r}(b_i) &= 0
\notag \\
\Delta_0(t_{ij}) &= 0 & \delta_{2r}(t_{ij}) &= \big[ t_{ij} , B_i^{2r} t_{ij} \big]
\label{CEE.c02}
\end{align}
Upon comparison with (\ref{8.cor.35}), we find that the action of $L^0$ obtained from
the single-valued global flat connection exactly matches that of $\LL^0_{\rm CEE}$
from the meromorphic CEE connection on arbitrary $X \in \hat \mt_{1,n}$,
\bea
\big[ \LL^0 , X \big] =  \big[ \LL^0_{\rm CEE} , X \big] \label{CEE.19}
\eea
Similar to the observation on $L_{ww}^0$ in Remark \ref{innerrmk},  the action (\ref{CEE.c02}) of $\delta_{2r}$ on $t_{ij}$ can be absorbed into inner derivations of $\hat \mt_{1,n}$ generated by, 
\bea
\ell(\tau) = \frac{1}{2}  \sum_{r=1}^\infty (2r{+}1)  \, G_{2r+2}(\tau) \sum_{p\neq q} B^{2r}_p  t_{pq}
 \label{CEE.20}
\eea
which trivializes the action of the following series on $t_{ij}$:
\bea
\label{CEE.21}
\LL^0_{\rm CEE}(\tau)  + \ell(\tau)  = \Delta_0 + \sum_{r=1}^\infty  (2r{+}1)  \, G_{2r+2}(\tau) \, \epsilon_{2r+2}
\eea
The coefficients of the holomorphic Eisenstein series $G_{k}(\tau)$ of modular weight $k\geq 4$ are then given by 
Tsunogai's derivations $\epsilon_{k}$ of $\mt_{1,n}$ \cite{CEE}. More specifically, by
combining the contributions from (\ref{CEE.c02}) and (\ref{CEE.20}) to
\bea
\epsilon_{2r}(X) = \delta_{2r-2}(X) + \frac{1}{2} \sum_{p\neq q}  \, \big[ B_p^{2r-2} t_{pq},  X \big]
\label{CEE.c03}
\eea
for arbitrary $X \in \hat \mt_{1,n}$ and $r\geq 2$, we arrive at $\epsilon_{2r}(t_{ij}) = 0$ by construction as well as
\begin{align}
\epsilon_{2r}(a_i) &=  \big[ a_i , B^{2r-1}_i a_i  \big] + \sum_{\ell=1}^{r-1} (-1)^\ell \, \big[B_i^\ell a_i , B_i^{2r-1-\ell} a_i  \big]
\notag \\
\epsilon_{2r}(b_i) &= B_i^{2r} a_i
\label{CEE.c04} 
\end{align}
The inner derivation (\ref{CEE.20}) exactly matches the restriction of $\ell_{ww}$ in (\ref{defbL0}) to
genus $h=1$ after converting $\delta_{ww} \rightarrow 2\pi i \partial_\tau$ as in (\ref{univh1})
and recovering Eisenstein series via (\ref{8.lem.1}). Together with (\ref{CEE.19}), this implies that
the redefined connection $ \bL_{ww}^0$ in Remark \ref{innerrmk} reduces to a series in Tsunogai's derivations
at $h=1$,
\bea 
\bL_{\alpha}^0(\tau) =   \Delta_0 + \sum_{r=1}^\infty  (2r{+}1)  \, G_{2r+2}(\tau) \, \epsilon_{2r+2}
\label{CEE.c05} 
\eea
We therefore regard the action (\ref{7.thm.2}) of $ \bL_{ww}^0$ on $a_{iI}$, $b_i^I$ (together with 
$[ \bL_{ww}^0, t_{ij} ] = 0$) as a generalization of the Tsunogai derivations acting via  (\ref{CEE.c04}) to higher genus. In particular, the separation of $\mt_{h,n}$ generators $a_{iI}$, $b_i^I$ with different values of $i \in \{ 1,\cdots,n \} $ is preserved in passing from (\ref{CEE.c04}) at genus one to (\ref{7.thm.2}) at arbitrary genus. By the lowest-degree terms,
\bea
[ {\bf L}_{ww}^0 , a_{iI}] = {\cal O}(a_i^2) \, , \ \ \ \ \ \
[ {\bf L}_{ww}^0 , b_{i}^I]  = \omega^I(w) \omega^J(w) a_{iJ}+  {\cal O}(b_i a_i)
\label{CEE.31} 
\eea
in (\ref{7.thm.2}), the degree
zero derivation $\Delta_0$ in (\ref{CEE.c02}) with customary notation $\epsilon_{0} = \Delta_0$ 
which is part of an $\mathfrak{sl}(2)$ algebra generalizes to generators of $\mathfrak{sp}(2h)$.
While the genus one connection (\ref{CEE.c05}) only features a single modular form $G_k(\tau)$
and associated derivation $\epsilon_k$ at a given degree in $b_i$, the action
(\ref{7.thm.2}) at higher genus generically introduces multiple independent modular tensors 
per degree in the place of $G_k(\tau)$. The classification of these modular tensors and
their relations (see section \ref{sec:77.lin}) as well as the
identification of elementary derivations of $\hat \mt_{h,n}$ at fixed $b_i$ degree 
entering (\ref{7.thm.2}) is relegated to future~work.

\subsubsection{Evaluating $\LL^1$}

For genus one, the $\zz$-dependent part $\LL^1$ of the $d \tau$ component (\ref{unig1}) of the extended
 flat connection (\ref{univh1}) is obtained by integrating (\ref{7.c.9}) against the Beltrami differential. The contribution from $\cU_{ww}$  vanishes for this case, and we are left with,
\bea
\LL^1(\zz| \tau) = 2 \pi i \, \p_\tau \cW(\zz| \tau)
\label{pulldtau}
\eea
where $\cW$ is given by the restriction of (\ref{7.c.7}) to genus one, 
\bea
\cW (\zz| \tau) = \half \sum_{k \not= \ell} \sum_{r=0}^\infty \cG_{(r)} (z_k- z_\ell|\tau) (B_k^1)^ r t_{k \ell} 
\eea
Recalling that $\p_\tau b_k^1=  \p_\tau t_{k\ell}  =0$ since  $\delta_{ww} b^1_k  = \delta_{ww} t_{k\ell}  =0$ in view of Definition \ref{3.def:3}, the only $\tau$-dependence above is through $\cG_{(r)}$, whose $\tau$-derivative is conveniently evaluated using the representation (\ref{8.lem.6}). Recalling that the coordinate $z$ depends on $\tau$ in view of the identification $z\cong z+\tau$ along the $\mB_1$ cycle but that the real coordinates $u,v$ in which we decompose $z = u + \tau v$ remain constant during variations in $\tau$, we find, 
\bea
2 \pi i \p_\tau \cG_{(r)} (z|\tau) =  \! \!
\sum _{(m,n) \not = (0,0)}  { (-1)^r (r{+}1) \tau_2^r  \over  ( m + \tau n)^{r+2} } \, e^{ 2 \pi i (n u - mv)}
=  - (r{+}1)   \, \tau_2^r \, f^{(r+2)} (z|\tau) \ \ 
\label{abstau2}
\eea
where the Kronecker-Eisenstein kernels $f^{(s)} (z|\tau) $ were identified
through their lattice-sum representation equivalent to (\ref{8.lem.6}) and (\ref{dgisf}),
\bea
f^{(s)} (z|\tau) = - (-1)^s
\sum _{(m,n) \not = (0,0)}  { e^{ 2 \pi i (n u - mv)} \over ( m + \tau n)^{s} } 
\label{latfs}
\eea
On these grounds, our final result for $\LL^1$ is given by (absorbing the factors
of $\tau_2$ on the right side of (\ref{abstau2}) into $\tau_2^r ( B_k^1)^r = B_k^r$),
\bea
\LL^1(\zz| \tau) =  - \half \sum _{k \not = \ell} \sum_{r=0}^\infty
 (r{+}1) \,  f^{(r+2)}(z_k - z_\ell|\tau) \, B_k^r \, t_{k \ell} 
\label{finL1}
\eea
which closely resembles the $\zz$-dependence of the analogous $d\tau$ component of the CEE connection 
\cite{CEE}: the latter is formally obtained from the single-valued and modular $\LL^1(\zz| \tau) $
in (\ref{finL1}) by converting $ f^{(r+2)} \rightarrow g^{(r+2)}$ to the meromorphic
Kronecker-Eisenstein kernels. This formal relation via $ f^{(r+2)} \rightarrow g^{(r+2)}$
signals that our single-valued and modular global flat connection (\ref{univh1}) is
obtained from a gauge transformation of the meromorphic CEE connection. In the 
case of a single variable $z_i =u_i + \tau v_i$ on the torus, the replacement $ f^{(r+2)} \rightarrow g^{(r+2)}$
is attained through a gauge transformation with logarithm $\sim u_i b_i$, see section 2.3 of \cite{DHoker:2025szl} 
for the case of the Brown-Levin connection
and section 5.1 of \cite{Schlotterer:2025qjv} for its extension to the moduli space ${\cal M}_{1,2}$.
In our setting of multiple variables $z_i$ on the torus with $i \in \{ 1,\cdots,n \} $, the analogous gauge transformation
between the meromorphic CEE connection and the single-valued and modular extended flat connection (\ref{univh1}) is given by the exponential of $\sum_{i=1}^n u_i b_i$.

\newpage

\appendix

\section{Combinatorial and Lie algebra relations}
\setcounter{equation}{0}
\label{sec:AA}

We collect a number of combinatorial and Lie algebraic relations that were proven in \cite{DHoker:2026ggx}. 

\subsection{Shuffle identities}

The following deconcatenation sums of shuffle products were proven in appendix A of~\cite{DHoker:2026ggx}, 
\begin{align}
\sum_{\vI = \vP \vQ } \vP \shuffle \theta(\vQ) &= \delta _{\vI, \emptyset} 
&
\delta_{\vI, \emptyset} &= \left\{ \begin{array}{cl} \emptyset & \hbox{ for } \,  \vI= \emptyset \\
0 & \hbox{ for } \,  \vI \neq \emptyset \end{array} \right.
\label{A.sh.1}
 \\
\sum_{ \vI = \vec{P} \vec{Q} } (\vR \shuffle \vec{P}) L (\vS \shuffle \vec{Q}) & = \vR L\vS \shuffle \vI  
& 
& \vI,\vR,\vS  \in \cW_h
\label{A.sh.2}
\\
\sum_{\vI = \vP \vQ}  \vQ \shuffle \big(\vR \shuffle  \theta(\vP) \big)L & =  \vR L \vI 
&
& \vI, \vR \in \cW_h
\label{A.sh.4}
\end{align}
where $\theta$ is the antipode defined in footnote \ref{other.4} and $\cW_h$ is the set of all words in an alphabet of~$h$ letters. By iterating (\ref{A.sh.2}) we obtain the general formula, 
\bea
\label{A.sh.5}
\vI \shuffle \vR L_1 \cdots L_\ell \vS = \sum_{\vI = \vX_1 \cdots \vX_{\ell+1}} 
\big ( \vX_1 \shuffle \vR \big ) L_1 \vX_2 L_2 \cdots L_{\ell-1} \vX_\ell L_\ell \big ( \vX_{\ell+1} \shuffle \vS \big )
\eea

\subsection{Lie algebra identities}
\label{app:lie}

The following Lie algebra relation  was originally stated and proven in Lemma E.2 of  \cite{DHoker:2026ggx}, 
\bea
\label{A.Lie.1}
\big [ B_i ^\vI \, a_{iK} + B_j ^ \vI \, a_{jK}, t_{ij} \big ]
= - \sum _{\vI = \vP K \vQ} 
\sum_{\vP = \vX \shuffle \vY} \Big [ B_i ^\vX \big ( B_i^\vQ + B_i^{\theta(\vQ)} \big ) B_i^{\theta (\vY)}  t_{ij}, t_{ij} \Big ]
\eea
while the Lie algebra relation below, for an arbitrary $X \in \hat \mt_{h,n}$, will also be needed,
\bea
\label{A.Lie.2}
{} \big [ a_{iI}, B_j^{\vJ} X \big ] = B_{j }^\vJ \,  \big [ a_{iI}, X \big ] 
+ \sum _{\vJ = \vR L \vS} \delta_I^L \, B_{j}^ \vR \big [ B_{j }^\vS X, t_{ij} \big ] 
\eea

\begin{proof}
The proof is by induction on the length $s$ of the word $\vJ=J_1 \cdots J_s$. The statement clearly holds when $s=0$ since there are then no contributions to the second term on the right. For $s \geq 1$ we peel off one factor of $B_j$ followed by $B_j^J a_{iI}= \delta^J_I t_{ij}$,
\bea
{}  \big [ a_{iI}, B_j^{J_1} \cdots B_j^{J_s} X \big ] = B_j^{ J_1} \big [ a_{iI}, B_j^{J_2} \cdots B_j^{ J_s} X \big ]
+ \delta_I^{J_1} \big [ B_j^{J_2} \cdots B_j^{J_s} X, t_{ij} \big ] 
\eea
and then use the validity of the relation at orders $\leq s-1$ to evaluate the first term. 
\end{proof}

\newpage

\section{Coincident limits of the Fay identities}
\setcounter{equation}{0}
\label{sec:A}

In this appendix, we construct the one coincident limit of the Fay identities involving three points $x,y,z \in \Sigma$ that was not constructed in \cite{DHoker:2024ozn}, namely from the  $(1,0)_x \otimes (1,0)_y \otimes (0,0)_z$ form of the Fay identities to a $(2,0)_w \otimes (0,0)_z$ form when $x, y \to w$. The starting point is the general Fay identity in Theorem 6.2 of \cite{DHoker:2024ozn} where the product of two DHS kernels with a common scalar point $z$  is expressed in terms of a sum of binary products where at most one factor is $z$-dependent. For the problem at hand, it will suffice to take the trace in the indices $K$ and $M$ of this identity, whose decomposition  into $\p \Phi $ and $\p \cG$ may be recast as follows after some simplifications, 
\footnote{For later convenience we change variables from $(x,y)$ in the original formula of Theorem 5.2  for the interchange identities to $(w,y)$ and from $(x,y,z)$ in the original formula of Theorem 6.2 for the Fay identities to $(w,y,x)$ here.}
\bea
\label{A.2}
&&
\p_w \cG_\vI (w,x) \p_y \cG_\vJ (y,x) = - { 1 \over h} \om^M(y) \p_w \cG_{\vI M \theta (\vJ)} (w,y)
 \\ && \qquad 
+ \sum _{\vI = \vP \vQ} \Big ( \p_w \cG_\vP (w,y) \, \p_y \cG_{\vJ \shuffle \vQ} (y,x) 
- \p_w \Phi _\vP{}^M (w) \, \p_y \cG _{\vJ \shuffle M\vQ} (y,x) \Big )
\no \\ && \qquad
+ \sum _{\vJ = \vR \vS} \Big ( \p_y \cG_\vR (y,w) \, \p_w \cG_{\vI \shuffle \vS} (w,x) 
- \p_y \Phi _\vR {}^M(y)  \, \p_w \cG_{\vI \shuffle M \vS} (w,x) \Big )
\no \\ && \qquad
+ { 1 \over h} \sum _{\vJ = \vR \vS} \Big ( \p_w \Phi _{\vI M \theta (\vS)K}{}^M(w) \, \p_y \Phi _\vR{}^K(y) 
+ \p_w \Phi _{\vI M \theta(\vS)}{}^K  (w) \, \p_y \Phi _{\vR K} {}^M(y) \Big )
\no
\eea
Next, we take the limit of the above Fay identity as $ y \to w$ with the help of the following limits of the individual components.

\subsection{Coincident limits of DHS kernels}
\label{app:B.cin}

The limits $y \rightarrow w$ of (\ref{A.2}) exist and for the terms of the form 
$\p_w \cG_\vP (w,y)$ and $\p_y \cG_\vR (y,w) $ boil down to the following combinations (see section 8.3 of \cite{DHoker:2024ozn}), 
\bea
\label{A.3}
\cC_\emptyset (w) & = & \lim_{y \to w} \Big \{ \p_w \cG(w,y) + { 1 \over w-y} \Big \} 
\no \\
\cC_I (w) & = & \lim_{y \to w} \Big \{ \p_w \cG_I(w,y) - \pi { \bar w - \bar y \over w-y} \bar \om_I(w)  \Big \} 
\no \\
\cC_\vI (w) & = & \lim_{y \to w} \p_w \cG_\vI (w,y)  \hskip 2in |\vI | \geq 2  
\eea
The $(1,0)$ form $\cC_\vI(w)$ is single-valued and well-defined for $\vI \not= \emptyset$, but $\cC_\emptyset(w)$ transforms as a connection under conformal transformations.  Next, we recall the decomposition of $\cC_\vI(w)$ with $\vI \not= \emptyset$ in terms of DHS kernels from Theorem 8.2 of \cite{DHoker:2024ozn},
\bea
\label{A.4}
\cC_\vI (w) = \sum_{\vI = \vP \vQ \vR} 
\Big \{ \p_w \Phi _{\vP \shuffle \theta(\vR)} {}^M(w) \, \mN_{M \vQ} 
+ \p_w \Phi _{\big ( \vP \shuffle \theta(\vR) \big  ) M \vQ} {}^M(w) \Big \}
\eea
where the modular tensors $ \mN_{M\vQ}$ defined by (\ref{defcn}) satisfy the following dihedral symmetry relations, 
\bea
\label{A.6}
\mN_{M \vQ} = \mN_{\vQ M} 
\hskip 1in 
 \mN_{\theta(M \vQ)} =  \mN _{M \vQ}
\eea
These properties allow us to derive the following lemma.

{\lem
\label{A.lem:1}
The coincident limit  $\p_w \cG_\vI(w,w) $ may be expressed in terms of DHS kernels,
\bea
\label{A.lem.1}
\p_w \cG_\vI (w,w) = \sum_{\vI = \vP \vQ \vR}  \p_w \Phi _{\big ( \vP \shuffle \theta(\vR) \big  ) M \big ( \vQ + \theta(\vQ) \big ) } {}^M(w) 
\eea \sm}

\begin{proof}
\vskip -0.2in
On the one hand, the definition of $\cC_\vI$ in (\ref{A.3}) implies,
\bea
\label{A.1}
\p_w \cG_\vI (w,w) = \cC_\vI (w) + \cC_{\theta(\vI)} (w)
\eea
On the other hand $\cC_{\theta(\vI)}(w)$ obtained from (\ref{A.4}) takes the form,
\bea
\label{A.5}
\cC_{ \theta(\vI)}  (w) = \sum_{\vI = \vP \vQ \vR} 
\Big \{ \p_w \Phi _{\vP \shuffle \theta(\vR)} {}^M(w) \,  \mN_{M \theta(\vQ)} 
+ \p_w \Phi _{ ( \vP \shuffle \theta(\vR) ) M \theta(\vQ)} {}^M(w) \Big \}
\eea
Using the consequence $\mN_{M \theta(\vQ)}= - \mN_{M \vQ}$ of (\ref{A.6}), we see that the contribution from $ \mN$ cancels in (\ref{A.1}), while the remaining terms may be regrouped as in (\ref{A.lem.1}). 
\end{proof}

\subsection{Coincident limits of the trace of the Fay identity}

In taking the limit $y \to w$ of the trace of the Fay identity in (\ref{A.2}), we let $\p_w \cG_{\vP}(w,y) \to \cC_\vP(w)$ and pick up the following extra contributions for $|\vP|=0,1$,
\bea
&&
- { 1 \over w-y} \Big ( \p_y \cG_{\vI \shuffle \vJ} (y,x) - \p_w \cG_{\vI \shuffle \vJ}(w,x) \Big )
\\ &&
+ \pi { \bar w - \bar y \over w-y} \Big ( \p_y \cG _{I_2 \cdots I_r \shuffle \vJ}(y,x) \,  \bar \om _{I_1} (w) 
+ \p_w \cG _{\vI \shuffle J_2 \cdots J_s} (w,x) \,  \bar \om _{J_1} (w) \Big )
\no 
\eea
Expanding the terms inside the parentheses on the first line to first order,
\bea
&&\p_y \cG_{\vI \shuffle \vJ} (y,x) - \p_w \cG_{\vI \shuffle \vJ}(w,x)
 \\ && \quad 
= (y-w)  \p_w^2 \cG_{\vI \shuffle \vJ}(w,x) + ( \bar y - \bar w) \pbw \p_w \cG_{\vI \shuffle \vJ}(w,x) + \cO\big(|y-w|^2 \big)
\no \\ && \quad 
= (y-w)  \p_w^2 \cG_{\vI \shuffle \vJ}(w,x)  
- \pi  ( \bar y - \bar w) 
\Big (  \p_w \cG_{I_2 \cdots I_r \shuffle \vJ}(w,x)   \bar \om_{I_1} (w) 
\no \\ && \hskip 2.7in
+ \p_w \cG_{\vI  \shuffle J_2 \cdots J_s }(w,x)  \bar \om_{J_1} (w)  \Big ) +  \cO \big(|y-w|^2 \big)
\no
\eea
All the terms proportional to $\bar \om_K(w)$ cancel one another in the limit leaving the sole contribution $\p_w^2 \cG_{\vI \shuffle \vJ}(w,x)$, so that the coincident limit of the Fay identity of (\ref{A.2}) becomes, 
\bea
\label{A.10}
&&
\p_w \cG_\vI (w,x) \p_w \cG_\vJ (w,x) 
= - { 1 \over h} \om^L(w) \cC_{\vI L \theta (\vJ)} (w) + \p_w^2 \cG_{\vI \shuffle \vJ}(w,x)
 \\ && \qquad \quad
+ \sum _{\vI = \vP \vQ} \Big ( \cC_\vP (w) \, \p_w \cG_{\vJ \shuffle \vQ} (w,x) 
- \p_w  \Phi _\vP{}^ L (w) \, \p_w \cG _{\vJ \shuffle L \vQ} (w,x) \Big )
\no \\ && \qquad \quad
+ \sum _{\vJ = \vR \vS} \Big ( \cC_\vR (w) \, \p_w \cG_{\vI \shuffle \vS} (w,x) 
- \p_w  \Phi_\vR{}^ L(w)  \, \p_w \cG_{\vI \shuffle L \vS} (w,x) \Big )
\no \\ && \qquad \quad
+ { 1 \over h} \sum _{\vJ = \vR \vS} \Big ( \p_w \Phi _{\vI L \theta (\vS)K}{}^L (w) \, \p_w \Phi _\vR{}^K(w)
+ \p_w \Phi _{\vI L \theta(\vS)}{}^K  (w) \, \p_w \Phi _{\vR K} {}^L (w) \Big )
\no
\eea
The double derivative $\p_w^2 \cG_{\vI \shuffle \vJ}(w,x)$ is not a genuine $(2,0)$ form, but in the above formula it is accompanied by the contributions from $\cC_\emptyset(w)$ from the second and third lines, so that the combination $\p_w^2 \cG_{\vI \shuffle \vJ}(w,x) + 2 \, \cC_\emptyset (w) \cG_{\vI \shuffle \vJ}(w,x)$ plays the role of a covariant derivative and transforms as a genuine $(2,0)$ form, as do all other contributions. In our application of (\ref{A.10}) to the integrand of (\ref{9.q.14}) below, both of $\p_w^2 \cG_{\vI \shuffle \vJ}(w,x)$ and $\cC_\emptyset (w) \cG_{\vI \shuffle \vJ}(w,x)$ drop out upon integration against $\kappa(x)$.

\newpage

\section{Proof of Proposition \ref{7.prop:1re}}
\setcounter{equation}{0}
\label{sec:B}

Starting from the definition of $\mD$ in (\ref{7.f.1}), we evaluate the combination $\delta _K^J \, \mD_{\vI L}$
by eliminating the volume form $\kappa$ in the integration using the relation,
\bea
\label{B.1}
\delta ^J_K \, \kappa (x) = \om^J(x) \bar \om_K(x) + { 1 \over \pi} \, \pbx \p_x \Phi _K {}^J(x)
\eea
to obtain, 
\bea
\label{B.2}
\delta _K^J \, \mD_{\vI L} 
& = & 
\int_\Sigma d^2 x \, \om^J(x)  \bar \om_K(x) \cG_{\vI L} (x,x) 
+{ 1 \over \pi} \int _\Sigma d^2 x \, \pbx \p_x \Phi _K{}^J(x) \,  \cG_{\vI L} (x,x) 
\eea
Integrating by parts in $\p_x$ and $\pbx$ in the second term, eliminating $\p_x \cG$ using the coincident limit formula (\ref{A.lem.1}) of Lemma \ref{A.lem:1}, taking the $\pbx$ derivative using the first line in (\ref{2.a.7}) and then carrying out the integral in $x$, we obtain,
\bea
\label{B.3}
\delta _K^J \, \mD_{\vI L} 
& = & 
\int_\Sigma d^2 x \, \om^J(x)  \bar \om_K(x) \cG_{\vI L} (x,x) 
- \sum_{\vI L = \vX \vY \vZ} \cA^J {}_{K \big ( \vX \shuffle \theta(\vZ) \big ) M \big ( \vY + \theta (\vY) \big )} {}^M
\eea 
in terms of the modular tensors $\cA$ of (\ref{7.f.1}). Next, we evaluate,
\bea
\label{B.4}
\cA^M{}_{K \vI L M}{}^J 
& = & \int _\Sigma d^2 x \int _\Sigma d^2 y \, \om^M(x) \, \bar \om_K(x) \, \cG_{\vI L}(x,y) \, 
\bar \om_M(y) \, \om^J(y) 
\eea
by eliminating the combination $\om^M(x) \bar \om_M(y)$ using the relation, 
\bea
\label{B.5}
\om^M(x) \bar \om_M(y) = \delta(x,y) - { 1 \over \pi} \, \pby \p_x \cG(x,y)
\eea
to obtain upon integrating by parts in $\p_x$ and $\pby$, 
\bea
\label{B.6}
\cA^M{}_{K \vI L M}{}^J 
= \int _\Sigma d^2 x  \, \bar \om_K(x) \Big ( \cG_{\vI L}(x,x)  \om^J(x) 
- { 1 \over \pi}  \int _\Sigma d^2 y \,   \cG(x,y) \, \pby \p_x \cG_{\vI L}(x,y) \, \om^J(y) \Big )
\quad
\eea
The contribution from the $y$-integral may be computed using the last equation in (\ref{2.a.7}), 
\bea
\label{B.7}
\cA^M{}_{K \vI L M}{}^J 
& = & \int _\Sigma d^2 x  \, \bar \om_K(x) \cG_{\vI L}(x,x)  \om^J(x) 
 \\ &&
+ \int _\Sigma d^2 x \int _\Sigma d^2 y \,   \cG(x,y) \, \bar \om_K(x) \, \bar \om_M(y) \Big ( \p_x \Phi _{\vI L}{}^M(x) 
- \delta ^M_L \p_x \cG_{\vI} (x,y) \Big ) \om^J(y)
\no
\eea
The $x$-integrals in the second line may both be performed in terms of DHS kernels. The first integral gives another instance of $\cA$ in (\ref{B.4}) while the second gives a term of the form of the first line after integrating over $x$ via (\ref{2.a.4}) and then renaming $y\rightarrow x$, 
\bea
\label{B.8}
\cA^M{}_{K \vI L M}{}^J - \cA^J {}_{M K \vI L}{} ^M 
& = & \int _\Sigma d^2 x  \, \om^J(x) \Big ( \bar \om_K(x) \cG_{\vI L}(x,x)  
- \bar \om_L(x)  \cG_{K \vI} (x,x) \Big )
\eea
The integrals on the right side do not line up with the definition (\ref{7.f.1}) of $\cD$ or $\cA$ tensors but can be eliminated by rewriting the left side of Proposition \ref{7.prop:1re} via (\ref{B.3}). In this way, we obtain,
\bea
\label{B.9}
\delta _K^J \, \mD_{\vI L} - \mD_{K \vI} \, \delta ^J_L
& = & 
- \cA^J{}_{ML \theta(\vI)K }{}^M  
- \sum_{\vI L = \vX \vY \vZ} \cA^J {}_{K \big ( \vX \shuffle \theta(\vZ) \big ) M \big ( \vY + \theta (\vY) \big )} {}^M
\no \\ &&
- \cA^J {}_{M K \vI L}{} ^M 
+ \sum_{K \vI  = \vX \vY \vZ} \cA^J {}_{L \big ( \vX \shuffle \theta(\vZ) \big ) M \big ( \vY + \theta (\vY) \big )} {}^M
\eea 
To match the right side of Proposition \ref{7.prop:1re}, we separate the first sum  into the contributions from $\vZ= \emptyset$ with $\vI L = \vX \vY$ and those from $\vZ \to \vZ L$ with $\vI = \vX \vY \vZ$, and the second sum into those from $\vX= \emptyset$ with $K \vI = \vY \vZ$ and those from $\vX \to K \vX$ with $\vI = \vX \vY \vZ$,
\bea
\label{B.10}
\delta _K^J \, \mD_{\vI L} - \mD_{K \vI} \, \delta ^J_L
& = & 
- \cA^J{}_{ML \theta(\vI)K }{}^M  - \cA^J {}_{M K \vI L}{} ^M 
 \\ &&
- \sum_{\vI L = \vX \vY} \cA^J {}_{K \vX M \big ( \vY + \theta (\vY) \big )} {}^M
+ \sum_{\vI  = \vX \vY \vZ} \cA^J {}_{K \big ( \vX \shuffle L \theta(\vZ) \big ) M \big ( \vY + \theta (\vY) \big )} {}^M
\no \\ &&
+ \sum_{K \vI  = \vY \vZ} \cA^J {}_{L \theta(\vZ) M \big ( \vY + \theta (\vY) \big )} {}^M
+ \sum_{\vI  = \vX \vY \vZ} \cA^J {}_{L \big ( K \vX \shuffle \theta(\vZ) \big ) M \big ( \vY + \theta (\vY) \big )} {}^M
\no
\eea 
The terms on the right may be combined into the expression given in (\ref{7.f.3}),
\bea
\label{B.11}
 \mD_{K \vI} \, \delta ^J_L - \delta _K^J \, \mD_{\vI L} 
& = & 
 \sum_{K \vI L = \vX \vY \vZ} \cA^J {}_{ \big ( \vX \shuffle \theta(\vZ) \big ) M \big ( \vY + \theta (\vY) \big )} {}^M
\eea 
To show this, we similarly separate the deconcatenation sum into the contributions from $\vZ= \emptyset$ with $K \vI L = \vX \vY$ and those from $\vZ \to \vZ L$ with $K \vI = \vX \vY \vZ$,
\bea
\label{B.12}
\delta _K^J \, \mD_{\vI L} - \mD_{K \vI} \, \delta ^J_L
& = & 
- \!\! \sum_{K \vI L = \vX \vY} \!\! \cA^J {}_{ \vX  M \big ( \vY + \theta (\vY) \big )} {}^M
+ \! \! \sum_{K \vI  = \vX \vY \vZ} \!\! \cA^J {}_{ \big ( \vX \shuffle L \theta(\vZ) \big ) M \big ( \vY + \theta (\vY) \big )} {}^M
\qquad
\eea 
Decomposing each sum into the contributions from $\vX = \emptyset$ with $K \vI L = \vY$ and $K \vI = \vY \vZ$, respectively, and those from $\vX \to K \vX$ with $\vI L - \vX \vY$ and $\vI = \vX \vY \vZ$, respectively, 
\bea
\label{B.13}
\delta _K^J \, \mD_{\vI L} - \mD_{K \vI} \, \delta ^J_L
& = & 
- \cA^J {}_{ M K \vI L} {}^M - \cA^J {}_{ M  \theta (K \vI L) } {}^M
+ \sum_{\vI  = \vX \vY \vZ} \cA^J {}_{ \big ( K \vX \shuffle L \theta(\vZ) \big ) M \big ( \vY + \theta (\vY) \big )} {}^M
\no \\ &&
- \sum_{\vI L = \vX \vY} \cA^J {}_{ K \vX  M \big ( \vY + \theta (\vY) \big )} {}^M
+ \sum_{K \vI  = \vY \vZ} \cA^J {}_{  L \theta(\vZ) M \big ( \vY + \theta (\vY) \big )} {}^M
\eea 
Using the defining recursion relation for the shuffle product in the argument of the  third term on the first line, we find perfect agreement with the expression obtained in (\ref{B.9}) and rewritten as (\ref{B.10}). Using (\ref{A.sh.2}), the result of (\ref{B.11}) may be re-expressed as follows, 
\bea
\label{B.14}
 \mD_{K \vI } \, \, \delta _L^J  - \delta _K^J \, \mD_{\vI L} 
&=&
 \sum_{K \vI L = \vX \vY} \Big (  \cA^J {}_{\vX \shuffle \theta(\vY)M }{}^M  + \cA^J {}_{\vX M \shuffle \theta(\vY)}{}^M \Big )
\eea

\newpage

\section{Proof of the second relation in (\ref{muder1})}
\label{app:new}

It suffices to prove the relation for $n=1$, as the case $n=-1$ follows by duality and the cases  $|n| >1$ follow by tensor product. 
For $n=1$, the $\delta_{\mu}$ variation in (\ref{muder1}) of a one-form $\om = \om_m \, d\xi^m$ solely requires the projection onto $(0,1)$ forms of the covariant derivative
\begin{align}
(\nabla \om )_{np} = \p_n \om_p - \Gamma ^q_{np} \om_q  \hskip 1in
\Gamma ^q _{np} = \half g^{qr} \Big ( \p_n g_{pr} + \p_p g_{nr} - \p_r g_{np} \Big )
\end{align}
with respect to the Christoffel connection $\Gamma $ for a Riemannian metric $g = g_{mn} d \xi^m d \xi^n$ expressed in arbitrary real or complex local coordinates $\xi^m$ with $\p_m = \p / \p \xi^m$. The $(0,1)$-projection relevant for $\delta_\mu \omega$ is then given by
\bea
(\nabla ^- \om)_{mp} =\half \big ( \delta _m{}^n + i J_m{}^n \big ) (\nabla \om)_{np}
\eea
and commutes with the covariant derivative.
Parametrizing the metric by $g = e^\phi |dz{-} \mu d \bar z|^2$ and retaining only terms of order 0 and 1 in $\mu$ and $\bar \mu$ (namely, neglecting terms of order $\mu^2, \bar \mu^2$ and $\mu \bar \mu$), the components of the metric and its inverse are given, to this order,  by,
\begin{align}
g_{z \bar z} & = \thalf e^\phi & g_{zz} & = - e^\phi \bar \mu & g_{\bar z \bar z} & = - e^\phi \mu
\no \\
g^{z \bar z} & = 2 e^{-\phi } & g^{zz} & = 4 \, e^{- \phi} \mu  & g^{\bar z \bar z} & = 4 \, e^{-\phi} \bar  \mu
\end{align}
The resulting components of the Christoffel connection are given by,
\bea
\Gamma ^{\bar z}  _{\bar z \bar z} & = & \pbz \phi + \mu \p_z \phi + \p_z \mu
\no \\
\Gamma ^{z} _{\bar z \bar z} & = &  - \pbz \mu + \mu \pbz \phi
\no \\
\Gamma ^z_{z \bar z} = \Gamma ^z_{\bar z z} & = & - \p_z \mu - \mu \p_z \phi
\eea
and their complex conjugates. Evaluating the projected covariant derivative on the coefficient function $\om_z$ of a $(1,0)$ form $\om = \om_z dz$ gives,
\bea
(\nabla ^- \om)_{\bar z z} & = &
\pbz \om_z + \mu \p_z \om_z - \half (\delta _{\bar z} {}^z + i J_{\bar z} {}^z) \Gamma _{z z} ^z \om_z
- \half (\delta _{\bar z} {}^{\bar z}  + i J_{\bar z} {}^{\bar z} ) \Gamma _{\bar z z} ^z \om_z
\eea
Using $\delta _{\bar z} {}^{\bar z}=1$, $\delta _{\bar z} {}^z=0$, $J_{\bar z} {}^z= - 2 i \mu$ and $J_{\bar z} {}^{\bar z} =-i$, the formula simplifies as follows,
\bea
(\nabla ^- \om)_{\bar z z} & = &
\pbz \om_z + \mu \p_z \om_z - \mu \Gamma ^z _{zz} \om_z   - \Gamma _{\bar z z} ^z \om_z 
\no \\
& = &
\pbz \om_z + \mu \p_z \om_z +(\p_z \mu)  \om_z
\eea
which is  (\ref{muder1}) for $n=1$.

\newpage

\section{Proof of Proposition \ref{3.prop:40}}
\setcounter{equation}{0}
\label{sec:Arak}

In this appendix, we prove Proposition \ref{3.prop:40} which gives the variational derivative of the Arakelov Green function  in equations (\ref{muder.4}) and (\ref{muder.5}). The starting point for the proof is the explicit expression for the Arakelov Green function $\cG(x,y)$ in terms of Abelian integrals and the prime form obtained in  \cite{DHoker:2017pvk}. The construction uses the ``string Green function" $G(x,y)$ which, like the Abelian integrals and the prime form, is not single-valued on $\Sigma$ but rather should be defined on a given fundamental domain $D$ for $\Sigma$, as follows, 
\bea
\label{Arak.1}
G(x,y) = - \ln |E(x,y)|^2 + 2 \pi \Big ( \Im \int ^x _y \om_I \Big ) Y^{IJ} \Big ( \Im \int ^x _y \om_J \Big )
\eea
The relation is as follows, 
\bea
\label{Arak.2}
\cG(x,y) = G(x,y) - \gamma (x) - \gamma (y) + \gamma _0
\eea
where $\gamma(x)$ and $\gamma_0$ are given in terms of the K\"ahler form $\kappa(z)$ of (\ref{2.a.2}) by, 
\bea
\label{Arak.3}
\gamma (x) = \int _D d^2 z \, \kappa (z) \, G(x,z) \hskip 1in \gamma_0 = \int _D d^2 z \, \kappa (z) \, \gamma (z)
\eea
While $G(x,y)$, $\gamma (x) $ and $\gamma_0$ depend on the choice of the fundamental domain $D$, the combination $\cG(x,y)$ is single-valued in $x,y \in \Sigma$ and independent of the choice of $D$. 

\sm

Using the variational derivative for the prime form derived in \cite{Verlinde:1986kw}, 
\bea
\label{Arak.4}
\delta _{ww} \ln E(x,y) = - \half \Big ( \p_w \ln E(w,x)  - \p_w \ln E(w,y) \Big )^2
\eea
and the formulas for the variational derivatives of $\om_I$, $\bar \om_I$ and $Y$ given in (\ref{muder.2}), one readily evaluates the variational derivative of $G(x,y)$, and we obtain,
\bea
\label{Arak.5}
\delta _{ww} G(x,y) & = & \half \Big ( \p_w G(w,x) - \p_w G(w,y) \Big )^2
\no \\
& = & \half \Big ( \p_w \cG(w,x) - \p_w \cG(w,y) \Big )^2
\eea
where we passed from the first line to the second using (\ref{Arak.2}). The variational derivatives of the combinations $\gamma(x)$ and $\gamma _0$ may then be evaluated using (\ref{muder.3}) and the above expression for $\delta_{ww} G(x,y)$ expressed in terms of the Arakelov Green function on the second line of (\ref{Arak.5}), and we obtain,
\bea
\delta _{ww} \gamma (x) & = & 
\half \big ( \p_w \cG(w,x) \big )^2 + \half \int _\Sigma d^2 z \, \kappa (z) \big ( \p_w \cG(w,z) \big )^2
\no \\ &&
- { 1 \over h} \om^I(w) \int _D d^2 z \, \p_z \p_w \cG(w,z)   \, \bar \om_I(z) \, \Big ( \cG (z,x) + \gamma (z) \Big )
\no \\
\delta _{ww} \gamma _0 & = & 
\int _\Sigma d^2 z \, \kappa (z) \big ( \p_w \cG(w,z) \big )^2
- { 2 \over h} \om^I(w) \int _D d^2 z \, \p_z \p_w \cG(w,z)   \, \bar \om_I(z) \, \gamma (z) 
\eea
Combining these contributions with the help of (\ref{Arak.2}), all $D$-dependent integrals, the 
squares of $ \p_w \cG(w,\cdot)$ and their integrals over $\Sigma$ against $\kappa(z)$
manifestly cancel one another, such that we obtain the following formula,
\bea
\delta_{ww} \cG(x,y) & = & 
- \p_w \cG(w,x) \, \p_w \cG(w,y) 
+ { 1 \over h} \om^I(w) \int _\Sigma d^2 z \, \p_z \p_w \cG(w,z) \, \bar \om_I(z) \, \cG(z,x) 
\no \\ &&
+ { 1 \over h} \om^I(w) \int _\Sigma d^2 z \, \p_z \p_w \cG(w,z) \, \bar \om_I(z) \, \cG(z,y) 
\eea
Since $\cG$ and the Abelian differentials are single-valued on $\Sigma$, we may integrate by parts in $z$ and express the result in terms of the DHS kernel $\cG_I$ defined in the first line of (\ref{2.a.4}) to obtain (\ref{muder.4}) and (\ref{muder.5}) and thereby proving Proposition \ref{3.prop:40}.

\newpage

\section{Proof of Proposition \ref{7.prop:2}}
\setcounter{equation}{0}
\label{sec:BB}

The proposition is proven by direct calculation of the variational derivatives of the defining formulas in (\ref{7.f.1}) using the variational derivatives of the primitives of DHS kernels in (\ref{7.c.8}) and the defining relations for their coincident limits $\cC_\vI(w)$ in (\ref{A.3}). For the modular tensors $\cA$, the proof of the variational derivative (\ref{7.prop.2a}) is straightforward. For the modular tensors $\mD_\vI$, we use their definition in (\ref{7.f.1}) and  the variational derivatives of $\kappa(x)$ and $\cG_\vI(x,y)$ given in (\ref{muder.3}) and (\ref{7.c.8}), respectively. Combining these ingredients, we obtain,
\bea
\label{9.q.12}
\delta_{ww} \mD_\vI & = & 
{ 1 \over h} \om^J(w) \int_\Sigma d^2 x \, \p_w \cG(w,x) \, \bar \om_J(x) \, \p_x \cG_\vI (x,x)
\no \\ &&
- \sum_{\vI = \vX \vY} \int _\Sigma d^2 x \, \kappa(x) \, \p_w \cG_{ \theta(\vX)}  (w,x) \, \p_w \cG_\vY (w,x)
\eea
Substituting the expression of (\ref{A.lem.1}) for $\p_x \cG_\vI (x,x)$ into the first~line of 
(\ref{9.q.12}) and using the coincident Fay identity of (\ref{A.10}) to carry out the integral on the second line, we obtain,
\begin{align}
\int _\Sigma d^2 x \, &\kappa (x) \, \p_w \cG_{\theta(\vX)} (w,x) \, \p_w \cG_\vY (w,x)
= 
- { 1 \over h} \om^L(w) \cC_{\theta(\vX) L \theta(\vY)} (w)
 \label{9.q.14} \\ 
 &+ { 1 \over h} \sum _{\vY = \vR \vS} \Big ( \p_w \Phi _{\theta(\vX) L \theta(\vS) M} {}^L(w) \p_w \Phi _\vR{}^M(w)
+ \p_w \Phi _{\theta(\vX) L \theta(\vS)} {}^M (w) \p_w \Phi _{\vR M }{}^L(w) \Big ) 
\notag
\end{align}
Combining these results proves the following expression,
\bea
\delta_{ww} \mD_\vI & = & 
 \frac{1}{h} \,  \om^J(w)  \bigg\{ \sum_{\vI = \vX \vY \vZ} \p_w \Phi_{J (\vX \shuffle \theta(\vZ)) M (\vY + \theta(\vY)) }{}^M(w)
+   \sum_{\vI = \vX \vY} \!    \cC_{\theta(\vX) J \theta(\vY)}(w) \bigg\} 
\label{9.q.15} \\ && 
- { 1 \over h} \! \sum_{\vI = \vX \vY \vZ} \! 
\Big ( \p_w \Phi _{\theta(\vX) L \theta(\vZ) M} {}^L(w) \p_w \Phi _\vY{}^M(w)
+ \p_w \Phi _{\theta(\vX) L \theta(\vZ)} {}^M (w) \p_w \Phi _{\vY M }{}^L(w) \Big ) 
\no 
\eea
or, after using $ \mD_\vI =  \mD_{\theta(\vI)}$ and rearranging terms, 
\bea
\delta_{ww} \mD_\vI & = & 
- { 1 \over h} \! \sum_{\vI = \vX \vY \vZ} \! 
\Big( \p_w \Phi_{ \vZ M \vX J }{}^M(w) \p_w \Phi_{\theta(\vY)}{}^J(w) 
+ \p_w \Phi_{ \vZ M \vX }{}^J(w) \p_w \Phi_{\theta(\vY) J}{}^M(w) \Big) 
\notag \\ && 
+ \frac{1}{h} \,  \om^J(w)  \bigg\{ 
  \sum_{\vI = \vX \vY} \!    \cC_{\vY J \vX}(w) 
+\sum_{\vI = \vX \vY \vZ} \p_w \Phi_{J (\vX \shuffle \theta(\vZ)) M (\vY + \theta(\vY)) }{}^M(w)
\bigg\} 
\label{9.q.15re}
\eea
The first line and the last term are already part of the desired expression (\ref{7.prop.2b}) for $\delta_{ww} \mD_\vI$,
and the leftover task in completing its proof is to derive the combinatorial identity,
\bea
  \sum_{\vI = \vX \vY} \!    \cC_{\vY J \vX}(w)  = 
 \sum_{\vI = \vX \vY} \Big( \om^M(w) \mN_{M \vY J \vX} + \p_w \Phi_{M \vY J \vX}{}^M(w) \Big)  
 \label{9.q.15nw}
\eea
for the combination of coincident DHS kernels $\cC_{\vI}(w)$ given by (\ref{A.4}).
While individual terms on the left side feature numerous combinations
$\p_w \Phi_{\vQ}{}^K(w) \mN_{K\vR}$ with non-empty $\vQ$, $\vR$ from (\ref{A.4}),
the right side exhibits considerable cancellations to be explained below.

\sm

The deconcatenation sums of (\ref{A.4}) and the left side of (\ref{9.q.15nw}) conspire~to 
\begin{align}
&\sum_{\vI = \vX \vY} \cC_{\vY J \vX}(w) 
= \sum_{\vI = \vX \vY} \sum_{\vY J \vX = \vA \vB \vC} \Big( \p_w \Phi_{\vA \shuffle \theta(\vC)}{}^M(w) 
 \mN_{M \vB}
+ \p_w\Phi_{(\vA \shuffle \theta(\vC)) M \vB}{}^M(w)\Big)
\notag \\
&\ \ = \sum_{\vI =  \vX_1 \vX_2 \vX_3 \vY} \Big( \p_w \Phi_{ \vY J \vX_1 \shuffle \theta(\vX_3)}{}^M(w) 
 \mN_{M \vX_2}
+ \p_w\Phi_{(\vY J  \vX_1 \shuffle \theta(\vX_3)) M \vX_2}{}^M(w)\Big) \notag \\
&\ \ \quad + \sum_{\vI = \vX_1 \vX_2 \vY_1 \vY_2} \Big( \p_w \Phi_{   \vY_1\shuffle \theta(\vX_2) }{}^M(w) 
\mN_{M \vY_2 J  \vX_1}
+ \p_w\Phi_{( \vY_1 \shuffle \theta(\vX_2)   ) M \vY_2 J \vX_1}{}^M(w)\Big) \notag \\
&\ \ \quad - \sum_{\vI = \vX \vY_1 \vY_2 \vY_3} \Big( \p_w \Phi_{ \vY_1 \shuffle \theta(\vX) J \theta(\vY_3) }{}^M(w) 
\mN_{M \vY_2}
+ \p_w\Phi_{(\vY_1 \shuffle \theta(\vX) J \theta(\vY_3)  ) M\vY_2}{}^M(w)\Big)
\label{smpcs.01}
\end{align}
where we have organized the deconcatenations $\vY J \vX= \vA \vB \vC$ into three cases:
\begin{itemize}
\item $J \in \vA \ \Rightarrow \ \vA = \vY J \vX_1 \, , \ \ \vB= \vX_2 \, , \ \ \vC = \vX_3 \, , \ \  \vX = \vX_1 \vX_2 \vX_3$;
\item $J \in \vB \ \Rightarrow \ \vA = \vY_1 \, , \ \ \vB= \vY_2 J  \vX_1 \, , \ \ \vC = \vX_2 \, , \ \ \vX = \vX_1 \vX_2 \, , \ \ \vY = \vY_1 \vY_2 $;
\item $J \in \vC \ \Rightarrow \ \vA = \vY_1 \, , \ \ \vB= \vY_2 \, , \ \ \vC = \vY_3 J \vX \, , \ \ \vY = \vY_1 \vY_2 \vY_3$.
\end{itemize}
The last three lines of (\ref{smpcs.01}) then admit the following major simplifications:
\begin{itemize}
\item[(i)] in the second line from below, the only contributions to $\sum_{\vI = \vX_1 \vX_2 \vY_1 \vY_2}$ stem from empty words $\vX_2$ and $\vY_1$ (using the standard property $\sum_{\vQ = \vX_2 \vY_1} \vY_1  \shuffle  \theta(\vX_2) = \emptyset \,\delta_{\vQ,\emptyset}$ of the antipode), leaving us with $\sum_{\vI = \vX \vY} (\om^M (w)\mN_{M \vY J \vX} + \p_w \Phi_{M \vY J \vX}{}^M(w))$;
\item[(ii)] the first and third line from below of (\ref{smpcs.01}) cancel each other as one can see by writing their respective $\mN$ terms as
\bea
\sum_{\vI = \vA \vB \vC} \mN_{M \vB} \, \bigg\{
 \sum_{\vC = \vR \vS} \p_w \Phi_{ \vS J \vA \shuffle \theta(\vR)  }{}^M(w)
-\sum_{\vA = \vR \vS} \p_w \Phi_{  \vS  \shuffle \theta(\vR) J \theta(\vC) }{}^M(w)
\bigg\}
\label{smpcs.02}
\eea
and inserting the combinatorial identities 
\bea
\! \! \!
\sum_{\vC = \vR \vS} \!  \vS J \vA  \! \shuffle \! \theta(\vR) = J \big( \theta(\vC) \! \shuffle  \! \vA \big)
\, , \ \ \ \
\sum_{\vA = \vR \vS} \! \vS \! \shuffle \!  \theta(\vR) J \theta(\vC) = J \big( \theta(\vC) \! \shuffle \! \vA \big)
\label{smpcs.03}
\eea
The analogous terms without factors of $ \mN$ cancel by the same mechanism.
\end{itemize}
With the simplifications (i) and (ii) in place, the last three lines of (\ref{smpcs.01}) reduce to the right side of
(\ref{9.q.15nw}) which, upon insertion into (\ref{9.q.15re}), produces the second line of the target expression (\ref{7.prop.2b}) and completes the proof of Proposition \ref{7.prop:2}.

\newpage

\section{Proof of Theorem \ref{7.thm:1}}
\setcounter{equation}{0}
\label{sec:C}

To prove Theorem \ref{7.thm:1} we shall use the expressions for the action of $L^0_{ww}$ on the generators in (\ref{7.b.12}) and (\ref{7.e.2}), as well as the defining relations in (\ref{7.f.1}) and variational derivatives in Proposition \ref{7.prop:2} of the modular tensor $\cA$, $\cD$ to recast the action of $L^0_{ww}$ in terms of $\delta_{ww} \cA$ and $\delta_{ww} \mD$.

\subsection{The action of $\LL^0_{ww}$ on $a_{iI}$}

The starting point is the defining relation in (\ref{7.e.2}) which we repeat here for convenience, 
\bea
\label{C.a.1}
{} \big [ \LL_{ww}^0 (\btau) , a_{iI} \big ] & = & - \int _\Sigma d^2 x_i \, \bar \om_I(x_i) \, \big [ \LL_{ww}^1(\xx| \btau) , J_i^{(1,0)} (\xx| \btau) \big ]
\eea
We also recall the expressions for $\LL_{ww}^1(\xx| \btau)$ of (\ref{7.c.9}) and $J_i^{(1,0)} (\xx|\btau)$ of (\ref{7.c.6}) in terms of the functions $\cW(\xx| \btau)$ and $\cU_{ww}(\xx| \btau) $ defined in (\ref{7.c.7}) and (\ref{7.c.10}), respectively,
\bea
\label{C.a.2}
\LL^1_{ww} (\xx| \btau) & = & \delta_{ww} \cW(\xx| \btau) - \cU_{ww}(\xx| \btau) 
\no \\
J^{(1,0)}_i (\xx| \btau) & = & \p_i \cW(\xx| \btau) + \om^K(x_i) a_{iK}
\eea
In the remainder of this subsection we shall suppress the arguments $\xx$ and $\btau$.   Substituting the expression for $L_{ww}^1$ into (\ref{C.a.1}) and pulling a variational derivative from acting on $\cW$ to acting on the entire integral, we obtain, 
\bea
\label{9.f.2}
\big [ L^0_{ww},  a_{iI} \big ] & = & 
- \delta_{ww} \int _\Sigma d^2 x_i \, \bar \om_I(x_i) \, \big [  \cW  ,  J_i ^{(1,0)} \big ]
+ \int _\Sigma d^2 x_i \, \bar \om_I(x_i) \, \big [  \cW ,  \p_i L^1_{ww}  \big ]
\no \\ &&
+ \int _\Sigma d^2 x_i \, \bar \om_I(x_i) \, \big [ \cU_{ww},  J_i ^{(1,0)}  \big ]
\eea
where in the second term we have used the relation $\delta_{ww} J^{(1,0)}_i = \p_i L^1_{ww}$.
Integrating the second term by parts in $\p_i$ and using the second relation in (\ref{C.a.2}) to recast $\p_i \cW$ in terms of $J_i^{(1,0)}$ and (\ref{C.a.1}) to re-express the integral of $[L^1_{ww}, J_i^{(1,0)}]$ in terms of $\big [ \LL_{ww}^0 , a_{iI} \big ] $ we obtain, 
\bea
\label{9.f.3}
\int _\Sigma d^2 x_i \, \bar \om_I(x_i) \, \big [  \cW  ,  \p_i L^1_{ww}  \big ] 
& = & - \big [ \LL_{ww}^0 , a_{iI} \big ] 
- \int _\Sigma d^2 x_i \, \bar \om_I(x_i) \, \big [L^1_{ww}  , \om^K(x_i) a_{iK} \big ] 
\eea
Eliminating the left side of (\ref{9.f.3})  from (\ref{9.f.2}), we obtain,  
\bea
{} [L^0_{ww} , a_{iI}] & = & 
- \half \delta_{ww} \int _\Sigma d^2 x_i \, \bar \om_I(x_i) \, \big [  \cW  ,  \, J_i ^{(1,0)}  \big ]
 \\ &&
- \half \int _\Sigma d^2 x_i \, \bar \om_I(x_i) \, \big [L^1_{ww} , \, \om^K(x_i) a_{iK} \big ] 
+ \half \int _\Sigma d^2 x_i \, \bar \om_I(x_i) \, \big [ \cU_{ww},  \, J_i ^{(1,0)}  \big ]
\no
\eea
Expressing $J_i^{(1,0)}$ in terms of $\p_i \cW$ and $L^1_{ww}$ in terms of $\delta_{ww} \cW$ using (\ref{C.a.2}) gives, \bea
{} [L^0_{ww} , a_{iI}] & = & 
- \half \delta_{ww} \int _\Sigma d^2 x_i \, \bar \om_I(x_i) \, \big [  \cW  ,  \p_i \cW  + \om^K(x_i) a_{iK} \big ]
 \\ &&
- \half \int _\Sigma d^2 x_i \, \bar \om_I(x_i) \, \big [\delta_{ww} \cW - \cU_{ww}  , \, \om^K(x_i) a_{iK} \big ] 
\no \\ &&
+ \half \int _\Sigma d^2 x_i \, \bar \om_I(x_i) \, \big [ \cU_{ww},  \p_i \cW  + \om^K(x_i) a_{iK} \big ]
\no
\eea
 Integrating by parts in $\p_i$ on the last line and using the relation,
\bea
\p_i \, \cU_{ww} + \delta _{ww} \, \om^K(x_i) a_{iK} =0
\eea
further simplifies the final form, which is summarized by the Lemma below.

{\lem 
\label{9.lem:20}
The action of $L^0_{ww}$ on $a_{iI}$ is given by,
\bea
\label{9.lem.20}
{} [L^0_{ww} , a_{iI}] & = &  - \delta_{ww} \cP_{iI} + \cQ_{iI ww}
\eea
where $\cP_{iI}$ is defined by,  
\bea
\label{9.p.1}
\cP _{iI} & = & \half \int _\Sigma d^2 x_i \, \bar \om_I(x_i) \, \big [  \cW  , \,  \p_i  \cW  \big ]
+ \int _\Sigma d^2 x_i \, \bar \om_I(x_i) \, \big [  \cW  , \,   \om^K(x_i) a_{iK}  \big ]
\eea
while $\cQ_{iI ww}$ is given by,
\bea
\label{9.p.1a}
\cQ _{iIww} = \int _\Sigma d^2 x_i \, \bar \om_I(x_i) \, \big [ \cU_{ww}, \,  \om^K(x_i) a_{iK} \big ]
= \half \delta_{ww} \, \cA^M {}_I {}^N \, \big [ a_{iM}, a_{iN} \big ]
\eea
Both functions $\cP_{iI}$ and $\cQ_{iI ww}$ are independent of $\xx$ and $ \delta_{ww} \cA^M {}_I {}^N$ is given in (\ref{7.f.6})}. 

\begin{proof} 
\vskip 0in
To verify that $\cP_{iI}$ is independent of $x_j$ we compute its derivative $\p_j$ for $j \not= i$,
\bea
\label{9.p.3}
\p_j \cP_{iI} = 
\int _\Sigma d^2 x_i \, \bar \om_I(x_i) \, \big [ \p_j  \cW ,  J_i^{(1,0)}  \big ]
=\int _\Sigma d^2 x_i \, \bar \om_I(x_i) \, \big [ J_j^{(1,0)} - \om^K(x_j) a_{jK} ,  J_i^{(1,0)}  \big ]
\quad
\eea
Using the flatness condition $[J^{(1,0)}_i , J^{(1,0)} _j  ]=0$ and carrying out the integration in the remaining term, we find that the $\p_j \cP_{iI}=0$ in view of the structure relation $ [a_{jK} , a_{iI} ] =0$.  Since $\cP_{iI}$ is a scalar in $x_j$ for $j \not= i$, the vanishing of its $\p_j$ derivative implies that $\cP_{iI}$ is independent of $\xx$. 

\sm

The function $\cQ_{iI ww}$ is readily  computed by inserting the expression (\ref{7.c.10}) for  $ \cU_{ww}$ into (\ref{9.p.1a}), and we have the following expression for $h\geq 2$,
\bea
\label{9.p.2}
\cQ_{iI ww} =  
 \Big ( \om^M(w) \p_w \Phi _I{}^N(w) 
+ { 1 \over h-1} \om^J(w) \p_w \Phi _J{}^M (w) \, \delta ^N_I \Big ) \, [a_{iM}, a_{iN}] 
\eea
whereas $\cQ_{iI ww} = 0$ for $h=1$.
One easily verifies that $\cQ_{iI ww} $ is holomorphic in $w$ and that  its variational curl vanishes. 
Comparing with (\ref{7.f.6})  allows us to write $\cQ_{iI ww} $ as on the right of (\ref{9.p.1a}).
\end{proof}

\subsubsection{Computation of $\cP_{iI}$}

We begin by simplifying the integrals over $x_i$ in (\ref{9.p.1}). Clearly, any term in $ \cW(\xx)$  that is independent of $x_i$ has a vanishing contribution to the first integral in (\ref{9.p.1}). Thus, in the first integral in (\ref{9.p.1}), we have effectively,
\bea
\label{9.p.6}
 \cW (\xx) & ~ \to ~ & \sum _{\vP \not = \emptyset} \Phi _\vP {}^M(x_i) B_i ^\vP \, a_{iM}
+ \sum _{k \not= i} \sum_\vP \cG_\vP (x_k, x_i) B_k^\vP \, t_{i k}
\no \\ 
\p_i  \cW (\xx) & = & \sum _{\vQ \not = \emptyset} \p_i \Phi _\vQ {}^N(x_i) B_i ^\vQ \, a_{iN}
+ \sum _{\ell \not= i} \sum_\vQ \p_i \cG_\vQ (x_\ell, x_i) B_\ell^\vQ \, t_{i \ell}
\eea
while  all terms in $ \cW(\xx)$ are to be retained in the second integral in (\ref{9.p.1}). To compute $\cP_{iI}$ we take advantage of a trick that will save us from many lengthy calculations. Since $\cP_{iI}$ is independent of $\xx$ in view of Lemma \ref{9.lem:20} it equals its integral in $x_k$ for all $k \not= i$ over $\Sigma$  against the product of all $\kappa(x_k)$,
\bea
\cP_{iI} = \int_{\Sigma^n} \bigg( \prod_k  \kappa(x_k) \bigg)  \, \cP_{iI} = \prod _{k \not= i} \int _\Sigma d^2 x_k \, \kappa (x_k) \cP_{iI} 
\eea
Combining both operations above, we find that all cross-terms in the first term cancel upon integration against $\kappa$, and in the second term only the contribution from $\Phi _\vP {}^M(x_i) B_i^\vP a_{iM}$ to $\cW$ is non-zero. As a result, $\cP_{iI}$ is given by,
\bea
\cP_{iI} & = & \half \int _\Sigma d^2 x_i \, \bar \om_I(x_i) 
\sum _{\vP \not = \emptyset, \vQ \not = \emptyset}
\Phi _\vP {}^M(x_i) \p_i \Phi _\vQ {}^N(x_i) \big [ B_i ^\vP \, a_{iM}   , B_i ^\vQ \, a_{iN} \big ]
 \\ &&
+ \int _\Sigma d^2 x_i \, \bar \om_I(x_i) \sum _{\vP \not = \emptyset} \Phi _\vP {}^M(x_i) \om^N(x_i) \big [ B_i ^\vP \, a_{iM}, a_{iN} \big ]
\no \\ &&
+ \half \sum _{k \not= i}  \int _\Sigma d^2 x_k \, \kappa (x_k) \int _\Sigma d^2 x_i \, \bar \om_I(x_i)  
\sum_{\vP, \vQ}  \cG_\vP (x_k, x_i) 
 \p_i \cG_\vQ (x_k, x_i) \big [ B_k^\vP \, t_{i k} , B_k^\vQ \, t_{ik} \big ] 
\no
\eea
Carrying out the integrals over $x_i$ gives again DHS kernels and $\cA$ tensors, 
\bea
\cP_{iI} & = & 
- \half \sum _{\vP \not = \emptyset, \vQ \not = \emptyset} \cA^M {}_{ \theta(\vP) I \vQ} {}^N 
 \big [ B_i ^\vP \, a_{iM}   , B_i ^\vQ \, a_{iN} \big ]
 - \sum _{\vP \not = \emptyset} \cA^M {}_{ \theta(\vP) I } {}^N 
 \big [ B_i ^\vP \, a_{iM}   , a_{iN} \big ]
\no \\ &&
+ \half \sum _{k \not= i}  \int _\Sigma d^2 x_k \, \kappa (x_k) 
\sum_{\vP, \vQ}  \cG_{\vP I \theta(\vQ)}  (x_k, x_k) 
 \big [ B_k^\vP \, t_{i k} , B_k^\vQ \, t_{ik} \big ] 
\eea
Using the reflection symmetry $\cA^M{}_{\vI}{}^N = \cA^N {}_{\theta(\vI)} {}^M$ the first two terms may be regrouped as a sum over $\vP$ and $\vQ$ constrained only by $\vP \vQ \not= \emptyset$, while in the last line, the integral over $x_k$ may be carried out in terms of the modular tensors $\mD$ defined in (\ref{7.f.1}),
\bea
\cP_{iI} = 
- \half \!  \sum _{\vP, \vQ; \vP \vQ \not = \emptyset} \!\! \! \cA^M {}_{ \theta(\vP) I \vQ} {}^N 
 \Big [ B_i ^\vP \, a_{iM}   , B_i ^\vQ \, a_{iN} \Big ]
+ \half  \sum _{k \not= i}   \sum_{\vP, \vQ}  \mD_{\vP I \theta(\vQ)}  \Big [ B_k^\vP \, t_{i k} , B_k^\vQ \, t_{ik} \Big ] 
\qquad
\eea
Taking the variational derivative, including the contribution from $\cQ_{iI ww}$ in (\ref{9.p.1a}) as the $\vP \vQ= \emptyset$ contribution to the first sum, and converting $B_k^\vP t_{ik} = B_i ^{\theta(\vP)} t_{ik}$ in the second line, we obtain the first formula in (\ref{9.q.20}) of Theorem \ref{7.thm:1}.

\subsection{The action of $\LL^0_{ww}$ on $b_i^I$}
\label{appd2}

The starting point is the expression of (\ref{7.b.12}), repeated here for convenience,
\bea
\label{C.b.1}
{} \big [ \LL_{ww}^0 (\btau)  , b_i^I \big ] & = & \theta \bPhi ^I(w;B_i) \, \bPhi ^J (w;B_i) \, a_{iJ} - \big [ \cV_{ww}^1 (\btau), b_i^I \big ]
\eea
where $\cV_{ww}^1$ is given in (\ref{7.b.9bb}).  Forming the generating function for the variational derivative of the $\cA$-tensors, derived in Proposition \ref{7.prop:2}, we find, 
\bea
\label{C.b.2}
\sum_{\vI} \delta _{ww} \cA^M{}_\vI {}^N B_i^\vI a_{iN}
& = & \theta \bPhi ^M(w;B_i) \bPhi^N(w;B_i) a_{iN}
 \\ &&
-{ 1 \over h} \sum _\vI \om^R(w) \p_w \Phi _{R \vI} {}^N(w) \Big (  B_i^K B_i^\vI -  B_i^{\theta(\vI)}  B_i^N \Big ) a_{iN}
\no \\ &&
+ { 1 \over h-1} \om^R(w) \Big (  \p_w \Phi _R{}^M(w) B_i^N - \p_w \Phi _R{}^N (w) B_i^M  \Big ) a_{iN}
\no
\eea
Comparing the last two lines with $B_i^M \cV_{ww}^1$, obtained from (\ref{7.b.9bb}), we find,
\bea
\label{C.b.3}
\sum_{\vI} \delta _{ww} \cA^M{}_\vI {}^N B_i^\vI a_{iN} = \theta \bPhi ^M(w;B_i) \bPhi^N(w;B_i) a_{iN} + B_i^M \cV_{ww}^1
\eea
Its right side equals that of  (\ref{C.b.1}) which proves the second line of (\ref{9.q.20}) in Theorem \ref{7.thm:1}.

\subsection{The action of $\LL^0_{ww}$ on $t_{ij}$}

The starting point is the expression in the second line of (\ref{7.e.2}), restated here as follows,
\bea
\label{C.c.1}
{} \big [ L_{ww}^0 (\btau) , t_{ij} \big ] = - \big [ \cV^1_{ww} (\btau) , t_{ij} \big ]
- \lim_{x_j \to x_i} \big [ \cV_{ww} (\xx| \btau)  , t_{ij} \big ]
\eea
The first term on the right may be evaluated using the definition of $\cV_{ww}^1$ given in (\ref{7.b.9bb}), whose last term under the sum manifestly commutes with $t_{ij}$ while in the remainder only the terms with $k=i,j$ contribute to give,
\bea
\label{C.c.1a}
\big [ \cV^1_{ww}  , t_{ij} \big ] & = &
- { 1 \over h} \om^J(w) \sum_{\vI} \p_w \Phi _{J \vI} {}^K (w) \Big [ B_i^\vI a_{iK} + B_j^\vI a_{jK} , t_{ij} \Big ]
\eea
To evaluate the commutator, we make use of (\ref{A.Lie.1}) and we find, 
\bea
\label{C.c.1b}
\big [ \cV^1_{ww}  , t_{ij} \big ] & = &
 { 1 \over h} \om^J(w) \sum_{\vI} \big [ B_i^\vI t_{ij}, t_{ij} \big ] 
 \sum _{\vI = \vX \vY \vZ} \p_w \Phi _{J (\vX \shuffle \, \theta(\vZ) ) K (\vY + \theta(\vY)) }{}^K
 \eea
The second term on the right of (\ref{C.c.1}) may be evaluated using the following lemma.

{\lem
\label{C.lem:1}
The limit of the commutator $[\cV_{ww}, t_{ij}]$ on the right of (\ref{C.c.1}) is independent of $\xx$ and may be recast as follows,
\bea
\label{C.c.2}
\lim _{x_i, x_j \to x} \big [ \cV_{ww} (\xx| \btau) , t_{ij} \big ]
& = & 
- \big [ \theta \bG(w,x;B_i) \, \theta \bG(w,x;B_j) t_{ij} , t_{ij} \big ] 
 \\ &&
- \big [ \theta \bG(w,x;B_i) \, \bPhi ^K(w;B_i) a_{iK} 
\no \\ && \hskip 0.4in
+ \theta \bG(w,x;B_j) \, \bPhi ^K(w;B_j) a_{jK} , t_{ij} \big ] 
\no
\eea
\sm}

\begin{proof}
\vskip -0.2in
To prove (\ref{C.c.2}), we render $\cV_{ww}$ explicit with the help of (\ref{7.b.9}), use its regularity as $x_j \to x_i$ to simplify the commutator, and retain  only those contributions that do not manifestly vanish in view of the structure relations of $\mt_{h,n}$. All dependence on points $k \not= i,j$ is found to vanish using the antipode identity $\theta \bG(w,x;B_i) t_{ik} = \bG(w,x;B_k) t_{ik}$ and its $i \to j$ counterpart  as well as the structure relation $[t_{ik} + t_{jk}, t_{ij}]=0$. We are left with (\ref{C.c.2}), which involves only the variables $w$ and $x$.   To prove that (\ref{C.c.2}) is independent of $x$, we prove that the left side has vanishing $x$-derivative,
\bea
\label{C.c.4}
\p_x \Big \{ \lim _{x_i, x_j \to x} \big [ \cV_{ww} (\xx| \btau)  , t_{ij}  \big ] \Big \} & = &
\lim _{x_i, x_j \to x} \big [ \p_i \cV_{ww} (\xx| \btau) + \p_j \cV_{ww} (\xx| \btau) , t_{ij}  \big ] 
 \\ & = &
 \delta _{ww} \Big \{ \lim _{x_i, x_j \to x} \big [ J_i^{(1,0)}  (\xx| \btau) + J_j^{(1,0)}  (\xx| \btau) , t_{ij}  \big ] \Big \} 
 \no
\eea
The limit of the commutator in the second line on the right side was shown to vanish in appendix A of \cite{DHoker:2026lgg}. Since the limit of the commutator on the left is a single-valued scalar in $x$ with vanishing $x$-derivative, it must be independent of $x$.
\end{proof}

\sm

Since equation (\ref{C.c.2}) is independent of $x$, it is equivalent  to its integral in $x$ over~$\Sigma$ against $\kappa(x)$. Since the integral  $\int d^2 x  \, \kappa(x) \bG(w,x;B_i)$ vanishes, we may recast (\ref{C.c.1}) with the help of (\ref{C.c.2}) as follows,
\bea
\label{C.c.5}
{} \big [ L_{ww}^0 (\btau) + \cV^1_{ww}(\btau), t_{ij} \big ] 
& = & 
\int _\Sigma d^2 x \, \kappa(x)   \big [ \theta \bG(w,x;B_i) \,  \bG(w,x;B_i) t_{ij} , t_{ij} \big ] 
\eea
Expanding in powers of $B_i$ and rewriting $\p_w \cG_{\theta(\vP)} (w,x) = \p_w \cG_\vP (x,w)$,
\bea
\label{C.c.6}
{} \big [ L_{ww}^0 (\btau) + \cV^1_{ww}(\btau), t_{ij} \big ] 
& = & 
\int _\Sigma d^2 x \, \kappa(x) \sum _{\vP, \vQ} \p_w \cG_\vP (x,w) \p_w \cG_\vQ (w,x) 
\big [ B_i ^{\vP \vQ} t_{ij}, t_{ij} \big ] 
\eea
Setting $y=x$ in the first equation of (\ref{7.c.8}) and then integrating  in $x$ over $\Sigma$ against $\kappa(x)$, we see that the last terms integrate to zero and we are left with, 
\bea
\label{C.c.7}
{} \big [ L_{ww}^0 (\btau) + \cV^1_{ww}(\btau), t_{ij} \big ] 
& = & 
- \sum _{\vI} \big [ B_i ^\vI t_{ij}, t_{ij} \big ]  \int _\Sigma d^2 x \, \kappa(x) \delta _{ww} \cG_\vI(x,x) 
\eea
Lemma \ref{A.lem:1} provides a formula for $\p_x \cG_\vI(x,x)$, which we repeat here for convenience, 
\bea
\label{C.c.9}
\p_x \cG_\vI(x,x) = \sum_{\vI = \vX \vY \vZ}  \p_x \Phi _{(\vX \shuffle \theta(\vZ)) M (\vY + \theta(\vY)) }{}^M(x) 
\eea
The primitive of this relation involves an integration constant which follows from integrating both sides over $x$ against $\kappa(x)$  and using the definition of $\mD_\vI$ in (\ref{7.f.1}),
\bea
\label{C.c.10}
\cG_\vI(x,x) = \mD_\vI + \sum_{\vI = \vX \vY \vZ} \Phi _{(\vX \shuffle \theta(\vZ)) M (\vY + \theta(\vY)) }{}^M(x) 
\eea
The integral of its variational derivative against $\kappa(x)$ gives,
\bea
\label{C.c.11}
\int _\Sigma d^2 x \, \kappa(x) \delta _{ww} \cG_\vI(x,x) 
& = & \delta_{ww} \,  \mD_\vI 
+ \sum_{\vI = \vX \vY \vZ} \int _\Sigma d^2 x \, \kappa(x) \delta_{ww} \Phi _{(\vX \shuffle \theta(\vZ)) M (\vY + \theta(\vY)) }{}^M(x) 
\no \\ & = &
 \delta_{ww} \,  \mD_\vI -{1 \over h} \om^J(w)  \sum_{\vI = \vX \vY \vZ} \p_w  \Phi _{J (\vX \shuffle \theta(\vZ)) M (\vY + \theta(\vY)) }{}^M(w)
 \eea
using the second equation of (\ref{7.c.8}) in passing to the second line and discarding the integrals of total derivatives in $x$. Putting all together, we obtain, 
\bea
\label{C.c.12}
{} \big [ L_{ww}^0 (\btau) , t_{ij} \big ] 
& = & - \sum _\vI  \delta_{ww} \,  \mD_\vI \big [ B_i^\vI t_{ij}, t_{ij} \big ] - [ \cV_{ww}^1 , t_{ij}] 
\no \\ &&
+{1 \over h} \om^J(w) \sum_{\vI} \big [ B_i^\vI t_{ij}, t_{ij} \big ]  \sum_{\vI = \vX \vY \vZ} \p_w  \Phi _{J (\vX \shuffle \theta(\vZ)) M (\vY + \theta(\vY)) }{}^M(w)
\eea
In view of (\ref{C.c.1b}), the second and third terms on the right side cancel one another, and we are left with the last expression in (\ref{9.q.20}). This concludes the proof of Theorem \ref{7.thm:1}.

\newpage

\section{Proof of equation (\ref{6.z.5cc}) of Theorem \ref{6.thm:1}}
\setcounter{equation}{0}
\label{sec:E}

In this last appendix, we provide the proof of equation (\ref{6.z.5cc}) of Theorem \ref{6.thm:1}, which we repeat here for convenience, after swapping  $(iI) \leftrightarrow (jJ)$ and reversing the overall sign,
\bea
\label{E.2}
\Big [ \tilde \LL_{ww}^0 (\tilde  a_{iI}) , \tilde b_j^J \Big ]  
+ \Big [ \tilde a_{iI},  \tilde  \LL_{ww}^0( \tilde b_j^J )\Big ]  + \delta ^J_I \,  \tilde L^0_{ww} ( \tilde t_{ij})   =0 
\eea
As long as we directly use the form of the $\tilde  \LL_{ww}^0$ action in (\ref{6.z.4}) without appealing to the relation (\ref{7.f.3}) between the modular tensors $\cA$ and $\cD$, only the first and second terms of (\ref{E.2}) receive contributions from $\cA$ while only the first and last terms receive contributions from $\mD$. The contributions from $\cA$ are given by,
\bea
\label{6.e.3}
{} \big [ \tilde \LL_{ww}^0 (  \tilde a_{iI} ) , \tilde b_j^J  \big ] _\cA
& = &
-  \sum_{\vP, \vQ }  \delta_{ww}  \cA^J{}_{\theta(\vP) I \vQ} {}^N \, 
\big [ \tilde B_i^\vP \tilde t_{ij} , \tilde B_i^\vQ \tilde a_{iN} \big ]
\no \\ 
\big [ \tilde a_{iI}, \tilde \LL_{ww}^0 (  \tilde b_j^J) \big ] &= &
 \sum_\vP  \delta_{ww}  \cA^J {}_\vP {}^N \,  \big [ \tilde a_{iI},  \tilde B_j^\vP \tilde a_{jN} \big ]
\eea
while the contributions from $\mD$ are given by,
\bea
\label{6.e.4}
\big [ \tilde  \LL_{ww}^0(  \tilde a_{iI} ), \tilde b_j^J  \big ]_\mD  
& = &
- \sum_{\vP , \vQ } 
 \delta_{ww} \mD_{ \theta(\vP) I \vQ} \,   \big [ \tilde B_i^{\vP J} \tilde  t_{ij} , \tilde B_i^\vQ \tilde  t_{ij}  \big ] 
\no \\ 
\delta^J_I \, \tilde  \LL_{ww}^0( \tilde  t_{ij}) &= &
- \delta ^J_I \,  \sum_\vP  \delta _{ww}  \mD_\vP   \, \big [  \tilde  B_i^\vP \tilde  t_{ij}, \tilde  t_{ij} \big ]
\eea
We begin by simplifying the second line in (\ref{6.e.3}) using the relation (\ref{A.Lie.2}) of appendix \ref{sec:AA}, 
\bea
\label{6.e.5}
\big [ \tilde a_{iI}, \tilde L_{ww}^0( \tilde  b_j^J ) \big ]  = \sum_{\vQ, \vR}   
 \delta_{ww} \cA^J {}_{\vR I \vQ}  {}^N \,  \tilde  B_j^\vR \big [ \tilde  B_j^\vQ \tilde  a_{jN}, \tilde  t_{ij}  \big ]
\eea
Using (\ref{A.Lie.1}) to convert $\tilde  B_j^\vQ \tilde  a_{jK}$ into $- \tilde  B_i^\vQ \tilde a_{iK}$ plus $\tilde a$-independent terms, and applying the generalized Leibniz property $\tilde B_j^{\vR}  [X , Y ] = \sum_{\vR = \vP \shuffle \vQ} [ \tilde B_j^{\vP} X ,  \tilde B_j^{\vQ}Y]$ \cite{DHoker:2026ggx}, we obtain,
\bea
\label{6.e.6}
\big [ \tilde a_{iI}, \tilde L_{ww}^0( \tilde b_j^J ) \big ] & = &
- \sum_{\vP, \vQ, \vR}  \delta_{ww} \cA^J {}_{(\vR \shuffle \, \theta(\vP) )  I \vQ}  {}^N  \,  
\big [ \tilde  B_j^\vR \tilde B_i^\vQ \tilde  a_{iN}, \tilde  B_i^\vP \tilde  t_{ij}  \big ]
 \\ &&
- \sum_{ \vS, \vU, \vV, \vY, \vX}  \delta_{ww} \cA^J {}_{(\vU \shuffle \vV)  I (\vX \shuffle \, \vS) M (\vY + \theta(\vY)) }  {}^M \, \Big [  \tilde  B_i^{\vX \vY \theta(\vU \vS)} \tilde  t_{ij}, \tilde  B_i^{ \theta(\vV)} \tilde  t_{ij} \Big ]
\no
\eea
The $\vR= \emptyset$ contribution to the first line cancels the $\big [ \tilde  \LL_{ww}^0 ( \tilde  a_{iI} ) , \tilde  b_j^J  \big ] _\cA$ term in (\ref{6.e.3}). The remaining $\vR \not= \emptyset $ contributions to the first term in (\ref{6.e.6}) may be parametrized by letting $\vR \to \vR M$. Furthermore, using the identity (\ref{A.sh.2}) in terms of the deconcatenation $\vZ = \theta(\vU \vS)$ in the last term in (\ref{6.e.6}), changing variables $\vV \to \theta (\vP)$, and factoring out common Lie algebra elements, we obtain, 
\bea
\label{6.e.10}
{} \big [ \tilde \LL_{ww}^0 (  \tilde a_{iI} ) , \tilde b_j^J  \big ] _\cA + \big [ \tilde a_{iI}, \tilde \LL_{ww}^0 (  \tilde b_j^J) \big ] 
& = &
\sum_{\vP, \vQ} \cR_{\vP, \vQ} \big [ \tilde  B_i^\vP \tilde  t_{ij} , \tilde  B_i^\vQ \tilde t_{ij} \big ]
\eea
where the coefficients $\cR_{\vP, \vQ}$ are given by,
\bea
\label{6.e.11}
\cR_{\vP, \vQ} & = & 
\half \sum _{\vQ= \vX \vY}  \delta_{ww} \cA^J {}_{ \big ( \theta(\vY) M \shuffle \, \theta(\vP) \big )  I \vX}  {}^M   
- (\vP \leftrightarrow \vQ)
\no \\ &&
+ \half \sum_{\vQ =  \vX \vY \vZ}  \delta_{ww} 
\cA^J {}_{ \big ( \theta(\vZ) \shuffle \theta(\vP) I \vX \big ) M \big ( \vY + \theta(\vY) \big )}  {}^M 
- (\vP \leftrightarrow \vQ)
\eea
We shall now show the following equality,
\bea
\label{6.e.12}
\cR_{\vP, \vQ} =  \half \sum_{\theta(\vP) I \vQ = \vX \vY \vZ} 
\delta_{ww} \cA^J {} _{\big ( \vX \shuffle \theta(\vZ) \big ) M \big ( \vY + \theta(\vY) \big ) } {}^M
\eea
by decomposing the deconcatenation sum in (\ref{6.e.12}) as follows. We parametrize the three different sectors according to  whether the letter $I$ belongs to $\vX$, $\vY$ or $\vZ$, as follows, 
\begin{align}
\label{6.e.13}
& (1) & I & \in \vX &  \vX & = \theta(\vP) I \vU & \vQ & = \vU \vY \vZ  
\no \\
& (2) & I & \in \vY &  \theta(\vP) & = \vX \vU   & \vQ & = \vV \vZ  &  \vY & = \vU I \vV 
\no \\
& (3) & I & \in \vZ &  \theta (\vP) & = \vX \vY \vU   & \vZ & = \vU I \vQ
\end{align}
The contributions to $\cR_{\vP, \vQ}$ from $(1)$ and $(3)$ give,
\bea
\label{6.e.14}
(1) & = & 
\half \sum_{\vQ = \vU \vY \vZ} \delta_{ww} \cA^J {} _{\big ( \theta(\vP) I \vU \shuffle \theta(\vZ) \big ) M \big ( \vY + \theta(\vY) \big ) } {}^M
\no \\
(3) & = & 
\half \sum_{\theta(\vP) = \vX \vY \vU} \delta_{ww} \cA^J {} _{\big ( \vX \shuffle \theta(\vU I \vQ)) \big ) M \big ( \vY + \theta(\vY) \big ) } {}^M
\eea
Changing the variable $\vU \to \vX$ in the expression for $(1)$ gives the first term on the second line of (\ref{6.e.11}). Changing variables $\theta(\vX) \to \vZ$, $\theta(\vY) \to \vY$ and $\theta(\vU) \to \vX$ in $(3)$ gives the second term on the second line of (\ref{6.e.11}). The contributions to $(2)$ from the $\vY$ and $\theta(\vY)$ parts may be separated as follows,
\bea
\label{6.e.15}
(2) & = & 
\half \sum_{\theta (\vP) = \vX \vU} \sum_{\vQ = \vV \vZ} 
\delta_{ww} \cA^J {} _{\big ( \vX \shuffle \theta(\vZ) \big ) M \vU I \vV } {}^M - (\vP \leftrightarrow \vQ) 
\eea
Using the deconcatenation relation (\ref{A.sh.2}) of appendix \ref{sec:AA} for $\vS=\emptyset$, 
the sums simplify,
\bea
\label{6.e.16}
(2) & = & 
\half \sum_{\vQ = \vV \vZ} 
\delta_{ww} \cA^J {} _{ \big ( \theta(\vP) \shuffle \theta(\vZ) M  \big ) I \vV } {}^M - (\vP \leftrightarrow \vQ) 
\eea
Changing variables $\vV \to \vS$ and $\vZ \to \vX$, we recover the first line in (\ref{6.e.11}), which completes the proof of the expression (\ref{6.e.12}). Assembling all contributions to (\ref{6.e.10}) and reinstating the two classes of terms in (\ref{6.e.4}) involving $\mD$, we find,
\begin{align}
\label{6.e.17}
\Big [ \tilde \LL_{ww}^0 & (\tilde  a_{iI})  , \tilde b_j^J \Big ]  
+ \Big [ \tilde a_{iI},  \tilde  \LL_{ww}^0( \tilde b_j^J )\Big ]  + \delta ^J_I \,  \tilde L^0_{ww}  ( \tilde t_{ij}) 
\no \\ 
& =  
\half \sum_{\vP, \vQ} \big [ \tilde  B_i^\vP \tilde  t_{ij}, \tilde  B_i^\vQ \tilde  t_{ij} \big ] 
\sum_{\theta(\vP) I \vQ = \vX \vY \vZ} 
\delta_{ww} \cA^J {} _{\big ( \vX \shuffle \theta(\vZ) \big ) M \big ( \vY + \theta(\vY) \big ) } {}^M
\no \\ & \quad 
 - \sum_{\vP , \vQ }  \big [ \tilde B_i^{\vP J}  \tilde t_{ij} , \tilde B_i^\vQ \tilde t_{ij}  \big ] \, \delta_{ww} \mD_{ \theta(\vP) I \vQ}  
- \sum_{\vP } \big [ \tilde B_i^{\vP}  \tilde t_{ij} ,  \tilde t_{ij}  \big ]  \, \delta ^J_I \,   \delta _{ww}  \mD_\vP  
\end{align}
We shall now analyze separately the contributions involving Lie algebra elements $[\tilde B_i^\vP \tilde t_{ij}, \tilde t_{ij}]$ and those given by $[\tilde B_i^\vP \tilde t_{ij}, \tilde B_i ^\vQ \tilde t_{ij}]$ with $\vP, \vQ \not= \emptyset$.

\subsection{Cancellation of the terms of the form $[\tilde B_i^\vP \tilde t_{ij}, \tilde t_{ij}]$ }

The contributions to (\ref{6.e.17}) that are proportional to $[\tilde B_i^\vP \tilde t_{ij}, \tilde t_{ij}]$ are given by, 
\bea
\label{6.e.18}
&& 
\sum_{\vP} \big [ \tilde B_i^\vP \tilde t_{ij}, \tilde t_{ij} \big ] \sum_{\theta(\vP) I  = \vX \vY \vZ} 
\delta_{ww} \cA^J {} _{\big ( \vX \shuffle \theta(\vZ) \big ) M \big ( \vY + \theta(\vY) \big ) } {}^M
\no \\ && \quad
+ \sum_{\vP} \big [ \tilde B_i^{\vP J}  \tilde t_{ij} , \tilde t_{ij}  \big ]  \delta_{ww} \mD_{ I \vP  }  
- \delta ^J_I  \sum_\vP \big [  \tilde B_i^\vP \tilde t_{ij}, \tilde t_{ij} \big ] \delta _{ww}  \mD_\vP  
\qquad
\eea
after using $ \mD_{ \theta(\vP) I}  = - \mD_{ I \vP}  $.
The contributions from $\vP = \emptyset$ and $\vQ = \emptyset$ in (\ref{6.e.17})  contribute equally to the first term in (\ref{6.e.18}) while those from $\vP = \emptyset$ to the first and third terms in (\ref{6.e.18}) cancel in view of $[\tilde t_{ij}, \tilde t_{ij} ]=0$. The remaining contributions with $\vP \not= \emptyset$ may be parametrized by letting $\vP \to \vP K$. The coefficient of $ \big [ \tilde B_i^{\vP K} \tilde t_{ij}, \tilde t_{ij} \big ] $ is then given by, 
\bea
\label{6.e.19}
- \sum_{I \vP K  = \vX \vY \vZ} 
\delta_{ww} \cA^J {} _{\big ( \vX \shuffle \theta(\vZ) \big ) M \big ( \vY + \theta(\vY) \big ) } {}^M
+  \delta ^J_K \delta_{ww} \mD_{ I \vP  }   - \delta ^J_I  \delta _{ww} \mD_{\vP K}  
\eea
which vanishes in view of the identity (\ref{7.f.3}) for $\vI \to \vP$.

\subsection{Cancellation of the terms $[\tilde B_i^\vP \tilde t_{ij}, \tilde B_i ^\vQ \tilde t_{ij}]$ with $\vP, \vQ \not= \emptyset$}

The contributions to (\ref{6.e.17}) that are proportional to $[\tilde B_i^\vP \tilde t_{ij}, \tilde B_i^\vQ \tilde t_{ij}]$ with $\vP, \vQ \not= \emptyset$ are, 
\bea
\label{6.e.20}
&& 
\half \sum_{\vP \not= \emptyset , \vQ \not= \emptyset} \big [ \tilde B_i^\vP \tilde t_{ij}, \tilde B_i^\vQ \tilde t_{ij} \big ] 
\sum_{\theta(\vP) I \vQ = \vX \vY \vZ} \!\!\!
\delta_{ww} \cA^J {} _{\big ( \vX \shuffle \theta(\vZ) \big ) M \big ( \vY + \theta(\vY) \big ) } {}^M
\no \\ && 
- \sum_{\vQ \not= \emptyset , \vP } \big [ \tilde B_i^{\vP J}  \tilde t_{ij} , \tilde B_i^\vQ \tilde t_{ij}  \big ]  \delta_{ww} \mD_{ \theta(\vP) I \vQ}  
\eea
Since $\vP \not= \emptyset $ in the first term, we parametrize $\vP$ by letting $\vP \to \vP K$, while in both terms we parametrize $\vQ \to \vQ L$ since $\vQ \not= \emptyset$. The coefficient of the Lie algebra element $\big [ \tilde B_i^{ \vP K} \tilde t_{ij}, \tilde B_i^{ \vQ L} \tilde t_{ij} \big ] $ in (\ref{6.e.20})  is then given by half of, 
\bea
\delta^J_L \, \delta_{ww} \mD_{K \theta(\vP) I \vQ} - \delta^J_K \, \delta_{ww} \mD_{ \theta(\vP) I \vQ L} 
- \sum_{K \theta(\vP) I \vQ L = \vX \vY \vZ} 
\delta_{ww} \cA^J {} _{\big ( \vX \shuffle \theta(\vZ) \big ) M \big ( \vY + \theta(\vY) \big ) } {}^M 
\quad
\eea
which vanishes in view of (\ref{7.f.3}) for $\vI \to \theta(\vP) I \vQ$.

\newpage

\section{Completing the proof of Lemma \ref{66.lem:5a}}
\setcounter{equation}{0}
\label{sec:F}

It remained to prove that the $x_i$-dependence of $\tilde Z^t_{ij}$ in Lemma \ref{66.lem:5a}  belongs to the ideal $\cI$ of (\ref{6.z.3}). 
The analytic part of this proof is almost identical to the analytic part of the proof given in appendix A of \cite{DHoker:2026ggx} but the algebraic part differs. We recall the function of interest,
\bea
\label{G.1}
\p_i \tilde Z^t_{ij}
= - \delta_{ww} \tilde J^{(1,0)}_{ij} 
\hskip 1in 
\tilde J^{(1,0)}_{ij}  = \lim_{x_j \to x_i} \big [ \tilde t_{ij}, \tilde J_i^{(1,0)}  + \tilde J_j^{(1,0)}  \big ]
\eea
which we shall show belongs to $\cI$.

\sm

Substituting the expression for $\tilde J^{(1,0)}_i$ from (\ref{66.a.2bb}) into (\ref{G.1}), we obtain, 
\bea
\label{G.2}
\tilde J^{(1,0)} _{ij} & = & \lim _{x_j \to x_i} \Big [ t_{ij}, \, \bPhi ^J(x_i; \tilde B_i) \tilde a_{iJ}
+ \bPhi ^J(x_j; \tilde B_j) \tilde a_{jJ}
\no \\ && \qquad ~
+ \sum _{k \not= i} \bG(x_i,x_k; \tilde B_i) \tilde t_{ik} + \sum _{\ell \not= j} \bG(x_j,x_\ell; \tilde B_j) \tilde t_{j\ell}  \Big ]
\eea
The contributions on the second line of (\ref{G.2}) for which $k \not= j$ and $\ell \not =i$  admit regular limits as $x_j \to x_i$ and  cancel one another. This may be  established by converting the generators $\tilde B_i$ and $\tilde B_j$ into $\tilde B_k$ using (\ref{6.z.2aa}); pulling the common factor $\theta  \bG(x_i,x_k; \tilde B_k) $ out of the commutator; and using the fact that the remaining commutator $[\tilde t_{ij}, \tilde t_{ik}{+} \tilde t_{jk}]$ vanishes by (\ref{6.z.2bb}). Thus, we are left with,
\bea
\label{G.3}
\tilde J^{(1,0)} _{ij} & = & \lim _{x_j \to x_i} 
\Big [ \tilde t_{ij}, \, \bPhi ^J(x_i; \tilde B_i) \tilde a_{iJ} + \bPhi ^J(x_j; \tilde B_j) \tilde a_{jJ}
\no \\ && \qquad \qquad
+ \bG(x_i,x_j; \tilde B_i)  \tilde t_{ij}  +  \bG(x_j,x_i; \tilde B_j) \tilde t_{ij}   \Big ]
\qquad
\eea  
To evaluate the limit  we expand in powers of $\tilde B_i$ and $\tilde B_j$. The only terms that are singular in the limit occur in $\bG$ at the zeroth order in $\tilde B$. They  are proportional to $[\tilde t_{ij}, \tilde t_{ij}]$ in the above expression and cancel trivially.  The remaining contributions may be arranged as, 
\bea
\label{G.4}
\tilde J^{(1,0)} _{ij} & = & \sum_{\vI} \p_i \Phi _\vI {}^J(x_i) \Big [ \tilde t_{ij}, \tilde B_i^\vI \tilde a_{iJ} + \tilde B_j^\vI \tilde a_{jJ} \Big ]
+  \sum_{\vI, \,  |\vI| \geq 2}  \p_i \cG_\vI  (x_i, x_i)  \Big [ \tilde t_{ij}, \, \tilde B_i^\vI \tilde t_{ij} \Big ]
\eea
where we used the identity $ \cG_\vI (x_j, x_i) = \cG_{\theta(\vI)} (x_i, x_j)$ to combine partial derivatives and to cancel terms with $|\vI |=1$. 
Clearly,  $ \tilde J^{(1,0)} _{ij} $ is a single-valued $(1,0)$ form in $x_i$ all of whose terms in the expansion are in the range of $\p_i$, except for the $\vI=\emptyset $ term $\om^J (x_i) [\tilde t_{ij}, \tilde a_{iJ} + \tilde a_{jJ}]$ in the first sum   which belongs to $\cI$ in view of (\ref{6.z.2dd}). Therefore, we have, 
\bea
\label{G.5}
\int _\Sigma d^2 x_i \, \bar \om^K(x_i) \tilde J^{(1,0)} _{ij}  \, \in \, \cI
\eea
Next, we prove that  $ \tilde J^{(1,0)} _{ij} $ is holomorphic in $x_i$. To compute its $\bar \p_i$ derivative it will be convenient to use formula (\ref{G.1}) prior to taking the limit. To obtain the full derivative, we then need to differentiate in both $x_i$ and $x_j$, 
\bea
\label{G.6}
\bar \p_i \tilde J^{(1,0)} _{ij}  = \lim _{x_j \to x_i} \Big ( \bar \p_i + \bar \p_j \Big ) 
\Big [ \tilde t_{ij}, \tilde J^{(1,0)} _i + \tilde J^{(1,0)}_j \Big ]
\eea
These derivatives can be worked out using (\ref{89.b.3}) and we find, 
\bea
\label{A.14}
\bar \p_i \tilde J^{(1,0)} _{ij}  & = &
\lim _{x_j \to x_i} \bigg \{ \Big [ \tilde t_{ij}, \big [ \tilde J^{(0,1)}_i 
+ \tilde J^{(0,1)} _j , \tilde J^{(1,0)}_i + \tilde J^{(1,0)} _j \big ] \Big ]
\no \\ && \qquad
- \pi \sum_{k \not=i} \delta (x_i, x_k) [\tilde t_{ij}, \tilde t_{ik} ] 
- \pi \sum_{k \not=j} \delta (x_j, x_k) [\tilde t_{ij}, \tilde t_{j k} ] \bigg \}
\eea
The sum of the $\delta$ functions vanishes in the limit $x_j \to x_i$  in view of the relation 
(\ref{6.z.2bb}), and the first line may be rearranged using the relation, 
\bea
\label{A.15}
\lim _{x_j \to x_i} \big [ \tilde t_{ij}, \tilde J^{(0,1)}_i + \tilde J^{(0,1)} _j \big] 
= - \pi \bar \om^I (x_i) [ \tilde t_{ij} , \tilde b_{iI} + \tilde b_{jI} ]=0
\eea
as follows,
\bea
\label{A.16}
\bar \p_i \tilde  J^{(1,0)} _{ij}  & = &
\Big [ \lim_{x_j\to x_i} ( \tilde J^{(0,1)}_i + \tilde J^{(0,1)} _j ), \tilde J^{(1,0)} _{ij} \Big ]
\eea
The proof that $\tilde J^{(1,0)} _{ij}$ belongs to $\cI$ may be completed by contradiction. We use the fact that $\tilde J^{(1,0)} _{ij}$ admits a Taylor expansion in powers of $\tilde b$, each term having positive $\tilde b$-degree. We denote the term of  lowest $\tilde b$-degree that does not belong to $\cI$ by $X$. Then, by equation (\ref{A.16}), $X$ must be holomorphic in $x_i$, if it exists, up to terms that belong to $\cI$,  since the right side of (\ref{A.16}) has $\tilde b$-degree at least one higher than $X$ and is therefore a linear combination of the holomorphic Abelian differentials $\bar \om^L(x_i)$ up to terms that belong to $\cI$. But the integral of this non-vanishing lowest $\tilde b$-degree term against $\bar \om^K(x_i)$ must belong to $\cI$,  as shown in (\ref{G.5}), thereby contradicting our initial assumption. As a result $\tilde J^{(1,0)} _{ij}$ must belong to $\cI$ to all $\tilde b$-degree.

\sm

Given that $\tilde J^{(1,0)} _{ij} \in \cI$ is preserved by its variational derivative, we have established
$\partial_i \tilde Z^t_{ij} \in \cI$ by (\ref{G.1}) and thus completed the proof of Lemma \ref{66.lem:5a}.

\newpage

\section{Proof of Lemma \ref{66.lem:25}}
\setcounter{equation}{0}
\label{sec:H}

To prove Lemma \ref{66.lem:25}, we shall evaluate the functions $\tilde Z^a_{iI}$ and $\tilde Z^t_{ij}$ of (\ref{66.e.2A}) and show that the result is related to the action of $\tilde L^0_{ww}$ on $\tilde a_{iI}$ and $\tilde t_{ij}$ given in (\ref{6.z.4})  as follows,
\begin{subequations}
\label{H.1}
\begin{align}
 \tilde \LL_{ww}^0  (\tilde a_{iI})  + \tilde Z^a_{iI}   & \, \in \, \cI
\label{H.1aa} \\
 \tilde \LL_{ww}^0  (\tilde t_{ij} )  + \tilde Z^t_{ij} & \,  \in \, \cI
\label{H.1bb}
\end{align}
\end{subequations}
Using the decomposition (\ref{66.a.2cc}) of $\tilde L_{ww}^1$ and the fact that $\tilde \cV^1_{ww}$ of (\ref{66.a.8}) is independent of $\xx$, we obtain the following simplified form,
\begin{subequations}
\label{H.2}
\begin{align}
\tilde Z^a_{iI} & = \big [ \tilde \cV_{ww}^1, \tilde a_{iI} \big ] + \hat Z^a_{iI} & \qquad
 \hat Z_{iI}^a & =   \int _\Sigma d^2 x_i \, \bar \om_I(x_i) \, \big [ \tilde \cV_{ww}  , \tilde J_i^{(1,0)}  \big ]
\label{H.2aa} \\
\tilde Z^t_{ij} & = \big [ \tilde \cV_{ww}^1, \tilde t_{ij} \big ] + \hat Z^t_{ij} & \qquad
\hat Z^t_{ij} & = \lim_{x_j \to x_i} \big [ \tilde \cV_{ww} , \tilde t_{ij} \big ]
\label{H.2bb}
\end{align}
\end{subequations}
The direct calculation of $\hat Z^a_{iI}$ and $\hat Z^t_{ij}$ from the expression (\ref{66.a.8}) for $\tilde \cV_{ww}$ is lengthy and tedious, but may be avoided by exploiting Lemma \ref{66.lem:5a}. Since $\tilde Z^a_{iI}$, $\tilde Z^t_{ij}$ and therefore $\hat Z^a_{iI}$, $\hat Z^t_{ij}$ are independent of $x_j$ for all $j \in \{ 1,\cdots,n \}$, we have the following immediate equality,
\begin{align}
\label{H.4}
\hat Z^a_{iI} & =  \int_{\Sigma^n}  \! \hat \kappa \, \hat Z^a_{iI} & 
\hat Z^t_{ij} & = \int_{\Sigma^n} \! \hat  \kappa \, \hat Z^t_{ij} &
 \int_{\Sigma^n}  \! \hat \kappa & =  \prod_\ell \int _\Sigma d^2 x_\ell \, \kappa(x_\ell) 
\end{align}
Many terms in the evaluation of $\hat Z^a_{iI}$ and $\hat Z^t_{ij}$ simplify using the relations, 
\begin{subequations}
\label{H.5}
\begin{align}
\int _\Sigma d^2 x_j \, \kappa (x_j) \,  \bG (w,x_j;B) & =  0
\label{H.5aa} \\
\int _\Sigma d^2 x_i \, \bar \om_I(x_i) \, \bG(x_i, x_j;B) & =  0 
\label{H.5bb} \\
\int _\Sigma d^2 x_i \, \bar \om_I(x_i) \, \bPhi^J (x_i;B) & =  \delta^J_I 
\label{H.5cc}
\end{align}
\end{subequations}

\subsection{Proof of (\ref{H.1aa})}

Substituting the explicit expressions for $\tilde \cV_{ww} $ and $ \tilde J_i^{(1,0)} $ into $\hat Z^a_{iI}$, we obtain,
\begin{align}
\label{H.6}
\hat Z^a _{iI}  = 
- \int_{\Sigma^n} \! \hat \kappa \int _\Sigma d^2 x_i \, \bar \om_I(x_i) \, \Big [ &
 \sum_k \theta \bG(w,x_k; \tilde B_k) \Big (  \bPhi^{L} (w; \tilde B_k) \tilde a_{k L } 
+ \thalf \sum_{\ell \not = k}    \theta \bG(w,x_\ell; \tilde B_\ell) \tilde t_{k\ell}  \Big ),
\no \\ & \qquad
\bPhi^J(x_i; \tilde B_i) \, \tilde a_{iJ} + \sum_{j \neq i} \bG (x_i,x_j; \tilde B_i) \, \tilde t_{ij}  \Big ]
\end{align}
The contribution to $\hat Z^a _{iI} $ from the first term inside the parentheses on the first line integrates to zero against the first term on the second line unless $k=i$ and against the second term on the second line: its contributions from  $k \not= i$ vanish in view of (\ref{H.5bb}) while those from $k=i$ vanish in view of  (\ref{H.5aa}). Similarly, the contribution from the second term inside the parentheses on the first line integrates to zero against the first term on the second line: its contributions from  $k, \ell  \not= i$ vanish in view of (\ref{H.5bb}) while those from $k=i$ or $\ell=i$ vanish in view of  (\ref{H.5aa}). Finally, the contributions from the second term on the first line and the second term on the second line vanish unless $(k,\ell)=(i,j)$ or $(k, \ell ) = (j,i)$ in view of (\ref{H.5aa}) and (\ref{H.5bb}). Thus, we are left with,
\begin{align}
\label{H.7}
\hat Z^a _{iI}  & = \hat Z^{a1} _{iI} + \hat Z^{a2} _{iI} 
\end{align}
where the following two terms will be separately simplified and combined with 
contributions from $ [ \tilde \LL_{ww}^0  , \tilde a_{iI} ]$ and $[ \tilde \cV_{ww}^1, \tilde a_{iI} ]$ in (\ref{H.1aa}) and (\ref{H.2aa}):
\begin{align}
\hat Z^{a1} _{iI} &= -  \sum_{j \not = i} \int _\Sigma d^2 x_j \,\kappa (x_j) \, \int _\Sigma d^2 x_i \, \bar \om_I(x_i) \, \Big [ 
 \theta \bG(w,x_i; \tilde B_i)  \bG(w,x_j; \tilde B_i) \tilde t_{ij} ,
\bG (x_i,x_j; \tilde B_i) \, \tilde t_{ij}  \Big ]
\no \\ 
\hat Z^{a2} _{iI} &=
-  \int _\Sigma d^2 x_i \, \bar \om_I(x_i) \, \Big [ 
 \theta \bG(w,x_i; \tilde B_i)   \bPhi^{ L }(w; \tilde B_i) \tilde a_{i L },  \bPhi^J(x_i; \tilde B_i) \, \tilde a_{iJ}  \Big ]
 \label{H.7parts}
\end{align}
Expanding both expressions in powers of $\tilde B$ and  integrating  over $x_i$ in terms of DHS kernels we obtain the following explicit expressions,
\begin{align}
\hat Z^{a1} _{iI} &= 
- \sum_{j \neq i} \sum_{\vP, \vQ} \big[ \tilde B_i^\vP  \tilde t_{ij} ,  \tilde B_i^\vQ \tilde t_{ij} \big]
\sum_{\vP = \vX \vY} \Xi_{\vY | \theta(\vX) I \vQ}(w) 
\no \\
 \hat Z^{a2} _{iI} &= 
- \sum_{\vP ,\vQ}   \big [  \tilde B_i^{\vP} \tilde a_{i L },  \tilde B_i^\vQ  \tilde a_{iJ}  \big ]
 \sum_{\vP = \vX \vY}  \, \p_w \Phi_\vY{}^L(w) \, \p_w \Phi_{\theta(\vX) I \vQ}{}^J(w) 
 \label{H.85}
\end{align}
where $\Xi_{\vR | \vS}$ is defined by,
\bea
\Xi_{\vR | \vS}(w) =  \int _\Sigma d^2 x \, \kappa(x) \, \p_w \cG_{ \vR}  (w,x) \, \p_w \cG_{\vS} (w,x)
= \Xi_{\vS | \vR}(w)
\label{H.83}
\eea
The function $\tilde L^0_{ww}(\tilde a_{iI})$, defined in (\ref{6.z.4}),  receives contributions involving the modular tensors $\cA$ and $\mD$. It will be convenient to decompose each one of these contributions in turn according to the decompositions of
$\delta _{ww} \cA$ given in (\ref{7.prop.2a}) and $\delta_{ww} \mD$
given by the following equivalent of (\ref{9.q.12}),
\bea
\delta_{ww} \mD_\vP = \frac{1}{h} \,  \om^J(w)   \sum_{\vP = \vX \vY \vZ} \p_w \Phi_{J (\vX \shuffle \theta(\vZ)) M (\vY + \theta(\vY)) }{}^M(w)
- \sum_{\vP = \vX \vY} \Xi_{ \theta(\vX) | \vY}(w)
\eea
namely,
\bea
\tilde L^0_{ww}(\tilde a_{iI})  =   \tilde L^0_{ww}(\tilde a_{iI}) _{\cA_1} + \tilde L^0_{ww}(\tilde a_{iI}) _{\cA_2} 
+  \tilde L^0_{ww}(\tilde a_{iI}) _{\mD_1} + \tilde L^0_{ww}(\tilde a_{iI}) _{\mD_2} 
\eea
The individual terms depending on $\cA$ are given as follows, 
\begin{align}
 \tilde L^0_{ww}(\tilde a_{iI})_{\cA_1}  &=  \omega^J(w) \,\bigg\{
 \frac{1}{h{-}1} \, \big[ \tilde a_{iM} , \tilde a_{iI} \big] \, \p_w \Phi_{J}{}^M(w)  \notag \\
 &\quad 
- \frac{1}{h} \sum_{\vQ}  \big[ \tilde a_{iM} , \tilde B_i^{\vQ L} \tilde a_{iN} \big]  
\,  \Big(  
\delta^M_I \, \p_w \Phi_{J \vQ L}{}^N(w) + \delta^N_L \, \p_w \Phi_{J \theta(\vQ) I}{}^M(w)
   \Big)
 \notag \\
 &\quad + \frac{1}{2h} \sum_{\vP, \vQ } \! \big [ \tilde B_i^{\vP K} \tilde a_{iM}, \tilde B_i^{\vQ L} \tilde a_{iN} \big ] \,   \Big( \delta^M_K \p_w \Phi_{J \theta(\vP) I \vQ L}{}^N(w)
 {-} \delta^N_L \p_w \Phi_{J \theta(\vQ) I \vP K}{}^M(w) \Big) \bigg\}
 \notag \\
 \tilde L^0_{ww}(\tilde a_{iI})_{\cA_2}  &=  \half \sum_{\vP, \vQ }  
\big [ \tilde B_i^\vP \tilde a_{iM}, \tilde B_i^\vQ \tilde a_{iN} \big ]
\sum_{\theta(\vP) I \vQ = \vX \vY}  \p_w \Phi _{\theta(\vX)} {}^M (w) \, \p_w \Phi _\vY{}^N(w)
\label{H.88}
\end{align}
while those depending on $\mD$ are given by, 
\begin{align}
 \tilde L^0_{ww}(\tilde a_{iI})_{\mD_1}  &= 
 - \frac{1}{2h} \sum_{\vP, \vQ }   \sum_{k\not= i} 
  \big [ \tilde B_i^\vP \tilde t_{ik} , \tilde B_i^\vQ \tilde t_{ik}  \big ] 
 \, \omega^J(w) \sum_{ \theta(\vP) I \vQ  = \vX \vY \vZ}
 \p_w \Phi_{J (\vX \shuffle \theta(\vZ)) M (\vY + \theta(\vY)) }{}^M(w)
 \notag \\
  \tilde L^0_{ww}(\tilde a_{iI})_{\mD_2} &=  \half \sum_{\vP, \vQ }   \sum_{k\not= i} 
  \big [ \tilde B_i^\vP \tilde t_{ik} , \tilde B_i^\vQ \tilde t_{ik}  \big ] \,  \sum_{ \theta(\vP) I \vQ  = \vX \vY} 
  \Xi_{\theta(\vX) | \vY}(w)
     \label{H.82}
\end{align}

\subsubsection{Cancellation of $\hat Z^{a1} _{iI}+\tilde L^0_{ww}(\tilde a_{iI})_{\mD_2}$}
 \label{secH.1.1}
 
To prove the cancellation of $\hat Z^{a1} _{iI} + \tilde L^0_{ww}(\tilde a_{iI})_{\mD_2}$, we decompose the deconcatenation sum in the second line of (\ref{H.82}) according to whether $I \in \vX$ or $I \in \vY$ which, by the antisymmetry of the accompanying commutator,  give equal contributions to the last line of,
\begin{align}
 \tilde L^0_{ww}(\tilde a_{iI})_{\mD_2} &=  \half \sum_{\vP, \vQ }   \sum_{k\not= i} 
  \big [ \tilde B_i^\vP \tilde t_{ik} , \tilde B_i^\vQ \tilde t_{ik}  \big ] \, 
\bigg\{
{-}\sum_{ \vQ = \vX \vY} \Xi_{\theta(\vX) I \vP | \vY}(w)
+ \sum_{ \vP = \vX \vY} \Xi_{ \vY |  \theta(\vX) I \vQ}(w)
\bigg\}   \notag\\
&=  \sum_{\vP, \vQ }   \sum_{k\not= i} 
  \big [ \tilde B_i^\vP \tilde t_{ik} , \tilde B_i^\vQ \tilde t_{ik}  \big ] \, 
 \sum_{ \vP = \vX \vY} \Xi_{ \vY |  \theta(\vX) I \vQ}(w)
    \label{H.84}
\end{align}
Combining this result with the expression for $\hat Z^{a1}_{iI} $ in (\ref{H.85}), we obtain,
\bea
\hat Z^{a1} _{iI} +  \tilde L^0_{ww}(\tilde a_{iI})_{\mD_2} =0
\eea

\subsubsection{Cancellation of $\hat Z^{a2} _{iI}+\tilde L^0_{ww}(\tilde a_{iI})_{\cA_2}$}
  \label{secH.1.2}
  
In the expression for $\tilde L^0_{ww}(\tilde a_{iI})_{\cA_2}$ in (\ref{H.88}) we decompose  the deconcatenation sum according to whether $I \in \vX$ versus $I \in \vY$ and exploit the antisymmetry of $ [ \tilde B_i^\vP \tilde a_{iM}, \tilde B_i^\vQ \tilde a_{iN} ]$ under $(\vP,M)\leftrightarrow (\vQ,N)$ to obtain,
\begin{align}
 \tilde L^0_{ww}(\tilde a_{iI})_{\cA_2}  &=   
 \sum_{\vP, \vQ }   \big [ \tilde B_i^\vP \tilde a_{iM}, \tilde B_i^\vQ \tilde a_{iN} \big ]  
 \sum_{\vP = \vX \vY} \p_w \Phi _{\vY} {}^M (w) \, \p_w \Phi _{\theta(\vX) I \vQ}{}^N(w)
\label{H.89}
\end{align}
Combining this result with the one for $\hat Z^{a2}_{iI}$ given in (\ref{H.85}), we obtain,
\bea
\hat Z^{a2} _{iI}+\tilde L^0_{ww}(\tilde a_{iI})_{\cA_2}=0
\eea

 \subsubsection{Proof of  $\tilde L^0_{ww}(\tilde a_{iI})_{\cA_1} + \tilde L^0_{ww}(\tilde a_{iI})_{\mD_1} 
 + [ \tilde \cV_{ww}^1, \tilde a_{iI} ] \in \cI$} 
 \label{secH.1.3}
 
In view of the expression (\ref{66.a.8}) for $\tilde \cV_{ww}^1$, we decompose
\begin{align}
\big[ \tilde \cV_{ww}^1, \tilde a_{iI} \big] = \big[ \tilde \cV_{ww}^1, \tilde a_{iI} \big]_1+ \big[ \tilde \cV_{ww}^1, \tilde a_{iI} \big]_2+ \big[ \tilde \cV_{ww}^1, \tilde a_{iI} \big]_3
   \label{H.92}
\end{align}
with (assuming $h\geq 2$ to set $c_h = \frac{1}{h{-}1}$)\footnote{All of $\tilde \cV_{ww}^1$, $ \tilde L^0_{ww}(\tilde a_{iI})_{\cA_1}$
and $ \tilde L^0_{ww}(\tilde a_{iI})_{\mD_1}$ vanish for $h=1$, and the cancellations of sections \ref{secH.1.1}
and \ref{secH.1.2} already suffice to prove Lemma \ref{66.lem:25} in that case.} 
\begin{align}
\big[ \tilde \cV_{ww}^1, \tilde a_{iI} \big]_1 &= \frac{1}{h{-}1} \, \omega^J(w) \,\p_w\Phi_J{}^M(w) \sum_{k=1}^n \big[ \tilde a_{iI} ,\tilde a_{kM}\big]
   \label{H.93} \\
   \big[ \tilde \cV_{ww}^1, \tilde a_{iI} \big]_2 &=  \frac{1}{h} \, \omega^J(w) \sum_{\vI \neq \emptyset }\p_w\Phi_{J\vI}{}^M(w)
    \big[ \tilde a_{iI} , \tilde B_i^\vI \tilde a_{iM}\big]
   \notag \\
   \big[ \tilde \cV_{ww}^1, \tilde a_{iI} \big]_3 &=  \frac{1}{h} \, \omega^J(w) \sum_{\vI \neq \emptyset }\p_w\Phi_{J\vI}{}^M(w)
  \sum_{k\neq i}  \big[ \tilde a_{iI} , \tilde B_k^\vI \tilde a_{kM}\big]
   \notag
\end{align}
By inspection, we see that the sum of $[ \tilde \cV_{ww}^1, \tilde a_{iI} ]_1 $ and the first line of $ \tilde L^0_{ww}(\tilde a_{iI})_{\cA_1} $ on the  right side of (\ref{H.88}) belong to $\cI$. Moreover, $[ \tilde \cV_{ww}^1, \tilde a_{iI} ]_2 $ cancels the first term on the second line of (\ref{H.88}). The remaining terms may be regrouped as follows,  
\begin{align}
 \tilde L^0_{ww}&(\tilde a_{iI})_{\cA_1}  + \big[ \tilde \cV_{ww}^1, \tilde a_{iI} \big] 
  \notag \\
 & = 
  \big[ \tilde \cV_{ww}^1, \tilde a_{iI} \big]_3  
  + \frac{1}{h} \,   \omega^J(w) \! \sum_{\vP, \vQ } \!  \sum_{k\neq i}\big [ \tilde B_i^{\vP } \tilde a_{iM}, \tilde B_i^{\vQ }  \tilde t_{ik} \big ] \,   \p_w \Phi_{J \theta(\vQ) I \vP }{}^{\! M}(w) 
  \ {\rm mod} \ \cI
 \label{H.94} 
 \end{align}
To proceed, we simplify $[ \tilde \cV_{ww}^1, \tilde a_{iI} ]_3$  in (\ref{H.93}) through the reformulation of the $\mt_{h,n}$ identity in Proposition E.1 of \cite{DHoker:2026ggx} within $\muu_{h,n}$\footnote{The inductive proof of Proposition E.1 of \cite{DHoker:2026ggx} relies on both of $[a_{iI}, a_{jJ}] = 0$ and $[a_{iI}+a_{jI}, t_{ij}] = 0$ for $j\neq i$ within the base cases of Lemmas E.2 and E.3 of the reference, leading to explicitly known terms in $\cI$ when transcribing the $\mt_{h,n}$ relation of the reference to the relation (\ref{H.95}) in $\muu_{h,n}$.}
\begin{align}
  \label{H.95} 
 \big[ \tilde a_{iI} , \tilde B_k^\vI \tilde a_{kM}\big] &= - \sum_{\vI = \vP L \vQ} \delta^L_I \, \big[ \tilde B_i^\vQ \tilde a_{iM} , \tilde B_i^{\theta(\vP)} \tilde t_{ik} \big]  
 - \sum_{\vI = \vP K \vQ L \vR} \delta^K_M \delta^L_I \, \tilde B_k^\vP \, \big[
 \tilde B_i^\vR \tilde t_{ik} , \tilde B_i^{\theta(\vQ)} \tilde t_{ik} \big] \notag \\
 &\quad + \sum_{\vI = \vP K \vQ L \vR} \delta^K_I \delta^L_M \sum_{\vX ,\vY} \delta_{\vQ,\vX \shuffle \vY}\, \tilde B_k^\vP \, \big[
 \tilde t_{ik} , \tilde B_i^{ \vX (\vR + \theta(\vR))\theta(\vY)} \tilde t_{ik} \big] 
  \ {\rm mod} \ \cI
\end{align}
where $k\neq i$. We shall decompose $[ \tilde \cV_{ww}^1, \tilde a_{iI} ]_3$ in (\ref{H.93}) according to the first,
second and third term of (\ref{H.95}) into
\bea
 \big[ \tilde \cV_{ww}^1, \tilde a_{iI} \big]_3 =  \big[ \tilde \cV_{ww}^1, \tilde a_{iI} \big]_{3a}
 +  \big[ \tilde \cV_{ww}^1, \tilde a_{iI} \big]_{3b} +  \big[ \tilde \cV_{ww}^1, \tilde a_{iI} \big]_{3c}  \ {\rm mod} \ \cI
   \label{H.96}
\eea
where
\begin{align}
\big[ \tilde \cV_{ww}^1, \tilde a_{iI} \big]_{3a} &=  - \frac{1}{h} \, \omega^J(w) \sum_{\vP,\vQ }\p_w\Phi_{J \vP I \vQ}{}^M(w)
  \sum_{k\neq i} \big[ \tilde B_i^\vQ \tilde a_{iM} , \tilde B_i^{\theta(\vP)} \tilde t_{ik} \big]   \label{H.97} \\
 \big[ \tilde \cV_{ww}^1, \tilde a_{iI} \big]_{3b} &= -  \frac{1}{h} \, \omega^J(w) \sum_{\vP,\vQ,\vR }\p_w\Phi_{J \vP M \vQ I \vR}{}^M(w)
  \sum_{k\neq i} \tilde B_k^\vP \, \big[
 \tilde B_i^\vR \tilde t_{ik} , \tilde B_i^{\theta(\vQ)} \tilde t_{ik} \big]
 \notag \\
 \big[ \tilde \cV_{ww}^1, \tilde a_{iI} \big]_{3c} &=  \frac{1}{h} \, \omega^J(w) \sum_{\vP,\vR,\vX,\vY }\p_w\Phi_{J \vP I (\vX \shuffle \vY) M \vR }{}^M(w)
  \sum_{k\neq i} \tilde B_k^\vP \, \big[
 \tilde t_{ik} , \tilde B_i^{ \vX (\vR + \theta(\vR))\theta(\vY)} \tilde t_{ik} \big] 
\notag
  \end{align}
One readily observes that $[ \tilde \cV_{ww}^1, \tilde a_{iI} ]_{3a}$ cancels the double sum
over $\vP,\vQ$ in the last line of (\ref{H.94}), and we arrive at the simplified form,
\begin{align}
 \tilde L^0_{ww}(\tilde a_{iI})_{\cA_1}  + \big[ \tilde \cV_{ww}^1, \tilde a_{iI} \big] &= 
 \big[ \tilde \cV_{ww}^1, \tilde a_{iI} \big]_{3b} +  \big[ \tilde \cV_{ww}^1, \tilde a_{iI} \big]_{3c}  \ {\rm mod} \ \cI
 \label{H.98}
 \end{align}
The final step is to cancel the remaining terms $ [ \tilde \cV_{ww}^1, \tilde a_{iI} ]_{3b}$, $[ \tilde \cV_{ww}^1, \tilde a_{iI} ]_{3c}$
against $ \tilde L^0_{ww}(\tilde a_{iI})_{\mD_1}$ in (\ref{H.82}) which shall decompose into 
\bea
 \tilde L^0_{ww}(\tilde a_{iI})_{\mD_1} =  \tilde L^0_{ww}(\tilde a_{iI})_{\mD_{1a}} +  \tilde L^0_{ww}(\tilde a_{iI})_{\mD_{1b}}+ \tilde L^0_{ww}(\tilde a_{iI})_{\mD_{1c}}
  \label{H.99}
\eea
according to the three cases $I \in \vX$, $I\in \vY$ or $I \in \vZ$ in the deconcatenation sum over
$\theta(\vP) I \vQ = \vX \vY \vZ$ in (\ref{H.82}) (also see (\ref{6.e.13}) for an earlier instance of this decomposition):
\begin{align}
  \tilde L^0_{ww}(\tilde a_{iI})_{\mD_{1a}} &= - \frac{ \omega^J(w)}{2h}  \sum_{\vP , \vQ} \sum_{k\neq i} \big[
 \tilde B_i^\vP \tilde t_{ik} , \tilde B_i^{\vQ} \tilde t_{ik} \big]  \! \! \! \sum_{\vQ = \vQ_1 \vQ_2 \vQ_3}  \! \! \!   \p_w\Phi_{J (\theta(\vQ_3)\shuffle \theta(\vP)I \vQ_1 ) M ( \vQ_2 + \theta(\vQ_2) ) }{}^M(w)
  \notag \\
    \tilde L^0_{ww}(\tilde a_{iI})_{\mD_{1b}} &= - \frac{1}{2h}\, \omega^J(w)  \sum_{\vP , \vQ} \sum_{k\neq i} \big[
 \tilde B_i^\vP \tilde t_{ik} , \tilde B_i^{\vQ} \tilde t_{ik} \big]  \notag \\
 &\qquad \times  \sum_{\vP = \vP_1 \vP_2  }  \sum_{ \vQ = \vQ_1 \vQ_2}  \p_w\Phi_{J (\theta(\vP_2)\shuffle \theta(\vQ_2) ) M ( \theta(\vP_1) I \vQ_1  -  \theta(\vQ_1) I \vP_1) }{}^M(w)
    \notag \\
     \tilde L^0_{ww}(\tilde a_{iI})_{\mD_{1c}} &=  \frac{ \omega^J(w)}{2h}  \sum_{\vP , \vQ} \sum_{k\neq i} \big[
 \tilde B_i^\vP \tilde t_{ik} , \tilde B_i^{\vQ} \tilde t_{ik} \big] 
 \! \! \! \sum_{\vP = \vP_1 \vP_2 \vP_3}  \! \! \!   \p_w\Phi_{J (\theta(\vP_3)\shuffle \theta(\vQ)I \vP_1 ) M ( \vP_2 + \theta(\vP_2) ) }{}^M(w)
  \label{H.100}
  \end{align}
After repeated use of $\tilde B_k^\vP \tilde t_{ik} = \tilde B_i^{\theta(\vP)} \tilde t_{ik} $
and $\tilde B_k^{\vR}  [X , Y ] = \sum_{\vR = \vP \shuffle \vQ} [ \tilde B_k^{\vP} X ,  \tilde B_k^{\vQ}Y]$ 
to rewrite the second and third line of (\ref{H.97}) as
\begin{align}
 \big[ \tilde \cV_{ww}^1, \tilde a_{iI} \big]_{3b} &=  -  \frac{1}{h} \, \omega^J(w) 
 \sum_{\vP , \vQ} \sum_{k\neq i} \big[
 \tilde B_i^\vP \tilde t_{ik} , \tilde B_i^{\vQ} \tilde t_{ik} \big] 
 \sum_{\vP = \vP_1 \vP_2 \atop{ \vQ = \vQ_1 \vQ_2} }
 \p_w\Phi_{J (\theta(\vP_2)\shuffle \theta(\vQ_2)) M \theta( \vQ_1) I \vP_1}{}^M(w)
\notag
  \\
 \big[ \tilde \cV_{ww}^1, \tilde a_{iI} \big]_{3c} &=   \frac{1}{h} \, \omega^J(w) 
 \sum_{\vP , \vQ} \sum_{k\neq i} \big[
 \tilde B_i^\vP \tilde t_{ik} , \tilde B_i^{\vQ} \tilde t_{ik} \big]  \notag \\
 &\qquad \times
 \sum_{\vQ = \vQ_1 \vQ_2 \vQ_3 \vQ_4 }
 \p_w\Phi_{J( \theta( \vP) \shuffle \theta(\vQ_4)) I (\vQ_1 \shuffle \theta(\vQ_3)) M (\vQ_2 + \theta(\vQ_2)) }{}^M(w)
  \label{H.111}
\end{align}
we find the following two cancellations:
\begin{align}
 \tilde L^0_{ww}(\tilde a_{iI})_{\mD_{1b}} +  \big[ \tilde \cV_{ww}^1, \tilde a_{iI} \big]_{3b} &= 0
 \notag \\
  \tilde L^0_{ww}(\tilde a_{iI})_{\mD_{1a}} + 
   \tilde L^0_{ww}(\tilde a_{iI})_{\mD_{1c}} + 
   \big[ \tilde \cV_{ww}^1, \tilde a_{iI} \big]_{3c} &= 0
  \label{H.101}
\end{align}
The first line is a consequence of antisymmetrizing the subscripts
of the $\Phi$-tensor in the first line of (\ref{H.111}) with respect to
$(\vP_1,\vP_2) \leftrightarrow (\vQ_1,\vQ_2)$ according to the accompanying
commutator. The second cancellation in (\ref{H.101}) arises from the simplification
\bea
\sum_{\vQ = \vQ_1 \vQ_2 \vQ_3 \vQ_4 } \! \! \! \!
 \p_w\Phi_{J( \theta( \vP) \shuffle \theta(\vQ_4)) I (\vQ_1 \shuffle \theta(\vQ_3)) M \cdots }{}^M
=  \! \! \! \sum_{\vQ = \vQ_1 \vQ_2 \vQ_3 } \! \! \! 
 \p_w\Phi_{J(  \theta(\vQ_3)  \shuffle \theta( \vP) I \vQ_1 ) M\cdots }{}^M
   \label{H.112}
\eea
of $[ \tilde \cV_{ww}^1, \tilde a_{iI}]_{3c} $ in (\ref{H.111}) due to the shuffle identity (\ref{A.sh.2}), 
followed by antisymmetrization in $(\vP \leftrightarrow \vQ)$.

\sm

By combining (\ref{H.98}), (\ref{H.99}) and (\ref{H.101}), we thus find
\bea
 \tilde L^0_{ww}(\tilde a_{iI})_{\cA_1}  +  \tilde L^0_{ww}(\tilde a_{iI})_{\mD_1} + \big[ \tilde \cV_{ww}^1, \tilde a_{iI} \big]
\in \cI
  \label{H.102}
  \eea
which, together with  the results $\hat Z^{a1} _{iI}   = -  \tilde L^0_{ww}(\tilde a_{iI})_{\mD_2}$ and $ \hat Z^{a2} _{iI} = -  \tilde L^0_{ww}(\tilde a_{iI})_{\cA_2} $ of the previous sections, implies the equivalent
\bea
 \tilde L^0_{ww}(\tilde a_{iI})  + [ \tilde \cV_{ww}^1, \tilde a_{iI}] +\hat Z^{a} _{iI}   \in \cI
   \label{H.113}
\eea
of (\ref{H.1aa}) and (\ref{H.2aa}).

\subsection{Proof of (\ref{H.1bb})}

Substituting the explicit expression for $\tilde \cV_{ww} $  into $\hat Z^t_{ij}$, we obtain,
\bea
\label{H.a.1}
\! \! \hat Z^t _{ij}  =
- \int_{\Sigma^n}  \! \!\hat \kappa  \lim _{x _j \to x_i} \! \bigg [ 
 \sum_k \theta \bG(w,x_k; \tilde B_k) \Big (  \bPhi^I (w; \tilde B_k) \tilde a_{kI} 
\! + \! \thalf \sum_{\ell \not = k}    \theta \bG(w,x_\ell; \tilde B_\ell) \tilde t_{k\ell}  \Big ), \tilde t_{ij} \bigg ]
\quad
\eea
The contribution from the first term vanishes in view of (\ref{H.5aa}) while the contribution from the second term is non-vanishing only when either $(k, \ell )=(i,j)$ or $(k, \ell )=(j,i)$ whose contributions are equal to one another. Thus, we have, 
\bea
\label{H.a.2}
\hat Z^t _{ij}  =
-  \int _\Sigma d^2 x \, \kappa (x)  \Big [ 
 \theta \bG(w,x; \tilde B_i) \, \bG(w,x; \tilde B_i) \tilde t_{ij} , \tilde t_{ij} \Big ]
\eea
where we have used the identity $\bG(w,x; \tilde B_i) \tilde t_{ij} =\theta \bG(w,x; \tilde B_j) \tilde t_{ij} $. Expanding the generating functions in a Lie series in $\tilde B$, 
\bea
\label{H.a.3}
\hat Z^t _{ij}  = - \sum _{\vI} \big [ \tilde B_i^\vI  \tilde t_{ij}, \tilde t_{ij} \big ] \sum _{\vI = \vX \vY} \int _\Sigma d^2 x \, \kappa (x) \, \p_w \cG_{\theta(\vX)} (w,x) \, \p_w \cG_\vY (w,x) 
\eea
One recognizes this combination as the second line in (\ref{9.q.12}) for $\delta_{ww} \mD_\vI$ and, using the expression for $\p_x \cG_\vI(x,x)$ given in Lemma \ref{A.lem:1}, one obtains an exact relation in $\hat{\muu}_{h,n}$, 
\bea
\label{H.a.4}
\hat Z^t _{ij}  = \sum _{\vI} \big [ \tilde B_i^\vI  \tilde t_{ij}, \tilde t_{ij} \big ] \Big \{ \delta _{ww} \mD_\vI
- { 1 \over h} \om^J(w) \sum_{\vI = \vP \vQ \vR} \p_w \Phi _{J ( \vP \shuffle \theta(\vR)) M (\vQ + \theta(\vQ))} {}^M(w) \Big \}
\eea
On the other hand, evaluating the commutator $\big [ \tilde \cV_{ww}^1, \tilde t_{ij} \big ] $ with the help of the expression for $\tilde \cV^1_{ww}$ given in (\ref{66.a.8}), all terms for $k \not= i,j$ cancel and the remaining contributions from the second term under the curly brackets belong to $\cI$, so that we are left with, 
\bea
\label{H.a.5}
\big [ \tilde \cV_{ww}^1, \tilde t_{ij} \big ] + { 1 \over h} \om^J(w) \sum_{\vI} \p_w \Phi _{J \vI} {}^M(w) 
\big [ \tilde B_i^\vI \tilde a_{iM} + \tilde B_j^\vI \tilde a_{jM}, \tilde t_{ij} \big ] \in \cI
\eea  
Using a proof by induction parallel to the one used to establish Lemma E.2 in \cite{DHoker:2026ggx}, one shows the following identity valid in $\hat{\muu}_{h,n}$, 
\bea
\label{H.a.6}
\big [ \tilde B_i ^\vI \, \tilde a_{iM} + \tilde B_j ^ \vI \, \tilde a_{jM}, \tilde t_{ij} \big ]
+ \sum _{\vI = \vY M \vQ} 
\sum_{\vY = \vP \shuffle \theta(\vR) } \Big [ \tilde B_i ^\vP \big ( \tilde B_i^\vQ + \tilde B_i^{\theta(\vQ)} \big ) \tilde B_i^\vR  \tilde  t_{ij}, \tilde  t_{ij} \Big ] \in \cI
\eea
consistently with the identity (\ref{A.Lie.1}) in $\hat \mt_{h,n}$.
Substituting this result in (\ref{H.a.5}), we obtain, 
\bea
\label{H.a.7}
\big [ \tilde \cV_{ww}^1, \tilde t_{ij} \big ] - { 1 \over h} \om^J(w) \sum_{\vI} 
\big [ \tilde B_i^\vI   \tilde t_{ij}, \tilde t_{ij} \big ] \sum_{\vI = \vP \vQ \vR}
\p_w \Phi _{J  (\vP \shuffle \theta(\vR) ) M (\vQ + \theta(\vQ))  } {}^M(w) 
 \in \cI
\eea  
Eliminating the double deconcatenation sums between (\ref{H.a.4}) and (\ref{H.a.7}) gives,
\bea
\big [ \tilde \cV_{ww}^1, \tilde t_{ij} \big ]  + \hat Z^t_{ij} - \sum _{\vI} \big [ \tilde B_i^\vI  \tilde t_{ij}, \tilde t_{ij} \big ]
 \delta _{ww} \mD_\vI \, \in \, \cI
\eea
Combining this result with (\ref{H.2bb}) and using the last relation in (\ref{6.z.4})
to identify the term involving $\delta _{ww} \mD_\vI$ as $+\tilde L^0_{ww}(\tilde t_{ij})$ then proves (\ref{H.1bb}).


\end{document}